%% file: main.tex
\documentclass{article}

\usepackage{graphicx}
\usepackage[dvipsnames]{xcolor}
\usepackage{amsthm}
\usepackage{amsmath}
\usepackage{amssymb}
\usepackage{amsfonts}
\usepackage{mathrsfs}
\usepackage{textcomp}
\usepackage{bm}
\usepackage{booktabs}
\usepackage{xspace}
\usepackage{listings}

\usepackage{xr}
\usepackage{hyperref}
\usepackage[noabbrev,capitalise,nosort,nameinlink]{cleveref}
\usepackage{natbib}
\usepackage{mathtools}
\usepackage{lineno}
\usepackage{multirow}
\usepackage{caption}
\usepackage{subcaption}
\usepackage[noend]{algpseudocode} 
\usepackage{algorithm}

\newtheorem{theorem}{Theorem}
\newtheorem{lemma}{Lemma}
\newtheorem{corollary}{Corollary}

\newtheorem{assumption}{Assumption}
\theoremstyle{remark}

\newtheorem{remark}{Remark}
\theoremstyle{definition}

\theoremstyle{plain}

\crefname{assumption}{Assumption}{Assumptions}
\Crefname{assumption}{Assumption}{Assumptions}
\crefname{proposition}{Proposition}{Propositions}
\Crefname{proposition}{Proposition}{Propositions}
\crefname{remark}{Remark}{Remarks}
\Crefname{remark}{Remark}{Remarks}
\crefname{theorem}{Theorem}{Theorems}
\Crefname{theorem}{Theorem}{Theorems}
\crefname{corollary}{Corollary}{Corollaries}
\Crefname{corollary}{Corollary}{Corollaries}
\crefname{lemma}{Lemma}{Lemmas}
\Crefname{lemma}{Lemma}{Lemmas}

\DeclareMathOperator{\Tr}{Tr}

\newcommand{\todo}[1]{\bgroup\color{red}#1\egroup}

\DeclareMathOperator*{\argmin}{arg\,min}

\DeclarePairedDelimiterX{\infdivx}[2]{(}{)}{%
  #1\;\delimsize\|\;#2%
}
\newcommand{\infdiv}{KL\infdivx}

\usepackage[a4paper, top=88pt, bottom=88pt, left=72pt, right=72pt, headsep=16pt, footskip=28pt]{geometry}
\hypersetup{
	colorlinks,
	linkcolor=blue,
	citecolor=blue,
	urlcolor=blue,
}
\newcommand*{\email}[1]{\bgroup\color{blue}\href{mailto:#1}{#1}\egroup}
\renewcommand*{\url}[1]{\bgroup\color{blue}\href{#1}{#1}\egroup}
\usepackage{enumitem}
\setlist[enumerate]{nosep}
\setlist[itemize]{nosep}
\usepackage{mleftright} \mleftright
\renewcommand{\qedsymbol}{$\blacksquare$}
\renewenvironment{proof}[1][\proofname]{\noindent{\sffamily\bfseries #1.} }{\hfill \qedsymbol \medskip}
\usepackage{scrlayer-scrpage, xhfill}
\usepackage[labelfont={sf,bf}]{caption}
\automark[section]{section}
\setkomafont{pageheadfoot}{\normalcolor\sffamily}
\setkomafont{pagenumber}{\normalfont\normalsize\sffamily}
\clearpairofpagestyles
\let\oldtitle\title
\renewcommand{\title}[1]{\oldtitle{#1}\newcommand{\theshorttitle}{#1}}

\let\oldauthor\author
\renewcommand{\author}[1]{\oldauthor{#1}\newcommand{\theshortauthor}{#1}}
\newcommand{\shortauthor}[1]{\renewcommand{\theshortauthor}{#1}}
\cohead{\xrfill[0.525ex]{0.6pt}~\theshorttitle~\xrfill[0.525ex]{0.6pt}}
\cehead{\xrfill[0.525ex]{0.6pt}~\theshortauthor~\xrfill[0.525ex]{0.6pt}}
\newcommand{\theabstract}[1]{\par\bgroup\noindent\textbf{\textsf{Abstract.}} #1\egroup}
\newcommand{\thekeywords}[1]{\par\smallskip\bgroup\noindent\textbf{\textsf{Keywords.}}\newcommand{\and}{ $\bullet$ } #1\egroup}
\newcommand{\themsc}[1]{\par\smallskip\bgroup\noindent\textbf{\textsf{2010 Mathematics Subject Classification.}}\newcommand{\and}{ $\bullet$ } #1\egroup}
\newcommand*{\affilref}[1]{\ref{affiliation#1}}
\newcommand*{\affiliation}[3]{
	\footnotetext[#1]{\label{affiliation#2} #3}
}

\newcommand{\thmmaintheorem}{Theorem~1\xspace}
\newcommand{\thmposteriormeancorollary}{Corollary~1\xspace}
\newcommand{\thmMLEcorollary}{Corollary~3\xspace}
\newcommand{\thmposteriorpredcorollary}{Corollary~2\xspace}
\newcommand{\thmposteriorconsistencycorollary}{Theorem~2\xspace}

\newcommand{\add}[1]{{#1}}
\newcommand{\del}[1]{}
\newcommand{\delp}[1]{}

\title{Bayesian Variational Inference for Mixed Data Mixture Models}

\author{%
	Junyang Wang\textsuperscript{\affilref{Imperial}}
	\and
    James Bennett\textsuperscript{\affilref{Imperial2}}
    \and
    Victor Lhoste\textsuperscript{\affilref{Imperial2}}
    \and
	Sarah Filippi\textsuperscript{\affilref{Imperial}}
	}
\shortauthor{J.~Wang, Sarah Filippi}

\begin{document}

\maketitle
\affiliation{1}{Imperial}{Department of Mathematics, Imperial College London,
London, United Kingdom.}
\affiliation{2}{Imperial2}{School of Public Health, Imperial College London,
London, United Kingdom.}

	\theabstract{\input{abstract}}

	\thekeywords{{Mixture model}\and{Mixed type data}\and{Variational inference}}


\input{introduction}

\input{methods}

\input{theory}

\input{results}

\input{conclusion}



\paragraph*{Acknowledgements}

This work was supported by the UK Medical Research Council (MRC), under the grant `Worldwide phenotypes and transitions in obesity-related multimorbidity', reference number MR/V034057/1. The authors are grateful to Majid Ezzati and Bin Zhou for useful discussions. For the purpose of Open Access, the author has applied a CC BY public copyright license to any Author Accepted Manuscript version arising from this submission.

\bibliographystyle{abbrvnat}
\bibliography{bibliography}

\clearpage
\setcounter{section}{0}
\renewcommand{\thesection}{S\arabic{section}}
\renewcommand{\thesubsection}{S\arabic{section}.\arabic{subsection}}
\renewcommand{\theHsection}{S\arabic{section}}
\counterwithin{lemma}{section}
\counterwithin{theorem}{section}
\renewcommand{\thetheorem}{\thesection.\arabic{theorem}}
\renewcommand{\thelemma}{\thesection.\arabic{lemma}}
\setcounter{equation}{0}
\renewcommand{\theequation}{S\arabic{equation}}
\renewcommand{\theHequation}{S\arabic{equation}}
\setcounter{figure}{0}
\renewcommand{\thefigure}{S\arabic{figure}}
\renewcommand{\theHfigure}{S\arabic{figure}}
\setcounter{table}{0}
\renewcommand{\thetable}{S\arabic{table}}
\renewcommand{\theHtable}{S\arabic{table}}
\input{appendix}

\end{document}

%% file: abstract.tex




Heterogeneous, mixed type datasets including both continuous and categorical variables are ubiquitous, and enrich data analysis by allowing for more complex relationships and interactions to be modelled. Mixture models offer a flexible framework for capturing the underlying heterogeneity and relationships in mixed type datasets. Most current approaches for modelling mixed data either forgo uncertainty quantification and only provide point estimates, or use MCMC, which incurs a very high computational cost that is not scalable to large datasets. This paper develops a coordinate ascent variational inference algorithm (CAVI) for mixture models on mixed (continuous and categorical) data, which circumvents the high computational cost of MCMC while retaining uncertainty quantification. We demonstrate our approach through simulation studies as well as an applied case study of the NHANES risk factor dataset. We provide theoretical justification for our method by establishing that the CAVI variational posterior mean converges locally to the true parameter value and lies within $O(1/n)$ of the maximum likelihood estimator. Building on this result, we show that the CAVI variational posterior contracts around the true parameter at rate $n^{-1/2}$.

%% file: introduction.tex
\section{Introduction} \label{sec: introduction}

Mixture models provide a probabilistic framework for modelling latent subpopulations within a population, without needing to know which subpopulation each individual data point belongs to. Mixture models are ubiquitous in statistical modelling and are commonly used for clustering and density estimation. Popular applications include image analysis 
\citep{Niknejad2015,Beecks2015}, finance \citep{Alexander2006, Durham2007}  and genomics \citep{Pritchard2000,Rau2015} . 

Commonly in data analysis, datasets include both continuous and categorical variables, and such datasets are known as mixed type data. For instance, healthcare datasets might include continuous variables such as a patient's blood pressure, weight and height, but also categorical variables such as smoking status. In such mixed datasets, the ability to model multimodal continuous and categorical data simultaneously is critical in understanding the relationship between the variables of interest, and the underlying subpopulations in the data.

Existing mixture models for mixed type data have been developed in various forms. \cite{Hunt2002} proposed a model where each mixture component is a product of multivariate Gaussian distributions for the continuous data, and a categorical distribution for the categorical data. The authors use an Expectation-Maximisation (EM) algorithm to calculate the maximum likelihood estimate of the parameters of interest, which was later extended in \cite{Hunt2003} to include missing data. \cite{Lawrence1996} uses a similar parametrisation, allowing the mean of the Gaussian components to depend on the value of the categorical data, but the covariance of the components are assumed to be independent of the categorical data as well as the mixture components in order to reduce the number of parameters estimated with EM.

There are also a number of works that assume categorical variables are generated from latent continuous variables, and the values of the categorical variables correspond to certain cutoffs in the latent continuous variable. Examples of these include \cite{Ranalli2017} and \cite{McParland2016}, both of which use EM for parameter estimation.

Approaches using Bayesian inference to obtain a posterior distribution over the model parameters rather than a point estimate have also been used widely in mixture models, especially for clustering (\citep{Celeux2020, Celeux2000, Schnatter2001, Miller2018,Richardson1997, Zhang2004, Teh2006}). However, mixed type data mixture models using Bayesian inference appear to be much more limited (\citep{Shu2019, Zhang2014}). MCMC for mixture models are however not without significant drawbacks. MCMC is computationally intensive particularly in higher dimensions, and diagnosing convergence can be challenging. Due to the non-identifiability of mixture components under symmetric priors, MCMC samples can exhibit label switching, where component labels permute during sampling \cite{Jasra2005}. This makes direct interpretation of component-specific parameters difficult without additional post-processing or constraints. 

Variational inference (VI) is a popular method used in Bayesian statistics and machine learning to approximate the posterior distribution when it is difficult or intractable to compute (\cite{Blei2017}). VI approximates the posterior distribution using optimisation, and is typically faster and more scalable than MCMC.  

Variational inference has been utilised for mixture models previously. For finite mixtures and continuous data, it is well known that a Normal-inverse Wishart prior on the mixture component mean and covariance and a Dirichlet prior on the mixture component probabilities, coupled with Gaussian data, leads to closed form updates using the Coordinate Ascent Variational Inference (CAVI) algorithm \citep{Bishop2006}. For finite mixtures and categorical data, \cite{Ahlmann2018} proposed a CAVI algorithm for multinomially distributed discrete data, and more recently \cite{Rao2025} applied a VI algorithm with variable selection for discrete biomedical data. Infinite mixture models using Dirichlet Process and variational inference include the seminal work of \cite{Blei2006} as well as further developments such as \citep{Chong2011, Hughes2013, Huynh2016}. 


In this paper, we propose a VI algorithm for finite mixture models on mixed type data, which allows both the continuous and categorical variables to influence the make up of each mixture component. We employ a mean-field variational approximation of the posterior with conjugate priors in order to obtain analytical updates of the variational posterior using the CAVI algorithm, with the derivation equations provided. 
To provide a rigorous foundation for our method, we study the asymptotic properties of the CAVI-derived variational posterior as the sample size tends to infinity. We prove that the CAVI posterior mean is locally consistent to the true parameter and that it remains asymptotically close to the maximum likelihood estimator (MLE) with a gap that vanishes at a rate of $O(1/n)$. Using these results, we show that the variational predictive density recovers the true data-generating likelihood, and the CAVI variational posterior contracts around the true parameter at rate $O(n^{-1/2})$.
We investigate the empirical performance of our algorithm extensively on a number of simulated datasets, where both point estimation and uncertainty quantification of the variational posterior are assessed and compared with a Gibbs sampler and existing EM methods. We also compare relative runtimes between CAVI, Gibbs and EM. Finally, we analyse data from the National Health and Nutrition Examination Survey (NHANES), a program conducted by the National Center for Health Statistics (NCHS) that assesses the health and nutritional status of the U.S. population through interviews and physical examinations, from which we use continuous and categorical risk factor variables such as body mass index (BMI) and smoking status. These risk factors frequently co-occur and are often etiologically linked; for example, obesity is a known risk factor for dyslipidemia, elevated blood pressure, and hyperglycemia, making it difficult to accurately assess health risks when considering each factor in isolation. We show our method can produce medically interpretable clusters of participants, and we visualise said clusters in a way that takes into account uncertainty using the posterior predictive distribution, something that is missed in most frequentist and non model based clustering methods.  

%% file: methods.tex

\section{Mixed data mixture model} \label{subsec: mixeddatamixturemodel}

Consider a set of $n$ independent observations $\{(x_i,c_i)\}_{1\leq i\leq n}$ where $x_i\in\mathbb{R}^q$  contains continuous observations while $c_i\in\mathbb{Z}^p$ is a $p$-dimensional vector of categorical data such that $c_{ij} \in \{1,2,\dots,d_j\}$. $d_j$ denotes the number of categorical values for the $j$-th categorical variable. We assume that the number of categorical parameters to be estimated remains small relative to the sample size $n$. 
In this paper we focus on a finite mixture model of $K$ components, where $K$ is a constant specified by the user.
A popular alternative would be the Dirichlet Process mixtures, which is more suitable if the number of data clusters is expected to increase with respect to sample size, whereas finite mixture models are more suitable if the data arise from a moderate number of clusters even as the sample size increases (\cite{Schnatter2019}). This paper focuses on the latter scenario. We consider a model where the likelihood for a single data point $(x_i, c_i)$ can be written as follows:
\begin{align} \label{eq: likelihoodVIzintegratedout} 
&p(x_i, c_i | \mu, \Lambda, \psi, \pi) \nonumber \\ 
&= \sum_{k=1}^{K} \pi_{k} \mathcal{N}(x_i|\mu_{k}, \Lambda_{k}^{-1} ) \prod_{j=1}^{p} \psi_{k,j,c_{ij}}\;.
\end{align}
Here, $\pi_{k}$ are the mixture weights, denoting the probability of each component $1 \leq k \leq K$; and  $\psi_{k,j,c_{ij}}$ denotes the probability of observing $c_{ij}$ for the $j$th categorical variable under component $k$. In addition, $\mu_{k}$ and $\Lambda_{k}$ denotes the component centre and component precision matrix respectively for component $k$, and we will also use $\Sigma_{k}=\Lambda_{k}^{-1}$ to denote the covariance matrix for component $k$. 

In this model formulation, conditional independence is assumed between the $p$ categorical variables in $c_i$ as well as between $c_i$ and $x_i$ within a given mixture component $k$. Note that conditional independence does not imply independence between the covariates in the marginal likelihood in \Cref{eq: likelihoodVIzintegratedout}. If in addition the model user wishes to model the dependence between categorical variables within each mixture, they can opt to merge categorical variables and relabel the categorical values accordingly under our current framework. For example, merging two binary categorical variables into one variable with four outcomes. 

For computational reasons, it is useful to introduce a latent allocation variable $z$ of size $n$ denoting the hidden mixture component of the observations. 
Let $x\in \mathbb{R}^{n \times q}$  and $c \in \mathbb{Z}^{n \times p}$ represent the matrix containing  respectively the continuous and categorical data. 
Incorporating $z$, the model can be written as: 
\begin{align} \label{eq: likelihoodVI} 
p(z|\pi) &= \prod_{i=1}^{n}\pi_{z_i}, \quad p(c|\psi,z) = \prod_{i=1}^{n}\prod_{j=1}^{p} \psi_{z_i,j,c_{ij}}, \nonumber \\
& p(x|c,z,\mu,\Lambda) = \prod_{i=1}^{n} \mathcal{N}(x_i|\mu_{z_i}, \Lambda_{z_i}^{-1} ). 
\end{align}
As we are presenting a Bayesian approach to infer the parameters of this mixture model, we now define the prior distributions over the model parameters, for which we consider the following:
\begin{align} \label{eq: priordist}
&\mu_{k} \sim \mathcal{N}(m_{k}, \beta^{-1} \Lambda_{k}^{-1} ), \quad \Lambda_{k} \sim \mathcal{W}(\nu, \Phi^{-1}) \nonumber \\
&\pi = (\pi_{1},\dots,\pi_{K}) \sim Dir(\alpha,\dots,\alpha), \nonumber \\
&\psi_{k,j} = (\psi_{k,j,1},\dots,\psi_{k,j,d_j}) \sim Dir(\eta_j,\dots,\eta_j).
\end{align}
The prior of $\mu_{k}$ and $\Sigma_{k}=\Lambda_{k}^{-1}$ is a joint Normal-Inverse-Wishart distribution while $\pi$ and $\psi_{k,j}$ take Dirichlet priors. The prior hyperparameters $m_{k}$, $\beta$, $\nu$, $\Phi$, $\alpha$, $\eta_j$ are specified by the user. The posterior distribution 
\begin{align}
 & p(\mu, \Lambda, \psi, \pi, z |x,c) & \nonumber \\ =&\frac{p(x|c,z,\mu,\Lambda)p(c|\psi,z)p(z|\pi)p(\mu, \Lambda, \psi, \pi)}{p(x,c)}  \label{eq:posterior}
\end{align}
induced by this prior is not available in closed form. In this paper, we derive a variational inference procedure to approximate the posterior. 


\section{Variational Inference} \label{subsec: VIbackground}
\subsection{Background on Variational Inference and Coordinate Ascent}
\label{subsec: CAVI}

In this section, we introduce the key concepts of Variational Inference and the Coordinate Ascent Variational Inference (CAVI) procedure that are essential for understanding the remainder of the paper. For more detailed explanations, we refer the reader to \cite{Bishop2006}. Suppose that we have some data $\bm{x}$ (not necessarily the data described in the previous section) generated by a likelihood function $p(\bm{x}|\bm{\theta})$ where $\bm{\theta}$ denotes some unknown parameters of interest, which have a prior distribution $p(\bm{\theta})$. In variational inference, the posterior distribution $p(\bm{\theta}|\bm{x})$ is approximated using a variational distribution $q$, which is chosen from a family of distributions $\mathcal{Q}$ that are computationally more tractable than the true posterior. The variational distribution $q$ is chosen by selecting the distribution from the variational family which minimises the Kullback-Leibler divergence (KL-divergence) between the variational distribution and the true posterior. Specifically:
\begin{equation} \label{eq: VIobjective}
q=\argmin_{q \in \mathcal{Q}} \infdiv{q(\bm{\theta})}{p(\bm{\theta}|\bm{x})}
\end{equation}
In practice, the variational family $\mathcal{Q}$ is often assumed to be parametrised by variational parameters $\lambda$, and the optimisation problem becomes equivalent to finding the parameter $\lambda$ which minimises $\infdiv{q_{\lambda}(\bm{\theta})}{p(\bm{\theta}|\bm{x})}$. One key advantage in this optimisation approach to Bayesian inference is that one no longer needs to evaluate the marginal distribution $p(\bm{x})$, which is typically the main computational challenge in conducting Bayesian inference through sampling. To see this, note that:
\begin{align} \label{eq: VUobjectiveparametrised}
\lambda^* &=\argmin_{\lambda} \infdiv{q_{\lambda}(\bm{\theta})}{p(\bm{\theta}|\bm{x})} \nonumber \\
&=\argmin_{\lambda} \mathbb{E}_{q_{\lambda}} [ \mathrm{ln}(q_{\lambda}(\bm{\theta}) )- \mathrm{ln} (p(\bm{x}|\bm{\theta}))- \mathrm{ln} (p(\bm{\theta}))] 
\end{align}
where $\mathrm{ln} (p(\bm{x}))$ can be dropped from the optimisation objective as it is independent of $q_{\lambda}$. The negative of the remaining optimisation objective $\mathbb{E}_{q_{\lambda}} [-\mathrm{ln}(q_{\lambda}(\bm{\theta}) )+ \mathrm{ln} (p(\bm{x}|\bm{\theta}))+ \mathrm{ln} (p(\bm{\theta}))] $ is commonly referred to as the Evidence Lower Bound (ELBO). 


A common assumption for the variational family is that the variational posterior factorises into a product of independent distributions:\begin{align} \label{eq: meanfieldapprox}
q(\bm{\theta})=\prod_{i=1}^{I} q_i(\bm{\theta}_i)
\end{align}
where $q_i$ is the posterior for the $i$th block of variables $\bm{\theta}_i$. This is known as the mean field assumption. One popular and convenient algorithm for solving the optimisation problem in  \Cref{eq: VUobjectiveparametrised} under the mean field assumption is Coordinate Ascent Variational Inference (CAVI). CAVI maximises the ELBO by iteratively updating each factor in the variational distribution while treating all other factors as fixed, until the ELBO has converged to a local maximum. 
It can be shown CAVI updates the $q_j$ in turn via \begin{equation}
q_j(\bm{\theta}_j) \propto \mathrm{exp}(\mathbb{E}_{\bm{\theta}_{-j}}[\mathrm{ln}(p(\bm{x},\bm{\theta}) ) ])\label{eq:cavi-gen}
\end{equation}while keeping all $q_i, i \neq j$ fixed.

\subsection{CAVI for mixed data mixture model} \label{subsec: ourCAVI}

In this section, we derive a variational inference procedure for the mixture model  defined in~\Cref{eq: likelihoodVIzintegratedout}, building on the general concepts and techniques introduced in the previous section. To approximate the posterior distribution $p(\mu, \Lambda, \psi, \pi, z |x,c)$ from~\Cref{eq:posterior}, we consider a variational distribution of the following form:
\begin{align} \label{eq: meanfieldapproxourmodel}
q(\mu, \Lambda, \psi, \pi, z)=q(\mu, \Lambda, \psi, \pi) q(z)\;.
\end{align}
This variational posterior assumes a general joint distribution for the `global' variables $\mu, \Lambda, \psi, \pi$, and a separate distribution for the `local' variables $z$ that is independent from the `global' variables. The exact distributional families of the variational posterior $q(\mu, \Lambda, \psi, \pi)$ and $q(z)$ need not to be specified in advance and will fall out during the derivation. Independence between $q(\mu, \Lambda, \psi, \pi)$ and $q(z)$ is a standard mean field assumption that allows us to split the variables into two blocks for the Coordinate Ascent Variational Inference (CAVI) algorithm. In other words, the algorithm  alternates between updating $q(\mu, \Lambda, \psi, \pi)$ and $q(z)$. By applying the CAVI update from equation~\eqref{eq:cavi-gen} to our variational posterior in~\eqref{eq: meanfieldapproxourmodel}, we obtain  (details of the derivations can be found in the supplement):
\begin{align} 
& q(\mu, \Lambda, \psi, \pi)  \propto \underbrace{\mathrm{exp}(\mathbb{E}_{z} [\mathrm{ln}(p(x|c,z,\mu,\Lambda)p(\mu|\Lambda)p(\Lambda))]) }_{\propto q(\mu,\Lambda)} \nonumber \\
&\times \underbrace{\mathrm{exp}(\mathbb{E}_{z}[\mathrm{ln}(p(z|\pi)p(\pi))] ) }_{\propto q(\pi)}\underbrace{\mathrm{exp}(\mathbb{E}_{z} [\mathrm{ln}(p(c|\psi,z)p(\psi))]  ) }_{\propto q(\psi)}  \nonumber 
\end{align}  
The variational distribution can therefore be written as the product of three terms which can be derived separately: $q(\mu, \Lambda, \psi, \pi)=q(\mu, \Lambda)q(\psi)q(\pi)$. The variational distribution over $\pi$,  $q(\pi)$, can be shown to take the form:
\begin{align} \label{eq: piupdatemaintext}
q(\pi) =  Dir(\pi|\hat{\alpha}_1,\dots,\hat{\alpha}_K ) \quad \text{with} \quad 
\hat{\alpha}_k &=\alpha+\sum_{i=1}^{n}r_{ik}
\end{align}
where $r_{ik}=\mathbb{E}_{z_i}[z_{ik}]=q(z_i=k)$ is the variational posterior probability of the $i$-th data being in component $k$ under the variational posterior probability.  The term 
$q(\mu, \Lambda)$ can be shown to take the following form:
\begin{align} \label{eq: mulambdaupdatemaintext}
q(\mu, \Lambda) &= \prod_{k}\mathcal{N}(\mu_{k}|\hat{m}_{k}, \hat{\beta}_{k}^{-1} \Lambda_{k}^{-1} ) \mathcal{W}(\Lambda_{k}|\hat{\nu}_{k}, \hat{\Phi}_{k}^{-1})) \nonumber \end{align}
where \begin{align}
&\hat{m}_{k} =\frac{\beta m_{k} + \sum_{j=1}^{n} r_{jk} x_j}{\beta + \sum_{j=1}^{n} r_{jk}}, \quad \hat{\beta}_{k}^{-1} =(\beta + \sum_{j=1}^{n} r_{jk} )^{-1} \nonumber \\
&\hat{\nu}_{k} =\nu+\sum_{i=1}^{n} r_{ik}, \quad \hat{\Phi}_{k} =\Phi- (\beta + \sum_{i=1}^{n} r_{ik} ) \hat{m}_{k}\hat{m}_{k}^T \nonumber \\
& +\beta m_{k} m_{k}^T + \sum_{i=1}^{n} r_{ik} x_i x_i^T
\end{align}
This gives for each $k$, independent Normal-Inverse-Wishart distributions \\
$\mu_{k}|\Lambda_{k} \sim \mathcal{N}(\mu_{k}|\hat{m}_{k}, \hat{\beta}_{k}^{-1} \Lambda_{k}^{-1} )$ and $\Lambda_{k} \sim \mathcal{W}(\Lambda_{k}|\hat{\nu}_{k}, \hat{\Phi}_{k}^{-1}))$.
Lastly, $q(\psi)$ can be shown to take the form:
\begin{align} \label{eq: psiupdatemaintext}
& q(\psi) = \prod_{k}\prod_{j} Dir(\psi_{k,j}|\hat{\eta}_{k,j,1},\dots,\hat{\eta}_{k,j,d_j} ) \nonumber \\
& \quad\text{with} \quad
\hat{\eta}_{k,j,g}=\eta_j+\sum_{i=1}^{n} I_{\{c_{ij}=g\}}r_{ik}
\end{align} 
After updating $q(\mu, \Lambda, \psi, \pi)$, we use its newly updated distribution to in turn update $q(z)$, which can be shown to take the form:
\begin{align} \label{eq: zupdatemaintext}
& q(z) \propto \mathrm{exp} \Bigl\{ \mathbb{E}_{-z} \mathrm{ln} p(x, c, z, \mu, \Lambda, \psi, \pi) \Bigr\} \nonumber \\
& \propto \mathrm{exp} \Bigl\{ \sum_{i=1}^{n}\sum_{k=1}^{K} ( z_{ik}\mathbb{E}_{\mu, \Lambda} [\mathrm{ln}(\mathcal{N}(x_i|\mu_{k}, \Lambda_{k}^{-1} ) )] \nonumber \\
& + \sum_{j=1}^{p} z_{ik}\mathbb{E}_{\psi} [\mathrm{ln}(\psi_{k,j,c_{ij}}) ] + z_{ik}\mathbb{E}_{\pi} [\mathrm{ln}(\pi_{k} ) ] ) \Bigr\} \nonumber \\
& \propto \prod_{i=1}^{n} \prod_{k=1}^{K} \rho_{ik}^{z_{ik}}
\end{align} 
where
\begin{align} \label{eq: rhoik main text} 
\rho_{ik} &=\mathrm{exp}\Bigl\{ \mathbb{E}_{\mu, \Lambda} [\mathrm{ln}(\mathcal{N}(x_i|\mu_{k}, \Lambda_{k}^{-1} ) )] \nonumber \\ &+\sum_{j=1}^{p}\mathbb{E}_{\psi} [\mathrm{ln}(\psi_{k,j,c_{ij}}) ]+\mathbb{E}_{\pi} [\mathrm{ln}(\pi_{k} ) ] \Bigr\}
\end{align}
The exact forms of each of the expectations inside $\rho_{ik}$ can be found in the supplement, and we also define the responsibilites $r_{ik}=\frac{\rho_{ik}}{\sum_{j=1}^{K} \rho_{ij} }$. \Cref{alg:ourCAVI} summarises the  CAVI updates for all the parameters. Details including the full derivation for each update and the exact forms of the ELBO is reserved for the supplement.



\begin{algorithm}
\caption{CAVI algorithm for mixed data mixture model} \label{alg:ourCAVI}
\begin{algorithmic}
\State Initialise $\hat{\alpha}_k$, $\hat{m}_{k}$, $\hat{\beta}_{k}$, $\hat{\Phi}_{k}$, $\hat{\nu}_{k}$, $\hat{\eta}_{k,j,g}$, $r_{ik}$
\While{$ELBO(\hat{\alpha}_k, \hat{m}_{k}, \hat{\beta}_{k}, \hat{\Phi}_{k}, \hat{\nu}_{k}, \hat{\eta}_{k,j,g}, r_{ik})$ not converged} 
    \State Update $\hat{\alpha}_k$ for $k=1, \dots K$ using \Cref{eq: piupdatemaintext}
    \State Update $\hat{m}_{k}$, $\hat{\beta}_{k}$, $\hat{\Phi}_{k}$, $\hat{\nu}_{k}$ for $k=1, \dots K$ using \Cref{eq: mulambdaupdatemaintext} 
    \State Update $\hat{\eta}_{k,j,g}$ for $k=1, \dots K$ using \Cref{eq: psiupdatemaintext} 
    \State Update $r_{ik}$ for $k=1, \dots K$ using  \Cref{eq: rhoik main text} 
    \State compute ELBO
\EndWhile \\
\Return{$\hat{\alpha}_k$, $\hat{m}_{k}$, $\hat{\beta}_{k}$, $\hat{\Phi}_{k}$, $\hat{\nu}_{k}$, $\hat{\eta}_{k,j,g}$, $r_{ik}$}
\end{algorithmic}
\end{algorithm}

\subsection{Posterior Predictive Distribution}

Now that we have the variational posterior, we can use the variational posterior predictive distribution to predict possible future (unobserved) data points given the observed data, while accounting for uncertainty in the model parameters. This can be used for purposes such as density estimation. We derive the variational posterior predictive distribution of a new observation $\tilde{x}, \tilde{c}$, which conveniently can be obtained in closed form in this case:
\begin{align} \label{eq: posteriorpred}
& q(\tilde{x},\tilde{c}|x,c) = \nonumber \\
&\int p(\tilde{x},\tilde{c}|\mu, \Lambda,\psi,\pi )q(\mu, \Lambda)q(\psi)q(\pi) d\mu d\Lambda d\psi d\pi \nonumber \\
=&\sum_{k=1}^{K} \mathbb{E}_{\pi} [\pi_{k}] \;\mathbb{E}_{\mu, \Lambda} [\mathcal{N}(\tilde{x}|\mu_{k}, \Lambda_{k}^{-1} )] \;\mathbb{E}_{\psi}[\prod_{j=1}^{p} \psi_{k,j,\tilde{c}_j}] \nonumber \\
=& \sum_{k=1}^{K} \bigl[ \frac{\hat{\alpha}_k}{\sum_{k'=1}^{K} \hat{\alpha}_{k'}} \;t_{\hat{\nu}_{k}-q+1}\Big(\tilde{x} \Big\vert \hat{m}_{k}, \frac{\hat{\Phi}_{k}(\hat{\beta}_{k}+1)}{\hat{\beta}_{k}(\hat{\nu}_{k}-q+1)} \Big) \nonumber \\
\times &\prod_{j=1}^{p} \frac{\hat{\eta}_{k,j,\tilde{c}_j}}{\sum_{g_j=1}^{d_j} \hat{\eta}_{k,j,g_j}} \bigl].
\end{align} 
where expectations such as $\mathbb{E}_{\pi}$ are with respect to the variational posterior distributions, $t_{\nu}(\tilde{x}|m,\Phi)$ denotes a multivariate student t density function with $\nu$ degrees of freedom, location $m$ and scale matrix $\Phi$.  

From the form of \eqref{eq: posteriorpred}, we can also obtain the marginals of the variational posterior predictive densities for each $j$th component of newly observed continuous data $\tilde{x}_j$ or categorical data $\tilde{c}_j$ given a particular mixture component $k$, i.e.  $q(\tilde{x}_j|\tilde{z}=k,\tilde{c},x,c)$ or $q(\tilde{c}_j|\tilde{z}=k,x,c)$. In particular, a component of a multivariate t distribution is an univariate t distribution with mean and variance equal to the corresponding components of the mean and covariance of the multivariate t distribution, so $q(\tilde{x}_j|\tilde{z}=k,\tilde{c},x,c)$ and $q(\tilde{c}_j|\tilde{z}=k,x,c)$ take the following forms:
\begin{align} \label{eq: posteriorpredmarginalx}
&q(\tilde{x}_j|\tilde{z}=k,\tilde{c},x,c) = \nonumber \\
& t_{\hat{\nu}_{k}-q+1}(\tilde{x}_j | \hat{m}_{k,j}, \frac{\hat{\Phi}_{k,jj}(\hat{\beta}_{k}+1)}{\hat{\beta}_{k}(\hat{\nu}_{k}-q+1)} ), \nonumber \\
&q(\tilde{c}_j|\tilde{z}=k,x,c) = \frac{\hat{\eta}_{k,j,\tilde{c}_j}}{\sum_{g_j=1}^{d_j} \hat{\eta}_{k,j,g_j}}.
\end{align}
These marginal distributions of the variational posterior predictive densities enable useful visualisations of the variational posterior predictive density in high dimensions that will be presented in \Cref{sec: results}. 

%% file: theory.tex
\section{Theory} \label{sec: theory}

This section is dedicated to examining the convergence of the proposed CAVI algorithm to the true, data generating parameters $\{\pi_k^{*}, \mu_k^{*}, \Lambda_k^{*}, \{\psi_{k,g}^{*}\}_g\}_k$ as the sample size $n$ tends to infinity. For ease of presentation we assume there is $p=1$ categorical variable with $d$ categorical values. \add{Our analysis relies on the following four assumptions.}

\add{%
\begin{assumption}[i.i.d.\ samples]\label{ass: wellspecified}
The observations $\{(x_i,c_i)\}_{i=1}^n$ are i.i.d.\ draws from the model described in 
\Cref{eq: likelihoodVIzintegratedout} at a fixed parameter
$\Theta^*=(\bm\pi^*,\bm\mu^*,\bm\Lambda^*,\bm\psi^*)$.
\end{assumption}
\begin{assumption}[Interior truth]\label{ass: interior} The data generating parameters satisfy: 
$\pi_k^{*}>0$, $\sum_{k=1}^{K} \pi_k^{*}=1$; $\psi_{k,g}^{*} > 0$, $\sum_{g=1}^{d} \psi_{k,g}^{*}=1$. Each $\Lambda_k^*$ is symmetric positive definite, for all $1\le k\le K$, $1\le g\le d$.
\end{assumption}
\begin{assumption}[Distinct components]\label{ass: distinct}
The pairs $(\mu_k^{*},\Lambda_k^{*})$, $1\le k\le K$, are pairwise distinct; that is, for $k\ne k'$ either $\mu_k^*\ne\mu_{k'}^*$ or $\Lambda_k^*\ne\Lambda_{k'}^*$.
\end{assumption}
\begin{assumption}[Correct support size]\label{ass: correctK}
The number of components $K$ and the number of categories $d$ are correctly specified.
\end{assumption}}


In recent years, a number of works have developed general theory to study the convergence of the optimal variational posterior from a frequentist perspective. These include studies on consistency of point estimators including \citep{Pierre2020, Pati2018, Yang2020, Zhang2020}, as well as Bernstein-von-Mises type theorems including \cite{Zhang2024} and \cite{Blei2019}. While these provide powerful general frameworks, they often require conditions that are difficult to verify in our setting. 

Instead, we build upon the approach of \cite{BoWang2006}, which allows for a direct analysis of the actual CAVI output. We extend Theorem 1 of \cite{BoWang2006} on consistency of the posterior mean to the case of mixed type data with likelihood \Cref{eq: likelihoodVIzintegratedout}, and further establish a contraction rate for the limiting CAVI variational posterior. Our approach relaxes a requirement in \cite{BoWang2006} that the responsibilities $r_{ik}$ and their derivatives converge uniformly as functions of $x_i$ \add{to compute the limit of the gradient of the CAVI map at $\Theta^*$.} We instead use a \add{strong law of large numbers (SLLN) for triangular arrays under a} dominated convergence condition (\add{Lemma S4.6} in the supplement), which is much easier to verify. \add{Furthermore, rather than working only with the gradient at $\Theta^*$, we show that the gradient converges uniformly over a fixed neighbourhood of $\Theta^*$, using expansions of the responsibilities that are uniform in $\Theta$ together with a uniform SLLN (Lemmas S4.2 and S4.3 in the supplement), so that the CAVI map is a contraction on this neighbourhood for all sufficiently large $n$.}

The CAVI algorithm derived in the previous section iteratively updates the variational hyperparameters $\hat{\alpha}_k$, $\hat{m}_{k}$, $\hat{\beta}_{k}$, $\hat{\Phi}_{k}$, $\hat{\nu}_{k}$, $\hat{\eta}_{k,g}$. Denoting with an upper-script $t$ the value of these hyperparameters at the $t$-th iteration of the CAVI algorithm, $\hat{\alpha}_k^t$, $\hat{m}_{k}^t$, $\hat{\beta}_{k}^t$, $\hat{\Phi}_{k}^t$, $\hat{\nu}_{k}^t$, $\hat{\eta}_{k,g}^t$ are deterministic functions of $\hat{\alpha}_k^{t-1}$, $\hat{m}_{k}^{t-1}$, $\hat{\beta}_{k}^{t-1}$, $\hat{\Phi}_{k}^{t-1}$, $\hat{\nu}_{k}^{t-1}$, $\hat{\eta}_{k,g}^{t-1}$ where the transformation is computed by evaluating the intermediate quantities $r_{ik}^{(t-1)}$. Instead of directly working with the variational hyper-parameters, we will focus on related estimates defined as follows:
\begin{align} \label{eq: simplifiedestimates} 
\hat{\pi}_k^{t}&=\frac{1}{n}\sum_{i=1}^{n} r_{ik}^{(t-1)}, \quad \hat{\mu}_k^{t}=\frac{\sum_{i=1}^{n} r_{ik}^{(t-1)}x_i}{\sum_{i=1}^{n} r_{ik}^{(t-1)}}, \nonumber \\
\hat{\psi}_{k,g}^{t}&=\frac{ \sum_{i=1}^{n} I_{\{c_i=g\}} r_{ik}^{(t-1)} }{ \sum_{i=1}^{n} r_{ik}^{(t-1)} } \nonumber \\
\hat{\Lambda}_k^{t}&=(\sum_{i=1}^{n} r_{ik}^{(t-1)} )(\sum_{i=1}^{n} r_{ik}^{(t-1)} (x_i-\hat{\mu}_k^{t})(x_i-\hat{\mu}_k^{t})^T )^{-1}
\end{align} 
The newly introduced CAVI estimates $\hat{\pi}_k^{t}, \hat{\mu}_k^{t}, \hat{\Lambda}_k^{t}, \hat{\psi}_{k,g}^{t}$ are closely related to the CAVI posterior means $\mathbb{E}_{\pi} [\pi_{k}]$, $\mathbb{E}_{\mu} [\mu_{k}]$, $\mathbb{E}_{\Lambda} [\Lambda_{k}]$, $\mathbb{E}_{\psi} [\psi_{k}]$ respectively. For example, $\hat{\pi}_k^{t}$ is equal to the CAVI posterior mean of $q(\pi_k)$ parametrised by $\hat\alpha_k^t$ in the case where the prior hyperparameter $\alpha=0$. Similar connections can be made for each variational posterior mean in the case where the prior hyperparameters $\beta, \nu, \Phi, \eta$ are set to $0$. Analysing these simplified estimates are sufficient in the sense that if they converge as $n \rightarrow \infty$, then the corresponding variational posterior means will asymptotically converge to the same limits (\Cref{thm: posteriormeancorollary}).

Like the variational hyperparameters, the newly introduced estimates $\hat{\pi}_k^{t}$,  $\hat{\mu}_k^{t}$,  $\hat{\Lambda}_k^{t}$,  $\hat{\psi}_{k,g}^{t}$ can be iteratively computed at each iteration by evaluating functions of the intermediate terms $r_{ik}^{(t-1)}$. Denoting by $\Theta^{t}$ the set of all quantities $\{\hat{\pi}_k^{t}, \hat{\mu}_k^{t}, \hat{\Lambda}_k^{t}, \{\hat{\psi}_{k,g}^{t}\}_g\}_k$ at iteration $t$, then as $r_{ik}^{(t-1)}$ are functions of the variational hyperparameters at $t-1$ and therefore $\Theta^{t-1}$, we can write $\Theta^t=T(\Theta^{(t-1)})$ for some deterministic transformation $T$ which encapsulates our proposed CAVI algorithm. In the following we define this function $T$ more formally. To do so we introduce some vectorized notation and define the required vector space for the domain of our function. We drop the upper script $t$ so that the quantities such as $\hat{\pi}_k$ represent a generic input to the function $T$. 

Let $\hat{\pi}_k \in \mathbb{R}$, $\hat{\mu}_k \in \mathbb{R}^q$, $\hat{\Lambda}_k \in S^q(\mathbb{R})$, the set of real, symmetric $q$ by $q$ matrices, and $\hat{\psi}_k \in \mathbb{R}^d$. Next, let $\mathbb{S}_{\pi}, \mathbb{S}_{\mu}, \mathbb{S}_{\Lambda},  \mathbb{S}_{\psi}$ be the K-fold direct sums of $\mathbb{R}$, $\mathbb{R}^q$, $S^q(\mathbb{R})$, $\mathbb{R}^d$ respectively. Then we let $\bm{\hat{\pi}}=(\hat{\pi}_1, \dots, \hat{\pi}_K) \in \mathbb{S}_{\pi}$, $\bm{\hat{\mu}}=(\hat{\mu}_1, \dots, \hat{\mu}_K) \in \mathbb{S}_{\mu}$, $\bm{\hat{\Lambda}} = (\hat{\Lambda}_1, \dots, \hat{\Lambda}_K) \in \mathbb{S}_{\Lambda}$, $\bm{\hat{\psi}} = (\hat{\psi}_1, \dots, \hat{\psi}_K)  \in \mathbb{S}_{\psi}$. Let $\Theta=(\bm{\hat{\pi}}, \bm{\hat{\mu}}, \bm{\hat{\Lambda}}, \bm{\hat{\psi}})$ belong to the direct sum space $\mathbb{S}_{\Theta}=\mathbb{S}_{\pi} \bigoplus \mathbb{S}_{\mu} \bigoplus \mathbb{S}_{\Lambda} \bigoplus \mathbb{S}_{\psi}$. We also let $\Theta^{*}=(\bm{\pi}^{*},\bm{\mu}^{*},\bm{\Lambda}^{*},\bm{\psi}^{*})$ be the true parameter, and the true data generating distribution be $P_{\Theta^*}$. As the notation implies, $\mathbb{S}_{\pi}, \mathbb{S}_{\mu}, \mathbb{S}_{\Lambda},  \mathbb{S}_{\psi}$ are vector spaces which contain possible values of the corresponding parameters $\bm{\hat{\pi}}, \bm{\hat{\mu}}, \bm{\hat{\Lambda}}, \bm{\hat{\psi}}$ and $\mathbb{S}_{\Theta}$ a vector space which contains the possible values of the overall parameter $\Theta$.


As the $r_{ik}^{(t-1)}$ are functions of the variational hyperparameters at $t-1$, they are also functions of $\Theta^{t-1}$, as the variational hyperparameters can be expressed in terms of $\Theta^{t-1}$. So from here on we will use notation such as $r_{ik}(\Theta^{t-1})=r_{ik}^{(t-1)}$ or $r_{ik}(\Theta)$ when highlighting the dependence of $r_{ik}$ on $\Theta$ instead\footnote{Note $r_{ik}(\Theta)$ is not well defined for all $\Theta$ as some $\Theta$ outside the support of the prior results in variational hyperparameters outside the hyperparameter space, but for those $\Theta$ we can trivially map $r_{ik}(\Theta)$ to a value like $1/K$.}. Let:
\begin{align} \label{eq: simplifiedestimatesfun}
&\Pi_k(\Theta)=\frac{1}{n}\sum_{i=1}^{n} r_{ik}(\Theta), \quad  M_k(\Theta) =\frac{\sum_{i=1}^{n} r_{ik}(\Theta) x_i}{\sum_{i=1}^{n} r_{ik}(\Theta) } \nonumber, \\
&\Psi_k(\Theta)_g  =\frac{ \sum_{i=1}^{n} I_{\{c_i=g\}} r_{ik}(\Theta) }{ \sum_{i=1}^{n} r_{ik}(\Theta) } ,  S_k(\Theta) = (\sum_{i=1}^{n} r_{ik}(\Theta) ) \nonumber \\ 
& \times (\sum_{i=1}^{n} r_{ik}(\Theta) (x_i-M_k(\Theta))(x_i-M_k(\Theta))^T )^{-1}.
\end{align}
These functions denote the maps which take the simplified estimates $\Theta^{(t)}$ from one iteration to the next. Vectorising these maps:
\begin{align}
\Pi(\Theta) &= \begin{bmatrix} \Pi_1(\Theta) \\  \vdots \\ \Pi_K(\Theta) \end{bmatrix},
M(\Theta) = \begin{bmatrix} M_1(\Theta) \\  \vdots \\ M_K(\Theta) \end{bmatrix} \nonumber \\
S(\Theta) &= \begin{bmatrix} S_1(\Theta) \\  \vdots \\ S_K(\Theta) \end{bmatrix},
\Psi(\Theta) = \begin{bmatrix} \Psi_1(\Theta) \\  \vdots \\ \Psi_K(\Theta) \end{bmatrix}
\end{align}
we can define the iterative procedure $\Theta^{(t)}=T(\Theta^{(t-1)})$, where:
\begin{align}
T(\Theta)=\begin{bmatrix} \Pi(\Theta) \\ M(\Theta) \\ S(\Theta) \\ \Psi(\Theta) \end{bmatrix}
\end{align}
For theoretical purposes we consider a generalised iterative procedure, the map $T^{\epsilon}$ defined for $0 < \epsilon < 2$ as: 
\begin{align} \label{eq: iterativeprocedure} 
\Theta^{(t)}= T^{\epsilon}(\Theta^{(t-1)})=(1-\epsilon)\Theta^{(t-1)}+\epsilon T(\Theta^{(t-1)})
\end{align}
Note $T^{1}(\Theta^{(t-1)})=T(\Theta^{(t-1)})$, so that $\epsilon=1$ is our CAVI algorithm. It is suggested in \cite{BoWang2006} that different choices of $\epsilon$ could improve convergence rate relative to the base case of $\epsilon=1$, but we do not investigate this in this paper. 

For our convergence analysis, we need to define a suitable norm on $\mathbb{S}_{\Theta}$. To define this norm, consider the following inner products:
\begin{align}
& \langle u_1,u_2 \rangle_{\pi_k}=u_1(\frac{1}{\pi_{k}^*})u_2  \qquad u_1,u_2 \in \mathbb{R} \nonumber \\
& \langle v_1,v_2 \rangle_{\mu_k}=v_1^T(\pi_k^* \Lambda_k^*)v_2 \qquad v_1,v_2 \in \mathbb{R}^q \nonumber \\
& \langle W_1,W_2 \rangle_{\Lambda_k}=\frac{\pi_k^*}{2}Tr(W_1\Lambda_k^{*-1}W_2\Lambda_k^{*-1}) \nonumber \\
&W_1,W_2 \in S^q(\mathbb{R}) \nonumber \\
& \langle a_1,a_2 \rangle_{\psi_k} = a_1^T diag(\frac{\pi_k^*}{\psi_k^*}) a_2 \quad a_1, a_2 \in \mathbb{R}^d
\end{align}
where $diag(\frac{\pi_k^*}{\psi_k^*})$ is a $d$ by $d$ diagonal matrix with entries $\frac{\pi_k^*}{\psi_{k,g}^*}$, $g=1,\dots,d$. We can easily show that $\langle W_1,W_2 \rangle_{\Lambda_k}$ is indeed an inner product.\footnote{$\langle W_1,W_2 \rangle_{\Lambda_k}$ is an inner product. Indeed, $Tr(W_2\Lambda_k^{*-1}W_1\Lambda_k^{*-1})=Tr(W_1\Lambda_k^{*-1}W_2\Lambda_k^{*-1})$ by the cyclic property of trace so symmetry is satisfied. For $\alpha, \beta \in \mathbb{R}, W_1,W_2,W_3 \in S^q(\mathbb{R})$,
\begin{align*}
&Tr((\alpha W_1+\beta W_3)\Lambda_k^{*-1}W_2\Lambda_k^{*-1}) \\
&=\alpha Tr(W_1\Lambda_k^{*-1}W_2\Lambda_k^{*-1}) + \beta Tr(W_3\Lambda_k^{*-1}W_2\Lambda_k^{*-1})
\end{align*}so linearity in the first argument is satisfied. Lastly, $$Tr(W\Lambda_k^{*-1}W\Lambda_k^{*-1})= Tr( (\Lambda_k^{*-\frac{1}{2}} W \Lambda_k^{*-\frac{1}{2}} )(\Lambda_k^{*-\frac{1}{2}} W \Lambda_k^{*-\frac{1}{2}})^T ) \geq 0 $$ so we have positive definiteness. }
Using the above inner products, we define $\langle B_1,B_2 \rangle$, on $B_1,B_2 \in \mathbb{S}_{\Theta}$ by simply adding the inner products of the constituent spaces: 
\begin{align}\label{eq: innerproductdef}
\langle .,. \rangle=\sum_{k=1}^{K}( \langle .,. \rangle_{\pi_k}+ \langle .,. \rangle_{\mu_k}+ \langle .,. \rangle_{\Lambda_k}+ \langle .,. \rangle_{\psi_k})
\end{align}
This inner product induces a vector norm on $\mathbb{S}_{\Theta}$ via  $\| B \|^2=\langle B,B \rangle$. \\

We can now state the main result:

\begin{theorem}\label{thm: maintheorem}
Almost surely  as the sample size $ n \rightarrow \infty$, the iterative procedure defined in \Cref{eq: iterativeprocedure} converges locally to the true data generating parameter $\Theta^{*}$ whenever $0 < \epsilon <2$. Locally here means whenever the starting values are sufficiently near $\Theta^{*}$.  
\end{theorem}

The proof of \Cref{thm: maintheorem} is contained in the supplement. The proof strategy involves showing that almost surely  as $ n \rightarrow \infty$, the map $T^{\epsilon}(\Theta)$ is locally contractive near $\Theta^{*}$. Specifically, there exists $0 \leq \lambda<1$ such that:
\begin{align}\label{eq: contractiondef}
\| T^{\epsilon}(\Theta)-T^{\epsilon}(\Theta^*) \| \leq \lambda \| \Theta-\Theta^* \|
\end{align}
for $\Theta$ sufficiently near $\Theta^{*}$. \add{To show this, by the integral form of the mean value theorem, we have:}
\add{
\begin{align}\label{eq: contractiontaylor}
& T^{\epsilon}(\Theta)-T^{\epsilon}(\Theta^*)
= \nonumber \\ 
&\left(\int_0^1 \nabla T^{\epsilon}\big(\Theta^*+s(\Theta-\Theta^*)\big)\,ds\right)(\Theta-\Theta^*),
\end{align}
}\add{valid whenever the segment $[\Theta^*,\Theta]$ lies in a region on which
$\nabla T^\epsilon$ is defined and continuous. Taking norms,}
\add{
\begin{align}\label{eq: contractiontaylor2}
&\| T^{\epsilon}(\Theta)-T^{\epsilon}(\Theta^*) \| \leq \nonumber \\
&\sup_{s\in[0,1]}\big\| \nabla T^{\epsilon}\big(\Theta^*+s(\Theta-\Theta^*)\big) \big\|_{op}\; \| \Theta-\Theta^* \|.
\end{align}
}
\add{It is therefore sufficient for the operator norm of $\nabla T^\epsilon(\Theta)$ to be
below $1$ in a compact neighbourhood of $\Theta^*$, almost surely for all large $n$. Writing $T^\epsilon_n$ to emphasise the dependence of $T^\epsilon$ on the sample size $n$, it is sufficient to prove the following three conditions over a closed ball
$\overline{B_{\delta}(\Theta^*)}$: \textup{(A1)} uniform almost-sure
convergence of $\nabla T^\epsilon_n$ to a continuous limit $A^\epsilon(\Theta)$;
\textup{(A2)} $\sup_{\overline{B_{\delta}(\Theta^*)}}\|A^\epsilon(\Theta)\|_{op}\le\lambda_0<1$;
\textup{(A3)} $T^\epsilon_n(\Theta^*)\to\Theta^*$. Details of the proof are reserved for the supplement.}

The vector norm that we use in the context of \Cref{eq: contractiondef} and \Cref{eq: contractiontaylor} is induced by the inner product  defined in \Cref{eq: innerproductdef}. Similarly, the operator norm used in \Cref{eq: contractiontaylor2} is defined as the operator norm with respect to the vector norm, so:
\begin{align}\label{eq: operatornorm}
\| A \|_{op}=\mathrm{sup}\{ \|AB \|: \|B \| \leq 1,  B \in \mathbb{S}_{\Theta} \}
\end{align}

\begin{remark}\label{rem: initialisation}
\add{
`Sufficiently near' in \Cref{thm: maintheorem} means there is a deterministic radius $\delta>0$,
depending only on $\Theta^*$ and the fixed prior hyperparameters, such that the
conclusion holds whenever the initial simplified estimates $\Theta^{(0)}$ lie in
the ball $\overline{B_{\delta}(\Theta^*)}$. This is the same kind of local guarantee as for
the EM algorithm on mixtures (e.g. \cite{Balakrishnan2017}) and, like EM, it does not guarantee convergence from
arbitrary starts, nor does it resolve the well-known label-switching
non-identifiability issue, which makes a global guarantee impossible without
restricting to a neighbourhood of one labelling of $\Theta^*$. In practice one
does not know $\Theta^*$, so it is recommended to initialise CAVI using outputs from a different estimator such as $K$-prototypes (which we do in \Cref{sec: results}) or EM.}
\end{remark}

The following corollaries are consequences of \Cref{thm: maintheorem} and rely on the same assumptions (\Cref{ass: wellspecified,ass: interior,ass: distinct,ass: correctK}). Proofs are given in the supplement. 

\begin{corollary}\label{thm: posteriormeancorollary}
 Almost surely as $n \rightarrow \infty$, the variational posterior mean obtained by running the CAVI algorithm converges as $t \rightarrow \infty$ to the true data generating parameter $\Theta^{*}$.
\end{corollary}

\begin{corollary}\label{thm: posteriorpredcorollary}
Almost surely as $n \rightarrow \infty$, the variational posterior predictive density described in \Cref{eq: posteriorpred} for $p=1$ converges as $t \rightarrow \infty$ pointwise to the likelihood density function evaluated at $\Theta^{*}$.
\end{corollary}

\add{Establishing the posterior contraction rate additionally uses a regularity
property of the mixture family: that it satisfies Conditions 1 and 2 of
\citet{Redner1984} in a neighbourhood of $\Theta^*$. Unlike
\Cref{ass: wellspecified,ass: interior,ass: distinct,ass: correctK}, this is not an
assumption we impose but a consequence of them, which we verify directly for our model.

\begin{lemma}[Redner--Walker regularity]\label{ass: RW}
Under \Cref{ass: wellspecified,ass: interior,ass: distinct,ass: correctK}, the finite
mixture family \Cref{eq: likelihoodVIzintegratedout} satisfies Conditions 1 and 2 of
\citet{Redner1984} in a neighbourhood of $\Theta^*$.
\end{lemma}
}

\add{
We prove \Cref{ass: RW} in Section S4.8 of the supplement. By Theorem 3.1 of
\citet{Redner1984} it implies the local maximum likelihood estimator is asymptotically
normal and converges to $\Theta^*$ at rate $n^{-1/2}$. With this, we can show:}

\begin{corollary}\label{thm: MLEcorollary}
Let $\Theta^{\infty}$ be the limit of $\Theta^{(t)}$ as $t \rightarrow \infty$, and $\tilde{\Theta}$ be the strongly consistent local maximum likelihood estimate of the parameter ${\Theta}$. Then, for sufficiently large sample size $n$, $\Theta^{\infty}=\tilde{\Theta}+O(\frac{1}{n})$ almost surely.
\end{corollary}

\begin{theorem}\label{thm: posteriorconsistencycorollary}
Let $Q_{n}^{\infty}(\Theta)$ denote the limiting CAVI posterior distribution for sample size $n$, obtained from initial values sufficiently close to $\Theta^{*}$. Then $Q_{n}^{\infty}(\Theta)$ is strongly consistent at $\Theta^*$; that is:
$$Q_{n}^{\infty}(\| \Theta-\Theta^*\|_2>\epsilon)  \rightarrow 0$$ 
for all $\epsilon>0$ as $n \rightarrow \infty$, almost surely. \add{In addition, the limiting CAVI posterior $Q_n^\infty$ (from initial
values sufficiently close to $\Theta^*$) contracts around $\Theta^*$ at rate
$n^{-1/2}$: for every deterministic sequence $M_n\to\infty$,
\begin{equation}\label{eq:vb_contraction_rate}
Q_n^\infty\!\left(\big\|\Theta-\Theta^*\big\|_2 > \frac{M_n}{\sqrt{n}}\right)
\ \xrightarrow[]{P_{\Theta^*}}\ 0.
\end{equation}
}
\end{theorem}

%% file: results.tex
\section{Experimental Results} \label{sec: results}

In this section, we conduct simulation experiments to assess and validate our method. For all simulated data in this section, the continuous part of the data was standardised before fitting the model. Throughout this section we take the following values for the prior hyperparameters: $m_{k}=0$ for all $k$, $\beta=1$, $\Phi=0.25I$, $\nu=q+K+1$, $\alpha=1/K$, $\eta_j=1/d_j$, where $d_j$ is the number of possible values of the $j$th categorical variable. The choice for $\alpha$ is motivated by \cite{Rousseau2011}, which showed that $\alpha<d_\Theta/2$, where $d_\Theta$ is the dimension of the component-specific parameter, causes the posterior of an overfitted mixture distribution with more components specified than the true number of components to set the weights of the extra components to $0$ as the sample size tends to infinity. We also initialise the CAVI algorithm using the $K$-prototypes algorithm, by setting the initial $r_{ik}$ to be equal to $0.9$ at the $k$ that $K$-prototypes predicts for each data point $(x_i,c_i)$, and $0.1$ otherwise.\footnote{Note that the $r_{ik}$ do not sum to $1$ over $k$ with this initialisation, but the initial $r_{ik}$ are not required to be probability vectors.} In this section, we will use $\Sigma_{k}=\Lambda_{k}^{-1}$ to refer to the covariance matrix of component $k$ 
to more directly compare the size of the component covariances to how much the component means $\{\mu_{k}\}_k$ are separated. 

\subsection{An illustrative one-dimensional example} \label{subsec: 1dexample}

For ease of visualisation, we begin with an example where the continuous data is of dimension $q=1$. 
We simulate $n=5000$ data points with the following model parameters: $p=q=1$,  $K^{*}=3$,   $\pi^{*}_{1}=0.4$,  $\pi^{*}_{2}=0.35$, $\pi^{*}_{3}=0.25$, $\mu^{*}_{1}=0$, $\mu^{*}_{2}=2$, $\mu^{*}_{3}=5$, $\Sigma^{*}_{1}=0.1$, $\Sigma^{*}_{2}=0.2$ and  $\Sigma^{*}_{3}=0.5$ (before standardisation). In addition, $\psi^{*}_{k,k}=0.6$ for all $1 \leq k \leq K$ and $\psi^{*}_{k,g}=0.2$ for all $1 \leq g \leq K$ with $g \neq k$.  Figures \ref{fig: simpleexamplemuSigma}, \ref{fig: simpleexamplePsi12} and \ref{fig: simpleexamplePsi3pi} compare the marginal posterior distributions obtained using our VI approach to the ones using a Gibbs sampler for the same model to examine whether the variational posterior obtained using our proposed CAVI algorithm is a suitable approximation of the posterior distribution. 
We observe that, for each parameter, the marginal variational distribution matches the Gibbs marginal posterior well, particularly in terms of posterior mode, but has slightly less variance than the Gibbs posterior, which is most noticeable for $\Sigma$. Overall the variational posterior provides a suitable estimate for the true, data generating value and performs similarly to the Gibbs sampler. From \Cref{fig: simpleexampleposteriorpred}, the variational posterior predictive distribution is shown to be close to the true likelihood functions for both continuous and categorical data, indicating that the variational posterior is a good fit for the data.

\begin{figure*}[h!]{
\includegraphics[width = 0.48\textwidth]{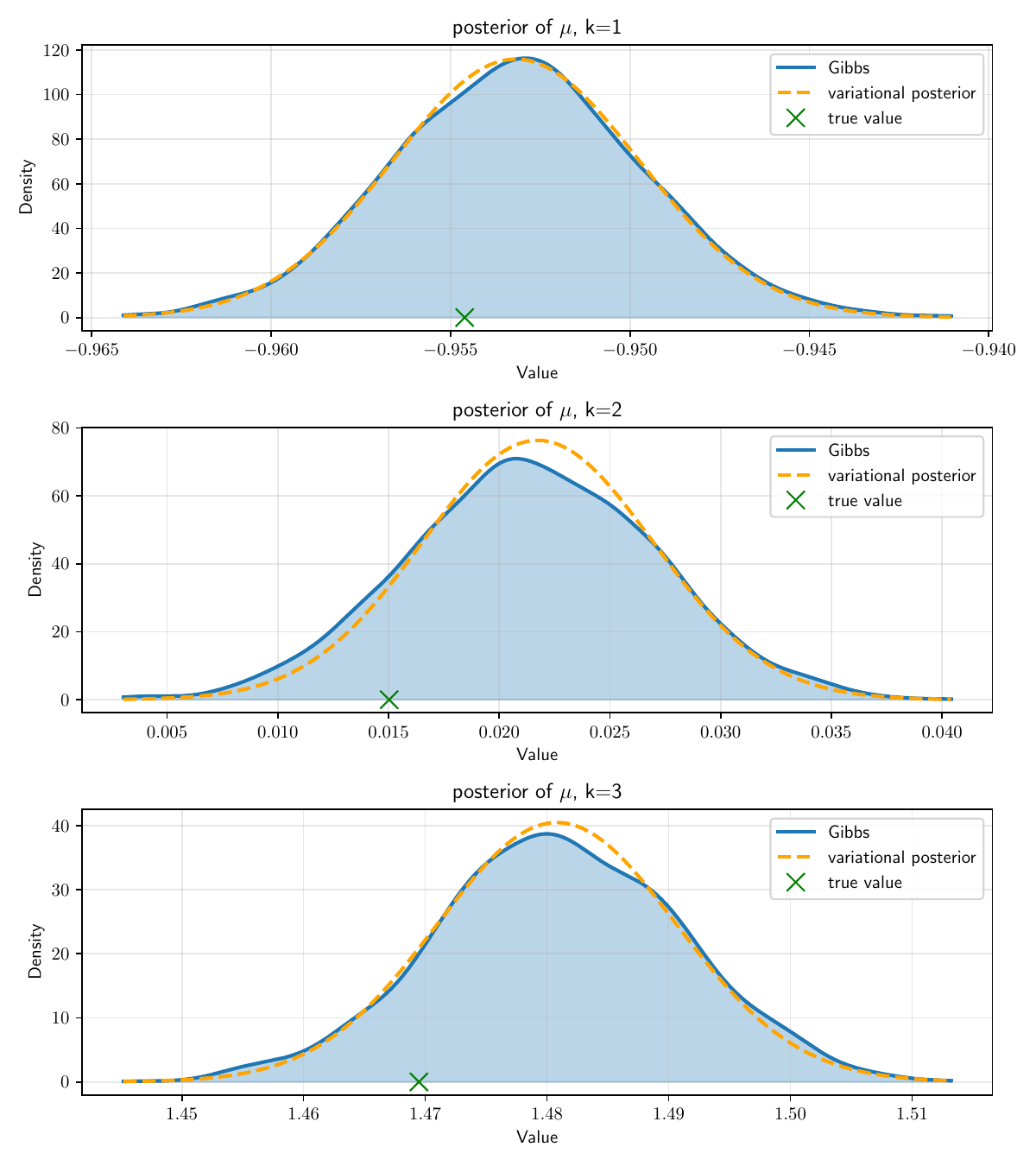}
\includegraphics[width = 0.48\textwidth]{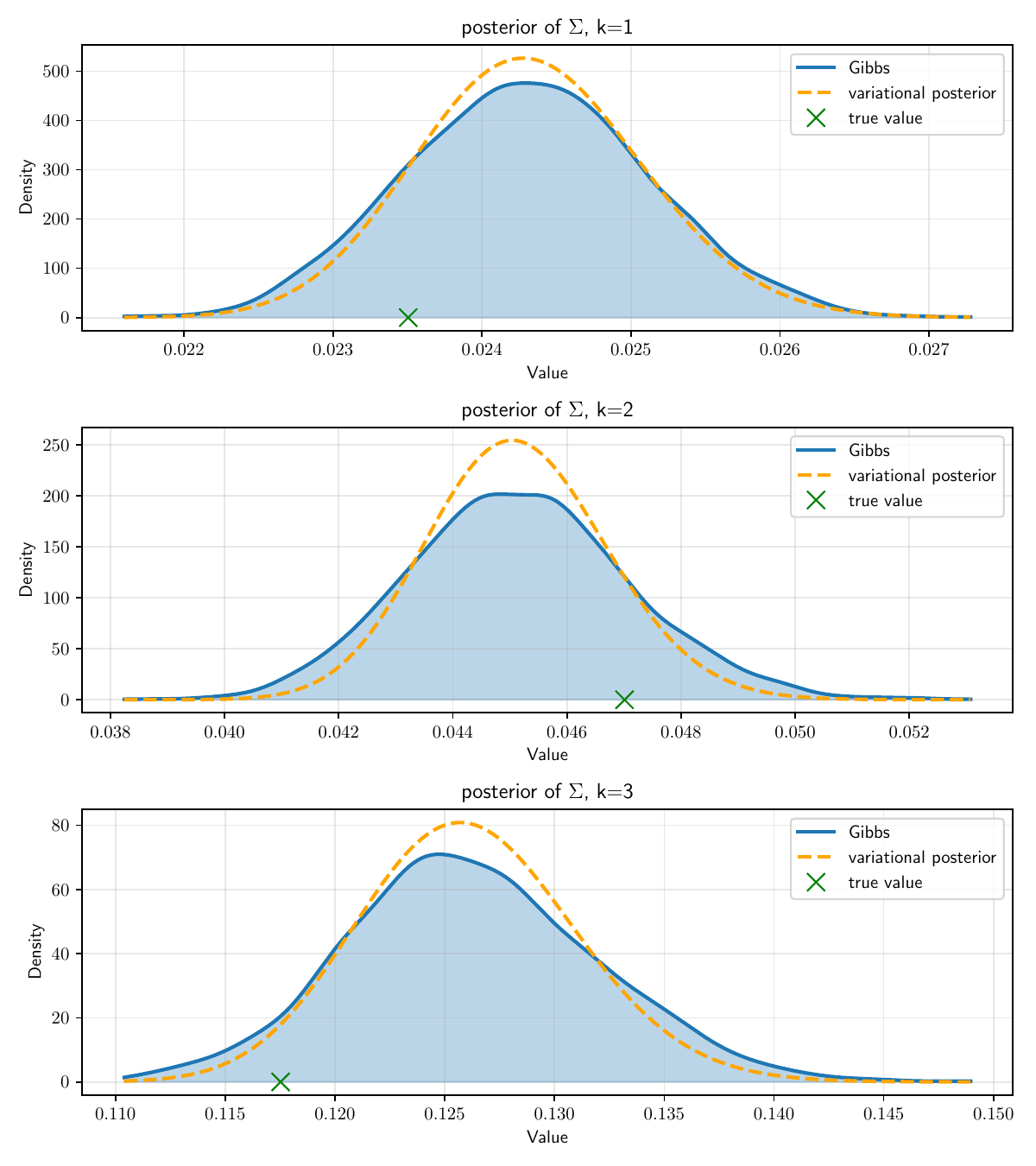}
}
\caption{Plot of marginal posterior distribution of $\mu$ and $\Sigma$ for both variational inference and Gibbs}
\label{fig: simpleexamplemuSigma}
\end{figure*}

\begin{figure*}[h!]{
\includegraphics[width = 0.48\textwidth]{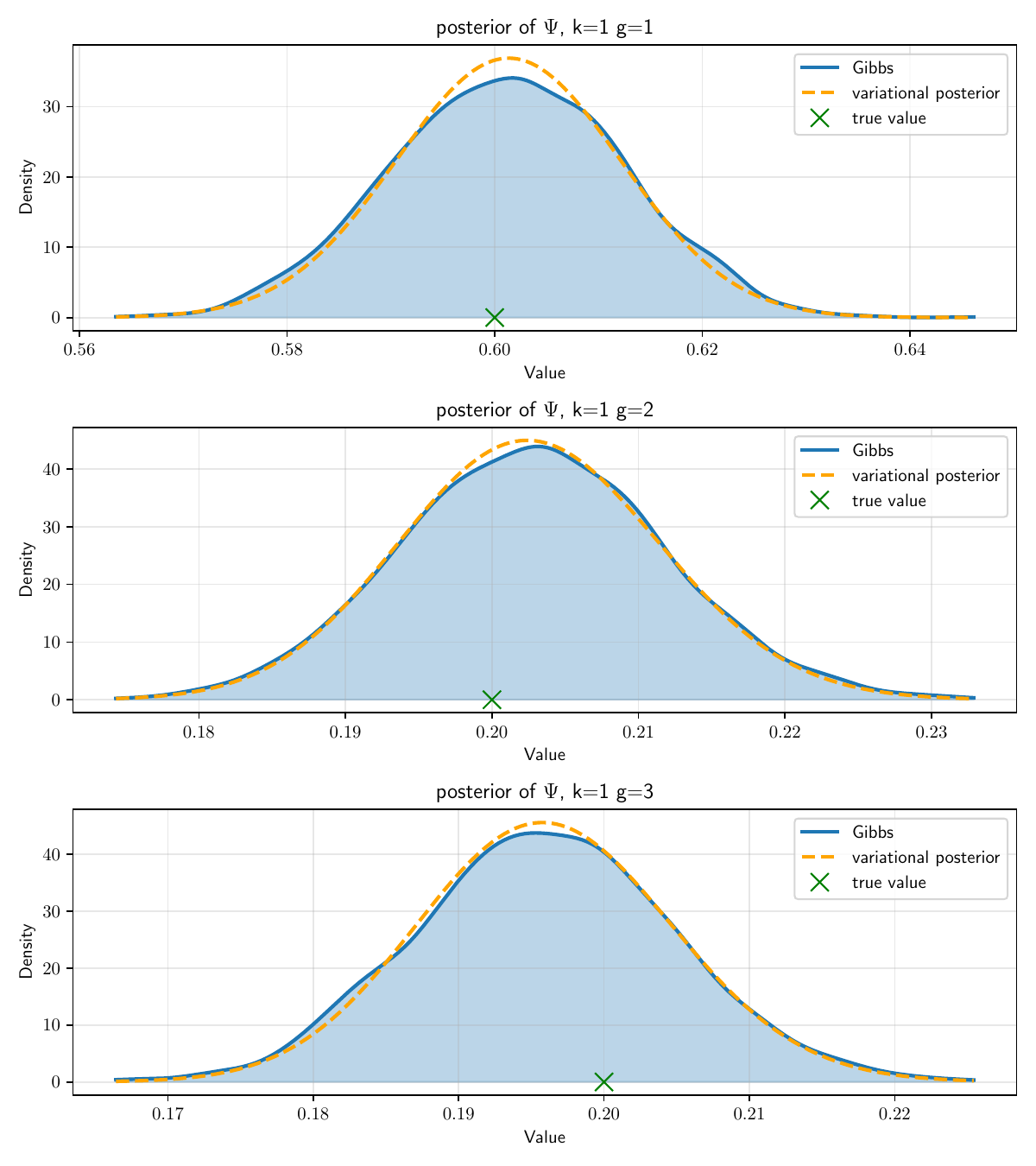}
\includegraphics[width = 0.48\textwidth]{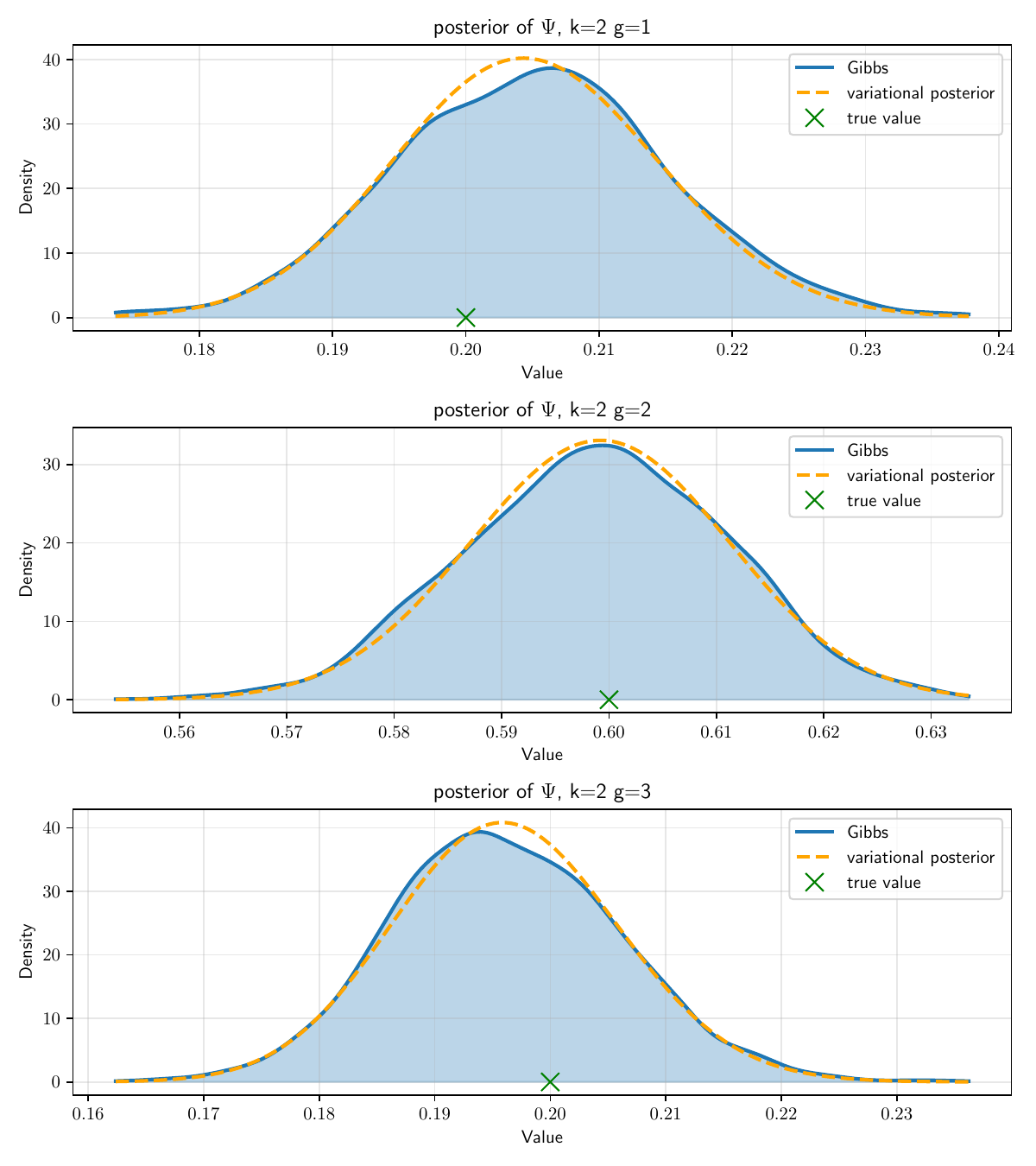}
}
\caption{Plot of marginal posterior distribution of $\psi_{1,g}, \psi_{2,g}$ for both variational inference and Gibbs}
\label{fig: simpleexamplePsi12}
\end{figure*}

\begin{figure*}[h!]{
\includegraphics[width = 0.48\textwidth]{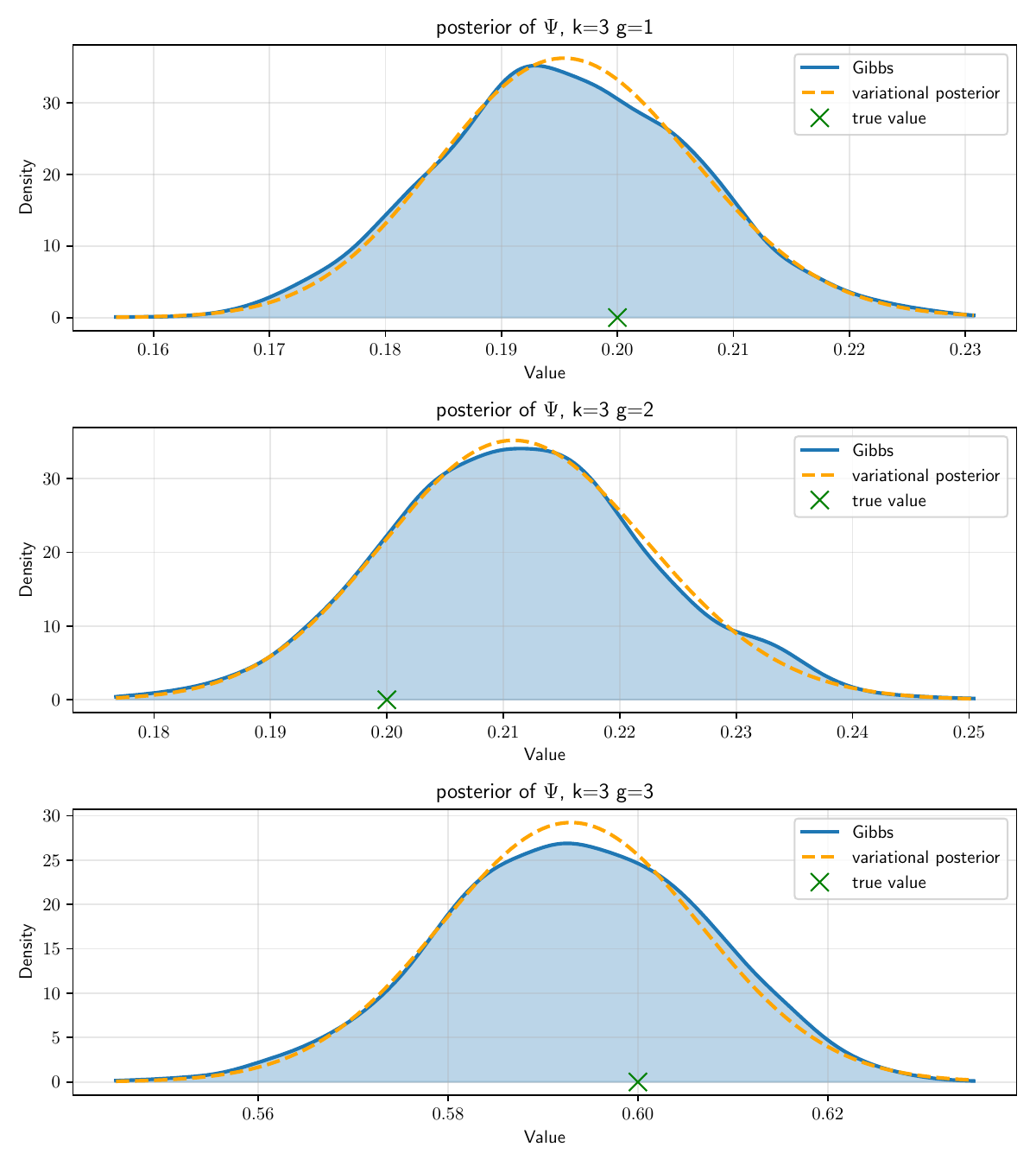}
\includegraphics[width = 0.48\textwidth]{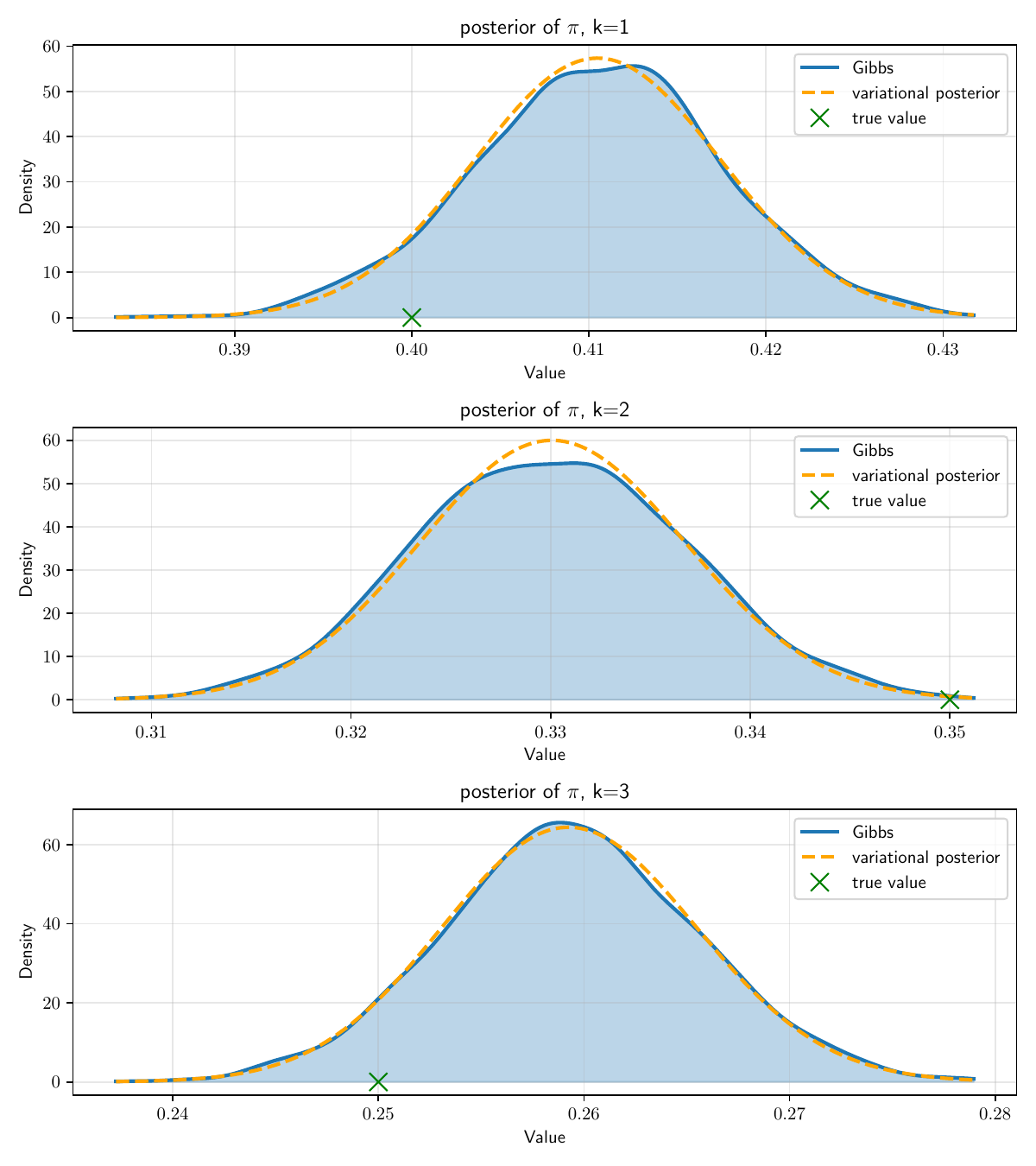}
}
\caption{Plot of marginal posterior distribution of $\psi_{3,g}, \pi$ for both variational inference and Gibbs}
\label{fig: simpleexamplePsi3pi}
\end{figure*}

\begin{figure*}[h!]{
\includegraphics[width = 0.55\textwidth]{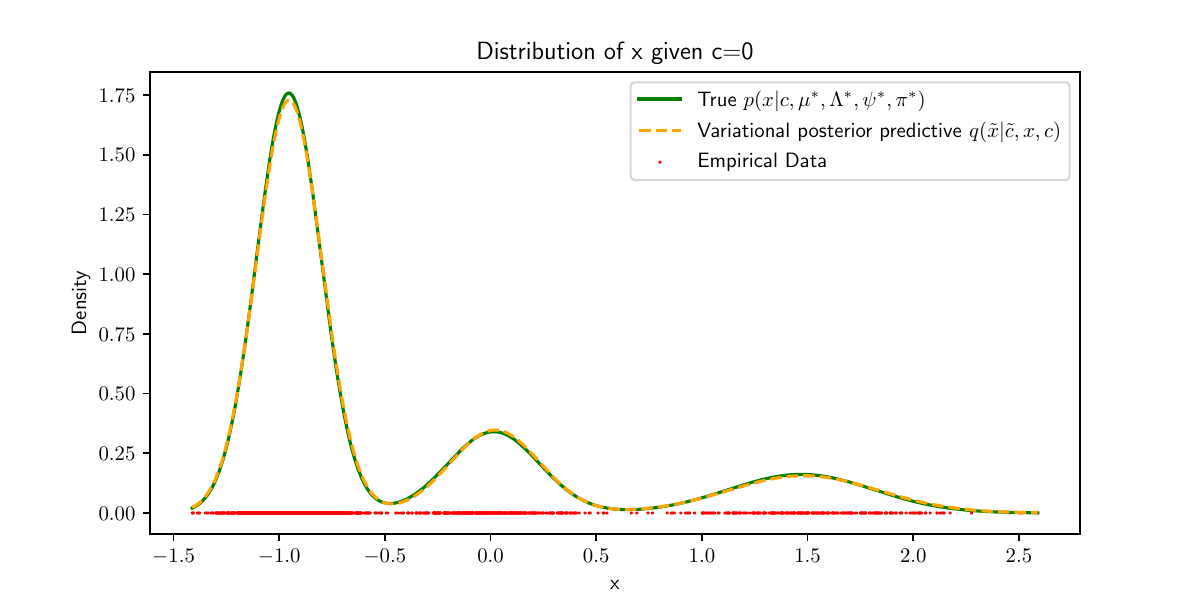}
\hspace{-1cm} 
\includegraphics[width = 0.55\textwidth]{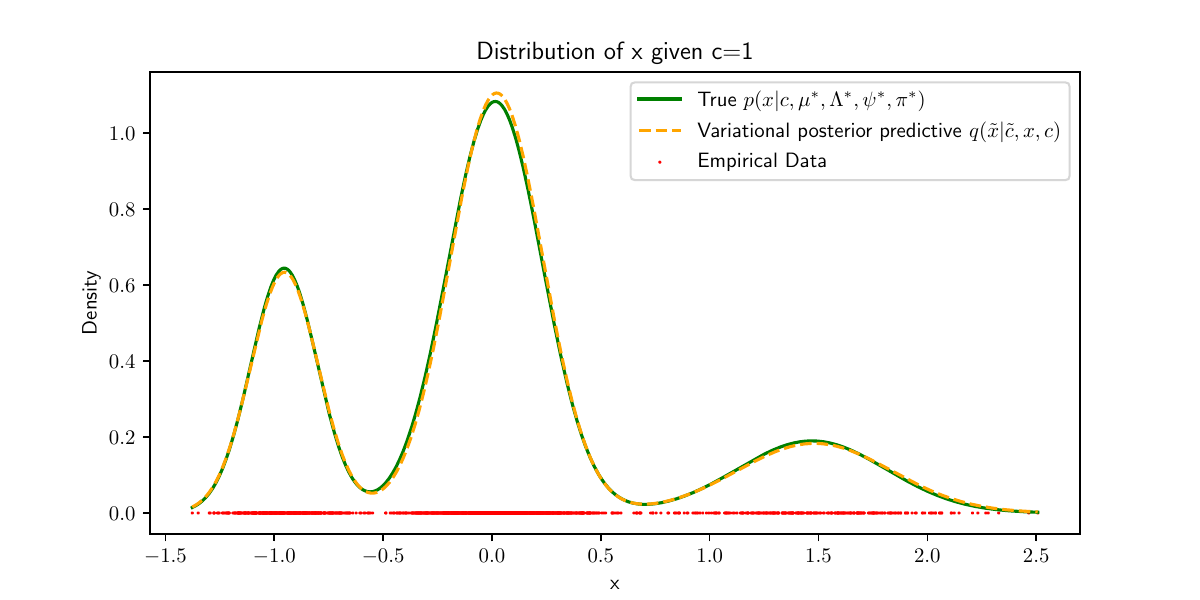}
\includegraphics[width = 0.55\textwidth]{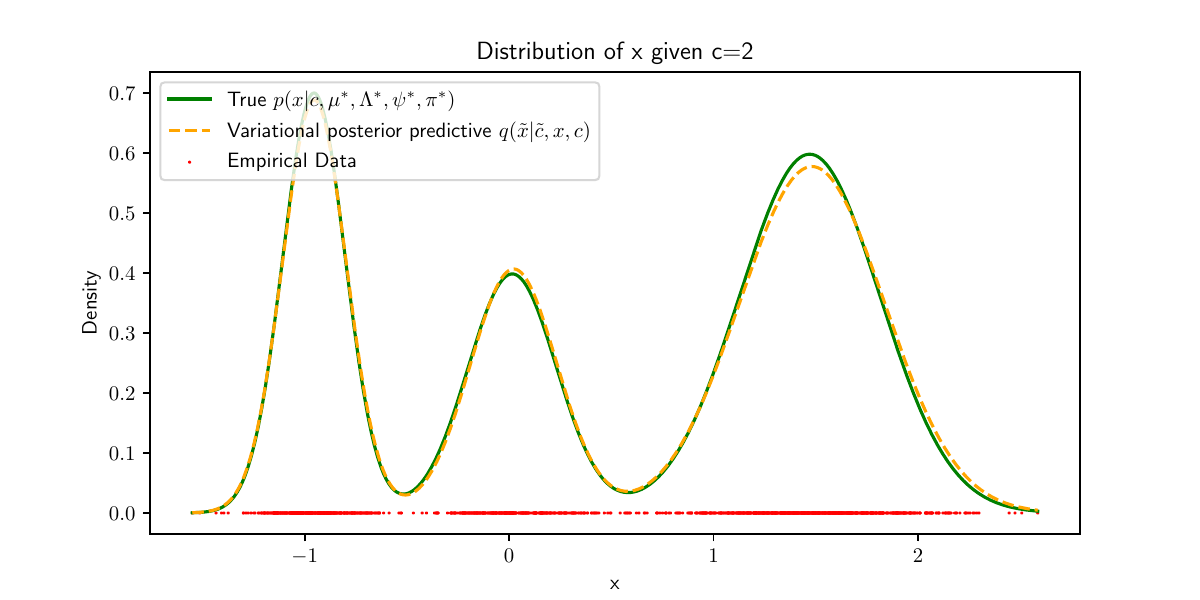}
\hspace{-1cm} 
\includegraphics[width = 0.4\textwidth]{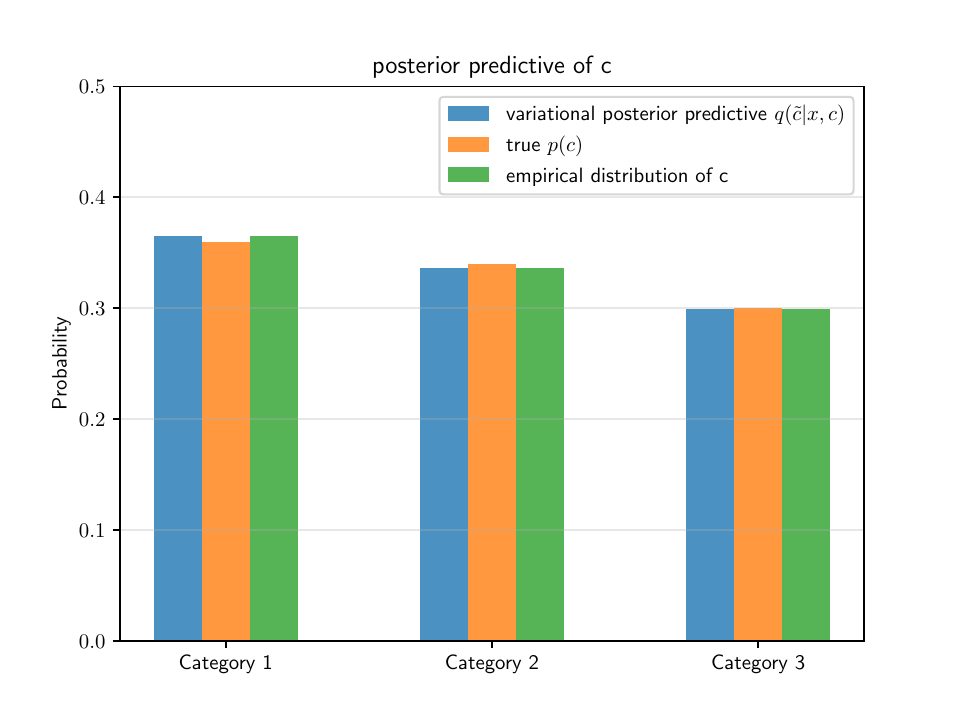}
}
\caption{Comparison of the variational posterior predictive distribution of $x$ given $c$, and $c$ to the empirical data and the likelihood function evaluated at the data generating parameters.}
\label{fig: simpleexampleposteriorpred}
\end{figure*}

\subsection{Simulation study}
We now conduct an extensive simulation study for three, higher dimensional scenarios, for data of various sample sizes $n$. In each scenario, we fix a set of true parameter values $(\mu^*, \Sigma^*, \psi^*, \pi^*)$ and $K^*$ components, and simulate both the component allocation $z_i^*$ as well as the observation $(x_i^*, c_i^*)$ for each data point $i=1,\dots,n$. Full detail of the data generating parameters for the three scenarios can be found in the supplement.

\subsubsection{Simulation set-ups and error metrics}\label{simulbench}
We consider the three following scenarios:
\begin{itemize}

\item In the first scenario, we simulate data with $\mu_{k}^{*}$ that are pairwise well separated, and $\Sigma_{k}^{*}$ with standard deviations that are small compared to the distance between the $\mu_{k}^{*}$. The categorical variable has weak correlation with the allocation variable $z$. This scenario aims to test the model when the mixtures are largely driven by well separated components in continuous variables, rather than categorical variables. 

\item For the second scenario, we instead use $\Sigma_{k}^{*}$ with standard deviations that are larger than the pairwise distance between $\mu_{k}^{*}$, so data generated heavily overlap. In contrast, the categorical variables are generated to be strongly dependent on $z$. This scenario aims to test the model when the mixtures are largely driven by the categorical variables rather than the continuous variables. 

\item For the third scenario, we simulate data that is a middle ground between scenarios 1 and 2. Specifically, we use $\mu_{k}^{*}$ with separations that are similar in magnitude to the standard deviations of $\Sigma_{k}^{*}$, causing some continuous data overlap. We generate a single categorical variable which depends significantly on $z$, but not as strongly as scenario 2. 

\end{itemize}

To assess the empirical performance of our algorithm, we employ several performance measures. First, to test parameter estimation, we calculate for $\pi, \mu, \Sigma, \psi$ the mean absolute error between the posterior mean of the variational posterior and the true parameter values:
\begin{align} \label{eq: paramerror} 
Error_{\mu}&=\frac{\sum_{k=1}^{K^{*}}  \left\Vert \mu^{*}_{k}-\hat{m}_{k} \right\Vert_1}{qK^{*} } \nonumber \\ Error_{\Sigma}&=\frac{\sum_{k=1}^{K^{*}} \left\Vert \Sigma^{*}_{k}-\frac{\hat{\Phi}_{k}}{\hat{\nu}_{k}-q-1} \right\Vert_1}{q^2K^{*} } \nonumber \\
Error_{\pi}&=\frac{\sum_{k=1}^{K^{*}} | \pi^{*}_{k}-\frac{\hat{\alpha}_k}{\sum_{k=1}^{K} \hat{\alpha}_k} | }{K^{*} }\\
Error_{\psi}&=\frac{\sum_{k=1}^{K^{*}} \sum_{j=1}^{p} \sum_{g_j=1}^{d_j} | \psi^{*}_{k,j,g_j}-\frac{\hat{\eta}_{k,j,g_j}}{\sum_{g_j'=1}^{d_j} \hat{\eta}_{k,j,g_j'} } | }{ (\sum_{j=1}^{p} d_j )K^{*} }\nonumber
\end{align}
To assess the density estimation, we sample an additional test dataset of size $n_{test}=\min(0.4n,2000)$ and calculate for each data point $(\tilde{x}_i, \tilde{c}_i)$ in the test dataset, the mean absolute error between the loglikelihood for \Cref{eq: likelihoodVIzintegratedout} evaluated at the true data generating parameters $(\mu^*, \Lambda^*, \psi^*, \pi^*)$, and log of the posterior predictive distribution in \Cref{eq: posteriorpred} using the final parameter estimates obtained from running \Cref{alg:ourCAVI} in the training set. We calculate the average of this error over the test set $Error_{logppd}$:
\begin{align} \label{eq: logppderror} 
&Error_{logppd}=\frac{1}{n_{test}}\sum_{i=1}^{n_{test}}|\mathrm{ln}(p(\tilde{x}_i, \tilde{c}_i | \mu^{*}, \Lambda^{*}, \psi^{*}, \pi^{*})) \nonumber \\
&-\mathrm{ln}(q(\tilde{x}_i, \tilde{c}_i|x,c)) | 
\end{align}

Lastly, to assess the uncertainty quantification over the parameters $(\mu, \Sigma, \psi, \pi)$, we examine the frequentist coverage of the posterior highest density regions (HDR) counting the proportion of $95\%$ variational posterior HDR containing the data generating parameter values over runs of the CAVI procedure on independent datasets. 
\add{Specifically, we evaluate this coverage at two levels: the joint coverage across all components for a given parameter (e.g., the joint HDR of $\mu = (\mu_1, \dots, \mu_K)$), as well as the component-specific marginal coverage for each individual cluster $k$ (i.e., $\mu_k, \Sigma_k$, and $\psi_k$). For the mixture weights $\pi$, only the joint coverage is reported, and we also report the joint coverage of $(\mu, \Sigma, \psi, \pi)$. We regard the per-cluster coverage as the more natural quantity in practice, since in a mixture model each cluster is typically interpreted on its own, and the uncertainty a practitioner reports is about the parameters of that cluster rather than a single region for the parameters of all $K$ clusters at once. The joint regions are also of much higher dimension, so their coverage drops sharply under even a small understatement of the variance in each direction, and is less informative about how large that understatement is.}
Conveniently, the variational posterior density $q(\mu, \Lambda, \psi, \pi)$ and the marginal densities $q(\mu)$, $q(\Lambda)$, $q(\psi)$, $q(\pi)$, $q(\mu_k)$, $q(\Lambda_k)$, $q(\psi_k)$ are all available in closed form, which allows us to easily calculate the HDR.

Note that the model likelihood in \Cref{eq: likelihoodVIzintegratedout} is invariant under any permutation $\sigma$  on ${1,\dots,K}$ 
\begin{align} \label{eq: permutedlikelihood} 
&\sum_{k=1}^{K} \pi_{\sigma(k)} \mathcal{N}(x_i|\mu_{\sigma(k)}, \Lambda_{\sigma(k)}^{-1} ) \prod_{j=1}^{p} \psi_{\sigma(k),j,c_{ij}} \nonumber \\
&=\sum_{k=1}^{K} \pi_{k} \mathcal{N}(x_i|\mu_{k}, \Lambda_{k}^{-1} ) \prod_{j=1}^{p} \psi_{k,j,c_{ij}}
\end{align}
This non identifiability of the likelihood usually causes the estimated model parameters from \Cref{alg:ourCAVI} to correspond to the true parameters up to a permutation. Because of this, when reporting error metrics listed or checking the frequentist coverage of posterior HDR on simulated data, we match the permuted parameters with the true parameters via linear sum assignment on an $L_1$ distance matrix between the estimated $\hat{m}_{k}$ and the true  $\mu^{*}_{k}$, and we use the resulting permutation $\sigma$ to relabel $\Lambda$, $\psi$, and $\pi$.

We calculate the aforementioned error metrics and coverage over $100$ independent runs for the three scenarios, for sample sizes $n=2500,5000,10000,20000$ assuming that the number of components $K=K^*$ is known a priori. 
To assess the competitiveness of our method, we compare the point estimation errors with the Expectation-Maximization (EM) algorithm of \cite{Hunt2002}, which unlike our method does not provide uncertainty quantification \add{by itself. So we use parametric bootstrap resampling to generate confidence regions for the EM estimates, under two different settings. The first setting, `EM (all steps)', generates $B=200$ datasets of size $n$ using the initial EM run's maximum likelihood estimate (MLE) $\tilde{\Theta}^{0}$. Then for each dataset $(x^{b},c^{b})$ from $1 \le b \leq B$, we run EM again until convergence to produce MLE estimates $\tilde{\Theta}^{b}$. The `EM (1 step)' on the other hand, only runs EM for 1 iteration for the bootstrap runs to obtain the $\tilde{\Theta}^{b}$. Confidence regions are generated via the bootstrap samples by assuming they are normally distributed, and using Hotelling's $T^2$ distribution. Essentially, the `EM (1 step)' option provides a way for EM to generate confidence regions for $\tilde{\Theta}^{0}$ with minimal extra computational cost in exchange for less accurate coverage, whereas the `EM (all steps)' option lets EM run until convergence on the bootstrap runs in order to obtain more accurate coverage, but at much greater computational cost.} 
We also compare both the point estimation error and coverage with MCMC, in this case a Gibbs sampler. MCMC is in a sense the gold standard in Bayesian inference as the empirical distribution of its samples converges to the true posterior with infinite computational power. For the Gibbs sampler, the posterior mean is used as the point estimate when calculating the errors. $Error_{logppd}$ is not calculated for the Gibbs sampler as the posterior predictive density is not available in closed form. \add{For Gibbs's coverage we calculate the overall HDR, using the joint density instead of the posterior as it is proportional to the posterior, while the true posterior is unavailable in closed form. We also calculate the marginal coverage for $\pi$, as well as $\mu_k, \Sigma_k$, and $\psi_k$ for each $k$. Details on the Gibbs sampler itself and how the marginal coverage is calculated can be found in Section S3 of the Supplement.

For the stopping criteria, the VI algorithm and EM algorithms are deemed to have converged if the change in ELBO (for VI) or log-likelihood (for EM) between consecutive iterations becomes less than $0.0001$. The Gibbs sampler is run until its parameter errors described in \Cref{eq: paramerror} post burn-in are within $2\%$ of those of VI on the same dataset, up to a maximum of $5001$ iterations (including burn-in). The burn-in period was chosen using Geweke's diagnostic \cite{Geweke1991} applied to the log joint density trace at 20 equally spaced candidate burn-ins over the first half of the chain. We took the smallest burn-in for which the absolute Geweke $z$-score was below $2$ at two consecutive candidates (with a minimum of 50 iterations).}

\add{Experiments in \Cref{simulbench} were run on the CX3 cluster of Imperial College London's High Performance Computing service, whose compute nodes have either two AMD EPYC 7742 or two Intel Xeon Platinum 8358 processors. Each run was a separate single-node job allocated 4 dedicated CPU cores and 24 GB of RAM, and all methods were implemented in Python 3.9.22. Additional experiments on choosing the number of mixture components $K$ (Section S6) and on the sensitivity to the prior hyperparameters (Section S7) can be found in the supplement.}

\subsubsection{Summary of simulation study results}

\begin{table*}[htbp]
\small
\centering
\add{
\resizebox{\textwidth}{!}{%
\begin{tabular}{ll ccccc}
\toprule
$n$ & Method & $\text{Error}_{\mu} (\times 10^{-3})$ & $\text{Error}_{\Sigma} (\times 10^{-3})$ & $\text{Error}_{\psi} (\times 10^{-3})$ & $\text{Error}_{\pi} (\times 10^{-3})$ & $\text{Error}_{\text{logppd}}$ \\ 
\midrule

\multirow{3}{*}{2500} 
 & VI    & 13.8, [9.79,18.5] & 7.59, [6.13,9.10] & 15.8, [10.5,21.5] & 6.36, [2.67,11.3] & 0.184, [0.166,0.204] \\
 & Gibbs & 13.8, [9.50,18.6] & 7.60, [6.15,9.03] & 15.8, [10.5,21.5] & 6.39, [2.75,11.2] & -- \\
 & EM    & 13.6, [9.18,18.1] & 7.20, [5.57,8.71] & 15.8, [10.4,21.5] & 6.36, [2.67,11.3] & -- \\
\midrule

\multirow{3}{*}{5000} 
 & VI    & 9.77, [7.60,13.1] & 5.15, [4.17,6.22] & 11.0, [7.91,14.2] & 4.47, [1.74,7.24] & 0.125, [0.109,0.140] \\
 & Gibbs & 9.78, [7.49,12.9] & 5.15, [4.21,6.20] & 11.0, [7.94,14.2] & 4.50, [1.78,7.26] & -- \\
 & EM    & 9.74, [7.59,12.8] & 4.96, [3.99,5.95] & 11.0, [7.92,14.3] & 4.47, [1.74,7.23] & -- \\
\midrule

\multirow{3}{*}{10000} 
 & VI    & 7.22, [5.18,9.68] & 3.61, [2.88,4.42] & 7.92, [5.51,10.4] & 3.06, [1.48,5.05] & 0.0862, [0.0788,0.0960] \\
 & Gibbs & 7.24, [5.15,9.57] & 3.61, [2.91,4.36] & 7.93, [5.53,10.4] & 3.03, [1.52,5.07] & -- \\
 & EM    & 7.19, [5.15,9.65] & 3.53, [2.80,4.35] & 7.92, [5.54,10.4] & 3.06, [1.48,5.05] & -- \\
\midrule

\multirow{3}{*}{20000} 
 & VI    & 5.06, [3.71,6.66] & 2.52, [2.01,3.03] & 5.66, [3.99,7.57] & 2.18, [1.08,3.48] & 0.0591, [0.0518,0.0646] \\
 & Gibbs & 5.05, [3.78,6.71] & 2.51, [1.99,3.03] & 5.68, [3.94,7.56] & 2.18, [1.04,3.42] & -- \\
 & EM    & 5.05, [3.74,6.69] & 2.48, [1.99,2.99] & 5.66, [4.00,7.56] & 2.18, [1.08,3.48] & -- \\
\bottomrule
\end{tabular}%
}
}
\caption{Scenario 1 -- error metrics averaged over 100 runs, along with 5\% and 95\% quantiles.}
\label{table:scenario1error}
\end{table*}

\begin{table*}[htbp]
\scriptsize
\centering
\add{
\setlength{\tabcolsep}{3pt}
\resizebox{\textwidth}{!}{%
\begin{tabular}{ll cccccccc}
\toprule
$n$ & Method & Overall & $\mu$ & $\Sigma$ & $\psi$ & $\pi$ & $\mu_k$ & $\Sigma_k$ & $\psi_k$ \\ 
\midrule

\multirow{4}{*}{2500} 
 & VI             & 0.00 & 0.980 & 0.00 & 0.920 & 0.920 & 0.944, [0.930,0.960] & 0.668, [0.00,0.980] & 0.940, [0.900,0.980] \\
 & Gibbs          & 0.00 & -- & -- & -- & 0.920 & 0.932, [0.920,0.950] & 0.658, [0.00,0.970] & 0.942, [0.920,0.980] \\
 & EM (all steps) & 0.880 & 0.970 & 0.840 & 0.920 & 0.910 & 0.950, [0.910,0.980] & 0.898, [0.870,0.930] & 0.930, [0.870,0.980] \\
 & EM (1 step)    & 0.890 & 0.970 & 0.830 & 0.920 & 0.910 & 0.948, [0.910,0.980] & 0.894, [0.870,0.910] & 0.930, [0.870,0.980] \\
\midrule

\multirow{4}{*}{5000} 
 & VI             & 0.0800 & 0.960 & 0.0200 & 0.980 & 0.970 & 0.960, [0.930,0.980] & 0.712, [0.00,0.960] & 0.944, [0.930,0.960] \\
 & Gibbs          & 0.170 & -- & -- & -- & 0.970 & 0.954, [0.930,0.980] & 0.706, [0.00,0.950] & 0.938, [0.930,0.950] \\
 & EM (all steps) & 0.900 & 0.940 & 0.890 & 0.980 & 0.960 & 0.962, [0.940,0.980] & 0.920, [0.890,0.960] & 0.948, [0.920,0.970] \\
 & EM (1 step)    & 0.890 & 0.940 & 0.890 & 0.980 & 0.960 & 0.960, [0.940,0.980] & 0.914, [0.890,0.940] & 0.948, [0.920,0.970] \\
\midrule

\multirow{4}{*}{10000} 
 & VI             & 0.370 & 0.890 & 0.200 & 0.960 & 0.950 & 0.928, [0.900,0.950] & 0.730, [0.00,0.970] & 0.930, [0.880,0.960] \\
 & Gibbs          & 0.560 & -- & -- & -- & 0.950 & 0.936, [0.930,0.950] & 0.728, [0.00,0.970] & 0.928, [0.860,0.960] \\
 & EM (all steps) & 0.960 & 0.920 & 0.900 & 0.960 & 0.940 & 0.938, [0.920,0.950] & 0.922, [0.870,0.960] & 0.930, [0.850,0.970] \\
 & EM (1 step)    & 0.960 & 0.920 & 0.890 & 0.960 & 0.940 & 0.936, [0.920,0.950] & 0.914, [0.870,0.950] & 0.930, [0.850,0.970] \\
\midrule

\multirow{4}{*}{20000} 
 & VI             & 0.610 & 0.920 & 0.570 & 0.960 & 0.970 & 0.934, [0.920,0.950] & 0.812, [0.270,0.950] & 0.938, [0.910,0.980] \\
 & Gibbs          & 0.860 & -- & -- & -- & 0.940 & 0.934, [0.920,0.960] & 0.810, [0.250,0.970] & 0.936, [0.910,0.960] \\
 & EM (all steps) & 0.930 & 0.930 & 0.920 & 0.930 & 0.940 & 0.932, [0.910,0.940] & 0.936, [0.910,0.970] & 0.940, [0.920,0.980] \\
 & EM (1 step)    & 0.910 & 0.920 & 0.920 & 0.930 & 0.940 & 0.932, [0.910,0.940] & 0.932, [0.910,0.970] & 0.940, [0.920,0.980] \\
\bottomrule
\end{tabular}%
}
}
\caption{Scenario 1 -- Frequentist coverage of 95\% HDR over 100 runs for various parameters, sample sizes $n$. Intervals for $\mu_k, \Sigma_k, \psi_k$ denote the min and max coverages over the $K$ variables.}
\label{table:scenario1coverage}
\end{table*}

\add{In terms of point estimation, from \Cref{table:scenario1error}, \Cref{table:scenario2error} and \Cref{table:scenario3error},} it can be seen that our VI approach accurately estimates the parameters $(\mu, \Sigma, \psi, \pi)$, and the error decreases as the sample size $n$ increases for all parameters. Similarly, the true log-likelihood is well estimated by the log posterior predictive density (lppd) and the error decreases as $n$ increases. For comparison with the size of the error, the lppd has an average absolute value of $4.6$, $11.6$ and $7.3$ in the test datasets for scenarios $1,2,3$ respectively. Overall, our approach performs very competitively to the EM algorithm of \cite{Hunt2002} as well as the Gibbs sampler in terms of the point estimation error for all scenarios.

 \add{
 In terms of coverage, from \Cref{table:scenario1coverage}, we observe that the joint (over k) coverages of VI are close to $95\%$ for all parameters and sample sizes except for $\Sigma$, which achieves $0.57$ coverage as sample size is increased to $20000$.
 Similar patterns can be observed for the per k coverages, though the coverage for the individual $\Sigma_k$ reaches about $0.81$. This could be due to $\Sigma$ having higher dimension compared to the other parameters, making it more difficult to estimate. The overall coverage appears to be bottlenecked by the coverage of $\Sigma$. The Gibbs coverage results are similar, also struggling to achieve the target $95\%$ joint coverage, albeit higher than VI on the overall joint coverage but very similar on the other coverages. On the other hand, both versions of EM are able to achieve good joint coverage in this example for all cases. This could be caused by the prior mean of $\mu_k$, which is at the origin, being far away from the true $\mu_k^{*}$ in this example, causing the parameter $\hat{\Phi}_k$ of $\Sigma_k$ to be inflated, whereas EM lacks a prior so does not have this issue. 
 }

\begin{table*}[htbp]
\small
\centering
\add{
\resizebox{\textwidth}{!}{%
\begin{tabular}{ll ccccc}
\toprule
$n$ & Method & $\text{Error}_{\mu} (\times 10^{-3})$ & $\text{Error}_{\Sigma} (\times 10^{-3})$ & $\text{Error}_{\psi} (\times 10^{-3})$ & $\text{Error}_{\pi} (\times 10^{-3})$ & $\text{Error}_{\text{logppd}}$ \\ 
\midrule

\multirow{3}{*}{2500} 
 & VI    & 39.8, [29.9,51.2] & 40.9, [35.3,46.7] & 13.6, [11.1,16.2] & 6.75, [2.91,10.9] & 0.206, [0.183,0.228] \\
 & Gibbs & 39.8, [29.2,51.4] & 40.9, [35.2,46.7] & 13.6, [11.1,16.4] & 6.60, [3.01,10.7] & -- \\
 & EM    & 39.7, [30.0,51.2] & 41.0, [35.0,46.9] & 13.6, [11.2,16.3] & 6.62, [2.94,10.8] & -- \\
\midrule

\multirow{3}{*}{5000} 
 & VI    & 28.7, [21.8,35.3] & 28.9, [24.3,33.4] & 9.55, [8.04,11.1] & 5.04, [2.48,8.17] & 0.142, [0.127,0.161] \\
 & Gibbs & 28.7, [21.5,35.5] & 28.9, [24.5,33.4] & 9.56, [7.98,11.1] & 5.02, [2.46,8.19] & -- \\
 & EM    & 28.7, [21.8,35.3] & 29.0, [24.4,33.4] & 9.55, [8.04,11.1] & 4.98, [2.46,8.22] & -- \\
\midrule

\multirow{3}{*}{10000} 
 & VI    & 19.7, [14.9,26.1] & 20.8, [18.1,23.9] & 6.77, [5.51,8.38] & 3.38, [1.60,5.75] & 0.0997, [0.0891,0.110] \\
 & Gibbs & 19.8, [15.0,26.1] & 20.9, [18.3,24.0] & 6.78, [5.49,8.37] & 3.39, [1.58,5.87] & -- \\
 & EM    & 19.7, [15.0,26.1] & 20.8, [18.0,24.1] & 6.78, [5.51,8.36] & 3.35, [1.56,5.70] & -- \\
\midrule

\multirow{3}{*}{20000} 
 & VI    & 13.5, [9.52,17.5] & 14.7, [12.3,16.9] & 4.67, [3.80,5.57] & 2.65, [1.31,4.16] & 0.0695, [0.0622,0.0765] \\
 & Gibbs & 13.5, [9.25,17.0] & 14.7, [12.3,16.9] & 4.68, [3.84,5.53] & 2.64, [1.30,4.20] & -- \\
 & EM    & 13.5, [9.50,17.5] & 14.7, [12.2,17.0] & 4.67, [3.80,5.56] & 2.64, [1.31,4.12] & -- \\
\bottomrule
\end{tabular}%
}
}
\caption{Scenario 2 -- error metrics averaged over 100 runs, along with 5\% and 95\% quantiles.}
\label{table:scenario2error}
\end{table*}

\begin{table*}[htbp]
\scriptsize
\centering
\add{
\setlength{\tabcolsep}{3pt}
\resizebox{\textwidth}{!}{%
\begin{tabular}{ll cccccccc}
\toprule
$n$ & Method & Overall & $\mu$ & $\Sigma$ & $\psi$ & $\pi$ & $\mu_k$ & $\Sigma_k$ & $\psi_k$ \\ 
\midrule

\multirow{4}{*}{2500} 
 & VI             & 0.230 & 0.810 & 0.570 & 0.470 & 0.900 & 0.906, [0.860,0.930] & 0.836, [0.670,0.910] & 0.784, [0.600,0.900] \\
 & Gibbs          & 0.870 & -- & -- & -- & 0.950 & 0.926, [0.880,0.950] & 0.890, [0.780,0.930] & 0.918, [0.880,0.970] \\
 & EM (all steps) & 0.930 & 0.910 & 0.830 & 0.870 & 0.970 & 0.938, [0.900,0.960] & 0.908, [0.880,0.920] & 0.908, [0.820,0.960] \\
 & EM (1 step)    & 0.410 & 0.830 & 0.570 & 0.0800 & 0.870 & 0.872, [0.800,0.920] & 0.798, [0.660,0.890] & 0.624, [0.300,0.780] \\
\midrule

\multirow{4}{*}{5000} 
 & VI             & 0.380 & 0.850 & 0.730 & 0.540 & 0.840 & 0.908, [0.860,0.960] & 0.898, [0.830,0.930] & 0.810, [0.770,0.840] \\
 & Gibbs          & 0.910 & -- & -- & -- & 0.910 & 0.930, [0.890,0.970] & 0.948, [0.930,0.970] & 0.930, [0.910,0.950] \\
 & EM (all steps) & 0.940 & 0.890 & 0.950 & 0.930 & 0.910 & 0.936, [0.890,0.970] & 0.942, [0.910,0.970] & 0.936, [0.920,0.950] \\
 & EM (1 step)    & 0.540 & 0.740 & 0.660 & 0.200 & 0.820 & 0.882, [0.830,0.940] & 0.858, [0.760,0.920] & 0.622, [0.410,0.740] \\
\midrule

\multirow{4}{*}{10000} 
 & VI             & 0.390 & 0.880 & 0.780 & 0.510 & 0.900 & 0.920, [0.880,0.970] & 0.900, [0.860,0.940] & 0.772, [0.600,0.840] \\
 & Gibbs          & 0.930 & -- & -- & -- & 0.940 & 0.958, [0.920,0.980] & 0.936, [0.910,0.960] & 0.920, [0.900,0.960] \\
 & EM (all steps) & 0.900 & 0.910 & 0.950 & 0.940 & 0.920 & 0.952, [0.920,0.980] & 0.942, [0.910,0.970] & 0.948, [0.900,0.980] \\
 & EM (1 step)    & 0.510 & 0.800 & 0.730 & 0.210 & 0.780 & 0.892, [0.850,0.920] & 0.846, [0.770,0.880] & 0.622, [0.380,0.740] \\
\midrule

\multirow{4}{*}{20000} 
 & VI             & 0.470 & 0.910 & 0.750 & 0.550 & 0.860 & 0.924, [0.910,0.940] & 0.900, [0.850,0.960] & 0.800, [0.720,0.860] \\
 & Gibbs          & 0.960 & -- & -- & -- & 0.920 & 0.952, [0.920,0.990] & 0.926, [0.900,0.960] & 0.936, [0.920,0.960] \\
 & EM (all steps) & 0.940 & 0.960 & 0.940 & 0.950 & 0.940 & 0.950, [0.930,0.970] & 0.920, [0.900,0.940] & 0.948, [0.940,0.960] \\
 & EM (1 step)    & 0.450 & 0.820 & 0.710 & 0.260 & 0.800 & 0.902, [0.870,0.930] & 0.846, [0.780,0.900] & 0.666, [0.460,0.760] \\
\bottomrule
\end{tabular}%
}
}
\caption{Scenario 2 -- Frequentist coverage of 95\% HDR over 100 runs for various parameters, sample sizes $n$. Intervals for $\mu_k, \Sigma_k, \psi_k$ denote the min and max coverages over the $K$ variables.}
\label{table:scenario2coverage}
\end{table*}


\add{In terms of coverage for scenario 2, from \Cref{table:scenario2coverage}, it can be seen VI's coverage for $\mu_k$, $\Sigma_k$ and $\pi$ remain close to or somewhat lower than $0.95$, whereas the coverage for $\psi_k$ is somewhat lower around $0.80$, which lowers its joint and overall coverages. This is likely caused by VI's mean-field assumption $q(z)q(\theta)$ ignoring the dependence between $z$ and $\theta$. In scenario 2, $z$ are driven mainly by the probabilities $\psi_k$, which could be why $\psi_k$ underestimates coverage the most. Nevertheless, VI achieves better coverage than EM (1 step) for almost all cases. Gibbs sampling and the fully refitted EM bootstrap attain close to $0.95$ in most cases. }

\begin{table*}[htbp]
\small
\centering
\add{
\resizebox{\textwidth}{!}{%
\begin{tabular}{ll ccccc}
\toprule
$n$ & Method & $\text{Error}_{\mu} (\times 10^{-3})$ & $\text{Error}_{\Sigma} (\times 10^{-3})$ & $\text{Error}_{\psi} (\times 10^{-3})$ & $\text{Error}_{\pi} (\times 10^{-3})$ & $\text{Error}_{\text{logppd}}$ \\ 
\midrule

\multirow{3}{*}{2500} 
 & VI    & 26.6, [20.6,34.4] & 19.4, [16.4,22.3] & 11.1, [7.80,14.3] & 6.60, [2.93,10.9] & 0.167, [0.148,0.193] \\
 & Gibbs & 26.7, [20.8,33.4] & 19.4, [16.4,22.5] & 11.1, [7.95,14.6] & 6.62, [2.98,10.9] & -- \\
 & EM    & 26.7, [20.5,34.5] & 19.6, [16.4,22.6] & 11.1, [7.91,14.3] & 6.60, [2.86,11.1] & -- \\
\midrule

\multirow{3}{*}{5000} 
 & VI    & 18.5, [14.7,23.9] & 13.7, [11.5,15.5] & 7.93, [5.71,10.5] & 4.87, [2.22,8.09] & 0.116, [0.101,0.130] \\
 & Gibbs & 18.6, [15.0,23.7] & 13.7, [11.5,15.6] & 7.90, [5.72,10.1] & 4.85, [2.23,8.08] & -- \\
 & EM    & 18.6, [14.6,24.0] & 13.8, [11.6,15.8] & 7.95, [5.75,10.6] & 4.86, [2.16,8.10] & -- \\
\midrule

\multirow{3}{*}{10000} 
 & VI    & 13.2, [9.82,17.4] & 9.66, [8.27,11.6] & 5.53, [4.23,6.98] & 3.23, [1.38,5.41] & 0.0807, [0.0710,0.0897] \\
 & Gibbs & 13.3, [9.87,17.6] & 9.67, [8.29,11.7] & 5.54, [4.26,7.07] & 3.23, [1.38,5.40] & -- \\
 & EM    & 13.2, [9.83,17.5] & 9.68, [8.30,11.7] & 5.54, [4.19,7.06] & 3.23, [1.39,5.39] & -- \\
\midrule

\multirow{3}{*}{20000} 
 & VI    & 9.50, [6.68,12.0] & 6.82, [5.62,7.89] & 3.96, [2.85,4.96] & 2.24, [1.01,3.71] & 0.0574, [0.0505,0.0639] \\
 & Gibbs & 9.53, [6.69,12.0] & 6.84, [5.64,7.84] & 3.97, [2.91,5.00] & 2.23, [0.966,3.63] & -- \\
 & EM    & 9.50, [6.67,12.0] & 6.84, [5.66,7.93] & 3.96, [2.85,4.95] & 2.24, [0.998,3.73] & -- \\
\bottomrule
\end{tabular}%
}
}
\caption{Scenario 3 -- error metrics averaged over 100 runs, along with 5\% and 95\% quantiles.}
\label{table:scenario3error}
\end{table*}

\begin{table*}[htbp]
\scriptsize
\add{
\centering
\setlength{\tabcolsep}{3pt}
\resizebox{\textwidth}{!}{%
\begin{tabular}{ll cccccccc}
\toprule
$n$ & Method & Overall & $\mu$ & $\Sigma$ & $\psi$ & $\pi$ & $\mu_k$ & $\Sigma_k$ & $\psi_k$ \\ 
\midrule

\multirow{4}{*}{2500} 
 & VI             & 0.480 & 0.880 & 0.490 & 0.850 & 0.900 & 0.914, [0.870,0.940] & 0.810, [0.650,0.890] & 0.914, [0.850,0.940] \\
 & Gibbs          & 0.890 & -- & -- & -- & 0.930 & 0.948, [0.910,0.970] & 0.912, [0.880,0.940] & 0.938, [0.910,0.970] \\
 & EM (all steps) & 0.900 & 0.970 & 0.880 & 0.910 & 0.930 & 0.956, [0.930,0.970] & 0.916, [0.890,0.940] & 0.934, [0.910,0.950] \\
 & EM (1 step)    & 0.380 & 0.830 & 0.490 & 0.690 & 0.870 & 0.884, [0.840,0.920] & 0.718, [0.510,0.790] & 0.858, [0.760,0.900] \\
\midrule

\multirow{4}{*}{5000} 
 & VI             & 0.520 & 0.900 & 0.580 & 0.860 & 0.920 & 0.916, [0.890,0.950] & 0.842, [0.730,0.900] & 0.912, [0.880,0.930] \\
 & Gibbs          & 0.890 & -- & -- & -- & 0.940 & 0.962, [0.950,0.980] & 0.914, [0.870,0.950] & 0.940, [0.920,0.950] \\
 & EM (all steps) & 0.950 & 0.940 & 0.910 & 0.950 & 0.950 & 0.956, [0.930,0.970] & 0.916, [0.890,0.930] & 0.942, [0.930,0.970] \\
 & EM (1 step)    & 0.460 & 0.810 & 0.520 & 0.750 & 0.920 & 0.904, [0.880,0.940] & 0.756, [0.630,0.830] & 0.850, [0.800,0.870] \\
\midrule

\multirow{4}{*}{10000} 
 & VI             & 0.680 & 0.820 & 0.740 & 0.900 & 0.900 & 0.904, [0.850,0.930] & 0.862, [0.780,0.900] & 0.926, [0.880,0.950] \\
 & Gibbs          & 0.920 & -- & -- & -- & 0.930 & 0.936, [0.910,0.970] & 0.954, [0.940,0.960] & 0.960, [0.910,0.980] \\
 & EM (all steps) & 0.980 & 0.920 & 0.930 & 0.970 & 0.980 & 0.944, [0.920,0.970] & 0.946, [0.930,0.960] & 0.950, [0.920,0.970] \\
 & EM (1 step)    & 0.630 & 0.760 & 0.630 & 0.750 & 0.900 & 0.880, [0.810,0.950] & 0.788, [0.640,0.850] & 0.880, [0.820,0.930] \\
\midrule

\multirow{4}{*}{20000} 
 & VI             & 0.640 & 0.870 & 0.730 & 0.860 & 0.950 & 0.912, [0.880,0.960] & 0.864, [0.780,0.910] & 0.904, [0.870,0.940] \\
 & Gibbs          & 0.950 & -- & -- & -- & 0.970 & 0.940, [0.900,0.970] & 0.942, [0.920,0.960] & 0.950, [0.910,0.970] \\
 & EM (all steps) & 0.950 & 0.980 & 0.950 & 0.970 & 0.980 & 0.956, [0.900,0.980] & 0.942, [0.910,0.980] & 0.952, [0.950,0.960] \\
 & EM (1 step)    & 0.450 & 0.780 & 0.510 & 0.790 & 0.950 & 0.882, [0.820,0.910] & 0.770, [0.620,0.860] & 0.860, [0.810,0.920] \\
\bottomrule
\end{tabular}%
}
}
\caption{Scenario 3 -- Frequentist coverage of 95\% HDR over 100 runs for various parameters, sample sizes $n$. Intervals for $\mu_k, \Sigma_k, \psi_k$ denote the min and max coverages over the $K$ variables.}
\label{table:scenario3coverage}
\end{table*}


\add{Scenario 3 sits between the first two. VI's coverages for $\mu_k$, $\psi_k$ and $\pi$ are close to 0.95, but those for $\Sigma_k$ are somewhat lower (0.81--0.86), and this appears to account for most of VI's lower joint and overall coverage. This pattern is counter to that of scenario 2, suggesting that the mixture allocations are determined more by the continuous variables than the categorical. As in Scenario 2, VI is able to achieve better coverage than EM (1 step) in almost all cases. Gibbs sampling and the fully refitted EM bootstrap attain close to $0.95$ in most cases. }

\add{\Cref{table:runtimecomparison} reports each method's computation time relative to VI. VI is the fastest way to obtain both point estimates and a measure of uncertainty. A single EM fit costs about the same as VI (0.66-1.03 times), but provides no uncertainty quantification on its own. Gibbs sampling needs 55-369 times the VI runtime to reach VI's estimation precision, although this ratio falls with $n$ (e.g. from 369 to 142 in Scenario 1). Bootstrapping EM is also costly: the one-step bootstrap takes 15-50 times the VI runtime, and the fully refitted bootstrap, which is needed for close to nominal coverage in Scenarios 2 and 3, takes 89-131 times. In the Supplement, we also include a table of the number of iterations each method took and an extended discussion. 
} 

\begin{table*}[htbp]
\small
\centering
\add{
\begin{tabular}{ll ccc}
\toprule
$n$ & Method & Scenario 1 & Scenario 2 & Scenario 3 \\ 
\midrule

\multirow{5}{*}{2500} 
 & VI (Baseline)& 1.00               & 1.00               & 1.00               \\
 & Gibbs        & 369, [115, 1040] & 167, [14.5, 526] & 313, [37.1, 586] \\
 & EM           & 1.01, [0.801, 1.21] & 0.684, [0.607, 0.792] & 0.676, [0.552, 0.799] \\
 & EM (1 step)  & 45.4, [37.6, 53.0] & 14.9, [12.4, 17.4] & 18.1, [14.0, 21.3] \\
 & EM (all steps)& 114, [92.2, 138] & 88.6, [71.7, 107] & 107, [82.5, 128] \\
\midrule

\multirow{5}{*}{5000} 
 & VI (Baseline)& 1.00               & 1.00               & 1.00               \\
 & Gibbs        & 294, [35.8, 863] & 167, [16.6, 435] & 170, [11.1, 385] \\
 & EM           & 1.03, [0.861, 1.21] & 0.836, [0.725, 0.957] & 0.676, [0.579, 0.804] \\
 & EM (1 step)  & 47.0, [39.6, 54.8] & 17.7, [15.0, 20.7] & 18.6, [14.9, 22.0] \\
 & EM (all steps)& 122, [103, 142] & 104, [87.2, 123] & 112, [92.4, 130] \\
\midrule

\multirow{5}{*}{10000} 
 & VI (Baseline)& 1.00               & 1.00               & 1.00               \\
 & Gibbs        & 187, [16.3, 622] & 107, [14.5, 295] & 117, [14.2, 272] \\
 & EM           & 1.00, [0.827, 1.23] & 0.935, [0.815, 1.05] & 0.664, [0.555, 0.750] \\
 & EM (1 step)  & 48.7, [41.9, 57.8] & 20.0, [17.1, 22.8] & 18.8, [15.8, 21.9] \\
 & EM (all steps)& 127, [107, 154] & 120, [104, 139] & 114, [98.7, 128] \\
\midrule

\multirow{5}{*}{20000} 
 & VI (Baseline)& 1.00               & 1.00               & 1.00               \\
 & Gibbs        & 142, [15.1, 391] & 64.4, [9.55, 214] & 55.4, [9.22, 174] \\
 & EM           & 0.989, [0.788, 1.23] & 1.02, [0.911, 1.15] & 0.666, [0.551, 0.798] \\
 & EM (1 step)  & 50.1, [40.0, 59.6] & 21.5, [18.5, 24.6] & 18.2, [14.7, 21.1] \\
 & EM (all steps)& 130, [105, 160] & 131, [115, 149] & 115, [96.3, 136] \\
\bottomrule
\end{tabular}
}
\caption{Computation time relative to VI}
\label{table:runtimecomparison}
\end{table*}

\subsection{Application: Identifying groups of individuals with comparable health phenotypes based on survey responses}

We apply our VI algorithm on the National Health and Nutrition Examination Survey (NHANES) dataset, an ongoing program conducted by the National Center for Health Statistics (NCHS) to monitor the health and nutritional status of the U.S. population. For this analysis we take a subset of the NHANES dataset from $1988$ to $2018$, for middle aged males aged between $40$ and $59$, as both sex and age group have strong epidemiological effects that can strongly drive cluster make up, resulting in less epidemiologically interesting clusters. We remove rows containing missing data, and details of the data cleaning procedure can be found in \cite{Lhoste2024}. Each row of the dataset represents a participant of the survey and contains a range of continuous and categorical risk factors including anthropometric measurements, blood pressure and heart rate, blood sugar level, cholesterol level and kidney function measurements. Specifically, for anthropometric measurements, we have: BMI, defined as a participant's weight divided by height squared ($\mathrm{kg} \mathrm{m}^{-2}$); height (cm) ; WHtR, defined as waist circumference divided by height, which is used as a measure of abdominal obesity. For blood pressure and heart rate, we have systolic blood pressure (SBP), diastolic blood pressures (DBP) and pulse rate, which is defined as the resting heart rate over one minute. For cholesterol level, we have both HDL and non-HDL cholesterol levels, where non-HDL is associated with higher risk of ischemic heart disease and stroke. Lastly, we have HbA1c which is a measurement of glucose levels in the blood in recent weeks and glomerular filtration rate (eGFR) as a measure of kidney function, where higher eGFR indicates better kidney function. For categorical variables, we include the participant's smoking status, classified as current, former, or never smoker. 


We investigate whether our proposed approach can produce epidemiologically heterogeneous and meaningful clusters that could then be used for planning prevention and clinical care, and to monitor population health over time. For the model likelihood, we use \Cref{eq: likelihoodVIzintegratedout}, and we make the same prior assumptions and hyperparameter choice as in \Cref{subsec: mixeddatamixturemodel}. \add{For the number of components $K$, we tested $3 \le K \le 12$, and the maximum ELBO was attained at $K=4$. However, from an epidemiological point of view, it is important that the clustering describes finer subgroups within the population, which in addition allows more suitable comparison with previous works. Therefore we have kept the choice of $K=10$ in line with previous analysis of the NHANES dataset \cite{Lhoste2024}}.

To examine the features of the clusters produced by our method, we use the variational posterior predictive marginal distributions, as described earlier in \Cref{eq: posteriorpredmarginalx}, because they are suitable for visualisation in high-dimensional settings and effectively summarize the composition of each cluster while accounting for posterior uncertainty. For each cluster, the variational posterior predictive marginal distributions are represented using radar plots for the continuous data in \Cref{fig: radarplotUSdata}, bar plots for the categorical data in \Cref{fig: barplotUSdata}. 

\begin{figure*}[h]{
\centering
\includegraphics[width = 1.00\textwidth]{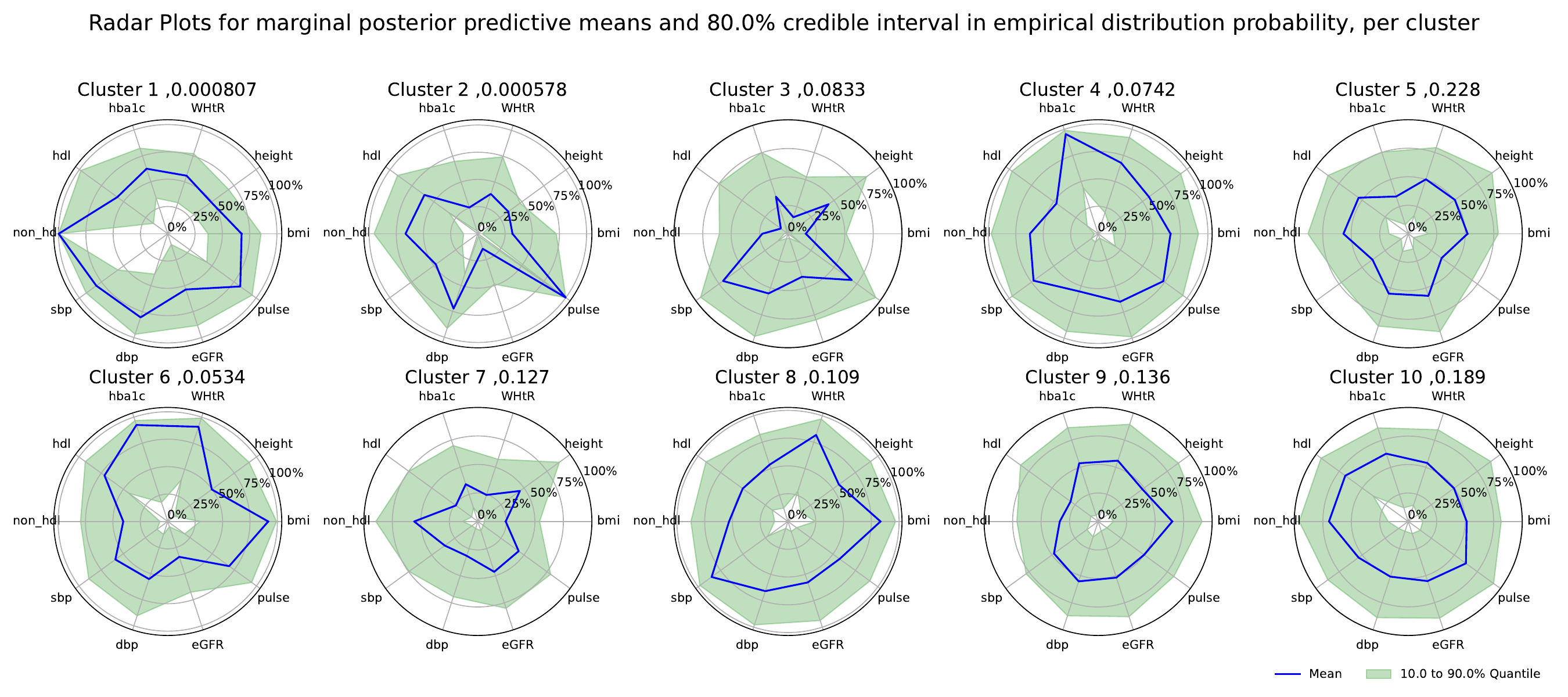}
}
\caption{Radar plots of the marginal posterior predictive distribution of each continuous variable, conditioned on cluster component. The posterior mean of each $\pi_k$ is also added atop each radar plot. The concentric circles and quantiles in each radar plot represents the quantiles of the empirical distribution of the whole data. For the sake of visual consistency, the distribution of height, eGFR and hdl are reversed so that higher values in the radar plot for all continuous variables correspond to higher risk of chronic diseases.}
\label{fig: radarplotUSdata}
\end{figure*}

\begin{figure*}[h]{
\centering
\includegraphics[width = 1.00\textwidth]{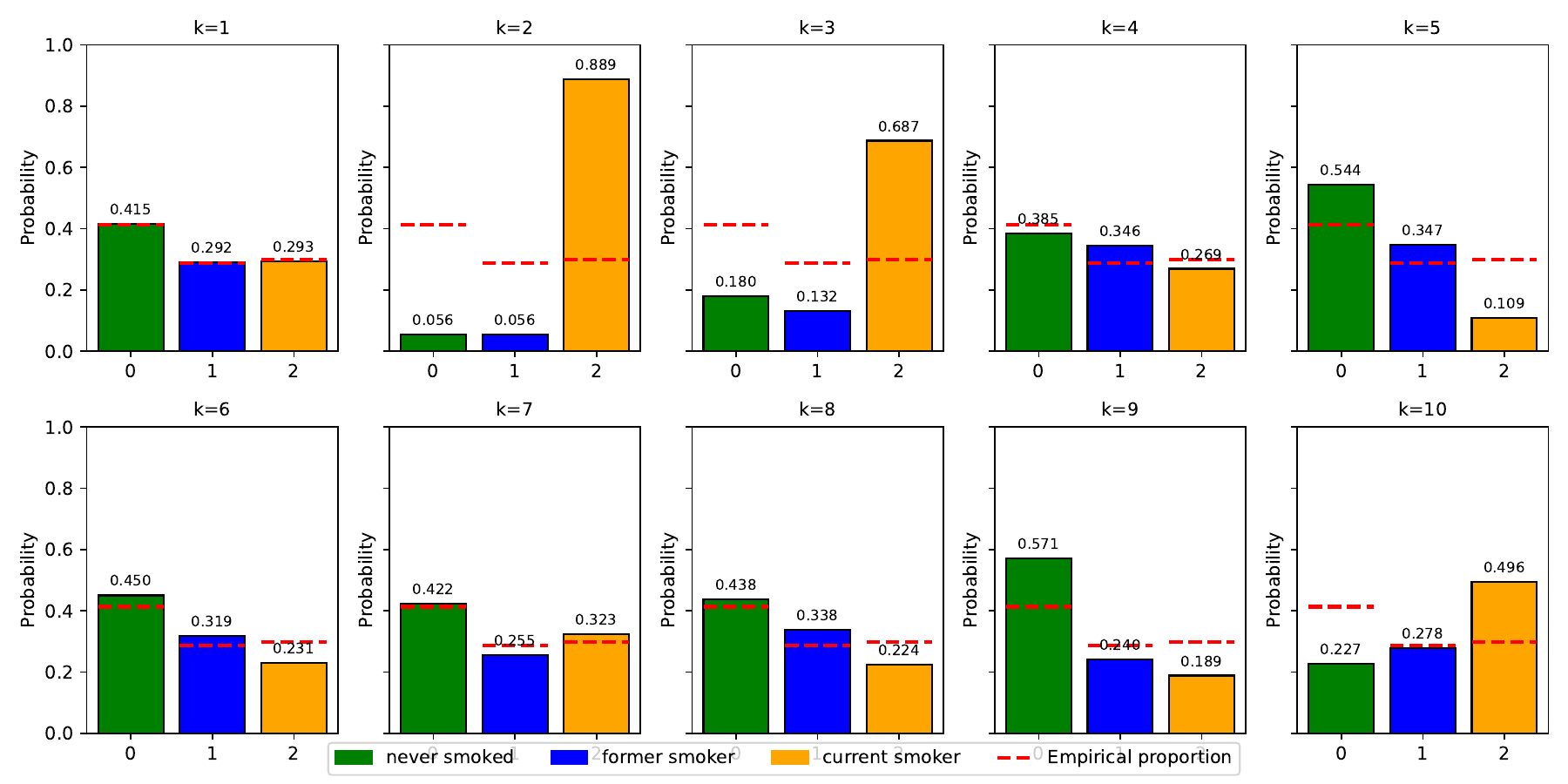}
}
\caption{Barplot of the marginal posterior predictive distribution of smoking status, conditioned on cluster component. }
\label{fig: barplotUSdata}
\end{figure*}

In \Cref{fig: radarplotUSdata}, one can see that clusters $1$ and $2$ have very small posterior mixture probabilities, suggesting that CAVI has effectively settled on fewer clusters than originally specified,  indicating some ability to correct for the over specification of $K$ even on real data. Cluster 5 is the largest cluster (in the sense of posterior mean of $\pi_k$), and is representative of a median adult, as the posterior predictive means of most continuous variables are close to the empirical median albeit with somewhat lower pulse, HbA1c and SBP, and its posterior predictive distribution for smoking status also does not differ greatly from the median. Clusters 9 and 10 also have posterior predictive means close to the empirical median for most continuous features, but distinguishes from cluster 5 and each other by cluster 9 having somewhat elevated BMI and healthier cholesterol levels (HDL and non-HDL) with mostly non smokers, and cluster 10 having elevated cholesterol levels and predominantly current smokers. 

Clusters 4, 6 and 8 have posterior predictive means at close to or exceeding 75\% quantile for multiple continuous features, and few if any below 50\% quantile. Clusters 4 and 6 are both characterized by very high posterior means for HbA1c, suggesting they are characterized by diabetic people. However, they differ in anthropometric measures and, to some extent, lipid profiles: Cluster 6 shows markedly elevated posterior means for WHtR and BMI, along with relatively high HDL, whereas cluster 4 exhibits posterior means for WHtR and BMI that are closer to the population median but still somewhat elevated, and close to median level of HDL. Cluster 8 is characterized by very high posterior means for systolic and diastolic blood pressure, BMI, and WHtR, but near-median levels of HbA1c and cholesterol, distinguishing it from clusters 4 and 6. Smoking distributions for all three clusters are similar. 

Clusters 7 and 3 are characterized by below-median values for most continuous variables, reflecting generally healthier profiles. Cluster 3 differs from cluster 7 by having higher posterior means for systolic and diastolic blood pressure and pulse. Individuals in cluster 3 are also more likely to be current smokers, whereas smoking status in cluster 7 is more evenly distributed, with no clear dominance of one category.

Several of our clusters share similarities with those reported by \cite{Lhoste2024} which conducted clustering analysis using K-means on the continuous part of the data only. Cluster 7, characterised by near-optimal levels across all risk factors, corresponds to the `low risk' phenotype and with similar prevalence (13\% in \cite{Lhoste2024} vs. a posterior mean weight of 0.127 here). Cluster 5 (posterior mean weight 0.228) represents individuals with near-median risk factors and appears analogous to the combined `mid-risk tall' and `mid-risk short' phenotypes in \cite{Lhoste2024} (jointly 28\%). Cluster 8 matches the `high blood pressure' phenotype, though with slightly less extreme blood pressure levels and somewhat higher prevalence. Cluster 4, defined by very high HbA1c, aligns with the `severe hyperglycemia' phenotype but with higher prevalence (posterior mean weight 0.07 vs. 3\% in \cite{Lhoste2024}). Cluster 6 corresponds to the `severe obesity' phenotype with high BMI and WHtR and elevated HbA1c; its near-median blood pressure and non-HDL values likely reflect higher rates of antihypertensive and statin use commonly observed in obese adults. Cluster 3 resembles the `low BMI, high HDL' phenotype, but with the additional insight of identifying high smoking prevalence.

%% file: conclusion.tex
\section{Discussion} \label{sec: discussion}


We have developed a scalable, coordinate ascent variational inference algorithm for mixed data mixture models that allows both the continuous and categorical variables to influence the make up of each component or cluster. We have conducted a simulation study under scenarios where the heterogeneity in the mixture components is driven by the continuous data, the categorical data, or a combination of both respectively. Through these simulation studies, we demonstrated that our proposed approach is able to conduct good density estimation, and estimate the underlying parameters well and competitively with Gibbs sampler and an EM algorithm. \add{The coverage of the VI credible regions is lower than that of the Gibbs sampler and the fully refitted EM bootstrap. This reflects a known property of the mean-field assumption: the factorisation $q(z)q(\Theta)$ between the cluster allocations and the parameters, which enables the closed-form CAVI updates and the frequentist asymptotic analysis of \Cref{sec: theory}, ignores the posterior dependence between $z$ and $\Theta$, and therefore understates the posterior uncertainty when the cluster allocations are uncertain, such as when the components overlap heavily in the continuous variables. The resulting loss of coverage on $\pi$ and per component $\mu_k$, $\Sigma_k$, $\psi_k$ is moderate: the coverage of VI is within about $0.05$ of that of the Gibbs sampler in 38 of the 48 settings in \Cref{table:scenario1coverage,table:scenario2coverage,table:scenario3coverage}, and very competitive with Gibbs in Scenario 1. In Scenarios 2 and 3 it is at least as high as that of EM with a one-step bootstrap for $\mu_k$, $\Sigma_k$, $\psi_k$ and $\pi$. Across the 3 scenarios, the one-step bootstrap takes 15 to 50 times as long as VI, while the Gibbs sampler and the fully refitted EM bootstrap require about two orders of magnitude the runtime of VI, making VI a very computationally efficient alternative. Moreover, the posterior predictive distribution is largely estimated accurately as can be seen in \Cref{table:scenario1error,table:scenario2error,table:scenario3error}, which is what we advocate to represent clusters in practical application. }

\add{Another practical advantage of our CAVI approach is its support for streaming data updates. Because the mean-field variational posterior retains conjugate distributions, its updated parameters take the exact functional form of the prior hyperparameters. When new data arrives, these posterior hyperparameters can be plugged directly in as the new prior hyperparameters, allowing the model to update sequentially without re-processing past observations or retaining historical data in memory. In contrast, the Gibbs sampler represents the posterior as sample draws, making direct hyperparameter updates impossible. Updating MCMC on streaming data without re-fitting the full combined dataset requires additional methods such as Sequential Monte Carlo or online MCMC methods, which come with their own computational challenges. }

We demonstrated the practical utility of our approach by applying it on a challenging dataset of risk factor data of US participants in the NHANES study. We showed our method can produce epidemiologically interpretable clusters, which is crucial when clustering epidemiological datasets. We visualised the heterogeneity between the different clusters using marginal posterior predictive distributions, which take into account posterior uncertainty, something that is often missed. However, while the medians of each continuous marginal posterior predictive distribution are distinguishable across clusters, the uncertainty quantification becomes very conservative for large probability credible intervals, causing the uncertainty quantification across different clusters to be less distinguishable.

\add{Theoretically, our analysis is of the CAVI iterations themselves rather than the idealised optimal variational posterior, so they apply to the estimates the algorithm actually returns, provided it is initialised in a fixed neighbourhood of $\Theta^*$ (\Cref{thm: maintheorem}). A central step of the proof is that the CAVI update and the EM update share the same M-step and differ only in the responsibilities: the CAVI responsibilities $r_{ik}$, which carry the prior hyperparameters and variational expectations, differ from the EM responsibilities by $O(1/n)$, together with their derivatives. CAVI therefore behaves asymptotically like EM: its fixed point lies within $O(1/n)$ of the consistent local MLE (\Cref{thm: MLEcorollary}), which is consistent with the near identical point estimation errors of the two methods in our simulations. It also inherits the $n^{-1/2}$ rate of the local MLE, which holds by Theorem 3.1 of \citet{Redner1984} since our model satisfies their Conditions 1 and 2 (\Cref{ass: RW}), and this gives the contraction rate of the variational posterior (\Cref{thm: posteriorconsistencycorollary}).}

One future research direction could include developing VI approaches that enforce more separation between the clusters in continuous space, which could produce more distinguishable uncertainty quantification across different clusters. Another research direction could be to consider infinite mixtures using Dirichlet processes, which is more suitable for data where the number of clusters is expected to increase with increasing sample size. In this paper, we focused on scenarios in which the number of categorical values in the data remains small relative to the sample size. In the case where the number of categorical values is large compared to the sample size, a possible future research direction is to investigate whether more parsimonious approaches for modelling the categorical component, such as the sparse log-linear formulation of \cite{Aliverti2022} and the Hamming-distance-based approach of \cite{Argiento2025}, can be incorporated into a variational inference mixed data framework.


%% file: appendix.tex
\newpage

\section*{Supplement to `Bayesian Variational Inference for Mixed Data Mixture Models'}

This supplement to the paper by Wang, Bennett, Lhoste and Filippi consists of several parts. \Cref{subsec: cavi} contains the derivation for the variational posterior update equations of the CAVI algorithm in the main text of the paper. \Cref{subsec: ELBOequations} contains the ELBO calculations for the CAVI algorithm. \Cref{sec: gibbsupdate} contains the Gibbs sampler updates for the model considered in the paper. The proofs of \thmmaintheorem, \thmposteriorconsistencycorollary, \Cref{main-ass: RW} and the corollaries in the main text of the paper are given in \Cref{sec: proofssection}, and auxiliary calculations needed for said proof in \Cref{subsec: gradientcalcs}. \Cref{subsec: simulatedparams} contains detail of the parameters used for the three scenarios in the Simulation study under the Experimental Results section of the main text of the paper. \add{\Cref{sec: chooseK} studies selecting the number of mixture components $K$ with the ELBO, \Cref{sec: sensitivity} examines sensitivity to the prior hyperparameters, and \Cref{sec: iterations} reports the number of iterations run by each algorithm.}

\section{Coordinate Ascent Variational Inference update equations} \label{subsec: cavi}

In this section we provide the CAVI (Coordinate Ascent Variational Inference) update equations for our mixed data mixture model. 

Recall the prior distributions are as follows:

\medskip

$\mu_{k} \sim \mathcal{N}(m_{k}, \beta^{-1} \Lambda_{k}^{-1} )$.

$\Lambda_{k} \sim \mathcal{W}(\nu, \Phi^{-1})$ or equivalently $\Sigma_{k} \sim \mathcal{IW}(\nu, \Phi)$, where $\Sigma_{k}=(\Lambda_{k})^{-1}$

$\pi=(\pi_{1},\pi_{2},\dots,\pi_{K}) \sim Dir(\alpha,\dots,\alpha)$

$\psi_{k,j}=(\psi_{k,j,1},\psi_{k,j,2},\dots,\psi_{k,j,d_j}) \sim Dir(\eta_j,\dots,\eta_j)$.

\medskip

The model (likelihood) terms can be written as follows:

$x|c,z,\mu,\Lambda \sim \prod_{i=1}^{n} \mathcal{N}(x_i|\mu_{z_i}, \Lambda_{z_i}^{-1} )= \prod_{i=1}^{n}\prod_{k=1}^{K}\mathcal{N}(x_i|\mu_{k}, \Lambda_{k}^{-1} )^{z_{ik}}$

$z|\pi = \prod_{i=1}^{n}\pi_{z_i}=\prod_{i=1}^{n}\prod_{k=1}^{K}\pi_{k}^{z_{ik}}$

$c|\psi,z = \prod_{i=1}^{n}\prod_{j=1}^{p} \psi_{z_i,j,c_{ij}}=\prod_{i=1}^{n}\prod_{j=1}^{p}\prod_{k=1}^{K}\psi_{k,j,c_{ij}}^{z_{ik}}$

where $z_{ik}=1$ if $z_i=k$, $z_{ik}=0$ otherwise.

The likelihood with $z$ marginalised out (for a single datapoint) is:

\begin{align} \label{eq: likelihoodVIappendix} 
p(x_i, c_i | \mu, \Lambda, \psi, \pi) &= \sum_{k=1}^{K} \pi_{k} \mathcal{N}(x_i|\mu_{k}, \Lambda_{k}^{-1} ) \prod_{j=1}^{p} \psi_{k,j,c_{ij}}
\end{align}  

\medskip

This leads to a joint distribution of the form:

\begin{align} \label{eq: jointdistributionVI} 
p(x, c, z, \mu, \Lambda, \psi, \pi) &= p(x|c,z,\mu,\Lambda)p(c|\psi,z)p(z|\pi)p(\mu|\Lambda)p(\Lambda)p(\psi)p(\pi) \nonumber \\
&= p(x|c,z,\mu,\Lambda)p(\mu|\Lambda)p(\Lambda)p(c|\psi,z)p(\psi)p(z|\pi)p(\pi) 
\end{align}  

We assume the variational posterior is of the form $q(\mu, \Lambda, \psi, \pi, z)=q(\mu, \Lambda, \psi, \pi) q(z)$. The exact distributional families of the variational posterior need not to be specified and will come out during the derivation. To derive its form we use CAVI with two blocks, so we alternate between updating $q(\mu, \Lambda, \psi, \pi)$ and $q(z)$.

\begin{align} \label{eq: variationalposteriorappendix} 
& q(\mu, \Lambda, \psi, \pi) \propto \mathrm{exp}(\mathbb{E}_{-\mu, -\Lambda, -\psi, -\pi} \;\mathrm{ln} p(x, c, z, \mu, \Lambda, \psi, \pi)\; ) \\
& \propto \mathrm{exp}(\mathbb{E}_{z} [\mathrm{ln}( p(x|c,z,\mu,\Lambda)p(\mu|\Lambda)p(\Lambda)p(c|\psi,z)p(\psi)p(z|\pi)p(\pi)   ) ] ) \nonumber \\
& \propto \underbrace{\mathrm{exp}(\mathbb{E}_{z} [\mathrm{ln}(p(x|c,z,\mu,\Lambda)p(\mu|\Lambda)p(\Lambda))]) }_{\propto q(\mu,\Lambda)}\underbrace{\mathrm{exp}(\mathbb{E}_{z}[\mathrm{ln}(p(z|\pi)p(\pi))] ) }_{\propto q(\pi)}\underbrace{\mathrm{exp}(\mathbb{E}_{z} [\mathrm{ln}(p(c|\psi,z)p(\psi))]  ) }_{\propto q(\psi)}  \nonumber 
\end{align}  

From this it's clear $q(\mu, \Lambda, \psi, \pi)=q(\mu, \Lambda)q(\psi)q(\pi)$, each term in this product can therefore be derived separately as follows:

For $q(\pi)$:

\begin{align} \label{eq: VIupdateforpi} 
q(\pi) & \propto  \mathrm{exp}(\mathbb{E}_{z}[\mathrm{ln}p(z|\pi)+\mathrm{ln}p(\pi)]) \nonumber \\
& \propto \mathrm{exp}(\mathbb{E}_{z}[\mathrm{ln}(\prod_{i=1}^{n}\prod_{k=1}^{K}\pi_{k}^{z_{ik}} )+\mathrm{ln} (Dir(\pi| \alpha,\dots,\alpha) )]) \nonumber \\
& \propto \mathrm{exp}(\mathbb{E}_{z}[\sum_{i=1}^{n}\sum_{k=1}^{K} z_{ik} \mathrm{ln}\pi_{k}+\sum_{k=1}^{K} (\alpha-1)\mathrm{ln} ( \pi_k )]) \nonumber \\
& \propto \mathrm{exp}(\sum_{k=1}^{K} (\alpha-1+\sum_{i=1}^{n}r_{ik} ) \mathrm{ln}\pi_{k}) \nonumber \\
& \propto \prod_{k=1}^{K} \pi_{k}^{(\alpha-1+\sum_{i=1}^{n}r_{ik} )}
\end{align}  

which by observation is a Dirichlet distribution $Dir(\alpha+\sum_{i=1}^{n}r_{i1},\dots,\alpha+\sum_{i=1}^{n}r_{iK} )$, where $r_{ik}=\mathbb{E}_{z_i}[z_{ik}]=q(z_i=k)$, the posterior probability of the ith data being in cluster $k$ under the variational posterior probability, which will be obtained in the CAVI update for $q(z)$ later on. 


For $q(\mu, \Lambda)$:

\begin{align} \label{eq: VIupdateformulambdapart1}
q(\mu, \Lambda) & \propto \mathrm{exp}(\mathbb{E}_{z} [\mathrm{ln}p(x|c,z,\mu,\Lambda)+\mathrm{ln}p(\mu|\Lambda)+\mathrm{ln}p(\Lambda)]) \nonumber \\
& \propto \mathrm{exp}(\mathbb{E}_{z} [\mathrm{ln}(\prod_{i=1}^{n}\prod_{k=1}^{K}\mathcal{N}(x_i|\mu_{k}, \Lambda_{k}^{-1} )^{z_{ik}} )]+ \mathrm{ln}(\prod_{k=1}^{K}  \mathcal{N}(\mu_{k}|m_{k}, \beta^{-1} \Lambda_{k}^{-1} ))+ \mathrm{ln}(\prod_{k=1}^{K} \mathcal{W}(\Lambda_{k}|\nu, \Phi^{-1}) ) ) \nonumber \\
& \propto \mathrm{exp}(\sum_{i=1}^{n}\sum_{k=1}^{K} r_{ik} \mathrm{ln}(\mathcal{N}(x_i|\mu_{k}, \Lambda_{k}^{-1} ) )+ \sum_{k=1}^{K}\mathrm{ln}(\mathcal{N}(\mu_{k}| m_{k}, \beta^{-1} \Lambda_{k}^{-1} ))+ \sum_{k=1}^{K} \mathrm{ln}( \mathcal{W}(\Lambda_{k,g}|\nu, \Phi^{-1}) ) ) \nonumber \\
& \propto \prod_{k=1}^{K}  \mathrm{exp}(\sum_{i=1}^{n} r_{ik} \mathrm{ln}(\mathcal{N}(x_i|\mu_{k}, \Lambda_{k}^{-1} ) )+\mathrm{ln}( \mathcal{N}(\mu_{k}| m_{k}, \beta^{-1} \Lambda_{k}^{-1} ))+ \mathrm{ln}( \mathcal{W}(\Lambda_{k}|\nu, \Phi^{-1}) ) )
\end{align}

From here we first derive a conjugate Gaussian posterior for each $\mu_{k}|\Lambda_{k}$, and subsequently the leftover terms can be `absorbed' into the Wishart prior to form a conjugate Wishart posterior for $\Lambda_{k}$. For each $\mu_{k}$, the relevant terms are:

\begin{align} \label{eq: VIupdateformulambdapart2}
q(\mu_{k}|\Lambda_{k}) & \propto \mathrm{exp}\Bigl\{ \sum_{i=1}^{n} r_{ik} \mathrm{ln}(\mathcal{N}(x_i|\mu_{k}, \Lambda_{k}^{-1} ) )+ \mathrm{ln}( \mathcal{N}(\mu_{k}|m_{k}, \beta^{-1} \Lambda_{k}^{-1} )) \Bigl\} \nonumber \\
& \propto \mathrm{exp}\Bigl\{ \sum_{i=1}^{n} r_{ik} [\frac{1}{2} \mathrm{ln}(det(\Lambda_{k}))-\frac{1}{2}(x_i-\mu_{k})^T \Lambda_{k} (x_i-\mu_{k}) ]+ \frac{1}{2}\mathrm{ln}(det(\beta \Lambda_{k})) \nonumber \\ & -\frac{1}{2}(\mu_{k}-m_{k})^T \beta \Lambda_{k} (\mu_{k}-m_{k}) \Bigl\} \nonumber \\
& \propto \mathrm{exp}\Bigl\{ -\frac{1}{2} [\mu_{k}^T(\beta + \sum_{i=1}^{n} r_{ik} ) \Lambda_{k}  \mu_{k}+(-\beta m_{k}^T-\sum_{i=1}^{n} r_{ik}x_i^T) \Lambda_{k} \mu_{k}) \nonumber \\ & +\mu_{k}^T \Lambda_{k} (-\beta m_{k}-\sum_{i=1}^{n} r_{ik}x_i)+m_{k}^T\beta \Lambda_{k} m_{k} + x_i^T \sum_{i=1}^{n} r_{ik}  \Lambda_{k}  x_i]  \nonumber \\ & + \sum_{i=1}^{n} r_{ik} (\frac{1}{2} \mathrm{ln}(det(\Lambda_{k})) + \frac{1}{2}\mathrm{ln}(det(\beta \Lambda_{k})) \Bigl\} \nonumber  \\
& \propto \mathrm{exp}\Bigl\{( \frac{1}{2}\mathrm{ln}(det(\beta \Lambda_{k}))-\frac{1}{2}(\mu_{k}-\hat{m}_{k})^T (\beta + \sum_{i=1}^{n} r_{ik} ) \Lambda_{k} (\mu_{k}-\hat{m}_{k}) \nonumber \\ & -  \frac{1}{2}[-\hat{m}_{k}^T (\beta + \sum_{i=1}^{n} r_{ik} ) \Lambda_{k} \hat{m}_{k}+m_{k}^T\beta \Lambda_{k} m_{k}+x_i^T \sum_{i=1}^{n} r_{ik}  \Lambda_{k}  x_i  - \sum_{i=1}^{n} r_{ik} \mathrm{ln}(det(\Lambda_{k})) ]  \Bigl\}
\end{align}

Where $\hat{m}_{k}=\frac{\beta m_{k} + \sum_{j=1}^{n} r_{jk} x_j}{\beta + \sum_{j=1}^{n} r_{jk}}$. The first two terms in the last expression of the derivation above forms a conjugate posterior 
$\mu_{k}|\Lambda_{k} \sim \mathcal{N}(\mu_{k}|\hat{m}_{k}, \hat{\beta}_{k}^{-1} \Lambda_{k}^{-1} )$, where $\hat{\beta}_{k}^{-1}=(\beta + \sum_{j=1}^{n} r_{jk} )^{-1}$ (the determinant term is correct up to a constant that does not depend on $\Lambda_{k}$). The second line in the last expression contains the leftover terms that need to be `absorbed' by the Wishart prior, which we will see in the derivation for $q(\Lambda_{k})$ as follows:

For $q(\Lambda_{k})$:

\begin{align} \label{eq: VIupdateforlambda}
q(\Lambda_{k}) & \propto \mathrm{exp}\Bigl\{ \mathrm{ln}(\mathcal{W}(\Lambda_{k}|\nu, \Phi^{-1})) - \frac{1}{2}[-\hat{m}_{k}^T (\beta + \sum_{i=1}^{n} r_{ik} ) \Lambda_{k} \hat{m}_{k}+ m_{k}^T\beta \Lambda_{k} m_{k} \nonumber \\ &
+x_i^T \sum_{i=1}^{n} r_{ik}  \Lambda_{k}  x_i  - \sum_{i=1}^{n} r_{ik} \mathrm{ln}(det(\Lambda_{k})) ]   \Bigl\} \nonumber \\
& \propto \mathrm{exp}\Bigl\{ \frac{\nu-q-1}{2} \mathrm{ln}(det(\Lambda_{k}))-\frac{1}{2}Tr(\Phi \Lambda_{k}) - \frac{1}{2}[-\hat{m}_{k}^T (\beta + \sum_{i=1}^{n} r_{ik} ) \Lambda_{k} \hat{m}_{k}+ m_{k}^T\beta \Lambda_{k} m_{k} \nonumber \\ &
+x_i^T \sum_{i=1}^{n} r_{ik}  \Lambda_{k}  x_i  - \sum_{i=1}^{n} r_{ik} \mathrm{ln}(det(\Lambda_{k})) ]   \Bigl\} \nonumber  \\
& \propto \mathrm{exp}\Bigl\{ \frac{\nu-q-1+\sum_{i=1}^{n} r_{ik}}{2}  \mathrm{ln}(det(\Lambda_{k}))-\frac{1}{2} Tr\Bigl(\Phi \Lambda_{k}-  (\beta + \sum_{i=1}^{n} r_{ik} ) \hat{m}_{k}\hat{m}_{k}^T \Lambda_{k} + \beta m_{k} m_{k}^T \Lambda_{k} \nonumber \\ &
+ \sum_{i=1}^{n} r_{ik} x_i x_i^T \Lambda_{k} \Bigl)  \Bigl\}
\end{align}

where in the last line we used the cyclic property of trace. The last line is also in the form of a Wishart distribution with degree of freedom $\hat{\nu}_{k}=\nu+\sum_{i=1}^{n} r_{ik}$ and scale matrix $\hat{\Phi}_{k}^{-1}$, where

$$\hat{\Phi}_{k}=\Phi- (\beta + \sum_{i=1}^{n} r_{ik} ) \hat{m}_{k}\hat{m}_{k}^T+ \beta m_{k} m_{k}^T + \sum_{i=1}^{n} r_{ik} x_i x_i^T $$


Lastly we have the update for $q(\psi)$:

\begin{align} \label{eq: VIupdateforpsi}
q(\psi) & \propto  \mathrm{exp} \Bigl\{\mathbb{E}_{z} [\mathrm{ln}p(c|\psi,z)+\mathrm{ln}p(\psi) ] \Bigl\} \nonumber \\
& \propto \ \mathrm{exp} \Bigl\{ \sum_{i=1}^{n}\sum_{k=1}^{K} \sum_{j=1}^{p} r_{ik}\mathrm{ln}(\psi_{k,j,c_{ij}}) + \sum_{k=1}^{K} \sum_{j=1}^{p} \mathrm{ln}(Dir(\psi_{k,j}|\eta_j,\dots,\eta_j) )  \Bigl\}
\end{align}

so for each $q(\psi_{k,j})$ we have:

\begin{align} \label{eq: VIupdateforpsi2}
q(\psi_{k,j}) & \propto  \mathrm{exp} \Bigl\{ \sum_{i=1}^{n} r_{ik}\mathrm{ln}(\psi_{k,j,c_{ij}}) + \mathrm{ln}(Dir(\psi_{k,j}|\eta_j,\dots,\eta_j) )  \Bigl\} \nonumber \\
q(\psi_{k,j}) & \propto  \mathrm{exp} \Bigl\{ \sum_{i=1}^{n} r_{ik}\mathrm{ln}(\psi_{k,j,c_{ij}}) + \sum_{g=1}^{d_j} (\eta_j-1) \mathrm{ln}(\psi_{k,j,g})  \Bigl\} \nonumber \\
q(\psi_{k,j}) & \propto  \mathrm{exp} \Bigl\{ \sum_{i=1}^{n} \sum_{g=1}^{d_j} I_{\{c_{ij}=g\}}r_{ik}\mathrm{ln}(\psi_{k,j,g}) + \sum_{g=1}^{d_j} (\eta_j-1) \mathrm{ln}(\psi_{k,j,g})  \Bigl\} \nonumber \\
q(\psi_{k,j}) & \propto  \mathrm{exp} \Bigl\{ \sum_{g=1}^{d_j} (\eta_j-1+\sum_{i=1}^{n} I_{\{c_{ij}=g\}}r_{ik} ) \mathrm{ln}(\psi_{k,j,g})  \Bigl\}
\end{align}

where $I_{\{c_{ij}=g\}}=1$ if $c_{ij}=g$, $I_{\{c_{ij}=g\}}=0$ otherwise. 

By observation, $q(\psi_{k,j}) \sim Dir(\eta_j+\sum_{i=1}^{n} I_{\{c_{ij}=1\}}r_{ik},\dots,\eta_j+\sum_{i=1}^{n} I_{\{c_{ij}=d_j\}}r_{ik} )$

\bigskip

Now that we have the update equations for $q(\mu, \Lambda, \psi, \pi)$, we can derive the update for $q(z)$. This will require the forms of the variational posterior $q(\mu, \Lambda, \psi, \pi)$ we just derived above in order to evaluate expectations with respect to $q(\mu, \Lambda, \psi, \pi)$, in order to obtain an expression for $r_{ik}$:

\begin{align} \label{eq: variationalposteriorqz} 
q(z) & \propto \mathrm{exp} \Bigl\{ \mathbb{E}_{-z} \mathrm{ln} p(x, c, z, \mu, \Lambda, \psi, \pi) \Bigl\} \nonumber \\
& \propto \mathrm{exp} \Bigl\{ \mathbb{E}_{\mu, \Lambda, \psi, \pi} [\mathrm{ln}( p(x|c,z,\mu,\Lambda)p(\mu|\Lambda)p(\Lambda)p(c|\psi,z)p(\psi)p(z|\pi)p(\pi)   ) ] \Bigl\} \nonumber \\
& \propto \mathrm{exp} \Bigl\{ \mathbb{E}_{\mu, \Lambda} [\mathrm{ln}( p(x|c,z,\mu,\Lambda)p(\mu|\Lambda)p(\Lambda) ) ] + \mathbb{E}_{\psi} [ \mathrm{ln}( p(c|\psi,z)p(\psi) )] + \mathbb{E}_{\pi} [  \mathrm{ln}(p(z|\pi)p(\pi) ) ] \Bigl\} \nonumber \\
& \propto \mathrm{exp} \Bigl\{ \mathbb{E}_{\mu, \Lambda} [\mathrm{ln}( p(x|c,z,\mu,\Lambda) ) ] + \mathbb{E}_{\psi} [ \mathrm{ln}( p(c|\psi,z) )] + \mathbb{E}_{\pi} [  \mathrm{ln}(p(z|\pi) ) ] \Bigl\} \nonumber \\
& \propto \mathrm{exp} \Bigl\{ \sum_{i=1}^{n}\sum_{k=1}^{K} z_{ik}\mathbb{E}_{\mu, \Lambda} [\mathrm{ln}(\mathcal{N}(x_i|\mu_{k}, \Lambda_{k}^{-1} ) )] + \sum_{i=1}^{n}\sum_{k=1}^{K} \sum_{j=1}^{p} z_{ik}\mathbb{E}_{\psi} [\mathrm{ln}(\psi_{k,j,c_{ij}}) ] + \sum_{i=1}^{n}\sum_{k=1}^{K} z_{ik}\mathbb{E}_{\pi} [\mathrm{ln}(\pi_{k} ) ] \Bigl\} 
\end{align}  

where in the penultimate expression we dropped terms which do not depend on $z$. At this point the expression of $q(z)$ essentially boils down to evaluating the three expectations $\mathbb{E}_{\mu, \Lambda} [\mathrm{ln}(\mathcal{N}(x_i|\mu_{k}, \Lambda_{k}^{-1} ) )]$, $\mathbb{E}_{\psi} [\mathrm{ln}(\psi_{k,j,c_{ij}}) ]$ and $\mathbb{E}_{\pi} [\mathrm{ln}(\pi_{k} ) ]$. 

\begin{align} \label{eq: variationalposterioremulambda} 
& \mathbb{E}_{\mu, \Lambda} [\mathrm{ln}(\mathcal{N}(x_i|\mu_{k}, \Lambda_{k}^{-1} ) )] = \mathbb{E}_{\mu, \Lambda} [\mathrm{ln}\Bigl\{ (2\pi)^{-q/2} det (\Lambda_{k})^{1/2} \mathrm{exp}(-\frac{1}{2} (x_i-\mu_{k})^T  \Lambda_{k} (x_i-\mu_{k})) \Bigl\}] \nonumber  \\
&= -\frac{q}{2}\mathrm{ln}(2\pi)+\frac{1}{2}\mathbb{E}_{\Lambda_{k}} [\mathrm{ln}(det (\Lambda_{k})) ]-\mathbb{E}_{\mu, \Lambda} [\frac{1}{2} (x_i-\mu_{k})^T  \Lambda_{k} (x_i-\mu_{k})] 
\end{align}  

It is well known that if $\Lambda \sim \mathcal{W}(\nu, \Phi^{-1})$, then $\mathbb{E}_{\Lambda}[\mathrm{ln}(det (\Lambda))]=q\mathrm{ln}(2)+\mathrm{ln}(det(\Phi^{-1}))+\sum_{i=1}^{q} \digamma(\frac{\nu+1-i}{2})$, where $\digamma(x)$ is the digamma function. Since $\Lambda_{k} \sim \mathcal{W}(\hat{\nu}_{k}, \hat{\Phi}_{k}^{-1})$, we have

$\mathbb{E}_{\Lambda_{k}} [\mathrm{ln}(det (\Lambda_{k})) ]=q\mathrm{ln}(2)+\mathrm{ln}(det(\hat{\Phi}_{k}^{-1}))+\sum_{i=1}^{q} \digamma(\frac{\hat{\nu}_{k}+1-i}{2})$

As for $\mathbb{E}_{\mu, \Lambda} [\frac{1}{2} (x_i-\mu_{k})^T  \Lambda_{k} (x_i-\mu_{k})]$, we have:

\begin{align} \label{eq: variationalposterioremulambda2} 
& \mathbb{E}_{\mu, \Lambda} [\frac{1}{2} (x_i-\mu_{k})^T  \Lambda_{k} (x_i-\mu_{k})] \nonumber  \\
&= \frac{1}{2}\iint_{}^{} (x_i-\mu_{k})^T  \Lambda_{k} (x_i-\mu_{k}) q(\mu_{k},\Lambda_{k})d\mu_{k} d\Lambda_{k} \nonumber \\
&= \frac{1}{2} \iint_{}^{} (x_i-\hat{m}_{k}+\hat{m}_{k}-\mu_{k})^T  \Lambda_{k} (x_i-\hat{m}_{k} \nonumber \\
&+\hat{m}_{k}-\mu_{k}) q(\mu_{k},\Lambda_{k})d\mu_{k} d\Lambda_{k} \nonumber \\
&= \frac{1}{2} \iint_{}^{} [(x_i-\hat{m}_{k})^T \Lambda_{k} (x_i-\hat{m}_{k})+2(x_i-\hat{m}_{k})^T \Lambda_{k} (\hat{m}_{k}-\mu_{k}) + \\ & (\hat{m}_{k}-\mu_{k})^T\Lambda_{k} (\hat{m}_{k}-\mu_{k}) ] q(\mu_{k}|\Lambda_{k})q(\Lambda_{k})d\mu_{k} d\Lambda_{k} \nonumber \\
&= \frac{1}{2} \iint_{}^{} [(x_i-\hat{m}_{k})^T \Lambda_{k} (x_i-\hat{m}_{k}) \nonumber \\
&+ Tr((\mu_{k}-\hat{m}_{k})(\mu_{k}-\hat{m}_{k})^T\Lambda_{k} ) ] q(\mu_{k}|\Lambda_{k})q(\Lambda_{k})d\mu_{k} d\Lambda_{k} \nonumber \\
&= \frac{1}{2} \iint_{}^{} [(x_i-\hat{m}_{k})^T \Lambda_{k} (x_i-\hat{m}_{k}) + Tr(\hat{\beta}_{k}^{-1}\Lambda_{k}^{-1}\Lambda_{k} ) ] q(\Lambda_{k}) d\Lambda_{k} \nonumber \\
&= \frac{1}{2} [(x_i-\hat{m}_{k})^T \hat{\nu}_{k} \hat{\Phi}_{k}^{-1} (x_i-\hat{m}_{k}) + q\hat{\beta}_{k}^{-1} ] 
\end{align}  

Next we evaluate $\mathbb{E}_{\pi} [\mathrm{ln}(\pi_{k} ) ]$ and $\mathbb{E}_{\psi} [\mathrm{ln}(\psi_{k,j,c_{ij}}) ]$. The expectation of the natural logarithm of a Dirichlet is well known, specifically if $X \sim Dir(\alpha_1,\dots,\alpha_K)$, then $\mathbb{E}(\mathrm{ln}(X_i) )=\digamma(\alpha_i)-\digamma(\sum_{j=1}^{K}\alpha_j)$. With this, we have:

\begin{align} \label{eq: Epilnpik} 
\mathbb{E}_{\pi} [\mathrm{ln}(\pi_{k} ) ]=\digamma(\alpha+\sum_{i=1}^{n}r_{ik})-\digamma(K\alpha+\sum_{j=1}^{K}\sum_{i=1}^{n}r_{ij})
\end{align}

\begin{align} \label{eq: Epsilnpsi} 
\mathbb{E}_{\psi} [\mathrm{ln}(\psi_{k,j,c_{ij}}) ]=\digamma(\eta_j+\sum_{s=1}^{n} I_{\{c_{sj}=c_{ij}\}}r_{sk} )-\digamma(d_j \eta_j+\sum_{g_j=1}^{d_j} \sum_{s=1}^{n} I_{\{c_{sj}=g_j\}}r_{sk})
\end{align}

Recalling that $q(\pi) \sim Dir(\alpha+\sum_{i=1}^{n}r_{i1},\dots,\alpha+\sum_{i=1}^{n}r_{iK} )$ 

and $q(\psi_{k,j}) \sim Dir(\eta_j+\sum_{i=1}^{n} I_{\{c_{ij}=1\}}r_{ik},\dots,\eta_j+\sum_{i=1}^{n} I_{\{c_{ij}=d_j\}}r_{ik} )$

We return to the expression for $q(z)$ now that we have all the required expectations:

\begin{align} \label{eq: variationalposteriorqzcontinued} 
q(z) & \propto \mathrm{exp} \Bigl\{ \sum_{i=1}^{n}\sum_{k=1}^{K} z_{ik}\mathbb{E}_{\mu, \Lambda} [\mathrm{ln}(\mathcal{N}(x_i|\mu_{k}, \Lambda_{k}^{-1} ) )] + \sum_{i=1}^{n}\sum_{k=1}^{K} \sum_{j=1}^{p} z_{ik}\mathbb{E}_{\psi} [\mathrm{ln}(\psi_{k,j,c_{ij}}) ] + \sum_{i=1}^{n}\sum_{k=1}^{K} z_{ik}\mathbb{E}_{\pi} [\mathrm{ln}(\pi_{k} ) ] \Bigl\} \nonumber \\
q(z) & \propto \prod_{i=1}^{n}\prod_{k=1}^{K} \mathrm{exp} \Bigl\{ z_{ik}\mathbb{E}_{\mu, \Lambda} [\mathrm{ln}(\mathcal{N}(x_i|\mu_{k}, \Lambda_{k}^{-1} ) )] + \sum_{j=1}^{p} z_{ik}\mathbb{E}_{\psi} [\mathrm{ln}(\psi_{k,j,c_{ij}}) ] + z_{ik}\mathbb{E}_{\pi} [\mathrm{ln}(\pi_{k} ) ] \Bigl\} \nonumber \\
q(z) & \propto \prod_{i=1}^{n}\prod_{k=1}^{K}  \mathrm{exp} (z_{ik}\mathrm{ln}(\rho_{ik}) ) = \prod_{i=1}^{n}\prod_{k=1}^{K} \rho_{ik}^{z_{ik}} 
\end{align}  

where 
\begin{align} \label{eq: rhoik} 
\rho_{ik} &=\mathrm{exp}\Bigl\{ \mathbb{E}_{\mu, \Lambda} [\mathrm{ln}(\mathcal{N}(x_i|\mu_{k}, \Lambda_{k}^{-1} ) )]+\sum_{j=1}^{p}\mathbb{E}_{\psi} [\mathrm{ln}(\psi_{k,j,c_{ij}}) ]+\mathbb{E}_{\pi} [\mathrm{ln}(\pi_{k} ) ] \Bigl\} 
\end{align}

By observation, we have $r_{ik}=\frac{\rho_{ik}}{\sum_{j=1}^{K} \rho_{ij} }$. With this, we can alternate between updating $q(\mu, \Lambda, \psi, \pi)$, which depend on $r_{ik}$ and the $r_{ik}$ which depend on $q(\mu, \Lambda, \psi, \pi)$. 

\add{
\subsection{Relation to the EM algorithm} \label{subsec: EMcomparison}

The classical way to fit the mixture model in \Cref{eq: likelihoodVIappendix}  is maximum likelihood via the EM algorithm. Our likelihood is the special case of the MULTIMIX model of \cite{Hunt2002} in which the continuous variables form a single multivariate normal block and each categorical variable forms its own discrete block. EM treats the allocations $z$ as missing data and alternates between two steps. Given current `global' parameter estimates $\tilde\Theta=(\tilde\pi,\tilde\mu,\tilde\Lambda,\tilde\psi)$, the E-step computes the component allocation probabilities
\begin{align} \label{eq: EMestep}
\tilde r_{ik}=p(z_i=k\mid x_i,c_i,\tilde\Theta)=\frac{\tilde\pi_k\,\mathcal{N}(x_i\mid\tilde\mu_k,\tilde\Lambda_k^{-1})\prod_{j=1}^{p}\tilde\psi_{k,j,c_{ij}}}{\sum_{l=1}^{K}\tilde\pi_l\,\mathcal{N}(x_i\mid\tilde\mu_l,\tilde\Lambda_l^{-1})\prod_{j=1}^{p}\tilde\psi_{l,j,c_{ij}}}.
\end{align}
The M-step updates $\tilde\Theta$ in closed form:
\begin{align} \label{eq: EMmstep}
&\tilde\pi_k=\frac{1}{n}\sum_{i=1}^{n}\tilde r_{ik},\quad
\tilde\mu_k=\frac{\sum_{i}\tilde r_{ik}x_i}{\sum_{i}\tilde r_{ik}},\quad
\tilde\psi_{k,j,g}=\frac{\sum_{i}\tilde r_{ik}I_{\{c_{ij}=g\}}}{\sum_{i}\tilde r_{ik}}, \nonumber\\
&\tilde\Lambda_k^{-1}=\frac{\sum_{i}\tilde r_{ik}(x_i-\tilde\mu_k)(x_i-\tilde\mu_k)^T}{\sum_{i}\tilde r_{ik}}.
\end{align}
The two steps are iterated until the log-likelihood stabilises.

Our proposed CAVI algorithm and EM have the same two-step structure. The update of $q(z)$ plays the role of the E-step and the update of $q(\mu,\Lambda,\psi,\pi)$ that of the M-step. The CAVI posterior means of the equivalent variables in \Cref{eq: EMmstep} are essentially of the same form, but using $r_{ik}$ in place of $\tilde r_{ik}$, with the addition of prior hyperparameters. Both algorithms are coordinate ascent: EM alternately maximises $\mathbb{E}_{\tilde{q}(z)}[\ln p(x,c,z\mid\Theta)]-\mathbb{E}_{\tilde{q}(z)}[\ln \tilde{q}(z)]$ over $\tilde{q}(z)$ and a point $\Theta$, while CAVI alternately maximises the ELBO over $q(z)$ and a distribution $q(\mu,\Lambda,\psi,\pi)$. There are several key distinctions between CAVI and EM: firstly, the ELBO incorporates a prior distribution whereas the EM objective does not. Secondly, the VI M-step returns a distribution whereas the EM M-step returns a point estimate. Lastly, the VI E-step also accounts for uncertainty by integrating over $q(\mu,\Lambda,\psi,\pi)$, whereas the EM E-step plugs a point estimate into $p(z_i=k\mid x_i,c_i,\tilde\Theta)$.
}

\section{Evaluating the ELBO} \label{subsec: ELBOequations}

Though the CAVI algorithm update equations do not require the ELBO, the termination condition of CAVI is based on checking whether the ELBO has converged (in practice, whether the difference in ELBO between subsequent iterations has reduced beneath a pre-specified threshold), which require us to evaluate it. This involves another long calculation though some of the terms have been calculated from deriving the CAVI update equations. Also, unlike the CAVI updates, one needs to be careful not to drop any `constant' terms here as we are no longer deriving probability distributions up to a constant. 

\begin{align} \label{eq: ELBO} 
ELBO(q)& =\mathbb{E}_{q}[\mathrm{ln}(p(x, c, z, \mu, \Lambda, \psi, \pi))-\mathrm{ln}(q(\mu, \Lambda, \psi, \pi, z))] \\
& = \mathbb{E}_{q}[\mathrm{ln}(p(x|c,z,\mu,\Lambda)p(\mu|\Lambda)p(\Lambda)p(c|\psi,z)p(\psi)p(z|\pi)p(\pi)) - \mathrm{ln}(q(\mu, \Lambda)q(\psi)q(\pi)q(z)) ] \nonumber \\
& = \mathbb{E}_{\mu, \Lambda,z} [\mathrm{ln}( p(x|c,z,\mu,\Lambda))]+\mathbb{E}_{\mu, \Lambda} [\mathrm{ln}(p(\mu|\Lambda))]+\mathbb{E}_{\Lambda} [\mathrm{ln}(p(\Lambda))] \nonumber \\
&+ \mathbb{E}_{\psi,z} [ \mathrm{ln}( p(c|\psi,z))]+ \mathbb{E}_{\psi}[\mathrm{ln}(p(\psi))] + \mathbb{E}_{\pi,z} [  \mathrm{ln}(p(z|\pi))]+ \mathbb{E}_{\pi} [\mathrm{ln}(p(\pi))] \nonumber \\
& - \mathbb{E}_{\mu, \Lambda}[\mathrm{ln}(q(\mu, \Lambda))]-\mathbb{E}_{\psi}[\mathrm{ln}(q(\psi))]-\mathbb{E}_{\pi}[\mathrm{ln}(q(\pi))]-\mathbb{E}_{z}[\mathrm{ln}(q(z))]
\end{align}  

There are 11 expectations here, which we evaluate in order:

\begin{align} \label{eq: ELBOterm1} 
& \mathbb{E}_{\mu, \Lambda,z} [\mathrm{ln}( p(x|c,z,\mu,\Lambda))]=\sum_{i=1}^{n}\sum_{k=1}^{K} r_{ik}\mathbb{E}_{\mu, \Lambda} [\mathrm{ln}(\mathcal{N}(x_i|\mu_{k}, \Lambda_{k}^{-1} ) )] \nonumber \\
&= \sum_{i=1}^{n}\sum_{k=1}^{K} r_{ik} \Bigl\{ -\frac{q}{2}\mathrm{ln}(2\pi)+\frac{1}{2}[q\mathrm{ln}(2)+\mathrm{ln}(det(\hat{\Phi}_{k}^{-1}))+\sum_{i=1}^{q} \digamma(\frac{\hat{\nu}_{k}+1-i}{2})] \nonumber \\
&-\frac{1}{2} [(x_i-\hat{m}_{k})^T \hat{\nu}_{k} \hat{\Phi}_{k}^{-1} (x_i-\hat{m}_{k}) + q\hat{\beta}_{k}^{-1} ] \Bigl\}
\end{align}  

\begin{align} \label{eq: ELBOterm2} 
& \mathbb{E}_{\mu, \Lambda} [\mathrm{ln}(p(\mu|\Lambda))]=\mathbb{E}_{\mu, \Lambda} [\mathrm{ln}(\prod_{k=1}^{K} \mathcal{N}(\mu_{k}|m_{k}, \beta^{-1} \Lambda_{k}^{-1} ))]=\sum_{k=1}^{K}\mathbb{E}_{\mu, \Lambda} [\mathrm{ln}(  \mathcal{N}(\mu_{k}|m_{k}, \beta^{-1} \Lambda_{k}^{-1}  ))] \nonumber \\
& = \sum_{k=1}^{K}\mathbb{E}_{\mu, \Lambda} [\mathrm{ln} \Bigl\{ (2\pi)^{-q/2} det(\beta \Lambda_{k})^{1/2} \mathrm{exp}(-\frac{1}{2} (\mu_{k}-m_{k})^T \beta \Lambda_{k} (\mu_{k}-m_{k}) ) \Bigl\} ] \nonumber \\
& = \sum_{k=1}^{K}\Bigl\{ -\frac{q}{2}\mathrm{ln}(2\pi)+ \frac{1}{2}\mathbb{E}_{\Lambda} [\mathrm{ln}(det(\beta \Lambda_{k}))]+\mathbb{E}_{\mu, \Lambda}[-\frac{1}{2} (\mu_{k}-m_{k})^T \beta \Lambda_{k} (\mu_{k}-m_{k}) ] \Bigl\} \nonumber \\
& = \sum_{k=1}^{K}\Bigl\{ -\frac{q}{2}\mathrm{ln}(2\pi)+ \frac{q}{2}\mathrm{ln}(\beta) + \frac{1}{2}\mathbb{E}_{\Lambda} [\mathrm{ln}(det(\Lambda_{k}))] \nonumber \\
&+\mathbb{E}_{\mu, \Lambda}[-\frac{1}{2} (\mu_{k}-\hat{m}_{k}+\hat{m}_{k}-m_{k})^T \beta \Lambda_{k} (\mu_{k}-\hat{m}_{k}+\hat{m}_{k}-m_{k}) ] \Bigl\} \nonumber \\
& = \sum_{k=1}^{K}\Bigl\{ -\frac{q}{2}\mathrm{ln}(2\pi)+ \frac{q}{2}\mathrm{ln}(\beta) + \frac{1}{2}\mathbb{E}_{\Lambda} [\mathrm{ln}(det(\Lambda_{k}))] +\mathbb{E}_{\mu, \Lambda}[-\frac{1}{2} (\mu_{k}-\hat{m}_{k})^T \beta \Lambda_{k} (\mu_{k}-\hat{m}_{k})] \nonumber \\
& -\mathbb{E}_{\mu, \Lambda}[(\mu_{k}-\hat{m}_{k})^T \beta \Lambda_{k} (\hat{m}_{k}-m_{k} ) ] -\mathbb{E}_{\Lambda}[\frac{1}{2}(\hat{m}_{k}-m_{k})^T\beta \Lambda_{k}(\hat{m}_{k}-m_{k})] \Bigl\} \nonumber \\
& = \sum_{k=1}^{K}\Bigl\{ -\frac{q}{2}\mathrm{ln}(2\pi)+ \frac{q}{2}\mathrm{ln}(\beta) + \frac{1}{2}\mathbb{E}_{\Lambda} [\mathrm{ln}(det(\Lambda_{k}))] -\frac{1}{2} (Tr(\mathbb{E}_{\mu, \Lambda}[(\mu_{k}-\hat{m}_{k})(\mu_{k}-\hat{m}_{k})^T \beta \Lambda_{k}])) \nonumber \\
& -\mathbb{E}_{\Lambda}[\frac{1}{2}(\hat{m}_{k}-m_{k})^T\beta \Lambda_{k}(\hat{m}_{k}-m_{k})] \Bigl\} \nonumber \\
& = \sum_{k=1}^{K}\Bigl\{ -\frac{q}{2}\mathrm{ln}(2\pi)+ \frac{q}{2}\mathrm{ln}(\beta) + \frac{1}{2}[q\mathrm{ln}(2)+\mathrm{ln}(det(\hat{\Phi}_{k}^{-1}))+\sum_{i=1}^{q} \digamma(\frac{\hat{\nu}_{k}+1-i}{2})] -\frac{q}{2} \hat{\beta}_{k}^{-1}\beta \nonumber \\
& -\frac{1}{2}(\hat{m}_{k}-m_{k})^T\beta \hat{\nu}_{k} \hat{\Phi}_{k}^{-1}(\hat{m}_{k}-m_{k}) \Bigl\}
\end{align}  

\begin{align} \label{eq: ELBOterm3}
& \mathbb{E}_{\Lambda} [\mathrm{ln}(p(\Lambda))]=\mathbb{E}_{\Lambda} [\mathrm{ln}(\prod_{i=1}^{K} \mathcal{W}(\Lambda_{k}|\nu, \Phi^{-1}))]=\sum_{i=1}^{K}\mathbb{E}_{\Lambda_{k}} [\mathrm{ln}(\mathcal{W}(\Lambda_{k}|\nu, \Phi^{-1}))] \nonumber \\
& = \sum_{i=1}^{K}\mathbb{E}_{\Lambda_{k}} [\mathrm{ln}( \frac{det(\Lambda_{k})^{(\nu-q-1)/2} \mathrm{exp}(-Tr(\Phi \Lambda_{k})/2)}{2^{\nu q/2} det(\Phi^{-1})^{\nu/2} \Gamma_q(\nu/2)} ) ] \nonumber \\
& = \sum_{i=1}^{K}\mathbb{E}_{\Lambda_{k}} [\frac{\nu-q-1}{2} \mathrm{ln}(det(\Lambda_{k}))-\frac{1}{2}Tr(\Phi \Lambda_{k})- \frac{\nu q}{2}\mathrm{ln}(2)-\frac{\nu}{2}\mathrm{ln}(det(\Phi^{-1}))-\mathrm{ln}(\Gamma_q(\nu/2))   ] \nonumber \\
& = \sum_{i=1}^{K}\Bigl\{\frac{\nu-q-1}{2} [q\mathrm{ln}(2)+\mathrm{ln}(det(\hat{\Phi}_{k}^{-1}))+\sum_{i=1}^{q} \digamma(\frac{\hat{\nu}_{k}+1-i}{2})]-\frac{1}{2}Tr(\Phi \hat{\nu}_{k} \hat{\Phi}_{k}^{-1}) \nonumber \\
&- \frac{\nu q}{2}\mathrm{ln}(2)-\frac{\nu}{2}\mathrm{ln}(det(\Phi^{-1}))-\mathrm{ln}(\Gamma_q(\nu/2)) \Bigl\}
\end{align}  

where $\Gamma_q(a)=\pi^{q(q-1)/4} \prod_{j=1}^{q} \Gamma(a+(1-j)/2)$

\begin{align} \label{eq: ELBOterm4}
\mathbb{E}_{\psi,z} [ \mathrm{ln}( p(c|\psi,z))] & = \mathbb{E}_{\psi,z} [\mathrm{ln}(\prod_{i=1}^{n}\prod_{k=1}^{K} \prod_{j=1}^{p}\psi_{k,j,c_{ij}}^{z_{ik}})] = \sum_{i=1}^{n}\sum_{k=1}^{K} \sum_{j=1}^{p} r_{ik}\mathbb{E}_{\psi} [\mathrm{ln}(\psi_{k,j,c_{ij}})] \nonumber \\
& = \sum_{i=1}^{n}\sum_{k=1}^{K} \sum_{j=1}^{p} r_{ik} \Bigl\{ \digamma(\eta_j+\sum_{s=1}^{n} I_{\{c_{sj}=c_{ij}\}}r_{sk} )-\digamma(d_j \eta_j+\sum_{g_j=1}^{d_j} \sum_{s=1}^{n} I_{\{c_{sj}=g_j\}}r_{sk}) \Bigl\}
\end{align}  

\begin{align} \label{eq: ELBOterm5}
&\mathbb{E}_{\psi}[\mathrm{ln}(p(\psi))]= \mathbb{E}_{\psi} [ \mathrm{ln}(\prod_{k=1}^{K} \prod_{j=1}^{p} Dir(\psi_{k,j}|\eta_j,\dots,\eta_j) )] = \sum_{k=1}^{K} \sum_{j=1}^{p} \mathbb{E}_{\psi} [\mathrm{ln}(Dir(\psi_{k,j}|\eta_j,\dots,\eta_j) )] \nonumber \\
&=\sum_{k=1}^{K} \sum_{j=1}^{p} \mathbb{E}_{\psi} [\mathrm{ln}(\frac{1}{B(\eta_j)}\prod_{g=1}^{d_j} \psi_{k,j,g}^{\eta_j-1})]=\sum_{k=1}^{K} \sum_{j=1}^{p} \mathbb{E}_{\psi} [\mathrm{ln}(\frac{1}{B(\eta_j)})+(\eta_j-1)\sum_{g=1}^{d_j} \mathrm{ln}(\psi_{k,j,g})] \nonumber \\
&=\sum_{k=1}^{K} \sum_{j=1}^{p} [\mathrm{ln}(\frac{1}{B(\eta_j)})+(\eta_j-1)\sum_{g=1}^{d_j} \Bigl\{ \digamma(\eta_j+\sum_{s=1}^{n} I_{\{c_{sj}=g\}}r_{sk} )-\digamma(d_j \eta_j+\sum_{g_j=1}^{d_j} \sum_{s=1}^{n} I_{\{c_{sj}=g_j\}}r_{sk}) \Bigl\}]
\end{align}  

where $B(\bm{\alpha})=\frac{\prod_{k=1}^{K}{\Gamma(\alpha_k)}}{\Gamma(\sum_{k=1}^{K}\alpha_k)}$, where with a slight abuse of notation we define $B(\eta_j)=\frac{\prod_{g=1}^{d_j}{\Gamma(\eta_j)}}{\Gamma(\sum_{g=1}^{d_j}\eta_j)}$ and $B(\alpha)=\frac{\prod_{k=1}^{K}{\Gamma(\alpha)}}{\Gamma(\sum_{k=1}^{K}\alpha)} $.

\begin{align} \label{eq: ELBOterm6}
\mathbb{E}_{\pi,z} [  \mathrm{ln}(p(z|\pi))]&=\mathbb{E}_{\pi,z} [\mathrm{ln}(\prod_{i=1}^{n}\prod_{k=1}^{K}\pi_{k}^{z_{ik}}) ]=\sum_{i=1}^{n}\sum_{k=1}^{K}r_{ik}\mathbb{E}_{\pi} [\mathrm{ln}(\pi_{k}) ] \nonumber \\
&=\sum_{i=1}^{n}\sum_{k=1}^{K}r_{ik}[\digamma(\alpha+\sum_{i=1}^{n}r_{ik})-\digamma(K \alpha+\sum_{j=1}^{K}\sum_{i=1}^{n}r_{ij})]
\end{align} 
 
\begin{align} \label{eq: ELBOterm7}
\mathbb{E}_{\pi} [\mathrm{ln}(p(\pi))]&=\mathbb{E}_{\pi} [\mathrm{ln}(Dir(\pi|\alpha,\dots,\alpha)) ]=\mathbb{E}_{\pi} [\mathrm{ln}(\frac{1}{B(\alpha)}\prod_{k=1}^{K} \pi_k^{\alpha-1} )) ] \nonumber \\
&= \mathrm{ln}(\frac{1}{B(\alpha)})+\sum_{k=1}^{K}(\alpha-1)\mathbb{E}_{\pi} [ \mathrm{ln}(\pi_k) ] \nonumber \\
&= \mathrm{ln}(\frac{1}{B(\alpha)})+\sum_{k=1}^{K}(\alpha-1)[\digamma(\alpha+\sum_{i=1}^{n}r_{ik})-\digamma(K \alpha+\sum_{j=1}^{K}\sum_{i=1}^{n}r_{ij})]
\end{align} 

\begin{align} \label{eq: ELBOterm8}
 \mathbb{E}_{\mu, \Lambda}[\mathrm{ln}(q(\mu, \Lambda))]&=\mathbb{E}_{\mu, \Lambda}[\mathrm{ln}(q(\mu|\Lambda)q(\Lambda) )]=\mathbb{E}_{\mu, \Lambda}[\mathrm{ln}(q(\mu|\Lambda))]+\mathbb{E}_{\Lambda}[\mathrm{ln}(q(\Lambda))] 
\end{align} 
We derive the two terms here separately for easier presentation.

\begin{align} \label{eq: ELBOterm8part1}
&\mathbb{E}_{\mu, \Lambda}[\mathrm{ln}(q(\mu|\Lambda))]=\mathbb{E}_{\mu, \Lambda}[\mathrm{ln}(\prod_{k=1}^{K} \mathcal{N}(\mu_{k}|\hat{m}_{k}, \hat{\beta}_{k}^{-1} \Lambda_{k}^{-1} )]=\sum_{k=1}^{K} \mathbb{E}_{\mu, \Lambda}[\mathrm{ln}( \mathcal{N}(\mu_{k}|\hat{m}_{k}, \hat{\beta}_{k}^{-1} \Lambda_{k}^{-1} )] \nonumber \\
&=\sum_{k=1}^{K} \mathbb{E}_{\mu, \Lambda}[\mathrm{ln} \Bigl\{ (2\pi)^{-q/2} det(\hat{\beta}_{k} \Lambda_{k})^{1/2} \mathrm{exp}(-\frac{1}{2} (\mu_{k}-\hat{m}_{k})^T \hat{\beta}_{k} \Lambda_{k} (\mu_{k}-\hat{m}_{k}) ) \Bigl\} ] \nonumber \\
& = \sum_{k=1}^{K}\Bigl\{ -\frac{q}{2}\mathrm{ln}(2\pi)+ \frac{1}{2}\mathbb{E}_{\Lambda} [\mathrm{ln}(det(\hat{\beta}_{k} \Lambda_{k}))]+\mathbb{E}_{\mu, \Lambda}[-\frac{1}{2} (\mu_{k}-\hat{m}_{k})^T \hat{\beta}_{k} \Lambda_{k} (\mu_{k}-\hat{m}_{k}) ] \Bigl\} \nonumber \\
& = \sum_{k=1}^{K}\Bigl\{ -\frac{q}{2}\mathrm{ln}(2\pi)+\frac{q}{2}\mathrm{ln}(\hat{\beta}_{k})+\frac{1}{2}\mathbb{E}_{\Lambda}[\mathrm{ln}(det(\Lambda_{k}))]+\mathbb{E}_{\mu, \Lambda}[-\frac{1}{2} Tr((\mu_{k}-\hat{m}_{k})(\mu_{k}-\hat{m}_{k})^T \hat{\beta}_{k} \Lambda_{k}) ] \Bigl\} \nonumber \\
& = \sum_{k=1}^{K}\Bigl\{ -\frac{q}{2}\mathrm{ln}(2\pi)+\frac{q}{2}\mathrm{ln}(\hat{\beta}_{k})+\frac{1}{2}\mathbb{E}_{\Lambda}[\mathrm{ln}(det(\Lambda_{k}))]-\frac{q}{2} \Bigl\} \nonumber \\
& = \sum_{k=1}^{K}\Bigl\{ -\frac{q}{2}\mathrm{ln}(2\pi)+\frac{q}{2}\mathrm{ln}(\hat{\beta}_{k})+\frac{1}{2}[q\mathrm{ln}(2)+\mathrm{ln}(det(\hat{\Phi}_{k}^{-1}))+\sum_{i=1}^{q} \digamma(\frac{\hat{\nu}_{k}+1-i}{2})]-\frac{q}{2} \Bigl\}
\end{align}

\begin{align} \label{eq: ELBOterm8part2}
& \mathbb{E}_{\Lambda}[\mathrm{ln}(q(\Lambda))]=\mathbb{E}_{\Lambda}[\prod_{k=1}^{K}\mathrm{ln}(\mathcal{W}(\Lambda_{k}|\hat{\nu}_{k}, \hat{\Phi}_{k}^{-1}))]=\sum_{k=1}^{K} \mathbb{E}_{\Lambda_{k}}[ \mathrm{ln}(\mathcal{W}(\Lambda_{k}|\hat{\nu}_{k}, \hat{\Phi}_{k}^{-1}))] \nonumber \\
& = \sum_{i=1}^{K}\mathbb{E}_{\Lambda_{k}} [\mathrm{ln}( \frac{det(\Lambda_{k})^{(\hat{\nu}_{k}-q-1)/2} \mathrm{exp}(-Tr(\hat{\Phi}_{k} \Lambda_{k})/2)}{2^{\hat{\nu}_{k} q/2} det(\hat{\Phi}_{k}^{-1})^{\hat{\nu}_{k}/2} \Gamma_q(\hat{\nu}_{k}/2)} ) ] \nonumber \\
& = \sum_{i=1}^{K}\mathbb{E}_{\Lambda_{k}} [\frac{\hat{\nu}_{k}-q-1}{2} \mathrm{ln}(det(\Lambda_{k}))-\frac{1}{2}Tr(\hat{\Phi}_{k} \Lambda_{k})- \frac{\hat{\nu}_{k} q}{2}\mathrm{ln}(2)-\frac{\hat{\nu}_{k}}{2}\mathrm{ln}(det(\hat{\Phi}_{k}^{-1}))-\mathrm{ln}(\Gamma_q(\hat{\nu}_{k}/2))   ] \nonumber \\
& = \sum_{i=1}^{K} \Bigl\{ \frac{\hat{\nu}_{k}-q-1}{2} [q\mathrm{ln}(2)+\mathrm{ln}(det(\hat{\Phi}_{k}^{-1}))+\sum_{i=1}^{q} \digamma(\frac{\hat{\nu}_{k}+1-i}{2})]-\frac{\hat{\nu}_{k} q}{2} \nonumber \\
&- \frac{\hat{\nu}_{k} q}{2}\mathrm{ln}(2)-\frac{\hat{\nu}_{k}}{2}\mathrm{ln}(det(\hat{\Phi}_{k}^{-1}))-\mathrm{ln}(\Gamma_q(\hat{\nu}_{k}/2)) \Bigl\}
\end{align} 

\begin{align} \label{eq: ELBOterm9}
&\mathbb{E}_{\psi}[\mathrm{ln}(q(\psi))] = \mathbb{E}_{\psi} [ \mathrm{ln}(\prod_{k=1}^{K} \prod_{j=1}^{p} Dir(\psi_{k,j}|\hat{\eta}_{k,j,1},\dots,\hat{\eta}_{k,j,d_j}) )] = \sum_{k=1}^{K} \sum_{j=1}^{p} \mathbb{E}_{\psi} [\mathrm{ln}(Dir(\psi_{k,j}|\hat{\eta}_{k,j,1},\dots,\hat{\eta}_{k,j,d_j}) )] \nonumber \\
&=\sum_{k=1}^{K} \sum_{j=1}^{p} \mathbb{E}_{\psi} [\mathrm{ln}(\frac{1}{B(\hat{\bm{\eta}}_{k,j})}\prod_{g=1}^{d_j} \psi_{k,j,g}^{\hat{\eta}_{k,j,g}-1})]=\sum_{k=1}^{K} \sum_{j=1}^{p} \mathbb{E}_{\psi} [\mathrm{ln}(\frac{1}{B(\hat{\bm{\eta}}_{k,j})})+\sum_{g=1}^{d_j}(\hat{\eta}_{k,j,g}-1) \mathrm{ln}(\psi_{k,j,g})] \nonumber \\
&=\sum_{k=1}^{K} \sum_{j=1}^{p} [\mathrm{ln}(\frac{1}{B(\hat{\bm{\eta}}_{k,j})})+\sum_{g=1}^{d_j}(\hat{\eta}_{k,j,g}-1) \Bigl\{ \digamma(\eta_j+\sum_{s=1}^{n} I_{\{c_{sj}=g\}}r_{sk} )-\digamma(d_j \eta_j+\sum_{g_j=1}^{d_j} \sum_{s=1}^{n} I_{\{c_{sj}=g_j\}}r_{sk}) \Bigl\}]
\end{align} 

where $\hat{\eta}_{k,j,g}=\eta_j+\sum_{i=1}^{n} I_{\{c_{ij}=g\}}r_{ik}$

\begin{align} \label{eq: ELBOterm10}
\mathbb{E}_{\pi}[\mathrm{ln}(q(\pi))]&=\mathbb{E}_{\pi} [\mathrm{ln}(Dir(\pi|\alpha+\sum_{i=1}^{n}r_{i1},\dots,\alpha+\sum_{i=1}^{n}r_{iK} )) ]=\mathbb{E}_{\pi} [\mathrm{ln}(\frac{1}{B(\bm{\alpha})}\prod_{k=1}^{K} \pi_k^{\alpha_k-1} )) ] \nonumber  \\
&= \mathrm{ln}(\frac{1}{B(\bm{\alpha})})+\sum_{k=1}^{K}(\alpha_k-1)\mathbb{E}_{\pi} [ \mathrm{ln}(\pi_k) ] \nonumber \\
&= \mathrm{ln}(\frac{1}{B(\bm{\alpha})})+\sum_{k=1}^{K}(\alpha_k-1)[\digamma(\alpha+\sum_{i=1}^{n}r_{ik})-\digamma(K\alpha+\sum_{j=1}^{K}\sum_{i=1}^{n}r_{ij})]
\end{align} 

where $\hat{\alpha}_k=\alpha+\sum_{i=1}^{n}r_{ik}$

\begin{align} \label{eq: ELBOterm11}
\mathbb{E}_{z}[\mathrm{ln}(q(z))]=\mathbb{E}_{z}[\mathrm{ln}(\prod_{i=1}^{n}\prod_{k=1}^{K} r_{ik}^{z_{ik}} )]=\sum_{i=1}^{n}\sum_{k=1}^{K}\mathbb{E}_{z}[z_{ik}\mathrm{ln}( r_{ik} )]=\sum_{i=1}^{n}\sum_{k=1}^{K}r_{ik}\mathrm{ln}( r_{ik} )
\end{align} 

One can alternatively calculate some of the ELBO terms in a somewhat simplified way by pairing them up and computing the KL divergence instead. Specifically: 

\begin{align} \label{eq: ELBOKLversion} 
ELBO(q)& =\mathbb{E}_{q}[\mathrm{ln}(p(x, c, z, \mu, \Lambda, \psi, \pi))-\mathrm{ln}(q(\mu, \Lambda, \psi, \pi, z))] \\
& = \mathbb{E}_{q}[\mathrm{ln}(p(x|c,z,\mu,\Lambda)p(\mu|\Lambda)p(\Lambda)p(c|\psi,z)p(\psi)p(z|\pi)p(\pi)) - \mathrm{ln}(q(\mu, \Lambda)q(\psi)q(\pi)q(z)) ] \nonumber \\
& = \mathbb{E}_{\mu, \Lambda,z} [\mathrm{ln}( p(x|c,z,\mu,\Lambda))]+\mathbb{E}_{\mu, \Lambda} [\mathrm{ln}(p(\mu|\Lambda))]+\mathbb{E}_{\Lambda} [\mathrm{ln}(p(\Lambda))] \nonumber \\
&+ \mathbb{E}_{\psi,z} [ \mathrm{ln}( p(c|\psi,z))]+ \mathbb{E}_{\psi}[\mathrm{ln}(p(\psi))] + \mathbb{E}_{\pi,z} [  \mathrm{ln}(p(z|\pi))]+ \mathbb{E}_{\pi} [\mathrm{ln}(p(\pi))] \nonumber \\
& - \mathbb{E}_{\mu, \Lambda}[\mathrm{ln}(q(\mu, \Lambda))]-\mathbb{E}_{\psi}[\mathrm{ln}(q(\psi))]-\mathbb{E}_{\pi}[\mathrm{ln}(q(\pi))]-\mathbb{E}_{z}[\mathrm{ln}(q(z))] \nonumber \\
& = \mathbb{E}_{\mu, \Lambda,z} [\mathrm{ln}( p(x|c,z,\mu,\Lambda))]+ \mathbb{E}_{\psi,z} [ \mathrm{ln}( p(c|\psi,z))]+ \mathbb{E}_{\pi,z} [  \mathrm{ln}(p(z|\pi))]-\mathbb{E}_{z}[\mathrm{ln}(q(z))] \nonumber \\
&+ \mathbb{E}_{\mu, \Lambda} [\mathrm{ln}(p(\mu|\Lambda))] - \mathbb{E}_{\mu, \Lambda}[\mathrm{ln}(q(\mu|\Lambda))]+\mathbb{E}_{\Lambda} [\mathrm{ln}(p(\Lambda))]-\mathbb{E}_{\Lambda} [\mathrm{ln}(q(\Lambda))] \nonumber \\
&+\mathbb{E}_{\pi} [\mathrm{ln}(p(\pi))]-\mathbb{E}_{\pi}[\mathrm{ln}(q(\pi))]+ \mathbb{E}_{\psi}[\mathrm{ln}(p(\psi))]-\mathbb{E}_{\psi}[\mathrm{ln}(q(\psi))] \nonumber \\
& = \mathbb{E}_{\mu, \Lambda,z} [\mathrm{ln}( p(x|c,z,\mu,\Lambda))]+ \mathbb{E}_{\psi,z} [ \mathrm{ln}( p(c|\psi,z))]+ \mathbb{E}_{\pi,z} [  \mathrm{ln}(p(z|\pi))]-\mathbb{E}_{z}[\mathrm{ln}(q(z))] \nonumber \\
&- \mathbb{E}_{\Lambda} [\infdiv{q(\mu|\Lambda)}{p(\mu|\Lambda)}]-\infdiv{q(\Lambda)}{p(\Lambda)}-\infdiv{q(\pi)}{p(\pi)}-\infdiv{q(\psi)}{p(\psi)}
\end{align} 

The KL terms are all KL divergences of well known distributions with analytical expressions, which can be evaluated as follows:

\begin{align} \label{eq: ELBOKLterm1} 
& \mathbb{E}_{\Lambda} [\infdiv{q(\mu|\Lambda)}{p(\mu|\Lambda)}] =\mathbb{E}_{\Lambda} [\infdiv{ \prod_{k=1}^{K} \mathcal{N}(\mu_{k}|\hat{m}_{k}, \hat{\beta}_{k}^{-1} \Lambda_{k}^{-1} ) }{ \prod_{k=1}^{K}  \mathcal{N}(\mu_{k}|m_{k}, \beta^{-1} \Lambda_{k}^{-1} ) } ] \nonumber \\
&=\sum_{k=1}^{K} \mathbb{E}_{\Lambda} [\infdiv{ \mathcal{N}(\mu_{k}|\hat{m}_{k}, \hat{\beta}_{k}^{-1} \Lambda_{k}^{-1} ) }{ \mathcal{N}(\mu_{k}|m_{k}, \beta^{-1} \Lambda_{k}^{-1} ) } ] \nonumber \\
&=\sum_{k=1}^{K} \frac{1}{2} \Bigl\{ \mathbb{E}_{\Lambda} [\mathrm{ln}(\frac{det(\beta^{-1} \Lambda_{k}^{-1})}{det(\hat{\beta}_{k}^{-1} \Lambda_{k}^{-1})}) -q + (\hat{m}_{k}-m_{k})^T \beta \Lambda_{k} (\hat{m}_{k}-m_{k})+Tr(\beta \Lambda_{k} \hat{\beta}_{k}^{-1} \Lambda_{k}^{-1}) ] \Bigl\} \nonumber \\
&=\sum_{k=1}^{K} \frac{1}{2} \Bigl\{ \mathbb{E}_{\Lambda} [\mathrm{ln}(\frac{\beta^{-q} det(\Lambda_{k}^{-1})}{\hat{\beta}_{k}^{-q} det( \Lambda_{k}^{-1})}) -q + (\hat{m}_{k}-m_{k})^T \beta \Lambda_{k} (\hat{m}_{k}-m_{k})+q\beta \hat{\beta}_{k}^{-1} ] \Bigl\} \nonumber \\
&=\sum_{k=1}^{K} \frac{1}{2} \Bigl\{-q\mathrm{ln}(\beta)+q\mathrm{ln}(\hat{\beta}_{k}) -q + (\hat{m}_{k}-m_{k})^T \beta \hat{\nu}_{k} \hat{\Phi}_{k}^{-1} (\hat{m}_{k}-m_{k})+q\beta \hat{\beta}_{k}^{-1} \Bigl\}
\end{align} 

\begin{align} \label{eq: ELBOKLterm2}
&\infdiv{q(\Lambda)}{p(\Lambda)} =\infdiv{\prod_{k=1}^{K} \mathcal{W}(\Lambda_{k}|\hat{\nu}_{k}, \hat{\Phi}_{k}^{-1})}{\prod_{i=1}^{K}\mathcal{W}(\Lambda_k|\nu, \Phi^{-1})} \nonumber \\
&=\sum_{k=1}^{K}\infdiv{ \mathcal{W}(\Lambda_{k}|\hat{\nu}_{k}, \hat{\Phi}_{k}^{-1})}{ \mathcal{W}(\Lambda_{k}|\nu, \Phi^{-1})} \nonumber \\
&=\sum_{k=1}^{K} \frac{1}{2}\Bigl\{\nu\mathrm{ln}(det(\Phi^{-1}))-\nu\mathrm{ln}(det(\hat{\Phi}_{k}^{-1}))+\hat{\nu}_{k} Tr(\Phi \hat{\Phi}_{k}^{-1} )+2\mathrm{ln}(\Gamma_q(\nu/2))-2\mathrm{ln}(\Gamma_q(\hat{\nu}_{k}/2)) \nonumber \\
&+(\hat{\nu}_{k}-\nu)\sum_{i=1}^{q} \digamma(\frac{\hat{\nu}_{k}+1-i}{2})-\hat{\nu}_{k} q  \Bigl\}
\end{align} 

\begin{align} \label{eq: ELBOKLterm3}
&\infdiv{q(\pi)}{p(\pi)}=\infdiv{Dir(\pi|\alpha+\sum_{i=1}^{n}r_{i1},\dots,\alpha+\sum_{i=1}^{n}r_{iK})}{Dir(\pi|\alpha,\dots,\alpha)} \nonumber \\
&=\mathrm{ln}(\Gamma(\sum_{k=1}^{K} (\alpha+\sum_{i=1}^{n}r_{ik} )) )-\mathrm{ln}(\Gamma(\sum_{k=1}^{K} \alpha))+\sum_{k=1}^{K} \mathrm{ln}(\Gamma(\alpha))- \sum_{k=1}^{K} \mathrm{ln}(\Gamma(\alpha+\sum_{i=1}^{n}r_{ik}) ) \nonumber \\
&+\sum_{k=1}^{K} (\alpha+\sum_{i=1}^{n}r_{ik} -\alpha )\Bigl\{ \digamma(\alpha+\sum_{i=1}^{n}r_{ik}))-\digamma(\sum_{k=1}^{K} (\alpha+\sum_{i=1}^{n}r_{ik})) \Bigl\} \nonumber \\
&=\mathrm{ln}(\Gamma(K \alpha+\sum_{k=1}^{K} \sum_{i=1}^{n}r_{ik} )) -\mathrm{ln}(\Gamma(K \alpha))+\sum_{k=1}^{K} \mathrm{ln}(\Gamma(\alpha))- \sum_{k=1}^{K} \mathrm{ln}(\Gamma(\alpha+\sum_{i=1}^{n}r_{ik}) ) \nonumber \\
&+\sum_{k=1}^{K} \sum_{i=1}^{n}r_{ik}\Bigl\{ \digamma(\alpha+\sum_{i=1}^{n}r_{ik})-\digamma(K \alpha+\sum_{k=1}^{K}\sum_{i=1}^{n}r_{ik}) \Bigl\}
\end{align} 


\begin{align} \label{eq: ELBOKLterm4}
& \infdiv{q(\psi)}{p(\psi)}=\infdiv{\prod_{k=1}^{K} \prod_{j=1}^{p} Dir(\psi_{k,j}|\hat{\eta}_{k,j,1},\dots,\hat{\eta}_{k,j,d_j})}{\prod_{k=1}^{K} \prod_{j=1}^{p} Dir(\psi_{k,j}|\eta_j,\dots,\eta_j)} \nonumber \\
&=\sum_{k=1}^{K} \sum_{j=1}^{p} \infdiv{ Dir(\psi_{k,j}|\hat{\eta}_{k,j,1},\dots,\hat{\eta}_{k,j,d_j})}{ Dir(\psi_{k,j}|\eta_j,\dots,\eta_j)} \nonumber \\
&=\sum_{k=1}^{K} \sum_{j=1}^{p} [\mathrm{ln}(\Gamma(\sum_{g_j=1}^{d_j} \hat{\eta}_{k,j,g_j}) )-\mathrm{ln}(\Gamma(\sum_{g_j=1}^{d_j} \eta_j))+\sum_{g_j=1}^{d_j} \mathrm{ln}(\Gamma(\eta_j))- \sum_{g_j=1}^{d_j} \mathrm{ln}(\Gamma(\hat{\eta}_{k,j,g_j}) ) \nonumber \\
&+\sum_{g_j=1}^{d_j} (\hat{\eta}_{k,j,g_j}-\eta_j)\Bigl\{ \digamma(\hat{\eta}_{k,j,g_j})-\digamma(\sum_{g_j=1}^{d_j}\hat{\eta}_{k,j,g_j}) \Bigl\} ] \nonumber \\
&=\sum_{k=1}^{K} \sum_{j=1}^{p} [\mathrm{ln}(\Gamma(\sum_{g_j=1}^{d_j} (\eta_j+\sum_{s=1}^{n} I_{\{c_{sj}=g_j\}}r_{sk}) ))-\mathrm{ln}(\Gamma(d_j \eta_j)) \nonumber \\
&\qquad +\sum_{g_j=1}^{d_j} \mathrm{ln}(\Gamma(\eta_j))- \sum_{g_j=1}^{d_j} \mathrm{ln}(\Gamma( \eta_j+\sum_{s=1}^{n} I_{\{c_{sj}=g_j\}}r_{sk}) ) \nonumber \\
&+\sum_{g_j=1}^{d_j} ( \eta_j+\sum_{s=1}^{n} I_{\{c_{sj}=g_j\}}r_{sk}-\eta_j)\Bigl\{ \digamma(\eta_j+\sum_{s=1}^{n} I_{\{c_{sj}=g_j\}}r_{sk}) \nonumber \\
&\qquad -\digamma(\sum_{g_j=1}^{d_j} (\eta_j+\sum_{s=1}^{n} I_{\{c_{sj}=g_j\}}r_{sk}) ) \Bigl\} ] \nonumber \\
&=\sum_{k=1}^{K} \sum_{j=1}^{p} [\mathrm{ln}(\Gamma(d_j \eta_j+\sum_{g_j=1}^{d_j}\sum_{s=1}^{n} I_{\{c_{sj}=g_j\}}r_{sk} ))-\mathrm{ln}(\Gamma(d_j \eta_j)) \nonumber \\
&\qquad +\sum_{g_j=1}^{d_j} \mathrm{ln}(\Gamma(\eta_j))- \sum_{g_j=1}^{d_j} \mathrm{ln}(\Gamma( \eta_j+\sum_{s=1}^{n} I_{\{c_{sj}=g_j\}}r_{sk}) ) \nonumber \\
&+\sum_{g_j=1}^{d_j} \sum_{s=1}^{n} I_{\{c_{sj}=g_j\}}r_{sk}\Bigl\{ \digamma(\eta_j+\sum_{s=1}^{n} I_{\{c_{sj}=g_j\}}r_{sk})-\digamma(d_j \eta_j+\sum_{g_j=1}^{d_j}\sum_{s=1}^{n} I_{\{c_{sj}=g_j\}}r_{sk})  \Bigl\} ]
\end{align}

where $\hat{\eta}_{k,j,g_j}=\eta_j+\sum_{s=1}^{n} I_{\{c_{sj}=g_j\}}r_{sk}$

\section{Gibbs sampler for the model} \label{sec: gibbsupdate}

With a Gibbs sampler we want to calculate the conditional distribution of every parameter conditioned on every other parameter. Recall the parameters are $\mu_{k}$, $\Sigma_{k}$, $\pi_{k}$, $\psi_{k,j}$, $z_i$.

First we update $z_i|x,c,\mu,\Sigma,\psi, \pi$, for each $1 \leq i \leq n$, which has conditional distribution:

\begin{align} \label{eq: zupdate} 
p(z_i=k|x,c,\mu,\Sigma,\psi, \pi) & \propto p(z_i=k|\pi)p(x_i|c_i, z_i, \mu, \Sigma)p(c_i|\psi,z_i=k) \nonumber \\
& \propto \pi_{k}N(x_i|\mu_{k},\Sigma_{k})\prod_{j=1}^{p} \psi_{k,j,c_{ij}}
\end{align} 

This is just a discrete distribution so you normalise this probability by dividing by the sum of the probabilities over $k$. 

Next we update $\psi_{k,j}$ for each $1 \leq k \leq K, 1 \leq j \leq p$

\begin{align} \label{eq: Psiupdate} 
p(\psi_{k,j}|x,c,z,\mu,\Sigma, \pi) & \propto p(\psi_{k,j})\prod_{i: z_i=k}^{}p(c_{ij}|z_i=k,\psi) \nonumber  \\
& \propto Dir(\psi_{k,j}|\eta_j,\dots,\eta_j)\prod_{i: z_i=k}^{} \psi_{k,j,c_{ij}}
\end{align} 

This is a Dirichlet prior with a categorical data likelihood, which gives a conjugate Dirichlet posterior. In particular, the posterior is of the form:

\begin{align} \label{eq: Psiposterior} 
p(\psi_{k,j}|x,c,z,\mu,\Sigma, \pi) \sim Dir(\eta_j+\sum_{i=1}^{n}I(c_{ij}=1,z_i=k),\dots,\eta_j+\sum_{i=1}^{n}I(c_{ij}=d_j,z_i=k))
\end{align}

Update $\mu_{k}, \Sigma_{k}$ in a block:

\begin{align} \label{eq: muSigmaupdate} 
p(\mu_{k},\Sigma_{k}|x, c, z, \pi, \psi) & \propto p(\mu_{k}|\Sigma_{k})p(\Sigma_{k})\prod_{i: z_i=k}^{}p(x_i|z_i=k, c_i, \mu, \Sigma) \nonumber  \\
& \propto \mathcal{N}(\mu_{k}|m_{k}, \beta_{0}^{-1}\Sigma_{k})\mathcal{IW}(\Sigma_{k}|\nu, \Phi)\prod_{i: z_i=k}^{} \mathcal{N}(x_i|\mu_{k}, \Sigma_{k})
\end{align} 

This is a Normal-inverse Wishart prior with a Gaussian likelihood, which is conjugate and gives another Normal-inverse Wishart posterior of the form:

\begin{align} \label{eq: muSigmaposterior} 
& p(\mu_{k}|\Sigma_{k},x, c, z, \pi, \psi) \sim \mathcal{N}(\mu_{k}|\frac{\beta_{0}m_{k}+\sum_{i: z_i=k}x_i}{\beta_{0}+n_{k}}, \frac{\Sigma_{k}}{\beta_{0}+n_{k}}) \nonumber \\
& p(\Sigma_{k}|x, c, z, \pi, \psi) \sim \mathcal{IW}(\Sigma_{k}|\nu+n_{k}, \Phi+\sum_{i: z_i=k} (x_i-\bar{x}_{k})(x_i-\bar{x}_{k})^T+\frac{\beta_{0}n_{k}}{\beta_{0}+n_{k}}(\bar{x}_{k}-m_{k})(\bar{x}_{k}-m_{k})^T )
\end{align} 

where $\bar{x}_{k}=\frac{1}{n_{k}}\sum_{i: z_i=k}x_i$

\medskip

Lastly, we update $\pi$:

\begin{align} \label{eq: piupdate} 
p(\pi|x, c, z, \mu, \gamma, \Sigma, \psi, s) & \propto p(\pi) \prod_{i=1}^{n}p(z_i|\pi) \nonumber \\
& \propto Dir(\alpha,\dots,\alpha)\prod_{i=1}^{n} \pi_{z_i}
\end{align} 

Again like the $\psi$ update this is a Dirichlet prior with categorical likelihood so we have conjugate posterior:

\begin{align} \label{eq: piposterior}
p(\pi|x, c, z, \mu, \gamma, \Sigma, \psi, s) \sim Dir(\alpha+\sum_{i=1}^{n}I(z_i=1),\dots,\alpha+\sum_{i=1}^{n}I(z_i=K))
\end{align}

\add{In the following we describe briefly how the Gibbs marginal credible regions are calculated. Let $B$ be the burn-in and $T$ the iteration at which the sampler is stopped, and let $\theta^{(s)}$, $s=1,\dots,S$ with $S=T-B+1$, be the retained Gibbs draws of $\theta \in \{\mu_{k}, \Sigma_{k}, \psi_{k}, \pi\}$, where $\psi_{k}=(\psi_{k,1},\dots,\psi_{k,p})$. The marginal posterior $p(\theta|x,c)$ has no closed form, but the conditional posteriors used by the sampler do, so we use the estimate:

\begin{align} \label{eq: cmde}
\hat{p}(\theta|x,c)=\frac{1}{S}\sum_{s=1}^{S} p(\theta|\xi^{(s)},x,c),
\end{align}

where $p(\theta|\xi,x,c)$ is the conditional posterior of $\theta$ in its Gibbs update, and $\xi^{(s)}$ is the value of the variables it conditions on in the update that generated $\theta^{(s)}$. So $p(\theta|\xi,x,c)$ for $\mu_{k}$  and $\Sigma_{k}$ is just \eqref{eq: muSigmaposterior} and \eqref{eq: piposterior} for $\pi$. For $\psi_{k}$, both the prior and the likelihood factorise over $j$ (the $\psi_{k,j}$ have independent Dirichlet priors, and the $c_{ij}$ are independent across $j$ given $z_i$, as in \eqref{eq: zupdate}), so
 
\begin{align} \label{eq: psikjoint}
p(\psi_{k}|x,c,z,\mu,\Sigma,\pi) \propto \prod_{j=1}^{p} \Big[ Dir(\psi_{k,j}|\eta_j,\dots,\eta_j)\prod_{i: z_i=k} \psi_{k,j,c_{ij}} \Big],
\end{align}
Its conditional posterior is therefore the product over $j$ of the Dirichlets in \eqref{eq: Psiposterior}.

The HDR threshold is estimated by $\hat{q}_{0.05}$, the empirical $5\%$ quantile of $\{\log\hat{p}(\theta^{(s)}|x,c)\}_{s=1}^{S}$, and the truth is covered if $\log\hat{p}(\theta^*|x,c)\geq\hat{q}_{0.05}$ ($\theta^{*}$ relabelled if necessary, using linear sum assignment on the $L_1$ distance matrix between the posterior mean of each $\mu_k$ and $\mu_k^{*}$).
}

\section{Proof of Theoretical Results} \label{sec: proofssection}

\subsection{Proof of \thmmaintheorem} \label{subsec: theorem1proof}

\add{First we state \Cref{lemma: localcontraction}, which shows if three conditions about the map $T^\epsilon_n$ are satisfied, then almost
surely for all sufficiently large $n$, $T^\epsilon_n$ is a contraction on a closed
ball $\overline{B_\delta(\Theta^*)}=\{\Theta:\|\Theta-\Theta^*\|\le\delta\}$, maps the ball to itself, and has a unique fixed point
$\Theta^\infty_n$ to which the CAVI iterates converge geometrically, with
$\|\Theta^\infty_n-\Theta^*\|\to0$. The rest of the proof will be devoted to verifying those three conditions for our
model. Throughout we write $T^\epsilon_n$ to emphasise the dependence of $T^\epsilon$ on the
sample of size $n$.}

\add{\paragraph{Order notation.} Asymptotic statements are as $n\to\infty$. A plain $O(\cdot)$ denotes a deterministic bound: its constant depends only on the prior hyperparameters and the neighbourhood of $\Theta^*$ under consideration, and any dependence on an observation $x_i$ is written explicitly, e.g.\ $O\big((1+\|x_i\|^2)/n\big)$. We write `$O(\cdot)$ almost surely' when the bound holds for all sufficiently large $n$ on an event of $P_{\Theta^*}$-probability one, and $O_p(\cdot)$ for boundedness in $P_{\Theta^*}$-probability. }

\begin{lemma}[\add{Local contraction, fixed point and rate}]\label{lemma: localcontraction}
\add{Let $\Theta^*\in\mathbb{S}_\Theta$ and, for each $n$, let $T^\epsilon_n$ be
$C^1$ on an open set containing $\overline{B_\delta(\Theta^*)}$. Suppose there
are a radius $\delta>0$, a constant $\lambda_0\in[0,1)$, a deterministic
continuous operator $A^\epsilon(\Theta)$ on
$\overline{B_\delta(\Theta^*)}$, and an almost-sure event $\Omega_0$ on which}
\begin{itemize}
\item[\textup{(A1)}] \add{$\displaystyle\sup_{\Theta\in\overline{B_\delta(\Theta^*)}}\big\|\nabla T^\epsilon_n(\Theta)-A^\epsilon(\Theta)\big\|_{op}\to0$ as $n\to\infty$ }
\item[\textup{(A2)}] \add{$\displaystyle\sup_{\Theta\in\overline{B_\delta(\Theta^*)}}\big\|A^\epsilon(\Theta)\big\|_{op}\le\lambda_0<1$ }
\item[\textup{(A3)}] \add{$\big\|T^\epsilon_n(\Theta^*)-\Theta^*\big\|\to0$ }
\end{itemize}
\add{Set $\lambda=\tfrac{1+\lambda_0}{2}<1$ and $\delta_n=\|T^\epsilon_n(\Theta^*)-\Theta^*\|$. Then on $\Omega_0$ there is $N(\omega)$ such that for all $n\ge N(\omega)$:}
\begin{itemize}
\item[\textup{(i)}] \add{\emph{(contraction)} $\|T^\epsilon_n(\Theta)-T^\epsilon_n(\Theta')\|\le\lambda\|\Theta-\Theta'\|$ for all $\Theta,\Theta'\in\overline{B_\delta(\Theta^*)}$;}
\item[\textup{(ii)}] \add{\emph{(self-map)} $T^\epsilon_n(\overline{B_\delta(\Theta^*)})\subseteq\overline{B_\delta(\Theta^*)}$;}
\item[\textup{(iii)}] \add{\emph{(fixed point)} $T^\epsilon_n$ has a unique fixed point $\Theta^\infty_n\in\overline{B_\delta(\Theta^*)}$ with $\|\Theta^\infty_n-\Theta^*\|\le\delta_n/(1-\lambda)\to0$;}
\item[\textup{(iv)}] \add{\emph{(geometric convergence)} for any start $\Theta^{(0)}\in\overline{B_\delta(\Theta^*)}$, $\Theta^{(t+1)}=T^\epsilon_n(\Theta^{(t)})$ stay in the ball and satisfy $\|\Theta^{(t)}-\Theta^*\|\le\lambda^t\|\Theta^{(0)}-\Theta^*\|+\delta_n/(1-\lambda)$.}
\end{itemize}
\end{lemma}
\begin{proof}
\add{On $\Omega_0$, (A1) gives $N_1$ with
$\sup_{\Theta\in\overline{B_\delta(\Theta^*)}}\big\|\nabla T^\epsilon_n(\Theta)-A^\epsilon(\Theta)\big\|_{op} \leq (1-\lambda_0)/2$
for $n \geq N_1$, so with (A2),
$\sup_{\overline{B_\delta}}\|\nabla T^\epsilon_n\|_{op}\le\lambda_0+(1-\lambda_0)/2=\lambda<1$.
Since $\overline{B_\delta(\Theta^*)}$ is convex, the integral mean value
theorem (\Cref{main-eq: contractiontaylor} in main text) and this bound give (i). For (ii), (i) with
$\Theta'=\Theta^*$ gives $\|T^\epsilon_n(\Theta)-\Theta^*\| \leq \lambda\delta+\delta_n$,
which is $\leq \delta$ once (A3) makes $\delta_n \leq (1-\lambda)\delta$ (for
$n\ge N_2$); take $N=\max(N_1,N_2)$. Statements (iii), (iv) are then the Banach
fixed-point theorem for the $\lambda$-contraction $T^\epsilon_n$ of the complete
set $\overline{B_\delta(\Theta^*)}$ into itself, with the rate obtained by
putting $\Theta=\Theta^\infty_n,\Theta'=\Theta^*$ in (i) and iterating
$\|\Theta^{(t+1)}-\Theta^*\|\le\lambda\|\Theta^{(t)}-\Theta^*\|+\delta_n$.
}
\end{proof}

\add{The rest of this section verifies (A1)-(A3): (A1) is shown by applying a uniform law of large numbers argument \Cref{lemma: uniformSLLN} to $\nabla T^\epsilon_n$ whose components consist of sums, products or divisions of averages of $n$ i.i.d random variables, but care needs to be taken as those random variables also depend on $n$; (A2) is shown by explicitly computing the limit at $\Theta^*$ of $\nabla T^\epsilon_n(\Theta^*) \to A^\epsilon(\Theta^*)$ as $n \to \infty$  (\Cref{eq: operatorlimits}), and checking $\|A^\epsilon(\Theta^*)\|_{op}<1$. Then
by continuity of $A^\epsilon$, (A2) is satisfied on a closed ball containing $\Theta^*$; (A3) follows from calculations in \Cref{subsec: gradientcalcs}. Together with
\Cref{lemma: localcontraction} these yield \Cref{main-thm: maintheorem}.}

To compute $\nabla T^{\epsilon}(\Theta)$, which is required for (A1) and (A2), we have:
\begin{align} \label{eq: Tgradient} 
\nabla T^{\epsilon}(\Theta)=(1-\epsilon)I+\epsilon
\begin{bmatrix}
\nabla_{\bm{\hat{\pi}}} \Pi & \nabla_{\bm{\hat{\mu}}} \Pi & \nabla_{\bm{\hat{\Lambda}}} \Pi & \nabla_{\bm{\hat{\psi}}} \Pi \\
\nabla_{\bm{\hat{\pi}}} M & \nabla_{\bm{\hat{\mu}}} M & \nabla_{\bm{\hat{\Lambda}}} M & \nabla_{\bm{\hat{\psi}}} M \\
\nabla_{\bm{\hat{\pi}}} S & \nabla_{\bm{\hat{\mu}}} S & \nabla_{\bm{\hat{\Lambda}}} S & \nabla_{\bm{\hat{\psi}}} S \\
\nabla_{\bm{\hat{\pi}}} \Psi & \nabla_{\bm{\hat{\mu}}} \Psi & \nabla_{\bm{\hat{\Lambda}}} \Psi & \nabla_{\bm{\hat{\psi}}} \Psi 
\end{bmatrix}
\end{align}


Because the gradient terms in $\nabla T^{\epsilon}(\Theta)$ involves $r_{ik}^{(t-1)}$ or the derivatives of $r_{ik}^{(t-1)}$, we need to calculate the limit of $r_{ik}^{(t-1)}$ and the derivatives of $r_{ik}^{(t-1)}$ as $n \rightarrow \infty$. We start with the limit of $r_{ik}^{(t-1)}$, recall that $r_{ik}^{(t-1)}=r_{ik}(\Theta^{(t-1)})= \frac{\rho_{ik}(\Theta^{(t-1)}) }{\sum_{j=1}^{K} \rho_{ij}(\Theta^{(t-1)})} $, where

\begin{align} \label{eq: rho_ikTheta}
\rho_{ik}(\Theta)=\mathrm{exp}\Bigl\{ \mathbb{E}_{\mu, \Lambda} [\mathrm{ln}(\mathcal{N}(x_i|\mu_{k}, \Lambda_{k}^{-1} ) )]+\mathbb{E}_{\psi} [\mathrm{ln}(\psi_{k,c_{i}}) ]+\mathbb{E}_{\pi} [\mathrm{ln}(\pi_{k} ) ] \Bigl\} 
\end{align}

Where the expectations $\mathbb{E}_{\mu, \Lambda}, \mathbb{E}_{\psi}, \mathbb{E}_{\pi}$ are with respect to the variational posterior with hyperparameter values computed at $\Theta$, which recall is related to $\Theta$ through the following relationships:

\begin{align} \label{eq: simplifiedtohyperparaappendix}
\hat{\alpha}_k&=n \hat{\pi}_k + \alpha \nonumber \\
\hat{m}_k&=(n\hat{\mu}_k \hat{\pi}_k+\beta m_k)/(n\hat{\pi}_k+\beta) \nonumber \\
\hat{\beta}_k&=n \hat{\pi}_k+\beta \nonumber \\
\hat{\nu}_k&=n \hat{\pi}_k+\nu \nonumber \\
\hat{\Phi}_k &=n \hat{\pi}_k (\hat{\Lambda}_k)^{-1}+ \frac{n \hat{\pi}_k }{ n \hat{\pi}_k+\beta} \beta (\hat{\mu}_k-m_k)(\hat{\mu}_k-m_k)^T+\Phi  \nonumber \\
\hat{\eta}_{k,g}&=n \hat{\pi}_k\hat{\psi}_{k,g}+\eta
\end{align}

\add{Before computing the $n\to\infty$ limits of $r_{ik}$ and its
$\Theta$-derivatives, we describe the general form of the latter, since it is used
both in \Cref{lemma: rho_iklimit} (for the derivative bound) and in the
verification of \textup{(A2)}. Writing
$r_{ik}=\frac{\mathrm{exp}(\ln(\rho_{ik}))}{\sum_{j=1}^{K} \mathrm{exp}(\ln(\rho_{ij}))}$
and applying the quotient rule,}
\begin{align}\label{eq: dr_ikdtheta_k'sum}
\frac{\partial r_{ik}}{\partial \theta_{k'} }&=\frac{\rho_{ik} \frac{\partial \ln(\rho_{ik}) }{\partial \theta_{k'}}(\sum_{j=1}^{K} \rho_{ij})-\sum_{j=1}^{K} \rho_{ij} \frac{\partial \ln(\rho_{ij}) }{\partial \theta_{k'}}\rho_{ik} }{ (\sum_{j=1}^{K} \rho_{ij})^2 } \nonumber \\
&=\frac{\partial \ln(\rho_{ik}) }{\partial \theta_{k'}} r_{ik}-\sum_{j=1}^{K} r_{ij} \frac{\partial \ln(\rho_{ij}) }{\partial \theta_{k'}} r_{ik}
\end{align}
Here $\theta_{k'}$ is any scalar parameter within $\Theta$ belonging to component $k'$, so we could have for example $\hat{\mu}_{k'l}$, the $l$th coordinate of $\hat{\mu}_{k'}$. Note if $\rho_{ik}$ is independent of $\theta_{k'}$ for $k' \neq k$, then
\begin{align}\label{eq: dr_ikdtheta_k'}
\frac{\partial r_{ik} }{\partial \theta_{k'}}=\frac{\partial \ln(\rho_{ik'}) }{\partial \theta_{k'}} \Bigl\{ \delta_{kk'}r_{ik}-r_{ik}r_{ik'} \Bigl\}
\end{align}
\add{so that $\nabla_\Theta r_{ik}$ is a fixed algebraic combination of the
responsibilities $\{r_{ij}\}_j$ and the log-derivatives
$\partial\ln\rho_{ij}/\partial\theta_{k'}$: once
$\partial\ln\rho_{ik'}/\partial\theta_{k'}$ is known, so is
$\partial r_{ik}/\partial\theta_{k'}$. Consequently every entry of
$\nabla T^\epsilon(\Theta)$ in \Cref{eq: Tgradient} is a finite algebraic combination (addition, multiplication, multiplication of inverse) of weighted empirical averages of the $r_{ik}(\Theta)$ and their $\Theta$-derivatives.}





\add{Because (A1) requires uniform convergence of $\nabla T^\epsilon$ on a compact neighbourhood of $\Theta^*$, we work
on a fixed closed ball around $\Theta^*$. So there exist constants $0<c_*<C_*<\infty$ and a radius
$\delta_1>0$ such that every $\Theta=(\bm{\hat\pi},\bm{\hat\mu},\bm{\hat\Lambda},\bm{\hat\psi})$
in the closed ball $\overline{B_{\delta_1}(\Theta^*)}$ satisfies}
\add{%
\begin{align}\label{eq: C1C3}
\hat\pi_k\ge c_*,\quad \hat\psi_{k,g}\ge c_*,\quad c_* I \preceq \hat\Lambda_k \preceq C_* I,
\end{align}
for all $1\le k\le K$, $1\le g\le d$. }

\begin{lemma}\label{lemma: rho_iklimit}
\add{Let $\rho^{\infty}_{ik}(\Theta)=\hat\pi_k\,N(x_i\mid\hat\mu_k,\hat\Lambda_k^{-1})\,\bar\psi_{k,c_i}$ denote the $n$-free counterpart of $\rho_{ik}(\Theta)$, where $\bar\psi_{k,c_i}=\hat\psi_{k,c_i}/\sum_{g}\hat\psi_{k,g}$ is the normalised categorical
weight. (The normalisation is immaterial at
$\Theta^*$, where $\sum_g\psi^*_{k,g}=1$, but $\overline{B_{\delta_1}(\Theta^*)}$ need not lie on the
categorical simplex, so $\sum_g\hat\psi_{k,g}$ is only near $1$; see the categorical block below.)
There is a constant $C<\infty$ (depending only on $\Theta^*$ and the
fixed prior hyperparameters) such that, for all $1\le k\le K$,}
\add{
\begin{align}\label{eq: remainder}
\sup_{\Theta\in \overline{B_{\delta_1}(\Theta^*)}}\big|\ln\rho_{ik}(\Theta)-\ln\rho^{\infty}_{ik}(\Theta)\big|
& \leq \frac{C(1+\|x_i\|^2)}{n}, \nonumber \\
\sup_{\Theta\in \overline{B_{\delta_1}(\Theta^*)}}\big\|\nabla_\Theta\big(\ln\rho_{ik}(\Theta)-\ln\rho^{\infty}_{ik}(\Theta)\big)\big\|
& \leq \frac{C(1+\|x_i\|^2)}{n}.
\end{align}
}
\add{Exponentiating the first bound in \Cref{eq: remainder}, and using the Gaussian decay of $\rho^{\infty}_{ik}$ to absorb the resulting factor $e^{C(1+\|x_i\|^2)/n}$ (see the proof), the responsibility numerators themselves converge uniformly at the same rate for all sufficiently large $n$,}
\begin{align}\label{eq: rhoerr}
\add{\sup_{\Theta\in \overline{B_{\delta_1}(\Theta^*)}}\big|\rho_{ik}(\Theta)-\rho^{\infty}_{ik}(\Theta)\big|
\;\le\; \frac{C(1+\|x_i\|^2)}{n}.}
\end{align}
\add{Consequently, writing $r^{\infty}_{ik}=\rho^{\infty}_{ik}\big/\sum_{j=1}^{K}\rho^{\infty}_{ij}$,}
\add{
\begin{align}\label{eq: rerr}
\sup_{\Theta \in \overline{B_{\delta_1}(\Theta^*)}}\big|r_{ik}(\Theta)-r^{\infty}_{ik}(\Theta)\big| 
& \le \frac{C(1+\|x_i\|^2)}{n},\nonumber \\
\sup_{\Theta \in \overline{B_{\delta_1}(\Theta^*)}}\big\|\nabla_\Theta r_{ik}(\Theta)-\nabla_\Theta r^{\infty}_{ik}(\Theta)\big\|
& \le \frac{C(1+\|x_i\|^4)}{n}.
\end{align}
}
\add{Note that $\rho^{\infty}_{ik}(\Theta)$ is the numerator of the E-step of the EM algorithm evaluated at $\Theta$, so $r^{\infty}_{ik}(\Theta)$ is the EM responsibility. Since the maps $\Pi,M,S,\Psi$ coincide with the EM M-step, \Cref{lemma: rho_iklimit} shows that $T^{\epsilon}_n$ and the corresponding EM map differ by $O(1/n)$, together with their derivatives.}
\end{lemma}

\begin{proof}
\add{We work throughout at a generic $\Theta\in\overline{B_{\delta_1}(\Theta^*)}$,
keeping $\hat\pi_k,\hat\mu_k,\hat\Lambda_k,\hat\psi_{k,g}$ as free symbols
constrained only by the margins \Cref{eq: C1C3}. Under \Cref{eq: C1C3}
every hyperparameter in \Cref{eq: simplifiedtohyperparaappendix} except $\hat{m}_k$ has the form
$na(\Theta)+b(\Theta)$ with $b$ bounded and leading coefficient $a(\Theta)$
bounded, continuously differentiable in $\Theta$ with bounded derivatives, and
kept away from degeneracy by \Cref{eq: C1C3} (the scalar coefficients bounded
below by a positive constant, the matrix coefficient $\hat\pi_k\hat\Lambda_k^{-1}$
of $\hat\Phi_k$ uniformly positive definite), the location $\hat m_k$ being the
bounded ratio exception, all the
$O(\cdot)$ below are uniform in $\Theta$ over the ball, and we track the
dependence on $x_i$ explicitly.}

\add{We bound each block remainder in the $C^1$ norm, that is $\sup_{\overline{B_{\delta_1}}}(|\cdot|+\|\nabla_\Theta\cdot\|)\le C(1+\|x_i\|^2)/n$, which gives both lines of
\Cref{eq: remainder} at once.}

\add{For the $O(1/n)$ derivative bounds we make use of the following fact: let $\lambda(\Theta)=na(\Theta)+b(\Theta)$ with
$a\in[c_*',C_*']\subset(0,\infty)$ and $a,b\in C^1$ with $\Theta$-derivatives
bounded uniformly on the ball (this covers the scalar hyperparameters and the
scalar denominators built from them, and $\hat\Phi_k$ obeys the
analogous bound with $\hat\pi_k\hat\Lambda_k^{-1}$ uniformly positive definite,
by \Cref{eq: C1C3}). Then, uniformly on $\overline{B_{\delta_1}(\Theta^*)}$,}
\begin{align}\label{eq: ratioderiv}
\add{\frac{\partial}{\partial\Theta}\Big(\frac{f(\Theta)}{\lambda(\Theta)}\Big)
=\frac{f'\lambda-f\lambda'}{\lambda^2}=O\!\Big(\frac1n\Big)\quad\text{for bounded }f\in C^1,}
\end{align}
\add{because $\lambda\ge nc_*'$ makes $\lambda^{-1},\lambda^{-2}=O(1/n),O(1/n^2)$
while $f,f',\lambda'=O(1)$ and $\lambda=O(n)$, so $f\lambda'/\lambda^2=O(1/n)$ and
$f'/\lambda=O(1/n)$. Similarly, the digamma function satisfies: 
\begin{align}\label{eq: digammaexp}
\add{\digamma\big(na(\Theta)+b(\Theta)\big)=\ln\big(na(\Theta)\big)+O\!\Big(\frac1n\Big),}
\end{align}}
\add{since
$\digamma(\lambda)=\ln\lambda+O(\lambda^{-1})$ and differentiating,
$\frac{\partial}{\partial\Theta}\digamma(\lambda)=\digamma'(\lambda)\,\lambda'
=(\lambda^{-1}+O(\lambda^{-2}))(na'+b')=\frac{a'}{a}+O(1/n)
=\frac{\partial}{\partial\Theta}\ln(na)+O(1/n)$, again uniformly on the ball. Both \Cref{eq: ratioderiv} and \Cref{eq: digammaexp} are used repeatedly below.}

\add{The three additive blocks of $\ln\rho_{ik}(\Theta)$ in
\Cref{eq: rho_ikTheta} are examined in turn. For each block we bound the difference with the corresponding term in $\ln\rho^{\infty}_{ik}(\Theta)$,
and then the derivative bound via \Cref{eq: ratioderiv} and \Cref{eq: digammaexp}.}

\add{\emph{Gaussian block.}} Recall that
\begin{align}
\mathbb{E}_{\mu, \Lambda} [\mathrm{ln}(\mathcal{N}(x_i|\mu_{k}, \Lambda_{k}^{-1} ) )] &=-\frac{q}{2}\mathrm{ln}(2\pi)+\frac{1}{2}\mathbb{E}_{\Lambda_{k}} [\mathrm{ln}(det (\Lambda_{k})) ]-\frac{1}{2} [(x_i-\hat{m}_{k})^T \hat{\nu}_{k} \hat{\Phi}_{k}^{-1} (x_i-\hat{m}_{k}) + q\hat\beta_{k}^{-1} ]. 
\end{align}
\add{Note $\hat m_k-\hat\mu_k=\frac{\beta(m_k-\hat\mu_k)}{n\hat\pi_k+\beta}$.
Its numerator $\beta(m_k-\hat\mu_k)$ is continuous in $\Theta$ and so bounded on
the compact ball, the denominator satisfies
$n\hat\pi_k+\beta\ge nc_*$ by  \Cref{eq: C1C3}, so the ratio is $\le C/n$ with a constant uniform over
the closed ball. The same bound holds for
$\nabla_\Theta(\hat m_k-\hat\mu_k)$ by applying \Cref{eq: ratioderiv} with denominator $\lambda(\Theta)=n\hat\pi_k+\beta$. Thus $\|\hat m_k-\hat\mu_k\|,\|\nabla_\Theta(\hat m_k-\hat\mu_k)\|\le C/n$.}

For \add{$\hat\nu_k\hat\Phi_k^{-1}$, define 
$E_n:=\hat\nu_k\hat\Phi_k^{-1}-\hat\Lambda_k$ and expand}
\begin{align}
\hat{\nu}_{k} \hat{\Phi}_{k}^{-1} &= (n \hat\pi_k+\nu)\Big(n \hat\pi_k \hat\Lambda_k^{-1}+ \tfrac{n \hat\pi_k }{ n \hat\pi_k+\beta} \beta (\hat\mu_k-m_k)(\hat\mu_k-m_k)^T+\Phi\Big)^{-1}  \nonumber \\
&= (I+X)^{-1}\,\tfrac{n \hat\pi_k+\nu}{n \hat\pi_k}\, \hat\Lambda_k,
\qquad \add{X=\tfrac{ \hat\Lambda_k }{ n \hat\pi_k+\beta} \beta (\hat\mu_k-m_k)(\hat\mu_k-m_k)^T+\tfrac{ \hat\Lambda_k }{n \hat\pi_k}\Phi.}
\end{align}
\add{By \Cref{eq: C1C3}, $\hat\Lambda_k$ lies in a compact subset and $\hat\pi_k\ge c_*$, so $\|X\|\le C/n$; the expansion
$(I+X)^{-1}=I-X+O(\|X\|^2)$ and $\tfrac{n\hat\pi_k+\nu}{n\hat\pi_k}=1+O(1/n)$ then give
$\hat\nu_k\hat\Phi_k^{-1}=\hat\Lambda_k(I+O(1/n))=\hat\Lambda_k+E_n$, i.e.}
\begin{align}\label{eq: Enbound}
\add{E_n=\hat\nu_k\hat\Phi_k^{-1}-\hat\Lambda_k,\qquad
\sup_{\Theta\in\overline{B_{\delta_1}(\Theta^*)}}\big(\|E_n\|+\|\nabla_\Theta E_n\|\big)\le C/n,}
\end{align}
\add{the derivative bound holding by the product rule. Writing
$P=(I+X)^{-1}g$ with $g=\frac{n\hat\pi_k+\nu}{n\hat\pi_k}=1+\frac{\nu}{n\hat\pi_k}$,
we have $E_n=(P-I)\hat\Lambda_k$, so
$\nabla_\Theta E_n=(\nabla_\Theta P)\hat\Lambda_k+(P-I)\nabla_\Theta\hat\Lambda_k$.
The first term is $O(1/n)$ since
$\nabla_\Theta P=[\nabla_\Theta(I+X)^{-1}]g+(I+X)^{-1}\nabla_\Theta g=O(1/n)$
(with $\nabla_\Theta(I+X)^{-1}=-(I+X)^{-1}(\nabla_\Theta X)(I+X)^{-1}$ and
$\nabla_\Theta X,\nabla_\Theta g=O(1/n)$ by \Cref{eq: ratioderiv}) and
$\hat\Lambda_k=O(1)$; the second is $O(1/n)$ since $P-I=(I+X)^{-1}g-I=O(1/n)$
(as $(I+X)^{-1}=I+O(1/n)$ and $g=1+O(1/n)$) while $\nabla_\Theta\hat\Lambda_k=O(1)$.} Also $q\hat{\beta}_k^{-1}=O(1/n)$ with $O(1/n)$ derivative by \Cref{eq: ratioderiv}.
\add{Writing $u=x_i-\hat\mu_k$ and $s=\hat\mu_k-\hat m_k$ (so
$\|s\|,\|\nabla_\Theta s\|\le C/n$ by the mean bound above), the quadratic form
expands as}
\begin{align}
(x_i-\hat m_k)^T \hat\nu_k\hat\Phi_k^{-1}(x_i-\hat m_k)
&= (x_i-\hat\mu_k)^T\hat\Lambda_k(x_i-\hat\mu_k)+ \underbrace{u^\top E_n u+2u^\top(\hat\Lambda_k+E_n)s+s^\top(\hat\Lambda_k+E_n)s}_{\text{remainder}}.\nonumber
\end{align}
\add{We bound the remainder in $C^1$ norm term by term. Each of the three terms pairs
has an $O(1/n)$ factor ($E_n$ or $s$, by \Cref{eq: Enbound} and the mean bound above)
with at most two factors $\|u\|=\|x_i-\hat\mu_k\|$. $\hat\mu_k$ is bounded on
the ball, so $\|u\|=O(1+\|x_i\|)$ and $\|u\|^2=O(1+\|x_i\|^2)$,
while $\|\hat\Lambda_k+E_n\|=O(1)$. Hence
$|u^\top E_nu|\le\|E_n\|\,\|u\|^2=O\big((1+\|x_i\|^2)/n\big)$,
$|2u^\top(\hat\Lambda_k+E_n)s|\le 2\|\hat\Lambda_k+E_n\|\,\|u\|\,\|s\|=O\big((1+\|x_i\|)/n\big)$, and
$|s^\top(\hat\Lambda_k+E_n)s|\le\|\hat\Lambda_k+E_n\|\,\|s\|^2=O(1/n^2)$; the largest is
$O\big((1+\|x_i\|^2)/n\big)$.
For the gradient, $\nabla_\Theta$ of each term is expanded by the product rule and
every resulting summand still carries exactly one $O(1/n)$ factor, because in this case
differentiating an $O(1/n)$ quantity leaves it $O(1/n)$ (\Cref{eq: Enbound} for
$E_n$, the mean bound for $s$) while differentiating a factor $u=x_i-\hat\mu_k$
gives $-\nabla_\Theta\hat\mu_k=O(1)$. For example, 
$\nabla_\Theta(u^\top E_nu)=(\nabla_\Theta u)^\top E_nu+u^\top(\nabla_\Theta E_n)u+u^\top E_n(\nabla_\Theta u)
=O\big((1+\|x_i\|)/n\big)+O\big((1+\|x_i\|^2)/n\big)+O\big((1+\|x_i\|)/n\big)=O\big((1+\|x_i\|^2)/n\big)$. Hence the whole
remainder obeys
$\sup_{\overline{B_{\delta_1}}}(|\cdot|+\|\nabla_\Theta\cdot\|)=O\big((1+\|x_i\|^2)/n\big)$.
}

\add{For the log-determinant we use a similar approach. 
$\mathbb{E}_{\Lambda_k}[\ln\det\Lambda_k]=\sum_{j=1}^q\digamma\!\big(\tfrac{\hat\nu_k+1-j}{2}\big)+q\ln2-\ln\det\hat\Phi_k$.
Under \Cref{eq: C1C3} each digamma argument is $\tfrac{\hat\nu_k+1-j}{2}=\tfrac{n\hat\pi_k}{2}(1+O(1/n))$, so $\digamma(\tfrac{\hat\nu_k+1-j}{2})=\ln(n\hat\pi_k/2)+O(1/n)$
and, since $\hat\Phi_k=n\hat\pi_k\hat\Lambda_k^{-1}+O(1)$,
$\ln\det\hat\Phi_k=q\ln(n\hat\pi_k)-\ln\det\hat\Lambda_k+O(1/n)$. Combining:}
\begin{align}
\tfrac{1}{2}\mathbb{E}_{\Lambda_{k}} [\mathrm{ln}(det (\Lambda_{k})) ] &=\tfrac{1}{2}\mathrm{ln}(det (\hat\Lambda_{k}))+O(1/n),
\end{align}
\add{uniformly. For the derivative, applying \Cref{eq: digammaexp} to each
$\digamma(\tfrac{\hat\nu_k+1-j}{2})$ (with $a=\tfrac{\hat\pi_k}{2}$, so
$\partial_\Theta\digamma=\partial_\Theta\ln(n\hat\pi_k/2)+O(1/n)$) and using that
$\partial_\Theta\ln\det\hat\Phi_k=q\,\partial_\Theta\ln(n\hat\pi_k)-\partial_\Theta\ln\det\hat\Lambda_k+O(1/n)$
by \Cref{eq: ratioderiv} applied entrywise to $\hat\Phi_k=n\hat\pi_k\hat\Lambda_k^{-1}+O(1)$
, the
$\partial_\Theta\ln(n\hat\pi_k)$ terms cancel, leaving $\partial_\Theta\tfrac12\ln\det\hat\Lambda_k+O(1/n)$.
Thus the log-determinant remainder is $O(1/n)$ in $C^1$ norm. 
Returning to the block, we therefore have:
\begin{align}
\mathbb{E}_{\mu, \Lambda} [\mathrm{ln}(\mathcal{N}(x_i|\mu_{k}, \Lambda_{k}^{-1} ) )]
&=-\tfrac{q}{2}\mathrm{ln}(2\pi)+\tfrac{1}{2}\mathrm{ln}(det (\hat\Lambda_{k}))-\tfrac{1}{2} (x_i-\hat\mu_k)^T \hat\Lambda_k (x_i-\hat\mu_k)+\zeta_{ik}(\Theta),
\end{align}
\add{where $\zeta_{ik}(\Theta)$ satisfies
$\sup_{\overline{B_{\delta_1}}}(|\zeta_{ik}|+\|\nabla_\Theta \zeta_{ik}\|)\le C(1+\|x_i\|^2)/n$.} 
}

\add{\emph{Weight block.}
$\mathbb{E}_\pi[\ln\pi_k]=\digamma(\hat\alpha_k)-\digamma(K\alpha+n)$. With $\hat\alpha_k=\alpha+n\hat\pi_k$, the analysed map uses the
$\Theta$-independent $\digamma(K\alpha+n)$ instead of $\digamma(\sum_{k'}\hat\alpha_{k'})$
(the two agree on the prior's support and this is
a definition of $T$, not a constraint forcing $\sum_{k'}\hat\pi_{k'}=1$ on the ball).
\add{By \Cref{eq: digammaexp} the first digamma is $\ln(n\hat\pi_k)+O(1/n)$ (with $a=\hat\pi_k\ge c_*$)
and the second is the $\Theta$-constant $\digamma(K\alpha+n)=\ln n+O(1/n)$,}}
\begin{align}
\mathbb{E}_{\pi} [\mathrm{ln}(\pi_{k} ) ] = \big(\ln(n\hat\pi_k)+O(1/n)\big)-\big(\ln n+O(1/n)\big)=\mathrm{ln}(\hat\pi_k)+O(1/n),
\end{align}
\add{uniformly. The $\Theta$-derivative is $O(1/n)$ by the derivative half of
\Cref{eq: digammaexp}: $\partial_\Theta\digamma(\hat\alpha_k)=\partial_\Theta\ln(n\hat\pi_k)+O(1/n)=\partial_\Theta\ln\hat\pi_k+O(1/n)$,
while $\partial_\Theta\digamma(K\alpha+n)=0$. Hence this block is $O(1/n)$ in
$C^1$ norm; it carries no $x_i$-dependence (degree $0$). }

\add{\emph{Categorical block.}} With $\hat\eta_{k,g}=\eta+n\hat\pi_k\hat\psi_{k,g}$,
$\mathbb{E}_\psi[\ln\psi_{k,c_i}]=\digamma(\hat\eta_{k,c_i})-\digamma(\sum_g\hat\eta_{k,g})$.
\add{Here we keep the true normaliser $\sum_g\hat\eta_{k,g}=d\eta+n\hat\pi_k\sum_g\hat\psi_{k,g}$
(i.e. we do not assume $\sum_g\hat\psi_{k,g}=1$, as the ball need not lie on the categorical
simplex). Applying \Cref{eq: digammaexp} to each digamma, with
$a=\hat\pi_k\hat\psi_{k,c_i}\ge c_*^2$ (first term, using \Cref{eq: C1C3})
and $a=\hat\pi_k\sum_g\hat\psi_{k,g}\ge dc_*^2$ (second term),}
\begin{align}
\mathbb{E}_{\psi} [\mathrm{ln}(\psi_{k,c_{i}}) ] = \big(\ln(n\hat\pi_k\hat\psi_{k,c_i})+O(1/n)\big)-\big(\ln(n\hat\pi_k\textstyle\sum_g\hat\psi_{k,g})+O(1/n)\big)= \ln\bar\psi_{k,c_i}+O(1/n),
\end{align}
\add{uniformly, where $\bar\psi_{k,c_i}=\hat\psi_{k,c_i}/\sum_g\hat\psi_{k,g}$. The
$\Theta$-derivative is again $O(1/n)$ by the derivative half of
\Cref{eq: digammaexp}, the $\partial_\Theta\ln(n\hat\pi_k)$ terms cancelling
between the two digammas and leaving $\partial_\Theta\ln\bar\psi_{k,c_i}+O(1/n)$.
So this block is $O(1/n)$ in $C^1$ norm and carries no $x_i$-dependence. }


\add{\emph{Combining the blocks.} Adding the three blocks and comparing with
$\ln\rho^\infty_{ik}(\Theta)=\ln\hat\pi_k+\ln N(x_i\mid\hat\mu_k,\hat\Lambda_k^{-1})+\ln\bar\psi_{k,c_i}$,
every $n$-dependent additional term is $O((1+\|x_i\|^2)/n)$ in value and gradient
uniformly over the ball, establishing \Cref{eq: remainder}. 

Exponentiating this
value bound gives $\rho^\infty_{ik}e^{-\varepsilon_{ik}}\le\rho_{ik}\le\rho^\infty_{ik}e^{\varepsilon_{ik}}$
with $\varepsilon_{ik}=C(1+\|x_i\|^2)/n$, so
$|\rho_{ik}-\rho^\infty_{ik}|\le\rho^\infty_{ik}(e^{\varepsilon_{ik}}-1)\le\rho^\infty_{ik}\,\varepsilon_{ik}\,e^{\varepsilon_{ik}}$
by $e^t-1\le t\,e^t$. The product $\rho^\infty_{ik}e^{\varepsilon_{ik}}$ is bounded uniformly in $x_i$ and over the ball for all sufficiently large $n$: the exponent of $e^{\varepsilon_{ik}}$ grows in $x_i$ like $(C/n)\|x_i\|^2$, while the Gaussian factor in $\rho^\infty_{ik}$ decays at least as fast as $\exp(-\tfrac{c_*}{2}\|x_i-\hat\mu_k\|^2)$ by \Cref{eq: C1C3}, and $C/n<c_*/2$ for $n$ large. Hence $|\rho_{ik}-\rho^\infty_{ik}|\le C(1+\|x_i\|^2)/n$, which is \Cref{eq: rhoerr}.

For the
responsibilities, the softmax is $\tfrac12$-Lipschitz in the max-norm of its
argument, so
$|r_{ik}-r^\infty_{ik}|\le\tfrac12\max_{k'}|\ln\rho_{ik'}-\ln\rho^\infty_{ik'}|\le C(1+\|x_i\|^2)/n$,
the first bound of \Cref{eq: rerr}. For the derivative bound, the 
identity \Cref{eq: dr_ikdtheta_k'sum} writes $\nabla_\Theta r_{ik}$ as a sum of
triple products of a score gradient $\nabla_\Theta\ln\rho_{ik'}$ (bounded by
$C(1+\|x_i\|^2)$) with two responsibility factors in $[0,1]$, telescoping against
the $n$-free version leaves, as dominant terms, a bounded score gradient
($O(1+\|x_i\|^2)$) times a responsibility difference ($O((1+\|x_i\|^2)/n)$),
whose product is $O((1+\|x_i\|^2)^2/n)\le C(1+\|x_i\|^4)/n$, the second bound of
\Cref{eq: rerr}. }

\add{Finally, at $\Theta=\Theta^*$,
$\ln\rho^\infty_{ik}(\Theta^*)=\ln(\pi_k^* N(x_i\mid\mu_k^*,\Lambda_k^{*-1})\psi_{k,c_i}^*)$,
so $\rho_{ik}(\Theta^*)=\pi_k^*p(x_i,c_i\mid z_i=k,\Theta^*)+O(1/n)$ and
$r_{ik}(\Theta^*)\to\pi_k^*p(x_i,c_i\mid z_i=k,\Theta^*)/\phi^*(x_i,c_i)$, which is
the pointwise statement used in \Cref{subsec: corollary2proof}. The marginal
variances of the variational posterior, which are used in
\Cref{subsec: corollary4proof}, are collected in \Cref{thm: viposteriormoments}.}

\end{proof}


\paragraph{Verification of \textup{(A1)}.} \add{Hypothesis
\textup{(A1)} of \Cref{lemma: localcontraction}, however, needs the convergence
to be \emph{uniform} over the ball. Consequently we need a uniform law of large numbers argument for a triangular array.
Each entry of $\nabla T^\epsilon(\Theta)$ 
is a finite algebraic combination (addition, multiplication, multiplication of inverse) of weighted empirical averages of the form}
\begin{align}\label{eq: Gnh}
\add{G_n(h;\Theta)=\frac1n\sum_{i=1}^n h(x_i,c_i;\Theta,n),}
\end{align}
\add{whose summand $h(x_i,c_i;\Theta,n)$ is a fixed algebraic combination of the
responsibilities $r_{ik}(\Theta)$ and their $\Theta$-derivatives (as established
via \Cref{eq: dr_ikdtheta_k'sum,eq: dr_ikdtheta_k'} above), weighted by
$1$, $x_i$, $x_ix_i^\top$ or $I_{\{c_i=g\}}$ (these are the entries of the matrix
\Cref{eq: Tgradient}). Crucially the summand depends on $n$: through the
variational hyperparameters $\hat\alpha_k=n\hat\pi_k+\alpha$,
$\hat\eta_{k,g}=n\hat\pi_k\hat\psi_{k,g}+\eta$, etc.\ in
\Cref{eq: simplifiedtohyperparaappendix}. So the responsibilities $r_{ik}(\Theta)$
are really $r_{ik}(\Theta;n)$. Thus $\{h(x_i,c_i;\Theta,n)\}_{i\le n}$ is a triangular array, and a classical uniform law of large numbers, which requires a single,
$n$-independent integrand, cannot be applied to \Cref{eq: Gnh} directly.}

\add{To address this, we remove the $n$-dependence first. Let}
\begin{align}\label{eq: hinfty}
\add{h_\infty(x_i,c_i;\Theta)}
\end{align}
\add{denote the $n$ free surrogate of $h$: the same algebraic combination,
but with each $r_{ik}(\Theta;n)$ replaced by its $n$-free counterpart
$r^\infty_{ik}(\Theta)=\rho^\infty_{ik}(\Theta)\big/\sum_{j=1}^{K}\rho^\infty_{ij}(\Theta)$ from
\Cref{lemma: rho_iklimit}, which contains no $n$. By \Cref{lemma: rho_iklimit},
$|r_{ik}-r^\infty_{ik}|$ and $\|\nabla_\Theta r_{ik}-\nabla_\Theta r^\infty_{ik}\|$
are $O((1+\|x_i\|^p)/n)$ uniformly on the ball with $p\le4$; multiplying by the
weights $1,x_i,x_ix_i^\top$ raises the degree to $p\le6$, so}
\begin{align}\label{eq: persamplebound}
\add{\sup_{\Theta\in\overline{B_{\delta_1}(\Theta^*)}}\big|h(x_i,c_i;\Theta,n)-h_\infty(x_i,c_i;\Theta)\big|
\;\le\;\frac{C(1+\|x_i\|^{p})}{n},\qquad p\le6.}
\end{align}
\add{Now the difference between $G_n(h;\Theta)$ and the surrogate average is
controlled by the triangle inequality and
\Cref{eq: persamplebound}, uniformly on the ball,}
\begin{align}\label{eq: reduction}
\add{\sup_{\Theta\in\overline{B_{\delta_1}(\Theta^*)}}\Big|G_n(h;\Theta)-\frac1n\sum_{i=1}^n h_\infty(x_i,c_i;\Theta)\Big|
\;\le\;\frac1n\sum_{i=1}^n\ \sup_{\Theta\in\overline{B_{\delta_1}(\Theta^*)}}\big|h-h_\infty\big|
\;\le\;\frac{C}{n}\cdot\frac1n\sum_{i=1}^n\big(1+\|x_i\|^{p}\big).}
\end{align}
\add{The factor $\tfrac1n\sum_i(1+\|x_i\|^p)$ is an ordinary empirical mean of
i.i.d.\ variables with finite expectation, so it converges almost surely. The right-hand side of \Cref{eq: reduction} is therefore $O(1/n)\to0$
almost surely. This step reduces the problem to the
$n$-free average $\tfrac1n\sum_i h_\infty(x_i,c_i;\Theta)$.}

\add{The surrogate $h_\infty(x_i,c_i;\Theta)$ is a fixed ($n$-independent),
$\Theta$-continuous function with integrable envelope $G(x)=C(1+\|x\|^6)$, so the
uniform strong law of large numbers applies to it:}

\begin{lemma}[\add{Uniform SLLN; \citealp{Jennrich1969}, Thm.~2}]\label{lemma: uniformSLLN}
\add{Let $\{X_i\}$ be i.i.d.\ and let $\{h(\cdot;\Theta):\Theta\in \mathcal{K}\}$
be a family of measurable functions indexed by a compact set $\mathcal{K}$, such
that $\Theta\mapsto h(x;\Theta)$ is continuous for each $x$ and
$\sup_{\Theta\in\mathcal K}|h(x;\Theta)|\le G(x)$ with $\mathbb{E}[G(X_1)]<\infty$.
Then $\Theta\mapsto \mathbb{E}[h(X_1;\Theta)]$ is continuous and}
\[
\add{\sup_{\Theta\in\mathcal K}\Big|\tfrac1n\textstyle\sum_{i=1}^n h(X_i;\Theta)-\mathbb{E}[h(X_1;\Theta)]\Big|\xrightarrow{a.s.}0.}
\]
\end{lemma}

\add{Applying
\Cref{lemma: uniformSLLN} with $\mathcal{K}=\overline{B_{\delta_1}(\Theta^*)}$,
$X_i=(x_i,c_i)$ and $h=h_\infty(\cdot;\Theta)$ gives
\[
\sup_{\overline{B_{\delta_1}(\Theta^*)}}\Big|\tfrac1n\textstyle\sum_i h_\infty(x_i,c_i;\Theta)-H(\Theta)\Big|\to0 \quad \text{a.s.},
\]
with $H(\Theta)=\mathbb{E}_{\phi^*}[h_\infty(x,c;\Theta)]$ continuous.
Combining this with the reduction bound \Cref{eq: reduction} through the
triangle inequality,}
\begin{align*}
\add{\sup_{\overline{B_{\delta_1}(\Theta^*)}}\big|G_n(h;\Theta)-H(\Theta)\big|
\le \underbrace{\sup_{\overline{B_{\delta_1}}}\Big|G_n(h;\Theta)-\tfrac1n\textstyle\sum_i h_\infty\Big|}_{O(1/n)\ a.s. \text{by }\Cref{eq: reduction}}
+\underbrace{\sup_{\overline{B_{\delta_1}}}\Big|\tfrac1n\textstyle\sum_i h_\infty-H(\Theta)\Big|}_{\to0\ \text{by \Cref{lemma: uniformSLLN}}}
\xrightarrow{a.s.}0.}
\end{align*}
\add{Doing this for each of the finitely many entries of \Cref{eq: Tgradient},
and using that finite algebraic combinations preserve uniform convergence
when all terms and required inverses are uniformly bounded (which is ensured by our chosen compact domain), gives}
\begin{align}\label{eq: uniformgradconv}
\add{\sup_{\Theta\in \overline{B_{\delta_1}(\Theta^*)}}\big\|\nabla T^\epsilon(\Theta)-A^\epsilon(\Theta)\big\|_{op}\xrightarrow{a.s.}0,}
\end{align}
\add{where $A^\epsilon(\Theta)$ is the continuous limit operator whose value at
$\Theta^*$ is given by \Cref{eq: operatorlimits}. This is exactly hypothesis
\textup{(A1)} of \Cref{lemma: localcontraction}.}

\paragraph{Verification of \textup{(A2)}.} \add{We verify \textup{(A2)} by
computing the limit operator
$A^\epsilon(\Theta^*)=\lim_{n\to\infty}\nabla T^\epsilon_n(\Theta^*)$ explicitly and
checking $\|A^\epsilon(\Theta^*)\|_{op}<1$. By the general identity
\Cref{eq: dr_ikdtheta_k'} it then suffices to compute the log-derivatives
$\frac{\partial \ln(\rho_{ik'}) }{\partial \theta_{k'}}$ at $\Theta=\Theta^*$ as
$n \rightarrow \infty$.} $\ln(\rho_{ik}(\Theta))$ as a function of $\Theta$ is the same as the continuous data case (see \cite{BoWang2006}, above Eq $8$) except with the addition of the $\mathbb{E}_{\psi} [\mathrm{ln}(\psi_{k,c_{i}}) ]$ term, and because 
\begin{align}  \label{eq: extraEpsi}
\mathbb{E}_{\psi} [\mathrm{ln}(\psi_{k,c_{i}}) ] &=\digamma(\eta+\sum_{s=1}^{n} I_{\{c_{s}=c_{i}\}}r_{sk} )-\digamma(d\eta+\sum_{g=1}^{d} \sum_{s=1}^{n} I_{\{c_{s}=g\}} r_{sk}) \nonumber \\
&=\digamma(\eta+ n \hat{\pi}_k\hat{\psi}_{k,c_i})-\digamma(d\eta+n \hat{\pi}_k \sum_{g=1}^{d} \hat{\psi}_{k,g})
\end{align}
is independent of $\bm{\hat{\mu}}, \bm{\hat{\Lambda}}$, we therefore have $\frac{\partial \ln(\rho_{ik}) }{\partial \hat{\mu}_{kl}}$, and $\frac{\partial \ln(\rho_{ik}) }{\partial \hat{\Lambda}_{kjl}}$ are unchanged from the continuous data case, which as $n \rightarrow \infty$ are of the form (see \cite{BoWang2006}):

\begin{align}
\frac{\partial \ln(\rho_{ik}) }{\partial \hat{\mu}_{kl}} \rightarrow \sum_{j=1}^{q} \hat{\Lambda}_{klj}(x_{ij} - \hat{\mu}_{kj} )
\end{align}
\begin{align}
\frac{\partial \ln(\rho_{ik}) }{\partial \hat{\Lambda}_{kjl}} \rightarrow \frac{1}{2}[\hat{\Lambda}_{kjl}^{-1}- (x_{ij} - \hat{\mu}_{kj})(x_{il} - \hat{\mu}_{kl}) ]
\end{align}

for $\frac{\partial \ln(\rho_{ik}) }{\partial \hat{\psi}_{k,g}} $, we use the following lemma:

\begin{lemma}\label{lemma: digammaderivative}
Let $\digamma$ be the digamma function on the positive reals, and $b_1>0, b_2 \geq 0, b_3 \geq 0$ be real constants and $n$ positive integer. Let $\lambda=b_1nx+b_2n+b_3$, then:

$$\frac{\partial \digamma(\lambda) }{\partial x } \sim  \frac{\partial \lambda  }{\partial x } (  \lambda-1)^{-1}=b_1n(b_1nx+b_2n+b_3-1)^{-1} \rightarrow (x+b_2/b_1)^{-1}$$ as $n \rightarrow \infty$.
\end{lemma}

\begin{proof}
See \cite{BoWang2006}, Appendix D, Pages 643-644, which proves the case for $\lambda=nx+b_3$ but the proof is the same except substituting $\lambda$ instead for $\lambda=b_1nx+b_2n+b_3$ and changing the chain rule derivative $\frac{\partial \lambda  }{\partial x }=b_1n$ accordingly.
\end{proof}

Using \Cref{lemma: digammaderivative}, for $\frac{\partial \ln(\rho_{ik}) }{\partial \hat{\psi}_{k,g}} $, we have:
\begin{align} \label{eq: dlnrho_ikdpsi_kg} 
& \frac{\partial \ln(\rho_{ik}) }{\partial \hat{\psi}_{k,g}}=\frac{\partial \digamma(\eta+ n \hat{\pi}_k\hat{\psi}_{k,c_i}) }{\partial \hat{\psi}_{k,g} }-\frac{\partial \digamma(d\eta+n \hat{\pi}_k \sum_{g=1}^{d} \hat{\psi}_{k,g}) }{\partial \hat{\psi}_{k,g} } \nonumber \\
\sim & \frac{\partial (\eta+ n \hat{\pi}_k\hat{\psi}_{k,c_i}) }{\partial \hat{\psi}_{k,g}  } (  \eta+ n \hat{\pi}_k\hat{\psi}_{k,c_i}-1)^{-1}-\frac{\partial (d\eta+n \hat{\pi}_k \sum_{g=1}^{d} \hat{\psi}_{k,g}) }{\partial \hat{\psi}_{k,g} } (d\eta+n \hat{\pi}_k \sum_{g=1}^{d} \hat{\psi}_{k,g}-1)^{-1}  \nonumber \\ 
\rightarrow & \frac{\delta_{g c_i}}{\hat{\psi}_{k,c_i} }-(n \hat{\pi}_k)(n \hat{\pi}_k \sum_{g=1}^{d} \hat{\psi}_{k,g})^{-1}=\frac{\delta_{g c_i}}{\hat{\psi}_{k,g} }-\frac{1}{\sum_{g=1}^{d} \hat{\psi}_{k,g}}
\end{align}
where $\delta_{g c_i}=I_{\{c_{i}=g\}}$. Lastly, $\frac{\partial \ln(\rho_{ik}) }{\partial \hat{\pi}_{k}}$ is a little more involved, as the additional $\mathbb{E}_{\psi} [\mathrm{ln}(\psi_{k,c_{i}}) ]$  term \Cref{eq: extraEpsi} does depend on $\hat{\pi}_{k}$. Again using \Cref{lemma: digammaderivative}, we have:
\begin{align}
\frac{\partial \mathbb{E}_{\psi} [\mathrm{ln}(\psi_{k,c_{i}}) ]}{\partial \hat{\pi}_{k}} &=\frac{\partial \digamma(\eta+ n \hat{\pi}_k\hat{\psi}_{k,c_i}) }{\partial \hat{\pi}_{k} }-\frac{\partial \digamma(d\eta+n \hat{\pi}_k \sum_{g=1}^{d} \hat{\psi}_{k,g} ) }{\partial \hat{\pi}_{k} }  \nonumber \\ 
\sim &  \frac{\partial (\eta+ n \hat{\pi}_k\hat{\psi}_{k,c_i}) }{\partial \hat{\pi}_{k} } (  \eta+ n \hat{\pi}_k\hat{\psi}_{k,c_i}-1)^{-1}-\frac{\partial (d\eta+n \hat{\pi}_k \sum_{g=1}^{d} \hat{\psi}_{k,g} ) }{\partial \hat{\pi}_{k} } (d\eta+n \hat{\pi}_k \sum_{g=1}^{d} \hat{\psi}_{k,g}-1)^{-1}  \nonumber \\
\rightarrow & (n \hat{\psi}_{k,c_i}) (n \hat{\pi}_k\hat{\psi}_{k,c_i} )^{-1} - n \sum_{g=1}^{d} \hat{\psi}_{k,g} (n \hat{\pi}_k \sum_{g=1}^{d} \hat{\psi}_{k,g})^{-1} = \hat{\pi}_k^{-1}-\hat{\pi}_k^{-1} = 0
\end{align}
It was shown in \cite{BoWang2006} that $\frac{\partial \mathbb{E}_{\mu, \Lambda} [\mathrm{ln}(\mathcal{N}(x_i|\mu_{k}, \Lambda_{k}^{-1} ) )]}{\partial \hat{\pi}_{k}} \rightarrow 0$ as $n \rightarrow \infty$. So $\frac{\partial \ln(\rho_{ik}) }{\partial \hat{\pi}_{k}}$ as in the continuous data case only depends on the $\mathbb{E}_{\pi} [\mathrm{ln}(\pi_{k} ) ]$ term as $n \rightarrow \infty$. Recall that 

\begin{align} \label{eq: Epilnpik2} 
\mathbb{E}_{\pi} [\mathrm{ln}(\pi_{k} ) ]=\digamma(\hat{\alpha}_k)-\digamma(\sum_{k=1}^{K}\hat{\alpha}_k)=\digamma(\alpha + n \hat{\pi}_k )-\digamma(K\alpha + n\sum_{k=1}^{K} \hat{\pi}_k)
\end{align}

As in \cite{BoWang2006}, rather than \Cref{eq: Epilnpik2} we instead use $\digamma(\alpha + n \hat{\pi}_k )-\digamma(K\alpha + n)$, which coincides with \Cref{eq: Epilnpik2} when $\sum_{k=1}^{K} \hat{\pi}_k=1$. This is always satisfied by the CAVI algorithm after the first update (and satisfied by the initial values too if the $r_{ik}$ are initialised to satisfy $\sum_{k=1}^{K} r_{ik}=1$), so $T(\Theta)$ still matches our proposed algorithm for $\Theta$ that lies in the support of the prior. This will simplify some of the following calculations by ensuring that $\frac{\partial \ln(\rho_{ik}) }{\partial \hat{\pi}_{k'}}=0$ if $k' \neq k$ \footnote{In contrast, we did not take $\sum_{g=1}^{d} \hat{\psi}_{k,g}$=1 in \Cref{eq: extraEpsi}, as this results in \Cref{eq: dlnrho_ikdpsi_kg}, and its expectation $\mathbb{E}[ (\frac{\delta_{g c}}{\psi_{k,g}^{*} }-\frac{1}{\sum_{g=1}^{d} \psi_{k,g}^{*}})\frac{\phi_k^*}{\phi^*} ]=0$, which will simplify some of the subsequent calculations.}.
By \Cref{lemma: digammaderivative}, we have:
\begin{align}
&\frac{\partial \ln(\rho_{ik}) }{\partial \hat{\pi}_{k}}  \rightarrow \frac{\partial \digamma(\alpha+ n \hat{\pi}_k) }{\partial \hat{\pi}_{k} }-\frac{\partial \digamma(K \alpha+n  ) }{\partial \hat{\pi}_{k} }  \nonumber \\
\sim &\frac{\partial (\alpha+ n \hat{\pi}_k) }{\partial \hat{\pi}_{k} } (\alpha+ n \hat{\pi}_k-1)^{-1}  \nonumber \\
\rightarrow & n(n \hat{\pi}_{k} )^{-1} = \frac{1}{\hat{\pi}_{k} }   
\end{align}

Now that we have computed the limits of the derivatives of $\mathrm{ln}(\rho_{ik} )$ as $n \rightarrow \infty$, we can return to our original objective of deriving the matrix of derivatives in \Cref{eq: Tgradient} as $n \rightarrow \infty$ at $\Theta=\Theta^*$. 
Let $\phi_k^*=p(x,c|z=k,\Theta^*)$ and $\phi^*=\sum_{k'=1}^{K}\pi_{k'}^*p(x,c|z=k',\Theta^*)$ denote the likelihood evaluated at an arbitrary point $x,c$ with the true parameter values $\Theta^*$ for a specific mixture component and with the mixture component integrated out, respectively. Also let:
\begin{align}
R_{k}^{\pi}&= \frac{\phi_{k}^* }{\phi^{*}}   \nonumber \\
R_{k}^{\mu}&= \Lambda^{*}_{k}(x - \mu^{*}_{k} )\frac{\phi_{k}^*\pi_{k}^* }{\phi^{*}}  \nonumber \\
R_{k}^{\Lambda}&= \frac{1}{2}[\Lambda_{k}^{*-1}- (x - \mu_{k}^{*})(x - \mu_{k}^{*})^T ]\frac{\phi_{k}^*\pi_{k}^* }{\phi^{*}}  \nonumber \\
R_{k}^{\psi}&: R_{kg'}^{\psi}= (\frac{\delta_{g' c}}{\psi_{k,g'}^*}-1)\frac{\phi_{k}^* \pi_{k}^* }{\phi^{*}}  \nonumber \\
L_{k}^{\pi}&= \frac{\pi_{k}^*\phi_{k}^* }{\phi^{*}} \nonumber \\
L_{k}^{\mu}&= (x - \mu^{*}_{k} )\frac{\phi_{k}^* }{\phi^{*}} \nonumber \\
L_{k}^{\Lambda}&= [\Lambda_{k}^{*} - \Lambda_{k}^{*}(x - \mu_{k}^{*})(x - \mu_{k}^{*})^T \Lambda_{k}^{*}]\frac{\phi_{k}^* }{\phi^{*}} \nonumber \\
L_{k}^{\psi}&: L_{kg'}^{\psi}=\frac{(\delta_{g' c}-\psi_{k,g'}^*)\phi_{k}^* }{ \phi^{*}}
\end{align}
and $R^{\theta}=[(R_{1}^{\theta})^T,\dots,(R_{K}^{\theta})^T ]$ and 
$L^{\theta}=
\begin{bmatrix}
L_{1}^{\theta} \\
\vdots \\
L_{K}^{\theta}
\end{bmatrix}$
for $\theta \in \{ \pi, \mu, \Lambda, \psi \}$.
 
One can then show the gradients inside \Cref{eq: Tgradient} take the following forms as $n \rightarrow \infty$:
\begin{align} \label{eq: operatorlimits}
\nabla_{\bm{\hat{\pi}}} \Pi (\Theta^*) \rightarrow I-\mathbb{E}[L^{\pi} R^{\pi}]  \nonumber \\
\nabla_{\bm{\hat{\mu}}} \Pi (\Theta^*) \rightarrow -\mathbb{E}[L^{\pi} R^{\mu}] \nonumber \\
\nabla_{\bm{\hat{\Lambda}}} \Pi (\Theta^*) \rightarrow -\mathbb{E}[L^{\pi} R^{\Lambda}] \nonumber \\ 
\nabla_{\bm{\hat{\psi}}} \Pi (\Theta^*) \rightarrow -\mathbb{E}[L^{\pi} R^{\psi}]  \nonumber \\
\nabla_{\bm{\hat{\pi}}} M (\Theta^*) \rightarrow -\mathbb{E}[L^{\mu} R^{\pi}] \nonumber \\
\nabla_{\bm{\hat{\mu}}} M (\Theta^*) \rightarrow I-\mathbb{E}[L^{\mu} R^{\mu}] \nonumber \\
\nabla_{\bm{\hat{\Lambda}}} M (\Theta^*) \rightarrow -\mathbb{E}[L^{\mu} R^{\Lambda}] \nonumber \\ 
\nabla_{\bm{\hat{\psi}}} M (\Theta^*) \rightarrow -\mathbb{E}[L^{\mu} R^{\psi}] \nonumber \\
\nabla_{\bm{\hat{\pi}}} S (\Theta^*) \rightarrow -\mathbb{E}[L^{\Lambda} R^{\pi}] \nonumber \\
\nabla_{\bm{\hat{\mu}}} S (\Theta^*) \rightarrow -\mathbb{E}[L^{\Lambda} R^{\mu}] \nonumber \\ 
\nabla_{\bm{\hat{\Lambda}}} S (\Theta^*) \rightarrow I-\mathbb{E}[L^{\Lambda} R^{\Lambda}] \nonumber \\
\nabla_{\bm{\hat{\psi}}} S (\Theta^*) \rightarrow -\mathbb{E}[L^{\Lambda} R^{\psi}] \nonumber \\
\nabla_{\bm{\hat{\pi}}} \Psi (\Theta^*) \rightarrow -\mathbb{E}[L^{\psi} R^{\pi}] \nonumber \\ 
\nabla_{\bm{\hat{\mu}}} \Psi (\Theta^*) \rightarrow -\mathbb{E}[L^{\psi} R^{\mu}] \nonumber \\ 
\nabla_{\bm{\hat{\Lambda}}} \Psi (\Theta^*) \rightarrow -\mathbb{E}[L^{\psi} R^{\Lambda}] \nonumber \\ 
\nabla_{\bm{\hat{\psi}}} \Psi (\Theta^*) \rightarrow I-C^{\psi}-\mathbb{E}[L^{\psi} R^{\psi}]
\end{align}
where $C^{\psi}: \mathbb{R}^{Kd} \rightarrow \mathbb{R}^{Kd}, C^{\psi}_{kgk'g'}=\delta_{kk'}\psi^{*}_{k,g}$. 
So if $a \in \mathbb{R}^{Kd}, u \in \mathbb{R}^{K}$:
$$\sum_{k'=1}^{K}\sum_{g'=1}^{d} C^{\psi}_{kgk'g'}a_{k'g'}=\sum_{k'=1}^{K}\sum_{g'=1}^{d}\delta_{kk'}\psi^{*}_{k,g} a_{k'g'}=\sum_{g'=1}^{d}\psi^{*}_{k,g} a_{kg'}$$

The expectation in \Cref{eq: operatorlimits} is taken with respect to $x,c$, which has density of the likelihood evaluated at the true data generating parameter $\Theta^*$. To evaluate the gradient limits in \Cref{eq: operatorlimits}, we use a strong law of large numbers like \Cref{lemma: modifiedSLLN}. This lemma has conditions that are much easier to verify than requiring the responsibilities $r_{ik}$ to converge uniformly as functions of $x_i$, which was assumed in \cite{BoWang2006}. As the gradient limit calculations are long, we reserve the details to \Cref{subsec: gradientcalcs}.

We now return to the main objective of showing that the operator $\nabla T^{\epsilon}(\Theta^*)$ converges as $n \rightarrow \infty$ (almost surely) to an operator $A^{\epsilon}(\Theta^*)$ with norm less than $1$, where the norm is induced by  $\langle .,. \rangle$. To do this, we show $A^{\epsilon}(\Theta^*) \prec I$ and $A^{\epsilon}(\Theta^*) \succ -I$, where for example $A^{\epsilon}(\Theta^*) \prec I$ means $I-A^{\epsilon}(\Theta^*)$ is positive definite. To see why this is sufficient to imply $\| A^{\epsilon}(\Theta^*) \|_{op}<1$, note that $A^{\epsilon}(\Theta^*) \prec I$ implies $\langle B,A^{\epsilon}(\Theta^*)B \rangle \leq \langle B, B \rangle $ for all $B$ with equality iff $B=0$, and similarly $A^{\epsilon}(\Theta^*) \succ -I$ implies $\langle B,A^{\epsilon}(\Theta^*)B \rangle \geq -\langle B, B \rangle $ with equality iff $B=0$. Combined this implies $| \langle B,A^{\epsilon}(\Theta^*)B \rangle | < | \langle B, B \rangle |=1$ for all $\| B \|=1$. Because the unit sphere $\{ B \in \mathbb{S}_{\pi} \bigoplus \mathbb{S}_{\mu} \bigoplus \mathbb{S}_{\Lambda} \bigoplus \mathbb{S}_{\psi}: \| B \|=1 \}$ is compact, $| \langle B,A^{\epsilon}(\Theta^*)B \rangle |$ attains its supremum on the unit sphere which is smaller than $1$. Lastly note that $A^{\epsilon}(\Theta^*)$ is self adjoint \footnote{One can easily show similar to \Cref{eq: BdotL} that $\langle L,B \rangle=R(B)$, and so $\mathbb{E}[\langle LR(B_1),B_2 \rangle ]=\mathbb{E}[R(B_1)\langle L,B_2 \rangle ]=\mathbb{E}[R(B_1)R(B_2)]=\mathbb{E}[\langle B_1,L \rangle R(B_2)]=\mathbb{E}[\langle B_1,L R(B_2) \rangle]$. Likewise, one can easily derive $\langle C^{\psi}B_1,B_2 \rangle$ and $\langle B_1,C^{\psi}B_2 \rangle$ similar to \Cref{eq: greaterthanminusIinequalityotherterms} and show $\langle C^{\psi}B_1,B_2 \rangle=\langle B_1,C^{\psi}B_2 \rangle$ } and recall from Hilbert space theory that the operator norm for self-adjoint operators $A$ is equal to $\mathrm{sup}\{ | \langle B,A B \rangle | : \|B \| = 1,  B \in \mathbb{S}_{\pi} \bigoplus \mathbb{S}_{\mu} \bigoplus \mathbb{S}_{\Lambda} \bigoplus \mathbb{S}_{\psi} \} $.

Let $u \in \mathbb{S}_{\pi}$, $v \in \mathbb{S}_{\mu}$, $W \in \mathbb{S}_{\Lambda}$, $a \in \mathbb{S}_{\psi}$, $B=\begin{bmatrix} u \\  v \\ W \\ a \end{bmatrix} \in \mathbb{S}_{\pi} \bigoplus \mathbb{S}_{\mu} \bigoplus \mathbb{S}_{\Lambda} \bigoplus \mathbb{S}_{\psi}$.

\add{The strict positivity below rests on the following linear-independence
fact, which is where \Cref{main-ass: distinct} enters.}
\begin{lemma}[\add{Gaussian polynomial linear independence}]\label{lemma: linind}
\add{Let $(\mu_1,\Lambda_1),\dots,(\mu_K,\Lambda_K)$ be pairwise distinct with each
$\Lambda_k$ symmetric positive definite, and let $\varphi_1,\dots,\varphi_K$ be
polynomials on $\mathbb{R}^q$. If $\sum_{k=1}^K \varphi_k(x)\,N(x\mid\mu_k,\Lambda_k^{-1})=0$
for almost every $x$, then $\varphi_k\equiv 0$ for every $k$.}
\end{lemma}
\begin{proof}
\add{Each summand is analytic, so the finite sum is analytic and
vanishing a.e.\ forces vanishing everywhere. Expanding
$-\tfrac12(x-\mu_k)^\top\Lambda_k(x-\mu_k)=-x^\top A_k x+b_k^\top x-\tfrac12\mu_k^\top\Lambda_k\mu_k$
with $A_k:=\tfrac12\Lambda_k$ and $b_k:=\Lambda_k\mu_k$, and absorbing the
constant $\tfrac{|\Lambda_k|^{1/2}}{(2\pi)^{q/2}}e^{-\frac12\mu_k^\top\Lambda_k\mu_k}$
into the polynomial (call the result $p_k$, which is $\equiv0$ iff
$\varphi_k\equiv0$), the hypothesis becomes}
\begin{equation}\label{eq: linindstar}
\add{\sum_{k=1}^{K}p_k(x)\,\exp\!\big(-x^\top A_kx+b_k^\top x\big)=0\qquad\forall x\in\mathbb{R}^q.}
\end{equation}
\add{Suppose, for contradiction, that not every $\varphi_k\equiv0$. Discard the
zero terms so that $p_k\not\equiv0$ for all remaining $k$. The pairs $(A_k,b_k)$
are pairwise distinct: $A_j=A_k$ gives $\Lambda_j=\Lambda_k$, and then $b_j=b_k$
gives $\Lambda_k\mu_j=\Lambda_k\mu_k$, hence $\mu_j=\mu_k$ ($\Lambda_k$
invertible), contradicting distinctness of $(\mu_k,\Lambda_k)$.}

\add{The idea is to evaluate \Cref{eq: linindstar} along a ray $x=tw$ and let
$t\to\infty$, choosing $w$ so that a single term dominates. To do this we need to choose $w$ outside certain problematic sets:}
\begin{itemize}
\item[]\add{$Z_{j,k}=\{w:w^\top(A_j-A_k)w=0\}$ for $j\ne k$ with $A_j\ne A_k$ (zero set of a nonzero quadratic form);}
\item[]\add{$H_{j,k}=\{w:(b_j-b_k)^\top w=0\}$ for $j\ne k$ with $A_j= A_k$ (a hyperplane, so $b_j\ne b_k$);}
\item[]\add{$\Omega_k=\{w:\overline{p}_k(w)=0\}$, where $\overline{p}_k$ is the nonzero leading homogeneous part of $p_k$ (zero set of a nonzero polynomial).}
\end{itemize}
\add{Each set is Lebesgue measure 0, so their finite union $E$ is measure 0 and we may fix
$w\notin E$. Put $x=tw$: the $k$th exponent is $-t^2\,(w^\top A_kw)+t\,(b_k^\top w)$.
Let $a_{\min}=\min_k w^\top A_kw$ and $S=\{k:w^\top A_kw=a_{\min}\}$. For
$j,k\in S$, $w\notin Z_{j,k}$ forces $A_j=A_k$, so all indices in $S$ share one
matrix; since the $(A_k,b_k)$ are distinct, the $b_k$, $k\in S$, are distinct, and
$w\notin H_{j,k}$ makes the numbers $b_k^\top w$ ($k\in S$) distinct. Let
$k^*\in S$ maximise $b_k^\top w$.}

\add{Divide \Cref{eq: linindstar} at $x=tw$ by
$\exp(-t^2 w^\top A_{k^*}w+t\,b_{k^*}^\top w)$:}
\begin{equation}\label{eq: dividedray}
\add{p_{k^*}(tw)+\sum_{k\ne k^*}p_k(tw)\,\exp\!\big(-t^2\,w^\top(A_k-A_{k^*})w+t\,(b_k-b_{k^*})^\top w\big)=0.}
\end{equation}
\add{For each $k\ne k^*$ the exponent tends to $-\infty$: if $k\notin S$ then
$w^\top(A_k-A_{k^*})w>0$ and the $-t^2$ term dominates; if $k\in S$ then
$A_k=A_{k^*}$ and $(b_k-b_{k^*})^\top w<0$ by choice of $k^*$, so the $-t$ term
dominates. As exponential decay outpaces the polynomial growth of $p_k(tw)$, every
such term $\to0$. Hence $p_{k^*}(tw)\to0$. But
$p_{k^*}(tw)=t^{d}\,\overline{p}_{k^*}(w)+O(t^{d-1})$ with $d=\deg p_{k^*}$ and
$\overline{p}_{k^*}(w)\ne0$ (as $w\notin\Omega_{k^*}$), so $|p_{k^*}(tw)|\to\infty$
(or equals the nonzero constant $|\overline{p}_{k^*}(w)|$ if $d=0$), a
contradiction. Therefore every $\varphi_k\equiv0$.}
\end{proof}

Let $R: \mathbb{S}_{\pi} \bigoplus \mathbb{S}_{\mu} \bigoplus \mathbb{S}_{\Lambda} \bigoplus \mathbb{S}_{\psi} \rightarrow \mathbb{R} $ such that:
\begin{align}
R(B)=\sum_{k=1}^{K} R_{k}^{\pi} u_k+ \sum_{k=1}^{K} (R_{k}^{\mu} )^T v_k+\sum_{k=1}^{K} Tr(R_{k}^{\Lambda} W_k)+\sum_{k=1}^{K} (R_{k}^{\psi} )^Ta_k
\end{align}

Similarly let $L=\begin{bmatrix} L_{\pi} \\  L_{\mu} \\ L_{\Lambda} \\ L_{\psi} \end{bmatrix}$. Note we deliberately constructed the inner product $\langle .,. \rangle$ so that:

\begin{align}\label{eq: BdotL}
\langle B,L \rangle &=  \sum_{k=1}^{K} u_k (\frac{1}{\pi_{k}^*}) L_{k}^{\pi} + \sum_{k=1}^{K} (v_k)^T (\pi_k^* \Lambda_k^*) L_{k}^{\mu} +\sum_{k=1}^{K}\frac{\pi_k^*}{2} Tr(W_k \Lambda_k^{*-1} L_{k}^{\Lambda} \Lambda_k^{*-1})+\sum_{k=1}^{K} a_k^T diag(\frac{\pi_k^*}{\psi_k^*}) L_{k}^{\psi} \nonumber \\
&= \sum_{k=1}^{K} u_k R_{k}^{\pi} + \sum_{k=1}^{K} v_k^T R_{k}^{\mu} +\sum_{k=1}^{K} Tr(R_{k}^{\Lambda} W_k)+ \sum_{k=1}^{K} a_k^T R_{k}^{\psi}  \nonumber \\
&= R(B) 
\end{align}

For the case $A^{\epsilon}(\Theta^*) \succ -I$. From \Cref{eq: operatorlimits}, we have that $A^{1}(\Theta^*)= I-C^{\psi}-\mathbb{E}[LR] $. Here we use a slight abuse of notation where we extend the map $C^{\psi} : \mathbb{S}_{\pi} \bigoplus \mathbb{S}_{\mu} \bigoplus \mathbb{S}_{\Lambda} \bigoplus \mathbb{S}_{\psi} \rightarrow \mathbb{S}_{\pi} \bigoplus \mathbb{S}_{\mu} \bigoplus \mathbb{S}_{\Lambda} \bigoplus \mathbb{S}_{\psi}$ by simply mapping the extended dimensions to $0$ so:

\begin{align} \label{eq: greaterthanminusIinequality}
& A^{\epsilon}(\Theta^*) \succ -I \iff (1-\epsilon)I + \epsilon A^{1}(\Theta^*) \succ -I \iff C^{\psi}+\mathbb{E}[LR] \prec \frac{2}{\epsilon} I  \nonumber \\
& \iff \langle B,C^{\psi}B \rangle + \mathbb{E}[\langle B,LR(B) \rangle ] \leq \frac{2}{\epsilon} \langle B,B \rangle \quad \textit{with equality iff }B=0 
\end{align} 
where
\begin{align} \label{eq: greaterthanminusIinequalityotherterms}
\langle B,C^{\psi}B \rangle &= \sum_{k=1}^{K}\sum_{g=1}^{d}a_{kg} \frac{\pi_k^*}{\psi_{k,g}^*}\sum_{g'=1}^{d}\psi^{*}_{k,g} a_{kg'} =  \sum_{k=1}^{K}\pi_k^* \sum_{g=1}^{d}a_{kg}\sum_{g'=1}^{d} a_{kg'} = \sum_{k=1}^{K}\pi_k^* (a_k^T \bm{1})^2  \\
\langle B,B \rangle &= \sum_{k=1}^{K}  u_k^2(\frac{1}{\pi_{k}^*}) +\sum_{k=1}^{K} v_k^T(\pi_k^* \Lambda_k^*)v_k+\sum_{k=1}^{K} \frac{\pi_k^*}{2}Tr(W_k\Lambda_k^{*-1}W_k\Lambda_k^{*-1})+\sum_{k=1}^{K} a_k^T diag(\frac{\pi_k^*}{\psi_k^*}) a_k \nonumber 
\end{align}
where $\bm{1}$ is a vector of dimension $d$ where all the entries are $1$. 

$\mathbb{E}[\langle B,LR(B) \rangle ]$ is a much more involved term, for which we will derive an upper bound for below:
\begin{align}
& \mathbb{E}[\langle B,LR(B) \rangle ]=\mathbb{E}[ (R(B))R(B) ] \\
&= \mathbb{E}[ R(B)^2]
\end{align} 
\begin{align}
&\mathbb{E}[ R(B)^2] \nonumber \\
&=\mathbb{E}[ \Bigl\{ \sum_{k=1}^{K} \frac{\pi_k^*\phi_k^*}{\phi^*} (u_k(\frac{1}{\pi_{k}^*})+(\Lambda^{*}_{k}(x - \mu^{*}_{k}))^T v_k \nonumber \\
&\qquad + Tr(\frac{1}{2}[\Lambda_{k}^{*-1}- (x - \mu_{k}^{*})(x - \mu_{k}^{*})^T ] W_k)+ \bm{\frac{\delta_{c}}{\psi_{k}^* }}^T a_k- a_k^T \bm{1} ) \Bigl\} ^2]  \nonumber \\ 
& \leq \mathbb{E}[ \sum_{k=1}^{K} \frac{\pi_k^*\phi_k^*}{\phi^*} (u_k(\frac{1}{\pi_{k}^*})+(\Lambda^{*}_{k}(x - \mu^{*}_{k}))^T v_k \nonumber \\
&\qquad + Tr(\frac{1}{2}[\Lambda_{k}^{*-1}- (x - \mu_{k}^{*})(x - \mu_{k}^{*})^T ] W_k)+ \bm{\frac{\delta_{c}}{\psi_{k}^* }}^T a_k-  a_k^T \bm{1} )^2]  \nonumber \\
& = \mathbb{E}[ \sum_{k=1}^{K} \frac{\pi_k^*\phi_k^*}{\phi^*} (u_k^2(\frac{1}{\pi_{k}^*})^2+v_k^T\Lambda^{*}_{k}(x - \mu^{*}_{k})(x - \mu^{*}_{k})^T \Lambda^{*}_{k} v_k+ Tr(\frac{1}{2}[\Lambda_{k}^{*-1}- (x - \mu_{k}^{*})(x - \mu_{k}^{*})^T ] W_k)^2  \nonumber \\
&+ (\bm{\frac{\delta_{c}}{\psi_{k}^* }}^T a_k- a_k^T \bm{1} )^2+ 2u_k(\frac{1}{\pi_{k}^*}) (x - \mu^{*}_{k})^T \Lambda^{*}_{k} v_k+2u_k(\frac{1}{\pi_{k}^*})Tr(\frac{1}{2}[\Lambda_{k}^{*-1}- (x - \mu_{k}^{*})(x - \mu_{k}^{*})^T ] W_k)  \nonumber \\
&+ 2u_k(\frac{1}{\pi_{k}^*}) (\bm{\frac{\delta_{c}}{\psi_{k}^* }}^T a_k- a_k^T \bm{1})+2(x - \mu^{*}_{k})^T \Lambda^{*}_{k} v_k Tr(\frac{1}{2}[\Lambda_{k}^{*-1}- (x - \mu_{k}^{*})(x - \mu_{k}^{*})^T ] W_k) \nonumber \\
&+2(x - \mu^{*}_{k})^T \Lambda^{*}_{k} v_k (\bm{\frac{\delta_{c}}{\psi_{k}^* }}^T a_k- a_k^T \bm{1} )+2Tr(\frac{1}{2}[\Lambda_{k}^{*-1}- (x - \mu_{k}^{*})(x - \mu_{k}^{*})^T ] W_k)  (\bm{\frac{\delta_{c}}{\psi_{k}^* }}^T a_k- a_k^T \bm{1} ) ]
\end{align} 
Here $\bm{\frac{\delta_{c}}{\psi_{k}^* }}$ is a vector of dimension $d$ with components $\frac{\delta_{cg}}{\psi_{k,g}^* }$. The inequality above is due to $| \sum_{k=1}^{K} \xi_k \eta_k | ^2 \leq \sum_{k=1}^{K} \xi_k^2 \eta_k$ for $\eta_k \geq 0$ and $\sum_{k=1}^K \eta_k=1$, and any real $\xi_k$. One can show this inequality by applying Cauchy-Schwarz on $\xi_k \sqrt{\eta_k}$ and $\sqrt{\eta_k}$. 
There are a lot of terms here but it will turn out the `diagonal' terms will simplify nicely, and the non diagonal terms will have expected value of $0$. 
\begin{align}  \label{eq: LRdiagonalterms}
&\mathbb{E}[ \sum_{k=1}^{K}  \frac{\pi_k^*\phi_k^*}{\phi^*} u_k^2(\frac{1}{\pi_{k}^*})^2 ]=\sum_{k=1}^{K} \pi_{k}^* u_k^2(\frac{1}{\pi_{k}^*})^2 = \sum_{k=1}^{K} u_k^2(\frac{1}{\pi_{k}^*}) \nonumber \\
&\mathbb{E}[\sum_{k=1}^{K}  \frac{\pi_k^*\phi_k^*}{\phi^*} (v_k^T\Lambda^{*}_{k}(x - \mu^{*}_{k})(x - \mu^{*}_{k})^T \Lambda^{*}_{k} v_k) ]=\sum_{k=1}^{K} v_k^T \pi^{*}_k\Lambda^{*}_{k} v_k  \nonumber \\
&\mathbb{E}[\sum_{k=1}^{K} \frac{\pi_k^*\phi_k^*}{\phi^*} (Tr(\frac{1}{2}[\Lambda_{k}^{*-1}- (x - \mu_{k}^{*})(x - \mu_{k}^{*})^T ] W_k)^2)]=\sum_{k=1}^{K} \frac{\pi_k}{2} Tr(\Lambda_{k}^{*-1} W_k \Lambda_{k}^{*-1} W_k) \nonumber \\
&\mathbb{E}[\sum_{k=1}^{K} \frac{\pi_k^*\phi_k^*}{\phi^*} ((\bm{\frac{\delta_{c}}{\psi_{k}^* }}^T a_k- a_k^T \bm{1} )^2 )]=\mathbb{E}[\sum_{k=1}^{K} \frac{\pi_k^*\phi_k^*}{\phi^*} (  a_k^T\bm{\frac{\delta_{c}}{\psi_{k}^* }} \bm{\frac{\delta_{c}}{\psi_{k}^* }}^T a_k-2a_k^T \bm{1}  \bm{\frac{\delta_{c}}{\psi_{k}^* }}^T a_k+(a_k^T \bm{1} )^2 )] \nonumber \\
&=\mathbb{E}[ \sum_{k=1}^{K} \frac{\pi_k^*\phi_k^*}{\phi^*} ( \sum_{g'=1}^{d}\sum_{g=1}^{d} a_{kg'} \frac{\delta_{cg'} }{\psi^*_{k,g'}}\frac{\delta_{cg} }{\psi^*_{k,g}} a_{kg} - 2a_k^T \bm{1} \sum_{g=1}^{d} \frac{\delta_{cg}}{\psi_{k,g}^*} a_{kg}+(a_k^T \bm{1} )^2  )] \nonumber \\
&= \sum_{k=1}^{K} a_k^T diag(\frac{\pi_k^*}{\psi_k^*}) a_k - \sum_{k=1}^{K}\pi_k^*(a_k^T \bm{1} )^2 
\end{align} 
The third expectation above can be derived by writing the Trace in summation convention and using Isserlis's theorem on the resulting fourth moment. Note the four `diagonal' terms above add up to precisely $\langle B,B \rangle -\sum_{k=1}^{K}\pi_k^*(a_k^T \bm{1} )^2$, and the `off-diagonal' terms add up to $0$:
\begin{align}  \label{eq: LRnondiagonalterms}
&\mathbb{E}[ \sum_{k=1}^{K}  \frac{\pi_k^*\phi_k^*}{\phi^*} 2u_k(\frac{1}{\pi_{k}^*}) (x - \mu^{*}_{k})^T \Lambda^{*}_{k} v_k]=0 \nonumber \\
&\mathbb{E}[ \sum_{k=1}^{K}  \frac{\pi_k^*\phi_k^*}{\phi^*} 2u_k(\frac{1}{\pi_{k}^*})Tr(\frac{1}{2}[\Lambda_{k}^{*-1}- (x - \mu_{k}^{*})(x - \mu_{k}^{*})^T ] W_k) ]=0  \nonumber \\
&\mathbb{E}[ \sum_{k=1}^{K}  \frac{\pi_k^*\phi_k^*}{\phi^*} 2u_k(\frac{1}{\pi_{k}^*}) (\bm{\frac{\delta_{c}}{\psi_{k}^* }}^T a_k- a_k^T \bm{1}) ]=0
\end{align} 
\begin{align}
&\mathbb{E}[ \sum_{k=1}^{K}  \frac{\pi_k^*\phi_k^*}{\phi^*} 2(x - \mu^{*}_{k})^T \Lambda^{*}_{k} v_k Tr(\frac{1}{2}[\Lambda_{k}^{*-1}- (x - \mu_{k}^{*})(x - \mu_{k}^{*})^T ] W_k) ]=0 
\end{align}
By Isserlis's Theorem.
\begin{align}
&\mathbb{E}[ \sum_{k=1}^{K}  \frac{\pi_k^*\phi_k^*}{\phi^*} 2(x - \mu^{*}_{k})^T \Lambda^{*}_{k} v_k (\bm{\frac{\delta_{c}}{\psi_{k}^* }}^T a_k- \frac{1}{2} a_k^T \bm{1} )]=0 \nonumber \\
&\mathbb{E}[ \sum_{k=1}^{K}  \frac{\pi_k^*\phi_k^*}{\phi^*} 2Tr(\frac{1}{2}[\Lambda_{k}^{*-1}- (x - \mu_{k}^{*})(x - \mu_{k}^{*})^T ] W_k)  (\bm{\frac{\delta_{c}}{\psi_{k}^* }}^T a_k- \frac{1}{2} a_k^T \bm{1} )]=0
\end{align}

By independence between $x$ and $c$. Putting together \Cref{eq: LRdiagonalterms} and \Cref{eq: LRnondiagonalterms} we have:
\begin{align}
\mathbb{E}[\langle B,LR(B) \rangle ] \leq \langle B,B \rangle - \sum_{k=1}^{K}\pi_k^*(a_k^T \bm{1} )^2 
\end{align}
With an upper bound on $\mathbb{E}[\langle B,LR(B) \rangle ]$ we can now return to \Cref{eq: greaterthanminusIinequality}:

\begin{align}\label{eq: greaterthanminusIinequalitypt2}
& \langle B,C^{\psi}B \rangle + \mathbb{E}[\langle B,LR(B) \rangle ] \nonumber \\
& \leq \sum_{k=1}^{K}\pi_k^*(a_k^T \bm{1} )^2 + \langle B,B \rangle  - \sum_{k=1}^{K}\pi_k^*(a_k^T \bm{1} )^2 \nonumber \\
& = \langle B,B \rangle 
\end{align} 
So the desired inequality is satisfied by picking $0<\epsilon<2$, which is enough for our CAVI algorithm ($\epsilon=1$).


For the case $A^{\epsilon}(\Theta^*)\prec I$. Again from \Cref{eq: operatorlimits}, we have that $A^{1}(\Theta^*)=I-C^{\psi}-\mathbb{E}[LR]$, so:

\begin{align}
A^{\epsilon}(\Theta^*) \prec I \iff (1-\epsilon)I + \epsilon A^{1}(\Theta^*)-I \prec 0 \iff C^{\psi}+\mathbb{E}[LR] \succ 0
\end{align}
\add{We must show $C^\psi+\mathbb{E}[LR]\succ 0$, i.e.\
$\langle B,C^\psi B\rangle+\mathbb{E}[\langle B,LR(B)\rangle]>0$ for all $B\ne0$.
Note it is \emph{not} the case that $\mathbb{E}[LR]$ alone is positive definite, it is
only positive semidefinite. Recall from \Cref{eq: BdotL} that
$\langle B,LR(B)\rangle=R(B)^2\ge0$, so
$\mathbb{E}[\langle B,LR(B)\rangle]=\mathbb{E}[R(B)^2]\ge0$ with equality iff
$R(B)=0$ almost surely under $\phi^*$. We first determine when $R(B)=0$ a.s.
Writing out $R(B)$ from \Cref{eq: operatorlimits}, for each categorical level
$c$ the identity $R(B)=0$ becomes, $\phi^*$-almost everywhere in $x$,}
\add{%
\begin{align}\label{eq: RBzero}
\sum_{k=1}^{K}\Big[\frac{u_k}{\pi_k^*}+ (x-\mu_k^*)^\top\Lambda_k^* v_k
+\tfrac12\Tr\big([\Lambda_k^{*-1}-(x-\mu_k^*)(x-\mu_k^*)^\top]W_k\big)
+\big(\tfrac{\delta_{c}}{\psi_k^*}-\mathbf{1}\big)^\top a_k\Big]\pi_k^*\phi_k^*(x,c)=0.
\end{align}}
\add{Let $\varphi_k(x,c)$ be the term inside square brackets in \Cref{eq: RBzero}, so that
\Cref{eq: RBzero} is \\ $\sum_{k}\varphi_k(x,c)\,\pi_k^*\phi_k^*(x,c)=0$.
Fix a level $c$. Each $\pi_k^*\phi_k^*(x,c)$ is $\psi_{k,c}^*\pi_k^*$ times
$N(x\mid\mu_k^*,\Lambda_k^{*-1})$, and $\varphi_k(\cdot,c)$ is a polynomial in
$x$ of degree $\le 2$. Hence \Cref{eq: RBzero} is a linear combination of
Gaussian densities with polynomial coefficients that vanishes a.e., and
\Cref{lemma: linind} (with distinct pairs $(\mu_k^*,\Lambda_k^*)$ from
\Cref{main-ass: distinct}) forces every polynomial coefficient to vanish. Matching
constant, linear and quadratic terms in $x$ gives, for every $k$,}
\add{%
\begin{align}\label{eq: RBkernel}
v_k=0,\quad W_k=0,\quad
\frac{u_k}{\pi_k^*}+\big(\tfrac{\delta_{c}}{\psi_k^*}-\mathbf{1}\big)^\top a_k=0\ \ \text{for all }c.
\end{align}}
\add{Averaging the last equation over $c$ (with weights $\psi_{k,g}^*$) and using
$\sum_g\psi_{k,g}^*\tfrac{\delta_{cg}}{\psi_{k,g}^*}=1=\mathbf{1}^\top\psi_k^*$
shows $u_k=0$; the remaining $c$-dependence then forces
$\tfrac{\delta_{c}}{\psi_k^*}{}^\top a_k=\mathbf{1}^\top a_k$ for all $c$, i.e.\
$a_{k,g}/\psi_{k,g}^*$ is the same for all $g$, i.e.\ $a_k\propto\psi_k^*$. Thus}
\add{%
\begin{align}\label{eq: LRkernel}
\ker \mathbb{E}[LR]=\big\{B=(u,v,W,a): u=0,\,v=0,\,W=0,\ a_k\propto\psi_k^*\ \forall k\big\}.
\end{align}}
\add{On this kernel, by \Cref{eq: greaterthanminusIinequalityotherterms},
$\langle B,C^\psi B\rangle=\sum_k\pi_k^*(a_k^\top\mathbf{1})^2$; and if
$a_k=t_k\psi_k^*$ then $a_k^\top\mathbf 1=t_k(\psi_k^*{}^\top\mathbf 1)=t_k>0$
unless $t_k=0$, so $\langle B,C^\psi B\rangle>0$ whenever $B\ne 0$ lies in
$\ker\mathbb{E}[LR]$. Therefore, for any $B\ne 0$: either
$R(B)\not\equiv0$, in which case $\mathbb{E}[R(B)^2]>0$ and $C^\psi\succeq0$ gives
$\langle B,C^\psi B\rangle+\mathbb{E}[R(B)^2]>0$. Or $B\in\ker\mathbb{E}[LR]$,
in which case $\mathbb{E}[R(B)^2]=0$ but $\langle B,C^\psi B\rangle>0$. In both cases
$\langle B,C^\psi B\rangle+\mathbb{E}[\langle B,LR(B)\rangle]>0$, so
$C^\psi+\mathbb{E}[LR]\succ0$, which is what was required. 

Having established
$\|A^\epsilon(\Theta^*)\|_{op}<1$, \textup{(A2)} holds on a possibly smaller ball
$\overline{B_\delta(\Theta^*)}$, $\delta\le\delta_1$, by continuity of $A^\epsilon$. Calculations in \Cref{subsec: gradientcalcs} show that $T^{\epsilon}(\Theta^*) \rightarrow \Theta^*$ as $n \rightarrow \infty$, establishing hypothesis \textup{(A3)}. As we have already established \textup{(A1)}, \Cref{lemma: localcontraction} completes the proof
of \Cref{main-thm: maintheorem}.}



\subsection{Details of gradient limit calculations}\label{subsec: gradientcalcs}


We start with the gradients of $\Psi$. As before, let $\theta_{k'}$ be any scaler variable in $\Theta$ belonging to component $k'$. We have:
\begin{align} \label{eq: dPsidpsi}
\frac{\partial \Psi_k(\Theta)_g}{\partial \theta_{k'}} &= \frac{\partial}{\partial \theta_{k'}}( \frac{\sum_{i=1}^{n} I_{\{c_i=g\}} r_{ik} }{ \sum_{i=1}^{n} r_{ik} } ) \nonumber \\
&= \frac{ (\sum_{i=1}^{n} I_{\{c_i=g\}} \frac{\partial  r_{ik}}{\partial \theta_{k'}} )( \sum_{i=1}^{n} r_{ik} ) -  (\sum_{i=1}^{n} \frac{\partial  r_{ik}}{\partial \theta_{k'} } )( \sum_{i=1}^{n} I_{\{c_i=g\}} r_{ik} ) }{ (\sum_{i=1}^{n} r_{ik})^2 } 
\end{align}
Starting with the case $\theta_{k'}=\hat{\psi}_{k',g'}$. We consider the two terms in the last line as $n \rightarrow \infty$ at $\Theta=\Theta^*$. Dividing each term in the top and bottom by $n$:
\begin{align} \label{eq: dPidpsi}
\frac{1}{n}\sum_{i=1}^{n} \frac{\partial  r_{ik}(\Theta^*)}{\partial \hat{\psi}_{k',g'}}&=\frac{1}{n}\sum_{i=1}^{n} \frac{\partial \ln(\rho_{ik'}(\Theta^*)) }{\partial \hat{\psi}_{k',g'}} \Bigl\{ \delta_{kk'}r_{ik}(\Theta^*)-r_{ik}(\Theta^*)r_{ik'}(\Theta^*) \Bigl\} \nonumber \\
& \rightarrow  \mathbb{E}[ (\frac{\delta_{g' c}}{\psi_{k',g'}^* }-1) (\delta_{kk'}\frac{\pi_k^*\phi_k^*}{\phi^*}-\frac{\pi_k^*\phi_k^* \pi_{k'}^* \phi_{k'}^*}{\phi^{*2}} ) ] \nonumber \\
& =  (\frac{\psi_{k,g'}^*}{\psi_{k',g'}^*}-1)\delta_{kk'}\pi_k^*-\mathbb{E}[ (\frac{\delta_{g' c}}{\psi_{k',g'}^* }-1) \frac{\pi_k^*\phi_k^* \pi_{k'}^* \phi_{k'}^*}{\phi^{*2}}  ]=-\mathbb{E}[ (\frac{\delta_{g' c}}{\psi_{k',g'}^* }-1) \frac{\pi_k^*\phi_k^* \pi_{k'}^* \phi_{k'}^*}{\phi^{*2}}  ]
\end{align}
where in the second line, we used the limit of $r_{ik}(\Theta^*)$ previously derived in \Cref{lemma: rho_iklimit}, and a modified version of the strong law of large numbers \Cref{lemma: modifiedSLLN}. The expectation here is taken with respect to $x,c$, which has density of the likelihood evaluated at the true data generating parameter $\Theta^*$. In the last line we used $\frac{\delta_{kk'}\pi_k^*}{{\psi_{k',g'}^*}}\mathbb{E}[\frac{\delta_{g' c}\phi_k^*}{\phi^*}]=\frac{\delta_{kk'}\pi_k^*}{{\psi_{k',g'}^*}}\psi_{k,g'}^*=\delta_{kk'}\pi_k^*$. 
Similarly:
\begin{align}
\frac{1}{n}\sum_{i=1}^{n}I_{\{c_i=g\}} \frac{\partial  r_{ik}(\Theta^*)}{\partial \psi_{k',g'}}&=\frac{1}{n}\sum_{i=1}^{n} I_{\{c_i=g\}} \frac{\partial \ln(\rho_{ik'}(\Theta^*)) }{\partial \psi_{k',g'}} \Bigl\{ \delta_{kk'}r_{ik}(\Theta^*)-r_{ik}(\Theta^*)r_{ik'}(\Theta^*) \Bigl\}  \nonumber \\
& \rightarrow  \mathbb{E}[ \delta_{g c} (\frac{\delta_{g' c}}{\psi_{k',g'}^*}-1) (\delta_{kk'}\frac{\pi_k^*\phi_k^*}{\phi^*}-\frac{\pi_k^*\phi_k^* \pi_{k'}^* \phi_{k'}^*}{\phi^{*2}} ) ]  \nonumber \\
& \rightarrow \delta_{gg'}\delta_{kk'}\pi_k^*-\delta_{kk'}\psi_{k,g}^{*}\pi_k^*-\mathbb{E}[ \delta_{gc}(\frac{\delta_{g' c}}{\psi_{k',g'}^* }-1) \frac{\pi_k^*\phi_k^* \pi_{k'}^* \phi_{k'}^*}{\phi^{*2}}  ]
\end{align}
where in the last line we used $\delta_{g c}\delta_{g' c}=\delta_{gg'}\delta_{g' c}$.
Additionally:
\begin{align}
&\Pi_k(\Theta^*)=\frac{1}{n}\sum_{i=1}^{n} r_{ik}(\Theta^*) \rightarrow \mathbb{E}[\frac{\pi_k^*\phi_k^*}{\phi^*} ] = \pi_k^{*}  \nonumber \\
&\frac{1}{n}\sum_{i=1}^{n} I_{\{c_i=g\}} r_{ik}(\Theta^*) \rightarrow \mathbb{E}[\delta_{gc}\frac{\pi_k^*\phi_k^*}{\phi^*}]= \pi_k^{*} \psi_{k,g}^{*}
\end{align}
Consequently, $\Psi_k(\Theta^*)_g \rightarrow \psi_{k,g}^{*}$. Putting everything together and plugging into \Cref{eq: dPsidpsi} and we get:
\begin{align} 
\frac{\partial \Psi_k(\Theta^*)_g}{\partial \psi_{k',g'}} & \rightarrow \delta_{gg'}\delta_{kk'}-\delta_{kk'}\psi_{k,g}^{*}-\mathbb{E}[ \delta_{gc}(\frac{\delta_{g' c}}{\psi_{k',g'}^* }-1) \frac{\phi_k^* \pi_{k'}^* \phi_{k'}^*}{\phi^{*2}}  ]+\mathbb{E}[ (\frac{\delta_{g' c}}{\psi_{k',g'}^* }-1) \frac{\psi_{k,g}^*\phi_k^* \pi_{k'}^* \phi_{k'}^*}{\phi^{*2}}  ]  \nonumber \\
& = \delta_{gg'}\delta_{kk'}-\delta_{kk'}\psi_{k,g}^{*} -\mathbb{E}[  \frac{(\delta_{gc}-\psi_{k,g}^{*})\phi_k^* }{\phi^{*}} (\frac{\delta_{g' c}}{\psi_{k',g'}^*}-1) \frac{\pi_{k'}^*\phi_{k'}^*}{\phi^{*}} ]
\end{align}
$\frac{\partial \Psi_k(\Theta)_g}{\partial \hat{\pi}_{k'}}$, $\frac{\partial \Psi_k(\Theta)_g}{\partial \hat{\mu}_{k'l}}$, $\frac{\partial \Psi_k(\Theta)_g}{\partial \hat{\Lambda}_{k'jl}}$ can be calculated similarly to \Cref{eq: dPsidpsi} and replacing $\theta_{k'}$ with the corresponding variable to be differentiated:

\begin{align}
\frac{\partial \Psi_k(\Theta^*)_g}{\partial \hat{\pi}_{k'}} &\rightarrow \mathbb{E}[ \delta_{g c} (\frac{1}{\pi_{k'}^{*}}) (\delta_{kk'}\frac{\pi_k^*\phi_k^*}{\phi^*}-\frac{\pi_k^*\phi_k^* \pi_{k'}^* \phi_{k'}^*}{\phi^{*2}} ) ] \frac{\pi_{k}^*}{\pi_{k}^{*2}}  \nonumber \\
& -\mathbb{E}[ (\frac{1}{\pi_{k'}^{*}} ) (\delta_{kk'}\frac{\pi_k^*\phi_k^*}{\phi^*}-\frac{\pi_k^*\phi_k^* \pi_{k'}^* \phi_{k'}^*}{\phi^{*2}} ) ]\frac{\pi_{k}^*\psi_{k,g}^*}{\pi_{k}^{*2} }  \nonumber \\
&=-\mathbb{E}[  \frac{(\delta_{gc}-\psi_{k,g}^{*})\phi_k^* }{\phi^{*}} \frac{\phi_{k'}^* }{\phi^{*}} ] 
\end{align}
Using Einstein's summation convention $\sum_{j=1}^{q} \Lambda_{k'lj}^{*}(x_{j} - \mu_{k'j}^{*})= \Lambda_{k'lj}^{*}(x_{j} - \mu_{k'j}^{*})$, we have:
\begin{align}
\frac{\partial \Psi_k(\Theta^*)_g}{\partial \hat{\mu}_{k'l}}  &\rightarrow \mathbb{E}[ \delta_{g c} \Lambda_{k'lj}^{*}(x_{j} - \mu_{k'j}^{*} )  (\delta_{kk'}\frac{\pi_k^*\phi_k^*}{\phi^*}-\frac{\pi_k^*\phi_k^* \pi_{k'}^* \phi_{k'}^*}{\phi^{*2}} ) ] \frac{\pi_{k}^*}{\pi_{k}^{*2}}  \nonumber \\
& -\mathbb{E}[ \Lambda_{k'lj}^{*}(x_{j} - \mu_{k'j}^{*} ) (\delta_{kk'}\frac{\pi_k^*\phi_k^*}{\phi^*}-\frac{\pi_k^*\phi_k^* \pi_{k'}^* \phi_{k'}^*}{\phi^{*2}} ) ]\frac{\pi_{k}^*\psi_{k,g}^*}{\pi_{k}^{*2} }  \nonumber \\
&= -\mathbb{E}[ \frac{(\delta_{gc}-\psi_{k,g}^{*})\phi_k^* }{\phi^{*}} \Lambda^{*}_{k'lj}(x_{j} - \mu^{*}_{k'j} )\frac{\phi_{k'}^*\pi_{k'}^* }{\phi^{*}} ]
\end{align}
\begin{align}
\frac{\partial \Psi_k(\Theta^*)_g}{\partial  \hat{\Lambda}_{k'jl}} &\rightarrow \mathbb{E}[ \delta_{g c} \frac{1}{2}(\Lambda_{k'jl}^{* -1}- (x_{j} - \mu_{k'j}^{*})(x_{l} - \mu_{k'l}^{*}) )
  (\delta_{kk'}\frac{\pi_k^*\phi_k^*}{\phi^*}-\frac{\pi_k^*\phi_k^* \pi_{k'}^* \phi_{k'}^*}{\phi^{*2}} ) ] \frac{\pi_{k}^*}{\pi_{k}^{*2}} \nonumber \\
& -\mathbb{E}[ \frac{1}{2}(\Lambda_{k'jl}^{* -1}- (x_{j} - \mu_{k'j}^{*})(x_{l} - \mu_{k'l}^{*}) )
 (\delta_{kk'}\frac{\pi_k^*\phi_k^*}{\phi^*}-\frac{\pi_k^*\phi_k^* \pi_{k'}^* \phi_{k'}^*}{\phi^{*2}} ) ]\frac{\pi_{k}^*\psi_{k,g}^*}{\pi_{k}^{*2} }  \nonumber \\
&=-\mathbb{E}[ \frac{(\delta_{gc}-\psi_{k,g}^{*})\phi_k^* }{\phi^{*}}  \frac{1}{2}[\Lambda_{k'jl}^{* -1}- (x_{j} - \mu_{k'j}^{*})(x_{l} - \mu_{k'l}^{*}) ]\frac{\phi_{k'}^*\pi_{k'}^* }{\phi^{*}}   ]
\end{align}

Next, for the gradients of $\Pi$:

\begin{align}  \label{eq: dPidpsishort}
\frac{\partial \Pi_k(\Theta^*)}{\partial \hat{\psi}_{k',g'} }&=\frac{\partial}{\partial \hat{\psi}_{k',g'}}(\frac{1}{n}\sum_{i=1}^{n} r_{ik}(\Theta^*) ) \nonumber \\ 
& \rightarrow -\mathbb{E}[ (\frac{\delta_{g' c}}{\psi_{k',g'}^* }-1) \frac{\pi_k^*\phi_k^* \pi_{k'}^* \phi_{k'}^*}{\phi^{*2}} ]
\end{align}
Recall this is the same as \Cref{eq: dPidpsi}.

\begin{align} 
\frac{\partial \Pi_k(\Theta^*)}{\partial \hat{\pi}_{k'} }&=\frac{\partial}{\partial \hat{\pi}_{k'}}(\frac{1}{n}\sum_{i=1}^{n} r_{ik}(\Theta^*) )  \nonumber \\ 
& \rightarrow \mathbb{E}[ \frac{1}{\pi_{k'}^{*}} (\delta_{kk'}\frac{\pi_k^*\phi_k^*}{\phi^*}-\frac{\pi_k^*\phi_k^* \pi_{k'}^* \phi_{k'}^*}{\phi^{*2}} ) ] \nonumber \\
&= \delta_{kk'}-\mathbb{E}[\frac{\pi_k^*\phi_k^* \phi_{k'}^*}{\phi^{*2}} ] 
\end{align}

\begin{align} 
\frac{\partial \Pi_k(\Theta^*)}{\partial \hat{\mu}_{k'l} }&=\frac{\partial}{\partial \hat{\mu}_{k'l}}(\frac{1}{n}\sum_{i=1}^{n} r_{ik}(\Theta^*) )  \nonumber \\ 
& \rightarrow \mathbb{E}[ \Lambda_{k'lj}^{*}(x_{j} - \mu_{k'j}^{*} ) (\delta_{kk'}\frac{\pi_k^*\phi_k^*}{\phi^*}-\frac{\pi_k^*\phi_k^* \pi_{k'}^* \phi_{k'}^*}{\phi^{*2}} ) ]  \nonumber \\
&= -\mathbb{E}[\Lambda_{k'lj}^{*}(x_{j} - \mu_{k'j}^{*} ) \frac{\pi_k^*\phi_k^*\pi_{k'}^* \phi_{k'}^*}{\phi^{*2}} ) ] 
\end{align}

\begin{align} 
\frac{\partial \Pi_k(\Theta^*)}{\partial \hat{\Lambda}_{k'jl} }&=\frac{\partial}{\partial \hat{\Lambda}_{k'jl}}(\frac{1}{n}\sum_{i=1}^{n} r_{ik}(\Theta^*) )  \nonumber \\ 
& \rightarrow \mathbb{E}[ \frac{1}{2}(\Lambda_{k'jl}^{* -1}- (x_{j} - \mu_{k'j}^{*})(x_{l} - \mu_{k'l}^{*}) ) (\delta_{kk'}\frac{\pi_k^*\phi_k^*}{\phi^*}-\frac{\pi_k^*\phi_k^* \pi_{k'}^* \phi_{k'}^*}{\phi^{*2}} ) ]  \nonumber \\
&= -\mathbb{E}[\frac{1}{2}(\Lambda_{k'jl}^{* -1}- (x_{j} - \mu_{k'j}^{*})(x_{l} - \mu_{k'l}^{*}) ) \frac{\pi_k^*\phi_k^*\pi_{k'}^* \phi_{k'}^*}{\phi^{*2}} ) ] 
\end{align}

Next, for the gradients of $M$:

\begin{align}
\frac{\partial M_k(\Theta)}{\partial \theta_{k'} }&=\frac{\partial}{\partial \theta_{k'} }( \frac{\sum_{i=1}^{n} r_{ik}x_i}{\sum_{i=1}^{n} r_{ik}} )  \nonumber \\ 
&=\frac{ (\sum_{i=1}^{n} \frac{\partial  r_{ik} }{\partial \theta_{k'} } x_i )( \sum_{i=1}^{n} r_{ik} ) -  (\sum_{i=1}^{n} \frac{\partial  r_{ik} }{\partial \theta_{k'} } )( \sum_{i=1}^{n} r_{ik} x_i ) }{ (\sum_{i=1}^{n} r_{ik} )^2 }
\end{align}
Note
\begin{align}
M_k(\Theta^*)=\frac{\sum_{i=1}^{n} r_{ik}(\Theta^*)x_i}{\sum_{i=1}^{n} r_{ik}(\Theta^*)} \rightarrow \frac{\mathbb{E}[  \frac{\pi_k^*\phi_k^* x}{\phi^{*}} ]}{\mathbb{E}[  \frac{\pi_k^*\phi_k^*}{\phi^{*}} ]}=\mu_k^{*}
\end{align}
For $\theta_{k'}=\hat{\psi}_{k',g'}$, we have:
\begin{align}
\frac{\partial M_k(\Theta^*)}{\partial \hat{\psi}_{k',g'} } & \rightarrow \frac{1}{\pi_k^{*2} } \Bigl\{ \mathbb{E}[ (\frac{\delta_{g' c}}{\psi_{k',g'}^* }-1) (\delta_{kk'}\frac{\pi_k^*\phi_k^*}{\phi^*}-\frac{\pi_k^*\phi_k^* \pi_{k'}^* \phi_{k'}^*}{\phi^{*2}} )x ] \pi_k^* \nonumber \\
& \qquad - \mathbb{E}[ (\frac{\delta_{g' c}}{\psi_{k',g'}^* }-1) (\delta_{kk'}\frac{\pi_k^*\phi_k^*}{\phi^*}-\frac{\pi_k^*\phi_k^* \pi_{k'}^* \phi_{k'}^*}{\phi^{*2}} ) ] \pi_k^* \mu_k^* \Bigl\}   \nonumber \\
& =  \delta_{kk'} \mathbb{E}[ (\frac{\delta_{g' c}}{\psi_{k',g'}^* }-1) \frac{\phi^*_k}{\phi^*} (x-\mu_k^*) ] - \mathbb{E}[ \frac{(x-\mu_k^*) \phi_k^*}{\phi^{*} } (\frac{\delta_{g' c}}{\psi_{k',g'}^* }-1) \frac{\pi_{k'}^* \phi_{k'}^*}{\phi^*}  ]   \nonumber \\
& =  - \mathbb{E}[ \frac{(x-\mu_k^*) \phi_k^*}{\phi^{*} } (\frac{\delta_{g' c}}{\psi_{k',g'}^* }-1) \frac{\pi_{k'}^* \phi_{k'}^*}{\phi^*}  ] 
\end{align}
For $\theta_{k'}=\hat{\pi}_{k'}$, we have:
\begin{align}
\frac{\partial M_k(\Theta^*)}{\partial \hat{\pi}_{k'} } & \rightarrow  \frac{1}{\pi_k^{*2} } \Bigl\{ \mathbb{E}[ \frac{1}{\pi_{k'}^{*}} (\delta_{kk'}\frac{\pi_k^*\phi_k^*}{\phi^*}-\frac{\pi_k^*\phi_k^* \pi_{k'}^* \phi_{k'}^*}{\phi^{*2}} )x ] \pi_k^* - \mathbb{E}[ \frac{1}{\pi_{k'}^{*}} (\delta_{kk'}\frac{\pi_k^*\phi_k^*}{\phi^*}-\frac{\pi_k^*\phi_k^* \pi_{k'}^* \phi_{k'}^*}{\phi^{*2}} ) ] \pi_k^* \mu_k^* \Bigl\}  \nonumber \\
& = \delta_{kk'} \mathbb{E}[ \frac{1}{\pi_{k'}^{*}} \frac{\phi^*_k}{\phi^*} (x-\mu_k^*) ] - \mathbb{E}[ \frac{(x-\mu_k^*) \phi_k^*}{\phi^{*} } (\frac{1}{\pi_{k'}^{*}}) \frac{\pi_{k'}^* \phi_{k'}^*}{\phi^*}  ]   \nonumber \\
& = - \mathbb{E}[ \frac{(x-\mu_k^*) \phi_k^*}{\phi^{*} } \frac{ \phi_{k'}^*}{\phi^*}  ] 
\end{align}
For $\theta_{k'}=\hat{\mu}_{k'l}$, we have:
\begin{align}
\frac{\partial M_k(\Theta^*)_{r}}{\partial \hat{\mu}_{k'l} } & \rightarrow  \frac{1}{\pi_k^{*2} } \Bigl\{ \mathbb{E}[ \Lambda_{k'lj}^{*}(x_{j} - \mu_{k'j}^{*} ) (\delta_{kk'}\frac{\pi_k^*\phi_k^*}{\phi^*}-\frac{\pi_k^*\phi_k^* \pi_{k'}^* \phi_{k'}^*}{\phi^{*2}} )x_{r} ] \pi_k^*  \nonumber \\
&- \mathbb{E}[ \Lambda_{k'lj}^{*}(x_{j} - \mu_{k'j}^{*} ) (\delta_{kk'}\frac{\pi_k^*\phi_k^*}{\phi^*}-\frac{\pi_k^*\phi_k^* \pi_{k'}^* \phi_{k'}^*}{\phi^{*2}} ) ] \pi_k^* \mu_{kr}^* \Bigl\}  \nonumber \\
& = \delta_{kk'} \mathbb{E}[ \Lambda_{k'lj}^{*}(x_{j} - \mu_{k'j}^{*} ) \frac{\phi^*_k}{\phi^*} (x_{r}-\mu_{kr}^*) ] - \mathbb{E}[ \frac{(x_{r}-\mu_{kr}^*) \phi_k^*}{\phi^{*} } (\Lambda_{k'lj}^{*}(x_{j} - \mu_{k'j}^{*} )) \frac{\pi_{k'}^* \phi_{k'}^*}{\phi^*}  ]  \nonumber \\
& = \delta_{kk'}\delta_{lr}- \mathbb{E}[ \frac{(x_{r}-\mu_{kr}^*) \phi_k^*}{\phi^{*} } (\Lambda_{k'lj}^{*}(x_{j} - \mu_{k'j}^{*} )) \frac{\pi_{k'}^* \phi_{k'}^*}{\phi^*}  ]
\end{align}
For $\theta_{k'}=\hat{\Lambda}_{k'jl}$, we have:
\begin{align}
\frac{\partial M_k(\Theta^*)_{r}}{\partial \hat{\Lambda}_{k'jl} } & \rightarrow  \frac{1}{\pi_k^{*2} } \Bigl\{ \mathbb{E}[ \frac{1}{2}(\Lambda_{k'jl}^{* -1}- (x_{j} - \mu_{k'j}^{*})(x_{l} - \mu_{k'l}^{*}) ) (\delta_{kk'}\frac{\pi_k^*\phi_k^*}{\phi^*}-\frac{\pi_k^*\phi_k^* \pi_{k'}^* \phi_{k'}^*}{\phi^{*2}} )x_{r} ] \pi_k^*  \nonumber \\
&- \mathbb{E}[ \frac{1}{2}(\Lambda_{k'jl}^{* -1}- (x_{j} - \mu_{k'j}^{*})(x_{l} - \mu_{k'l}^{*}) ) (\delta_{kk'}\frac{\pi_k^*\phi_k^*}{\phi^*}-\frac{\pi_k^*\phi_k^* \pi_{k'}^* \phi_{k'}^*}{\phi^{*2}} ) ] \pi_k^* \mu_{kr}^* \Bigl\}  \nonumber \\
& = \delta_{kk'} \mathbb{E}[ \frac{1}{2}(\Lambda_{k'jl}^{* -1}- (x_{j} - \mu_{k'j}^{*})(x_{l} - \mu_{k'l}^{*}) ) \frac{\phi^*_k}{\phi^*} (x_{r}-\mu_{kr}^*) ]  \nonumber \\
& - \mathbb{E}[ \frac{(x_{r}-\mu_{kr}^*) \phi_k^*}{\phi^{*} } (\frac{1}{2}[\Lambda_{k'jl}^{* -1}- (x_{j} - \mu_{k'j}^{*})(x_{l} - \mu_{k'l}^{*}) ]) \frac{\pi_{k'}^* \phi_{k'}^*}{\phi^*}  ]  \nonumber \\
& = - \mathbb{E}[ \frac{(x_{r}-\mu_{kr}^*) \phi_k^*}{\phi^{*} } (\frac{1}{2}[\Lambda_{k'jl}^{* -1}- (x_{j} - \mu_{k'j}^{*})(x_{l} - \mu_{k'l}^{*}) ]) \frac{\pi_{k'}^* \phi_{k'}^*}{\phi^*}  ]
\end{align}

Finally, for the gradients of $S$:

\begin{align}
&\frac{\partial S_k(\Theta)}{\partial \theta_{k'} }=\frac{\partial}{\partial \theta_{k'}}\Bigl\{ (\sum_{i=1}^{n} r_{ik} )(\sum_{i=1}^{n} r_{ik} (x_i-M_k)(x_i-M_k)^T )^{-1} \Bigl\}  \nonumber \\ 
&= (\sum_{i=1}^{n} \frac{\partial r_{ik}}{\partial \theta_{k'} } )(\sum_{i=1}^{n} r_{ik} (x_i-M_k)(x_i-M_k)^T )^{-1}  \nonumber \\
&- (\sum_{i=1}^{n} r_{ik} ) (\sum_{i=1}^{n} r_{ik} (x_i-M_k)(x_i-M_k)^T )^{-1} \nonumber \\
&\qquad \times (\sum_{i=1}^{n} \frac{\partial }{\partial \theta_{k'} }\Bigl\{ r_{ik}(x_i-M_k)(x_i-M_k)^T \Bigl\} ) (\sum_{i=1}^{n} r_{ik} (x_i-M_k)(x_i-M_k)^T )^{-1} 
\end{align}

where we used $M_k=M_k(\Theta)$ as a shorthand. Before continuing we need to evaluate some intermediate limits:
\begin{align}
&\frac{1}{n}\sum_{i=1}^{n} r_{ik}(\Theta^*) (x_i-M_k(\Theta^*))(x_i-M_k(\Theta^*))^T \nonumber \\
&= \frac{1}{n}\sum_{i=1}^{n} r_{ik}(\Theta^*) (x_i-\mu_k^{*}+\mu_k^{*}-M_k(\Theta^*))(x_i-\mu_k^{*}+\mu_k^{*}-M_k(\Theta^*))^T  \nonumber \\
&=\frac{1}{n}\sum_{i=1}^{n} r_{ik}(\Theta^*) (x_i-\mu_k^{*})(x_i-\mu_k^{*})^T+\frac{1}{n}\sum_{i=1}^{n} r_{ik}(\Theta^*)(x_i-\mu_k^{*})(\mu_k^{*}-M_k(\Theta^*))^T  \nonumber \\
&+(\mu_k^{*}-M_k(\Theta^*))\frac{1}{n}\sum_{i=1}^{n} r_{ik}(\Theta^*)(x_i-\mu_k^{*})^T+(\mu_k^{*}-M_k(\Theta^*))(\frac{1}{n}\sum_{i=1}^{n} r_{ik}(\Theta^*))(\mu_k^{*}-M_k(\Theta^*))^T  \nonumber \\
& \rightarrow  \mathbb{E}[\frac{\pi_k^*\phi_k^*}{\phi^{*}}(x-\mu_k^{*})(x-\mu_k^{*})^T ] = \pi_k^* (\Lambda_k^{*})^{-1}
\end{align}
Consequently, $S_k(\Theta^*) \rightarrow \Lambda_k^{*}$. Using a similar argument, $\frac{1}{n}\sum_{i=1}^{n} r_{ik}(\Theta^*) (x_i-M_k(\Theta^*)) \rightarrow 0$

\begin{align}
&\frac{1}{n}\sum_{i=1}^{n} \frac{\partial }{\partial \theta_{k'}}\Bigl\{ r_{ik}(x_i-M_k)(x_i-M_k)^T \Bigl\} )  \nonumber \\
&=\frac{1}{n}\sum_{i=1}^{n} \frac{\partial r_{ik} }{\partial \theta_{k'}} (x_i-M_k)(x_i-M_k)^T + (\frac{-\partial M_k }{\partial \theta_{k'}})\frac{1}{n}\sum_{i=1}^{n} r_{ik}(x_i-M_k)^T + \frac{1}{n}\sum_{i=1}^{n} r_{ik}(x_i-M_k)(\frac{-\partial M_k }{\partial \theta_{k'}})^T \nonumber \\
&=\frac{1}{n}\sum_{i=1}^{n} \frac{\partial r_{ik} }{\partial \theta_{k'}} (x_i-\mu_k^{*})(x_i-\mu_k^{*})^T+(\mu_k^{*}-M_k)\frac{1}{n}\sum_{i=1}^{n} \frac{\partial r_{ik} }{\partial \theta_{k'}}(x_i-\mu_k^{*})^T  \nonumber \\ 
&+\frac{1}{n}\sum_{i=1}^{n} \frac{\partial r_{ik} }{\partial \theta_{k'}} (x_i-\mu_k^{*})(\mu_k^{*}-M_k)^T+(\mu_k^{*}-M_k)\frac{1}{n}\sum_{i=1}^{n} \frac{\partial r_{ik} }{\partial \theta_{k'}}(\mu_k^{*}-M_k)^T  \nonumber \\
&+(\frac{-\partial M_k }{\partial \theta_{k'}})\frac{1}{n}\sum_{i=1}^{n} r_{ik}(x_i-M_k)^T + \frac{1}{n}\sum_{i=1}^{n} r_{ik}(x_i-M_k)(\frac{-\partial M_k }{\partial \theta_{k'}})^T  \nonumber \\
& \sim \frac{1}{n}\sum_{i=1}^{n} \frac{\partial r_{ik} }{\partial \theta_{k'}} (x_i-\mu_k^{*})(x_i-\mu_k^{*})^T \textit{   at }\Theta=\Theta^{*}, n \rightarrow \infty
\end{align}
Putting everything together, for the case $\theta_{k'}=\hat{\psi}_{k',g'}$, we therefore have:
\begin{align}
& \frac{\partial S_k(\Theta^*)}{\partial \hat{\psi}_{k',g'} }\rightarrow  \mathbb{E}[ (\frac{\delta_{g' c}}{\psi_{k',g'}^* }-1) (\delta_{kk'}\frac{\pi_k^*\phi_k^*}{\phi^*}-\frac{\pi_k^*\phi_k^* \pi_{k'}^* \phi_{k'}^*}{\phi^{*2}} )]  (\pi_k^{* -1} \Lambda_k^*)   \nonumber \\
&  - \pi_k^{*}(\pi_k^{* -1} \Lambda_k^*) \mathbb{E}[ (\frac{\delta_{g' c}}{\psi_{k',g'}^*}-1) (\delta_{kk'}\frac{\pi_k^*\phi_k^*}{\phi^*}-\frac{\pi_k^*\phi_k^* \pi_{k'}^* \phi_{k'}^*}{\phi^{*2}} )(x-\mu_k^{*})(x-\mu_k^{*})^T ] (\pi_k^{* -1} \Lambda_k^*)  \nonumber \\
&  = \delta_{kk'} \mathbb{E}[  (\frac{\delta_{g' c}}{\psi_{k',g'}^*}-1) \frac{\phi_{k}^*}{\phi^{*} } (\Lambda_k^*-\Lambda_k^*(x-\mu_k^{*})(x-\mu_k^{*})^T \Lambda_k^*) ]  \nonumber \\
&  -\mathbb{E}[  \frac{\phi_{k}^*}{\phi^{*} } (\Lambda_k^*-\Lambda_k^*(x-\mu_k^{*})(x-\mu_k^{*})^T \Lambda_k^*) (\frac{\delta_{g' c}}{\psi_{k',g'}^*}-1)(\frac{\pi_{k'}^* \phi_{k'}^*}{ \phi^*} ) ]  \nonumber \\
&  = -\mathbb{E}[  \frac{\phi_{k}^*}{\phi^{*} } (\Lambda_k^*-\Lambda_k^*(x-\mu_k^{*})(x-\mu_k^{*})^T \Lambda_k^*) (\frac{\delta_{g' c}}{\psi_{k',g'}^*}-1)(\frac{\pi_{k'}^* \phi_{k'}^*}{\phi^*} ) ]  
\end{align}
For $\theta_{k'}=\hat{\pi}_{k'}$, we have:

\begin{align}
& \frac{\partial S_k(\Theta^*)}{\partial \hat{\pi}_{k'} }\rightarrow  \mathbb{E}[ \frac{1}{\pi_{k'}^{*}} (\delta_{kk'}\frac{\pi_k^*\phi_k^*}{\phi^*}-\frac{\pi_k^*\phi_k^* \pi_{k'}^* \phi_{k'}^*}{\phi^{*2}} )]  (\pi_k^{* -1} \Lambda_k^*) \nonumber \\
&  - \pi_k^{*}(\pi_k^{* -1} \Lambda_k^*) \mathbb{E}[ (\frac{1}{\pi_{k'}^{*}}) (\delta_{kk'}\frac{\pi_k^*\phi_k^*}{\phi^*}-\frac{\pi_k^*\phi_k^* \pi_{k'}^* \phi_{k'}^*}{\phi^{*2}} )(x-\mu_k^{*})(x-\mu_k^{*})^T ] (\pi_k^{* -1} \Lambda_k^*)  \nonumber \\
& =  -\mathbb{E}[ (\Lambda_k^*-\Lambda_k^*(x-\mu_k^{*})(x-\mu_k^{*})^T\Lambda_k^*)\frac{\phi_k^* \phi_{k'}^*}{\phi^{*2}} ]
\end{align}

For $\theta_{k'}=\hat{\mu}_{k'l}$, we have:





\begin{align}
& \frac{\partial S_k(\Theta^*)}{\partial \hat{\mu}_{k'l} }\rightarrow  \mathbb{E}[ \Lambda_{k'lj}^{*}(x_{j} - \mu_{k'j}^{*} ) (\delta_{kk'}\frac{\pi_k^*\phi_k^*}{\phi^*}-\frac{\pi_k^*\phi_k^* \pi_{k'}^* \phi_{k'}^*}{\phi^{*2}} )]  (\pi_k^{* -1} \Lambda_k^*)  \nonumber \\
&  - \pi_k^{*}(\pi_k^{* -1} \Lambda_k^*) \mathbb{E}[ \Lambda_{k'lj}^{*}(x_{j} - \mu_{k'j}^{*} ) (\delta_{kk'}\frac{\pi_k^*\phi_k^*}{\phi^*}-\frac{\pi_k^*\phi_k^* \pi_{k'}^* \phi_{k'}^*}{\phi^{*2}} )(x-\mu_k^{*})(x-\mu_k^{*})^T ] (\pi_k^{* -1} \Lambda_k^*) \nonumber \\
& =  -\mathbb{E}[ (\Lambda_k^*-\Lambda_k^*(x-\mu_k^{*})(x-\mu_k^{*})^T\Lambda_k^*)\Lambda_{k'lj}^{*}(x_{j} - \mu_{k'j}^{*} ) \frac{\phi_k^* \pi_{k'}^*\phi_{k'}^*}{\phi^{*2}} ]
\end{align}

For $\theta_{k'}=\hat{\Lambda}_{k'jl}$, we have:

\begin{align}
& \frac{\partial S_k(\Theta^*)_{rs} }{\partial \hat{\Lambda}_{k'jl} }\rightarrow  \mathbb{E}[ \frac{1}{2}(\Lambda_{k'jl}^{* -1}- (x_{j} - \mu_{k'j}^{*})(x_{l} - \mu_{k'l}^{*}) )(\delta_{kk'}\frac{\pi_k^*\phi_k^*}{\phi^*}-\frac{\pi_k^*\phi_k^* \pi_{k'}^* \phi_{k'}^*}{\phi^{*2}} )]  (\pi_k^{* -1} \Lambda_{krs}^*) \nonumber \\
&  - \pi_k^{*}(\pi_k^{* -1} \Lambda_{kra}^*) \mathbb{E}[ \frac{1}{2}(\Lambda_{k'jl}^{* -1}- (x_{j} - \mu_{k'j}^{*})(x_{l} - \mu_{k'l}^{*}) ) \nonumber \\
&\qquad \times (\delta_{kk'}\frac{\pi_k^*\phi_k^*}{\phi^*}-\frac{\pi_k^*\phi_k^* \pi_{k'}^* \phi_{k'}^*}{\phi^{*2}} )(x_{a}-\mu_{ka}^{*})(x_{b}-\mu_{kb}^{*}) ] (\pi_k^{* -1} \Lambda_{kbs}^*) \nonumber \\
&=-\delta_{kk'}\Lambda_{kra}^*\frac{1}{2}[\Lambda_{kjl}^{* -1}\Lambda_{kab}^{* -1}- (\Lambda_{kjl}^{* -1}\Lambda_{kab}^{* -1}+\Lambda_{kja}^{* -1}\Lambda_{klb}^{* -1}+\Lambda_{kjb}^{* -1}\Lambda_{kla}^{* -1})  ] (\Lambda_{kbs}^*) \nonumber \\
& -\mathbb{E}[ (\Lambda_{krs}^*-\Lambda_{kra}^*(x_{a}-\mu_{ka}^{*})(x_{b}-\mu_{kb}^{*})\Lambda_{kbs}^*)\frac{1}{2}(\Lambda_{k'jl}^{* -1}- (x_{j} - \mu_{k'j}^{*})(x_{l} - \mu_{k'l}^{*}) ) \frac{\phi_k^* \pi_{k'}^*\phi_{k'}^*}{\phi^{*2}} ] \nonumber \\
&=\frac{1}{2}\delta_{kk'}(\delta_{rj}\delta_{ls}+\delta_{js}\delta_{rl}) \nonumber \\
& -\mathbb{E}[ (\Lambda_{krs}^*-\Lambda_{kra}^*(x_{a}-\mu_{ka}^{*})(x_{b}-\mu_{kb}^{*})\Lambda_{kbs}^*)\frac{1}{2}(\Lambda_{k'jl}^{* -1}- (x_{j} - \mu_{k'j}^{*})(x_{l} - \mu_{k'l}^{*}) ) \frac{\phi_k^* \pi_{k'}^*\phi_{k'}^*}{\phi^{*2}} ]
\end{align}

\begin{lemma}\label{lemma: modifiedSLLN}
If ${X_n}$ is a sequence of independent and identically distributed random variables and $F_n, F_0$ are measurable functions such that $F_n(.) \rightarrow F_0(.)$ pointwise, and there exists a measurable function $G$ such that $|F_n(x)| \leq G(x), \mathbb{E}[G^2(X_1)]<\infty$, then almost surely:

$$\frac{1}{n} \sum_{i=1}^{n} F_n(X_i) \rightarrow \mathbb{E}[F_0(X_1)].$$
\end{lemma}

\begin{proof}
Decompose the sum as:
$$\frac{1}{n} \sum_{i=1}^{n} F_n(X_i) = \frac{1}{n} \sum_{i=1}^{n} F_n(X_i)-\mathbb{E}[F_n(X_1)]+\mathbb{E}[F_n(X_1)].$$

The second term $\mathbb{E}[F_n(X_1)]$ converges to $\mathbb{E}[F_0(X_1)]$ by the dominated convergence theorem, since $F_n(x) \rightarrow F_0(x)$ pointwise, and $|F_n(x)| \leq G(x)$ and $G(X_1)$ has finite expectation. The first term $\frac{1}{n} \sum_{i=1}^{n} F_n(X_i)-\mathbb{E}[F_n(X_1)]$ converges to $0$ almost surely by Theorem 2 of \cite{Hu1989}, where the required conditions are satisfied. Specifically, consider the array of random variables $Y_{n,i}=F_n(X_i)-\mathbb{E}[F_n(X_1)]$. By construction $\mathbb{E}[Y_{n,i}]=0$ and the $Y_{n,i}$ are row-wise independent (for fixed $n$, the $Y_{n,i}$ are independent), and the $Y_{n,i}$ are uniformly bounded by $H(X_1)=G(X_1)+\mathbb{E}[G(X_1)] $, in the sense that $P(|Y_{n,i}|>t) \leq P(|H(X_1)|>t)$ for all $t>0$. The requirement $\mathbb{E}[H^2(X_1)]<\infty$ is satisfied because $ \mathbb{E}[G^2(X_1)]<\infty$.

\end{proof}

\noindent\Cref{lemma: modifiedSLLN} in our context is straightforward to apply because the $r_{ik}$ are bounded between $0$ and $1$, so the construction of a dominating function $G$ is immediate. Furthermore, derivatives of $r_{ik}$ only introduce bounded or polynomial factors of $x_i$, so it's still straightforward to find dominating functions for the derivatives of $r_{ik}$, as mixtures of independent Gaussian and finite discrete distributions possess finite moments of all orders. In contrast, Lemma 2 of \cite{BoWang2006} requires the $r_{ik}$ and its derivatives to be uniformly convergent as functions of $x_i$ as $n \rightarrow \infty$ instead, which appears much more cumbersome to verify compared to the relaxed requirements in \Cref{lemma: modifiedSLLN}.

\subsection{Proof of \thmposteriormeancorollary}\label{subsec: corollary1proof}

Using \Cref{eq: simplifiedtohyperparaappendix}, at iteration $t$, the variational posterior means take the following forms. Using the convergence of $\Theta^{(t)}$ in \thmmaintheorem, the variational posterior means therefore converge to the limits below as $t, n \rightarrow \infty$:

\begin{align}\label{eq: posteriormeans}
\mathbb{E}_{\pi}[\pi_k]&=\frac{\hat{\alpha}_k^{t}}{\sum_{k=1}^{K} \hat{\alpha}_k^{t}}=\frac{n \hat{\pi}_k^{t} + \alpha}{n \sum_{k=1}^{K}\hat{\pi}_k^{t} + K\alpha} \rightarrow \pi_k^{*} \nonumber \\
\mathbb{E}_{\mu,\Lambda}[\mu_k]&=\hat{m}_k^{t}=(n\hat{\mu}_k^{t} \hat{\pi}_k^{t}+\beta m_k)/(n\hat{\pi}_k^{t}+\beta) \rightarrow \mu_k^{*} \nonumber \\
\mathbb{E}_{\mu,\Lambda}[\Lambda_k]&=\hat{\nu}_k^{t} (\hat{\Phi}_k^{t})^{-1}= (n \hat{\pi}_k^{t}+\nu)(n \hat{\pi}_k^{t})^{-1} ( (\hat{\Lambda}_k^{t})^{-1}+ \frac{1 }{ n \hat{\pi}_k^{t}+\beta} \beta (\hat{\mu}_k^{t}-m_k)(\hat{\mu}_k^{t}-m_k)^T+\frac{\Phi}{n \hat{\pi}_k})^{-1} \nonumber \\
& \rightarrow \Lambda_k^{*} \nonumber \\
\mathbb{E}_{\psi}[\psi_{k,g}]&=\frac{\hat{\eta}_{k,g}^{t}}{ \sum_{g=1}^{d}\hat{\eta}_{k,g}^{t} }=\frac{n\hat{\pi}_k^{t}\hat{\psi}_{k,g}^{t}+\eta}{n\hat{\pi}_k^{t}\sum_{g=1}^{d}\hat{\psi}_{k,g}^{t}+d\eta} \rightarrow  \psi_{k,g}^{*}
\end{align}

\subsection{Proof of \thmposteriorpredcorollary}\label{subsec: corollary3proof}
Recall the variational posterior predictive density evaluated at iteration $t$ is
\begin{align} \label{eq: posteriorpredappendix}
q(\tilde{x},\tilde{c}|x,c) &= \sum_{k=1}^{K} \frac{\hat{\alpha}_k^{t}}{\sum_{k'=1}^{K} \hat{\alpha}_{k'}^{t}} \;t_{\hat{\nu}_{k}^{t}-q+1}\Big(\tilde{x} \Big\vert \hat{m}_{k}^{t}, \frac{\hat{\Phi}_{k}^{t}(\hat{\beta}_{k}^{t}+1)}{\hat{\beta}_{k}^{t}(\hat{\nu}_{k}^{t}-q+1)} \Big)\; \frac{\hat{\eta}_{k,\tilde{c}}^{t}}{\sum_{g=1}^{d} \hat{\eta}_{k,g}^{t}}
\end{align} 
we already established $\frac{\hat{\alpha}_k^{t}}{\sum_{k'=1}^{K} \hat{\alpha}_{k'}^{t}} \rightarrow \pi_k^*$, $\frac{\hat{\eta}_{k,\tilde{c}}^{t}}{\sum_{g=1}^{d} \hat{\eta}_{k,g}^{t}} \rightarrow \psi_{k,\tilde{c}}^*$. For the remaining middle term, we know $\hat{m}_{k}^{t} \rightarrow \mu_k^*$, $\frac{\hat{\Phi}_{k}^{t}(\hat{\beta}_{k}^{t}+1)}{\hat{\beta}_{k}^{t}(\hat{\nu}_{k}^{t}-q+1)} \rightarrow \Lambda_k^{*-1}$ and $\hat{\nu}_{k}^{t}-q+1 \rightarrow \infty$. So $t_{\hat{\nu}_{k}^{t}-q+1}\Big(\tilde{x} \Big\vert \hat{m}_{k}^{t}, \frac{\hat{\Phi}_{k}^{t}(\hat{\beta}_{k}^{t}+1)}{\hat{\beta}_{k}^{t}(\hat{\nu}_{k}^{t}-q+1)} \Big) \rightarrow \mathcal{N}(\tilde{x}|\mu_k^*, \Lambda_k^{*-1})$. Putting it all together we have $q(\tilde{x},\tilde{c}|x,c) \rightarrow  \sum_{k=1}^{K} \pi_k^* \mathcal{N}(\tilde{x}|\mu_k^*, \Lambda_k^{*-1}) \psi_{k,\tilde{c}}^*$ as required. 

\subsection{Proof of \thmMLEcorollary}\label{subsec: corollary2proof}
Let $\tilde{\Theta}$ be the strongly consistent local maximum likelihood
estimate of the parameter $\Theta$. One can show through
differentiating the marginal likelihood that $\tilde{\Theta}$ satisfies the following equations:
\begin{align}\label{eq: MLEequations}
&\tilde{\pi}_k-\tilde{\Pi}_k(\tilde{\Theta})=0,\quad
\tilde{\Pi}_k(\Theta)=\frac1n\sum_{i=1}^n r_{ik}^{\infty}(\Theta)
\nonumber\\
&\tilde{\mu}_k-\tilde{M}_k(\tilde{\Theta})=0,\quad
\tilde{M}_k(\Theta)=
\frac{\frac1n\sum_{i=1}^n x_i r_{ik}^{\infty}(\Theta)}
     {\frac1n\sum_{i=1}^n r_{ik}^{\infty}(\Theta)}
\nonumber\\
&\tilde{\Lambda}_k-\tilde{S}_k(\tilde{\Theta})=0,\quad
\tilde{S}_k(\Theta)=
\left(\frac1n\sum_{i=1}^n r_{ik}^{\infty}(\Theta)\right)
\left(\frac1n\sum_{i=1}^n r_{ik}^{\infty}(\Theta)
(x_i-\tilde{M}_k(\Theta))(x_i-\tilde{M}_k(\Theta))^T\right)^{-1}
\nonumber\\
&\tilde{\psi}_{k,g}-\tilde{\Psi}_k(\tilde{\Theta})_g=0,\quad
\tilde{\Psi}_k(\Theta)_g=
\frac{\frac1n\sum_{i=1}^n I_{\{c_i=g\}}r_{ik}^{\infty}(\Theta)}
     {\frac1n\sum_{i=1}^n r_{ik}^{\infty}(\Theta)}.
\end{align}
where $r_{ik}^{\infty}(\Theta)=\rho_{ik}^{\infty}(\Theta)/
\sum_{k'=1}^K\rho_{ik'}^{\infty}(\Theta)$, as in
\Cref{lemma: rho_iklimit}. 
Note $\rho_{ik}^{\infty}(\Theta)=p_{ik}(\Theta)=\hat{\pi}_k p(x_i,c_i|z_i=k,\Theta)$ whenever $\sum_{g=1}^{d} \hat{\psi}_{k,g}=1$, which $\tilde{\Theta}$ satisfies because $\sum_{g=1}^{d} \tilde{\psi}_{k,g}=1$.

Similar to the notation for the function $T$ in \thmmaintheorem, let
\begin{align}
\tilde{\Pi}(\Theta) = \begin{bmatrix} \tilde{\Pi}_1(\Theta) \\  \vdots \\ \tilde{\Pi}_K(\Theta) \end{bmatrix},
\tilde{M}(\Theta) = \begin{bmatrix} \tilde{M}_1(\Theta) \\  \vdots \\ \tilde{M}_K(\Theta) \end{bmatrix},
\tilde{S}(\Theta) = \begin{bmatrix} \tilde{S}_1(\Theta) \\  \vdots \\ \tilde{S}_K(\Theta) \end{bmatrix},
\tilde{\Psi}(\Theta) = \begin{bmatrix} \tilde{\Psi}_1(\Theta) \\  \vdots \\ \tilde{\Psi}_K(\Theta) \end{bmatrix}
\end{align}
and
\begin{align}
\tilde{T}(\Theta)=
\begin{bmatrix}
\tilde{\Pi}(\Theta)\\\tilde{M}(\Theta)\\
\tilde{S}(\Theta)\\\tilde{\Psi}(\Theta)
\end{bmatrix}.
\end{align}
We have $\tilde{\Theta}-\tilde{T}(\tilde{\Theta})=0$.
Let $\Theta^{\infty}$ be the limit of $\Theta^{(t-1)}$ as $t\to\infty$,
which exists and is the fixed point $\Theta_n^{\infty}$ of
\Cref{lemma: localcontraction} with $\epsilon=1$.
Then $\Theta^{\infty}-T(\Theta^{\infty})=0$. We have
\begin{align}
&\Theta^{\infty}-T(\Theta^{\infty})=0 \\
&\Theta^{\infty}-\tilde{T}(\Theta^{\infty})+\tilde{T}(\Theta^{\infty})-T(\Theta^{\infty})=0 \\
&\Theta^{\infty}-\tilde{T}(\Theta^{\infty})=O(\frac{1}{n}) \quad a.s. \\
&\tilde{\Theta}-\tilde{T}(\tilde{\Theta})+\Theta^{\infty}-\tilde{\Theta}-\add{\bar{A}_n}(\Theta^{\infty}-\tilde{\Theta})=O(\frac{1}{n}) \quad a.s.\\
&\add{(I-\bar{A}_n)}(\Theta^{\infty}-\tilde{\Theta})=O(\frac{1}{n}) \quad a.s.
\end{align}
where $\bar{A}_n:=\int_0^1\nabla\tilde{T}
\bigl(\tilde{\Theta}+s(\Theta^{\infty}-\tilde{\Theta})\bigr)\,ds$.
In the above, we used
$\tilde{T}(\Theta^{\infty})-T(\Theta^{\infty})=O(1/n)$ almost surely
in the third line, which follows from $\tilde{T}$ having the same form as $T$ except with $r_{ik}$ replaced by $r_{ik}^{\infty}$, and $|r_{ik}(\Theta)-r_{ik}^{\infty}(\Theta)|
\le C(1+\|x_i\|^2)/n$ uniformly on the ball established back in \Cref{lemma: rho_iklimit}.
The weighted averages in $T$ and their counterparts in $\tilde{T}$ therefore differ by $O(1/n)$
almost surely and converge by \Cref{lemma: uniformSLLN}.
$T$ and $\tilde{T}$ are finite algebraic combinations of these weighted averages, which, together with the required inverses, are uniformly bounded on $\overline{B_\delta(\Theta^*)}$, so $\tilde{T}(\Theta^{\infty})-T(\Theta^{\infty})=O(1/n)$.
On the fourth line, we used the integral form of the mean value theorem,
which is valid since $\tilde{\Theta}$ and $\Theta^{\infty}$ both lie in the
convex ball $\overline{B_\delta(\Theta^*)}$ for all sufficiently large $n$,
almost surely.

Finally, by \thmmaintheorem we have as $n \rightarrow \infty$, $\Theta^{\infty} \rightarrow \Theta^{*}$ almost surely and it is known the MLE $\tilde{\Theta} \rightarrow \Theta^{*}$ almost surely.
Moreover, the entries of $\nabla\tilde{T}$ are built from the EM
responsibilities $r_{ik}^{\infty}$ in place of $r_{ik}$, so they are finite
algebraic combinations of the empirical averages with $n$-free summands treated in the verification of \textup{(A1)},
and \Cref{lemma: uniformSLLN} gives
$\sup_{\Theta\in\overline{B_\delta(\Theta^*)}}
\|\nabla\tilde{T}(\Theta)-A^1(\Theta)\|_{op}\to0$ almost surely.
Together with the continuity of $A^1$, this gives
$\bar{A}_n\to A^1(\Theta^*)$ almost surely.
We have shown in the proof of \thmmaintheorem that
$I-A^1(\Theta^*)$ is positive definite, so $I-\bar{A}_n$ is invertible
with $\|(I-\bar{A}_n)^{-1}\|_{op}$ bounded for all sufficiently large $n$,
almost surely. Consequently,
$\Theta^{\infty}=\tilde{\Theta}+O(\frac{1}{n})$ almost surely.

\newpage

\subsection{Proof of \thmposteriorconsistencycorollary}\label{subsec: corollary4proof}

For posterior consistency, we require $Q_{n}^{\infty}(\| \Theta-\Theta^*\|_2>\epsilon)$ to converge to $0$ for all $\epsilon>0$ as $n \rightarrow \infty$. Using Markov's inequality:
\begin{align} \label{eq: posteriorconsistency}
Q_{n}^{\infty}(\| \Theta-\Theta^*\|_2>\epsilon)&=Q_{n}^{\infty}(\| \Theta-\Theta^*\|^2_2>\epsilon^2) \leq \frac{1}{\epsilon^2}\mathbb{E}_{Q_{n}^{\infty}}(\| \Theta-\Theta^*\|^2_2) \\
&= \frac{1}{\epsilon^2}\mathbb{E}_{Q_{n}^{\infty}}(\| \Theta-\mathbb{E}_{Q_{n}^{\infty}}[\Theta]+\mathbb{E}_{Q_{n}^{\infty}}[\Theta]-\Theta^*\|^2_2) \\
& \leq \frac{2}{\epsilon^2} \mathbb{E}_{Q_{n}^{\infty}}(\| \Theta-\mathbb{E}_{Q_{n}^{\infty}}[\Theta] \|^2_2) + \frac{2}{\epsilon^2}\mathbb{E}_{Q_{n}^{\infty}}(\| \mathbb{E}_{Q_{n}^{\infty}}[\Theta]-\Theta^{*} \|^2_2)
\end{align}
The $\mathbb{E}_{Q_{n}^{\infty}}(\| \Theta-\mathbb{E}_{Q_{n}^{\infty}}[\Theta]  \|^2_2)$ term is a sum of the marginal variances of the variational posterior which were all shown to be $O(\frac{1}{n})$ in \Cref{thm: viposteriormoments}. The $\mathbb{E}_{Q_{n}^{\infty}}(\| \mathbb{E}_{Q_{n}^{\infty}}[\Theta] -\Theta^*\|^2_2)=\| \mathbb{E}_{Q_{n}^{\infty}}[\Theta] -\Theta^*\|^2_2$ term converges to $0$ almost surely by \thmposteriormeancorollary as $n \rightarrow \infty$. \add{This proves the consistency statement of \Cref{main-thm: posteriorconsistencycorollary}, using only \Cref{main-ass: wellspecified,main-ass: interior,main-ass: distinct,main-ass: correctK}.}

\add{We now prove the contraction rate, which
additionally invokes \Cref{main-ass: RW}.} For the contraction rate, let $\epsilon=M/\sqrt{n}$, then 
$\frac{2}{\epsilon^2} \mathbb{E}_{Q_{n}^{\infty}}(\| \Theta-\mathbb{E}_{Q_{n}^{\infty}}[\Theta] \|^2_2)
=\frac{2n}{M^2}\,O(n^{-1})
=O(\frac{1}{M^2})$ almost surely.
For the second term, using Theorem 3.1 from \cite{Redner1984}\add{, whose Conditions 1, 2 we verify for our model in \Cref{subsec: RWverification} (\Cref{main-ass: RW})}, we have $\tilde{\Theta}-\Theta^*=O_p(n^{-1/2})$, combining this with \thmMLEcorollary we have:
$\mathbb{E}_{Q_{n}^{\infty}}[\Theta]-\Theta^{*}=\tilde{\Theta}-\Theta^*+O(\frac{1}{n}) \ a.s.=O_p(n^{-1/2})$. Hence 
$$\frac{2}{\epsilon^2}\| \mathbb{E}_{Q_{n}^{\infty}}[\Theta] -\Theta^*\|^2_2=\frac{2n}{M^2}O_p(n^{-1})=O_p(\frac{1}{M^2}).$$ Finally, let $M=M_n \rightarrow \infty$ and we have $Q_n^\infty\!\left(\big\|\Theta-\Theta^*\big\|_2 > \frac{M_n}{\sqrt{n}}\right)
\ \xrightarrow[]{P_{\Theta^*}}\ 0.$

\subsection{Moments of the CAVI posterior}\label{thm: viposteriormoments}

Uniformly over $\overline{B_{\delta_1}(\Theta^*)}$, we have the following moments:

$\mu_k$ has marginal variational posterior distribution 
$t_{\hat{\nu}_k-q+1}(\mu_k|\hat{m}_k,\frac{\hat{\Phi}_k}{\hat{\beta}_k (\hat{\nu}_k-q+1)})$, which has mean $\hat{m}_k$ and covariance:

$$Cov(\mu_{ki},\mu_{ka})=\frac{\hat{\nu}_k-q+1}{(\hat{\nu}_k-q-1)\hat{\beta}_k (\hat{\nu}_k-q+1)}[\hat{\Phi}_k]_{ia}=O(\frac{1}{n})$$

For $q(\Lambda_k)$, we have $\mathbb{E}_{\Lambda_{k}}(\Lambda_k)=\hat{\nu}_k \hat{\Phi}_{k}^{-1} = \hat{\Lambda}_k+O(\frac{1}{n}) $, \\
$Cov\left( [\Lambda_k]_{ij}, [\Lambda_k]_{ab} \right)
= \hat{\nu}_k \left( [\hat{\Phi}_k^{-1}]_{ia} [\hat{\Phi}_k^{-1}]_{jb} + [\hat{\Phi}_k^{-1}]_{ib} [\hat{\Phi}_k^{-1}]_{ja} \right)
= O\left(\frac{1}{n}\right)$

For $q(\pi)$, we note that $\mathbb{E}_{\pi}(\pi_k)= \frac{\hat{\alpha}_k} {\sum_{k'=1}^{K} \hat{\alpha}_{k'} }=\frac{n \hat{\pi}_k + \alpha }{n + K\alpha } = \hat{\pi}_k+O(\frac{1}{n}) $ and
$$ Cov(\pi_k,\pi_l)=\frac{\delta_{kl}\hat{\alpha}_k}{(\sum_{k'=1}^{K} \hat{\alpha}_{k'})(\sum_{k'=1}^{K} \hat{\alpha}_{k'}+1) }-\frac{\hat{\alpha}_k \hat{\alpha}_l}{(\sum_{k'=1}^{K} \hat{\alpha}_{k'})(\sum_{k'=1}^{K} \hat{\alpha}_{k'})(\sum_{k'=1}^{K} \hat{\alpha}_{k'}+1) }= O(\frac{1}{n}) $$

For $q(\psi_k)$, we have $\mathbb{E}_{\psi} [\psi_{k,c_{i}} ]= \frac{\hat{\eta}_{k,c_{i} }}{ \sum_{g=1}^{d} \hat{\eta}_{k,g} } = \frac{n \hat{\pi}_k\hat{\psi}_{k,c_i}+\eta}{ n \hat{\pi}_k+d\eta } = \hat{\psi}_{k,c_i}+O(\frac{1}{n}) $.
$$Cov(\psi_{k,g}, \psi_{k,h} )=\frac{\delta_{gh}\hat{\eta}_{k,g}}{(\sum_{g'=1}^{d} \hat{\eta}_{k,g'})(\sum_{g'=1}^{d} \hat{\eta}_{k,g'}+1) }-\frac{\hat{\eta}_{k,g} \hat{\eta}_{k,h}}{(\sum_{g'=1}^{d} \hat{\eta}_{k,g'})(\sum_{g'=1}^{d} \hat{\eta}_{k,g'})(\sum_{g'=1}^{d} \hat{\eta}_{k,g'}+1) }= O(\frac{1}{n}) $$

\add{\subsection{Proof of Lemma 1: verification of the Redner--Walker conditions}\label{subsec: RWverification}
\thmposteriorconsistencycorollary uses Theorem 3.1 of \citet{Redner1984} for the
existence of a strongly consistent, asymptotically normal root of the likelihood
equations at rate $n^{-1/2}$, which requires its Conditions 1 and 2 to hold
in a neighbourhood of $\Theta^*$. We verify
these conditions explicitly in this section.

We first state the two conditions as in \cite{Redner1984}, adapted to our
notation. Let $\vartheta=(\vartheta_1,\dots,\vartheta_P)$ denote the parameter of the reduced parameter space introduced below, with true value
$\vartheta^*$, and let $p(x,c\mid\vartheta)$ be the mixture density
\Cref{eq: likelihoodVIappendix}. Writing
$S(x,c\mid\vartheta):=\nabla_\vartheta\log p(x,c\mid\vartheta)$ for the score,
Redner and Walker's Fisher information matrix is
\begin{equation}\label{eq: RWfisher}
\mathcal I(\vartheta)=\sum_{c=1}^{d}\int_{\mathbb{R}^q} S(x,c\mid\vartheta)\,S(x,c\mid\vartheta)^{\top}\,p(x,c\mid\vartheta)\,\mathrm{d}x
=\mathbb{E}_{\vartheta}\!\big[S\,S^{\top}\big].
\end{equation}

\smallskip
\noindent\textbf{Condition 1 (\citealp{Redner1984}).} For all $\vartheta$ in a
neighbourhood of $\vartheta^*$, for almost every $x\in\mathbb{R}^q$, and for all
$i,j,k\in\{1,\dots,P\}$, the partial derivatives
$\dfrac{\partial p}{\partial\vartheta_i}$,
$\dfrac{\partial^2 p}{\partial\vartheta_i\,\partial\vartheta_j}$ and
$\dfrac{\partial^3 \log p}{\partial\vartheta_i\,\partial\vartheta_j\,\partial\vartheta_k}$
exist and satisfy
\begin{equation}\label{eq: RWcond1}
\Big|\frac{\partial p(x,c\mid\vartheta)}{\partial\vartheta_i}\Big|\le F_i(x),\qquad
\Big|\frac{\partial^2 p(x,c\mid\vartheta)}{\partial\vartheta_i\,\partial\vartheta_j}\Big|\le F_{ij}(x),\qquad
\Big|\frac{\partial^3 \log p(x,c\mid\vartheta)}{\partial\vartheta_i\,\partial\vartheta_j\,\partial\vartheta_k}\Big|\le F_{ijk}(x),
\end{equation}
where $F_i$ and $F_{ij}$ are integrable over $\mathbb{R}^q$ (for each $c$) and
$F_{ijk}$ satisfies $\int F_{ijk}(x)\,p(x,c\mid\vartheta)\,\mathrm{d}x<\infty$.

\smallskip
\noindent\textbf{Condition 2 (\citealp{Redner1984}).} The Fisher information
matrix $\mathcal I(\vartheta)$ in \Cref{eq: RWfisher} is well defined and
positive definite at $\vartheta^*$.

\smallskip
\noindent We now introduce the reduced space and verify each condition in turn.

\paragraph{Reduced parametrisation.} The full coordinates
$(\bm\pi,\bm\psi)$ are constrained by $\sum_k\pi_k=1$ and $\sum_g\psi_{k,g}=1$,
so the family is not identifiable in those coordinates and the Fisher
information is singular. We therefore work in the reduced parametrisation that
eliminates $\pi_K$ and $\psi_{k,d}$, writing

\[\vartheta=(\pi_1,\dots,\pi_{K-1};\ \{\mu_k,\Lambda_k\}_k;\ \{\psi_{k,g}\}_{g<d})\in\mathbb{R}^P,\]
which is a smooth unconstrained space of the model in a neighbourhood of
$\Theta^*$. Identifiability holds in this space up to the (finite) label
permutations, which \citet{Redner1984} handle by restricting to a neighbourhood
of a single labelling of $\Theta^*$.

\paragraph{Verification of Condition 1.} With $p(x,c\mid\vartheta)$ the mixture
density \Cref{eq: likelihoodVIappendix}, the map $\vartheta\mapsto p$ is
$C^3$ on the reduced space: it is a finite sum of products of Gaussian densities
(smooth in $(\mu_k,\Lambda_k)$ in a neighbourhood of $\Theta^*$) and categorical probabilities
(smooth in $\psi$). The dominating functions $F_i,F_{ij},F_{ijk}$ of
\Cref{eq: RWcond1} exist because Gaussian mixtures have finite moments of all
orders and, on $\overline{B_{\delta_1}(\Theta^*)}$, the densities and their
$\vartheta$-derivatives are bounded by $C(1+\|x\|^{6})f_0(x,c)$ for an integrable
envelope $f_0$. $\partial p,\partial^2 p$ are
(polynomial)$\times$(Gaussian density) and Lebesgue-integrable in $x$, while $\partial^3\log p$
is polynomially bounded, $O(1+\|x\|^{6})$, so it has
finite expectation because $\phi^*$ has moments of all orders. As the sum over $c$ ranges over the finite set
$\{1,\dots,d\}$, summing over it preserves both integrability and finiteness of
the expectation. This verifies Condition~1.

\paragraph{Verification of Condition 2.} For any constant vector $b\in\mathbb{R}^P$,
\begin{equation}\label{eq: fisherquadform}
b^\top\mathcal I(\vartheta^*)\,b=b^\top\mathbb{E}_{\Theta^*}\!\big[SS^\top\big]b=\mathbb{E}_{\Theta^*}\!\big[(b^\top S)^2\big]\ge0,
\end{equation}
so $\mathcal I(\vartheta^*)$ is positive semidefinite, and it is positive definite
iff $\mathbb{E}_{\Theta^*}[(b^\top S)^2]>0$ for every $b\ne0$. $\mathbb{E}_{\Theta^*}[(b^\top S)^2]=0$ holds iff
$b^\top S(x,c)=0$ for $\phi^*$-almost every $(x,c)$. We show this holds.

Write $b=(b^\pi;\,\{b^\mu_k\}_k;\,\{b^\Lambda_k\}_k;\,\{b^\psi_k\}_k)$, where
$\phi_k^*(x,c)=N(x\mid\mu_k^*,\Lambda_k^{*-1})\psi_{k,c}^*$ and
$\phi^*=\sum_k\pi_k^*\phi_k^*$. The components of the score at $\vartheta^*$,
$S=S(x,c\mid\vartheta^*)$, are, in the reduced space, for
$k<K$ (weights), all $k$ (means and precisions), and $g<d$ (categorical),
\[
\frac{\partial\ln p}{\partial\pi_k}=\frac{\phi_k^*-\phi_K^*}{\phi^*},\qquad
\frac{\partial\ln p}{\partial\mu_k}=\frac{\pi_k^*\phi_k^*}{\phi^*}\,\Lambda_k^*(x-\mu_k^*),\qquad
\frac{\partial\ln p}{\partial\Lambda_k}=\frac{\pi_k^*\phi_k^*}{\phi^*}\,\tfrac12[\Lambda_k^{*-1}-(x-\mu_k^*)(x-\mu_k^*)^\top],
\]
\[
\frac{\partial\ln p}{\partial\psi_{k,g}}=\frac{\pi_k^*\phi_k^*}{\phi^*}\Big(\frac{\delta_{gc}}{\psi_{k,g}^*}-\frac{\delta_{dc}}{\psi_{k,d}^*}\Big),
\]
where $\partial\ln p/\partial\pi_k$ carries the reduced space constraint
$\pi_K=1-\sum_{k<K}\pi_k$ (hence the $-\phi_K^*$). Forming the dot product,
\begin{equation}\label{eq: bTS}
\begin{split}
b^\top S={}&\sum_{k=1}^{K-1}b^\pi_k\,\frac{\phi_k^*-\phi_K^*}{\phi^*}
+\sum_{k=1}^{K}\frac{\pi_k^*\phi_k^*}{\phi^*}\Big[(b^\mu_k)^\top\Lambda_k^*(x-\mu_k^*) \\
&+\tfrac12\Tr\big([\Lambda_k^{*-1}-(x-\mu_k^*)(x-\mu_k^*)^\top]b^\Lambda_k\big)
+\sum_{g<d}b^\psi_{k,g}\Big(\tfrac{\delta_{gc}}{\psi_{k,g}^*}-\tfrac{\delta_{dc}}{\psi_{k,d}^*}\Big)\Big].
\end{split}
\end{equation}
Exactly
as in the derivation of the kernel of $\mathbb{E}[LR]$ in the proof of
\thmmaintheorem\ (see \Cref{eq: RBzero}), multiplying $b^\top S=0$ through by
$\phi^*(x,c)>0$ turns it into
$\sum_{k}\varphi_k(x,c)\,\pi_k^*\phi_k^*(x,c)=0$, with $\varphi_k$ as in \Cref{eq: RBzero}, under the relabelling below:
\begin{equation}\label{eq: cond2identify}
v_k:=b^\mu_k,\quad W_k:=b^\Lambda_k,\quad
u_k:=\begin{cases}b^\pi_k,&k<K,\\[2pt]-\sum_{k'<K}b^\pi_{k'},&k=K,\end{cases}\quad
a_{k,g}:=\begin{cases}b^\psi_{k,g},&g<d,\\[2pt]-\sum_{g'<d}b^\psi_{k,g'},&g=d,\end{cases}
\end{equation}
the values of $u_K$ and $a_{k,d}$ due to $\pi_K=1-\sum_{k<K}\pi_k$ and
$\psi_{k,d}=1-\sum_{g<d}\psi_{k,g}$. Each
$b^\pi_k(\phi_k^*-\phi_K^*)$ contributes $b^\pi_k$ to $u_k$ and $-b^\pi_k$ to $u_K$,
and likewise each
$b^\psi_{k,g}\big(\delta_{gc}/\psi^*_{k,g}-\delta_{dc}/\psi^*_{k,d}\big)$
 contributes $b^\psi_{k,g}$ to $a_{k,g}$ and $-b^\psi_{k,g}$ to $a_{k,d}$, and by construction
$\mathbf 1^\top a_k=0$. This is
the identity already analysed in \Cref{eq: RBzero}, \Cref{eq: RBkernel}:
\Cref{lemma: linind} forces every $\varphi_k(\cdot,c)\equiv0$, and matching
powers of $x$ gives, for all $k$ and all levels $c$,
\begin{equation}\label{eq: cond2kernel}
W_k=0,\qquad v_k=0,\qquad \frac{u_k}{\pi_k^*}+\Big(\tfrac{\delta_{\cdot c}}{\psi_k^*}-\mathbf 1\Big)^{\!\top}a_k=0.
\end{equation}
Up to \Cref{eq: cond2kernel} the computation is identical to that for
$\mathbb{E}[LR]$, differing only in the last step, because we are now in
the reduced $\psi$ space. Now $\mathbf 1^\top a_k=0$ by
construction, so $a_k=t_k\psi_k^*$ gives $t_k=t_k\mathbf 1^\top\psi_k^*=\mathbf 1^\top a_k=0$;
hence $a_k=0$, i.e.\ $b^\psi_k=0$. The $\psi_k^*$-weighted average of the last equation in
\Cref{eq: cond2kernel} over $c$ then gives $u_k=0$ for every $k$, so
$b^\pi_k=0$ for $k<K$ (and $u_K=-\sum_{k<K}b^\pi_k=0$ is automatic), therefore
$b^\pi=0$. Every block of $b$ vanishes, so $b=0$ and
$\mathcal I(\vartheta^*)\succ0$.

}

\section{Detail of data generating parameters for simulated examples} \label{subsec: simulatedparams}

\subsection{Scenario 1: a continuous data dominated example} \label{subsec: continuousdataexample}

For this scenario, we simulate data as follows: $\pi^{*}$ are generated from normalising a random vector of size $K^{*}=5$, where each component is generated from $U(0.5,2.0)$. This ensures the ratio between the most likely component and the least likely component is capped to at most $4$ and not too skewed. We use a single categorical variable with integer values from $1$ to $d_1=4$, where the true probability  $\psi^{*}$ for the categorical variable is generated from a Dirichlet distribution with all $d_1$ hyperparameters equal to $5$. To ensure the continuous data are well separated between components, we generate $K^{*}$ vectors $\mu^{*}_{k}$ that are a distance of at least $4K^{*}$ and at most $8K^{*}$ apart. To generate $\Sigma^{*}_{k}$, for each $k$ we sample diagonal matrices $D_{k}$ with eigenvalues between $k^2/2$ and $k^2$ (so `standard deviation' of at most $k\le K^{*}$), then transforming it via conjugation with a random orthogonal matrix $U_{k}$ to obtain a non diagonal covariance matrix  $\Sigma^{*}_{k}= U_{k} D_{k} U^{T}_{k}$. This procedure generates components where the separation between the component centres is relatively large compared to the standard deviations of the covariance matrices, so the data sampled from different components are likely to be well separated.  

\subsection{A categorical data dominated example} \label{subsec: categorydataexample}

For this scenario, again $\pi^{*}$ are generated from normalising a random vector of size $K^{*}=5$, where each component is generated from $U(0.5,2)$. There are $p=2K^{*}$ binary categorical variables, which we generate by having them strongly depend on the component label as follows: if the component label is equal to $k$, then the $(2k-1)$th and $2k$th categorical variables have $0.9$ probability of equalling $1$, while all other categorical variables have $0.9$ probability of equalling $0$. In other words, $\psi^{*}_{k,2k-1,1}=0.9, \psi^{*}_{k,2k,1}=0.9, \psi^{*}_{k,j,1}=0.1$ for $j \neq 2k-1, 2k$. For the continuous data, we set $q=K^{*}$, with the component centres $\mu^{*}_{k}=e_{k}$, where $e_{k}$ is the unit vector in the $k$th direction. Each $\Sigma^{*}_{k}$ is chosen to be a diagonal matrix, with the $k$th diagonal set to $9$ while other diagonals are set to $4$. In this example, the distance between any two component centres is $\sqrt{2}$, but the standard deviations of the covariance matrices are $2$ or $3$. This leads to the points generated from different components to overlap significantly with each other while still being distinct, to ensure the components are largely driven by the value of the categorical variables.

\subsection{A third example} \label{subsec: mixofbothexample}

For this scenario, the true component probabilities $\pi^{*}$ are generated as for the two previous scenarios. We define a single categorical variable (i.e. $p=1$) with integer values between $1$ and $d_1=K^*$ that significantly correlates with the component label. In particular, for each $k$, $\psi^{*}_{k,1,k}=0.75$ while $\psi^{*}_{k,1,l}=0.25/(d_1-1)$ for all $l \neq k$, so under the likelihood the categorical variable satisfies $P(c=z)=0.75$. For the continuous data, we set $q=5, \mu^{*}_{k}=4K^{*} e_{k}$. $\Sigma^{*}_{k}$ is generated as in scenario 1, except the eigenvalue of the diagonal matrices $D_{k}$ is sampled uniformly to be between $(1.6K^{*})^{2}/2$ and $(1.6K^{*})^{2}$, significantly higher than in scenario 1, then transformed via conjugation with an orthogonal matrix to obtain a non diagonal matrix. The overall idea of this example is to represent a middle ground between scenario 1 and 2: have components that are less well separated in continuous space than scenario 1, in the sense that the standard deviations of the covariance matrices are relatively larger in size compared to the separation between the component centres, but compensated by a categorical variable that depends strongly on the component label $z$, though not as much as in scenario 2. 
\add{
\section{Choosing the number of mixture components $K$}\label{sec: chooseK}

In this section, we conduct a small simulation study on using the ELBO to select the number of mixture components $K$, which is typically unknown in practice. We use scenario 3 as described in \Cref{subsec: mixofbothexample} for this purpose. We fitted $K=1,\dots,10$ to $20$ datasets at each of $n=2500$ and $n=10000$,
using the same priors as described in Section 5 of the main text. The maximum ELBO
plotted in \Cref{fig:chooseK-elboraw} is obtained at the true $K^*=5$ in all 40 datasets. We also tested adding $\ln(K!)$ to the ELBO as recommended in Section 10.2.4 of \citet{Bishop2006} for mixture models, but it made no difference in terms of the optimal K selected, as seen in \Cref{fig:chooseK-elbo_lnKfact}. The normalised ELBO by $n$ in \Cref{fig:chooseK_combined} shows a clear
maximum at $K^*=5$ for both $n=2500$ and $n=10000$. In addition, we note that even when $K$ is overspecified, the variational posterior predictive distribution is still able to enjoy good density estimation compared to the model with the correctly specified $K$, as seen in \Cref{fig:chooseK_combined}. This is consistent with the theory that a small choice of the prior hyperparameter $\alpha$ for the mixture weights will lead to the model emptying the extra components and converge on the number of true components \cite{Rousseau2011}.
}

\begin{figure}[t]
  \centering
  \includegraphics[width=1.00\textwidth]{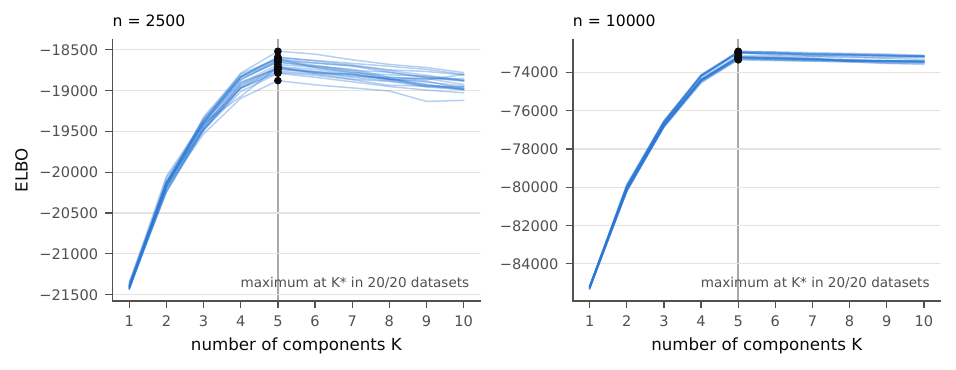}
  \caption{ELBO against the number of
    fitted components $K$ for Scenario 3 ($K^*=5$) for 20 different runs}
  \label{fig:chooseK-elboraw}
\end{figure}

\begin{figure}[t]
  \centering
  \includegraphics[width=1.00\textwidth]{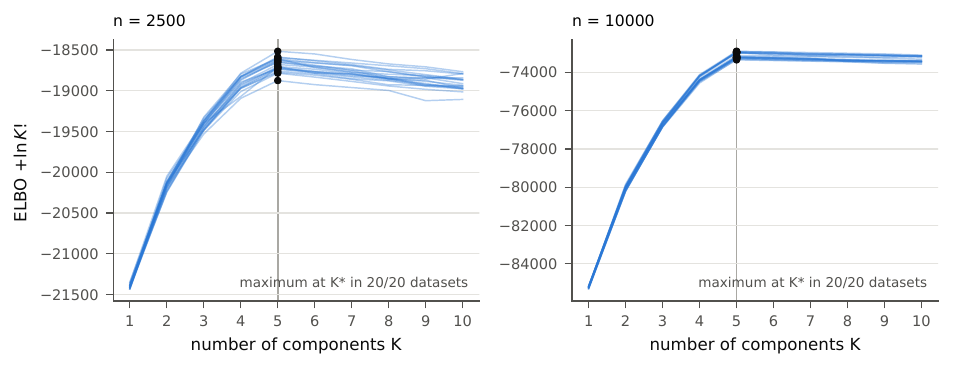}
  \caption{ELBO$+\ln(K!)$ against the number of
    fitted components $K$ for Scenario 3 ($K^*=5$) for 20 different runs}
  \label{fig:chooseK-elbo_lnKfact}
\end{figure}

\begin{figure}[t]
  \centering
  \includegraphics[width=0.48\textwidth]{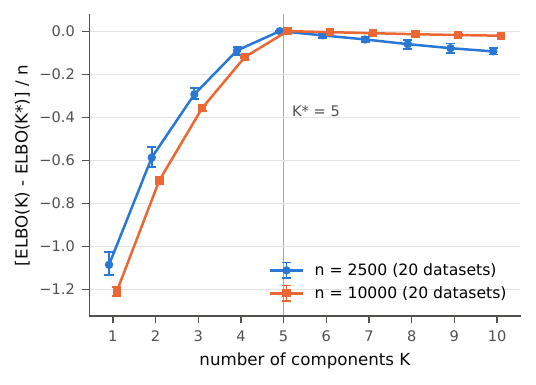}
  \hfill
  \includegraphics[width=0.48\textwidth]{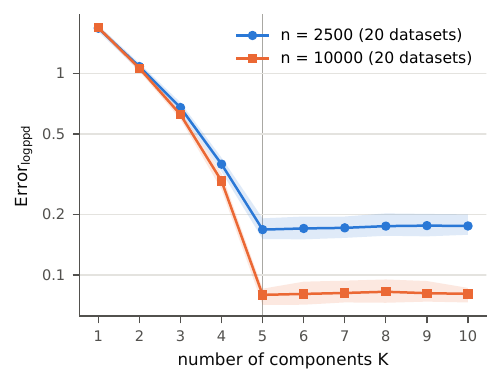}
  \caption{\textbf{Left:} Normalised ELBO, $[\mathrm{ELBO}(K)-\mathrm{ELBO}(K^*)]/n$, against the number of fitted components $K$ for Scenario 3 ($K^*=5$). Points are the mean over 20 runs, error bars are 5\% and 95\% quantiles over datasets. \textbf{Right:} Error between the log posterior predictive density and the true log likelihood, against the number of fitted components}
  \label{fig:chooseK_combined}
\end{figure}

\add{
\section{Sensitivity analysis of prior hyperparameters}\label{sec: sensitivity}
}
\add{
In this section, we conduct a small simulation study to examine the CAVI algorithm's sensitivity to the choice of hyperparameters. Again we use scenario 3 as described in \Cref{subsec: mixofbothexample}. Starting from the default hyperparameters ($m=0$, $\beta=1$, $\Phi=sI$ with $s=0.25$, $\nu=q+K+1=11$,
$\alpha=1/K=0.2$, $\eta_j=1/d_j=0.2$), we changed the following hyperparameters one at a time:
\begin{itemize}
  \item $\beta\in\{0.01,0.1,10,100\}$
  \item $\Phi=sI$ with $s\in\{0.01,1,10,100\}$
  \item $\nu\in\{7,9,25,50,200\}$
  \item $\alpha\in\{0.01,1,10\}$
  \item $\eta_j\in\{0.01,1,10\}$
\end{itemize}
This gives 19 alternative settings plus the baseline for comparison. Each hyperparameter setting was fitted to the same 100
Scenario 3 datasets at $n=2500$ and $n=10000$. We opted to keep $m=0$ because with standardised data it is a natural choice. We compare the point estimation error of the latent parameters as well as frequentist coverage examined in Section 5 of the main text.
}

\begin{figure}[t]
  \centering
  \includegraphics[width=\textwidth]{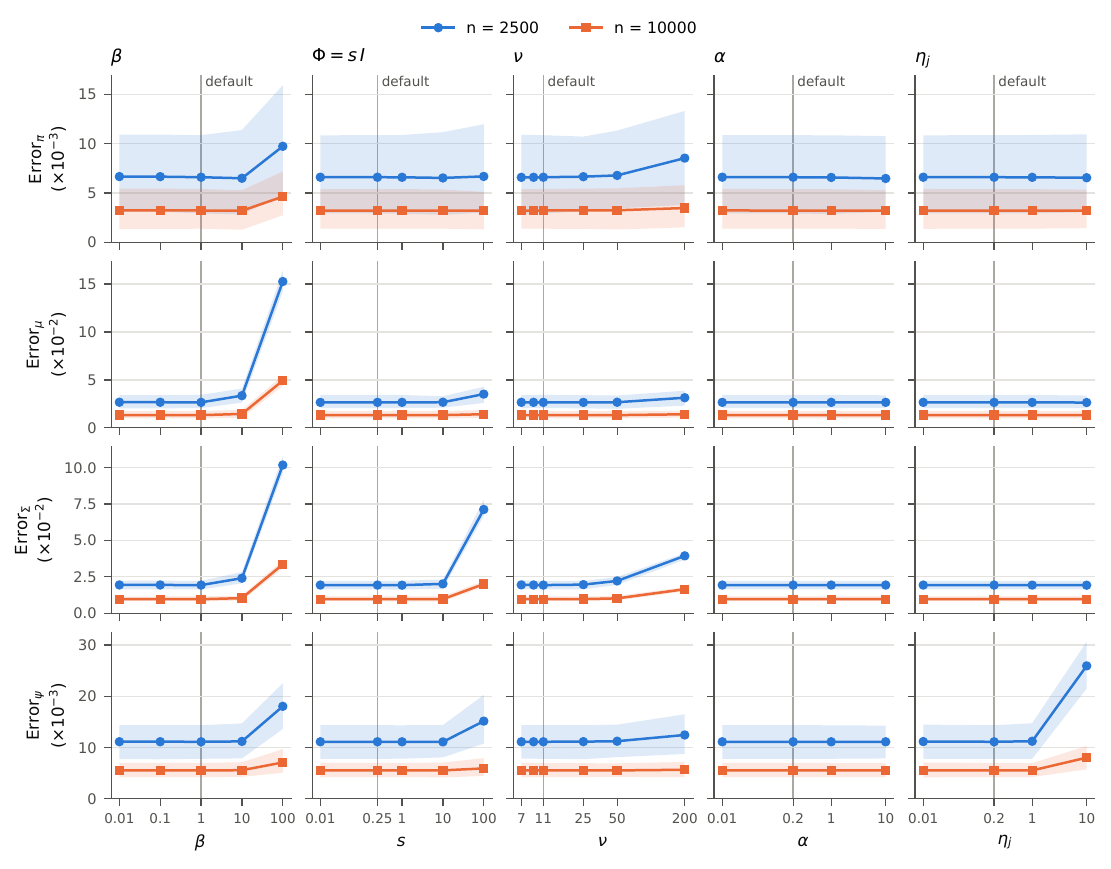}
  \caption{Sensitivity analysis for Scenario 3: $\mathrm{Error}_\pi$,
    $\mathrm{Error}_\mu$, $\mathrm{Error}_\Sigma$ and $\mathrm{Error}_\psi$. Lines are the mean errors over 100 datasets, uncertainty bands are 5\% and
    95\% quantiles over 100 datasets. Each column varies one hyperparameter with the
    others at their defaults (vertical line).}
  \label{fig:sens-err}
\end{figure}

\begin{figure}[t]
  \centering
  \includegraphics[width=\textwidth]{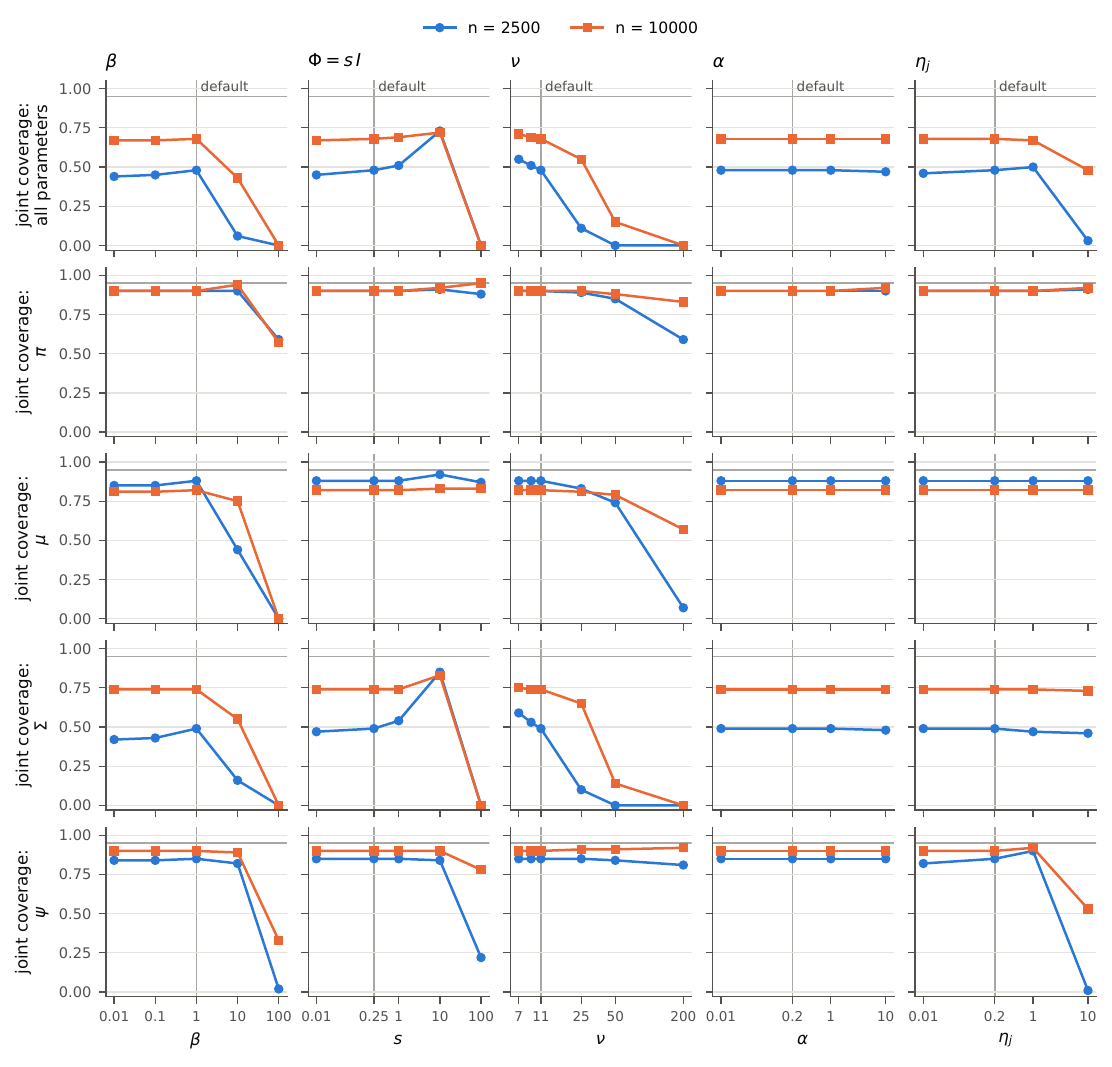}
  \caption{Sensitivity analysis for Scenario 3: frequentist coverage of the
    95\% HDR for all variables and each individual variable, joint over the $K$ components. Each column varies one hyperparameter with the
    others at their defaults (vertical line). }
  \label{fig:sens-cov}
\end{figure}

\begin{figure}[t]
  \centering
  \includegraphics[width=\textwidth]{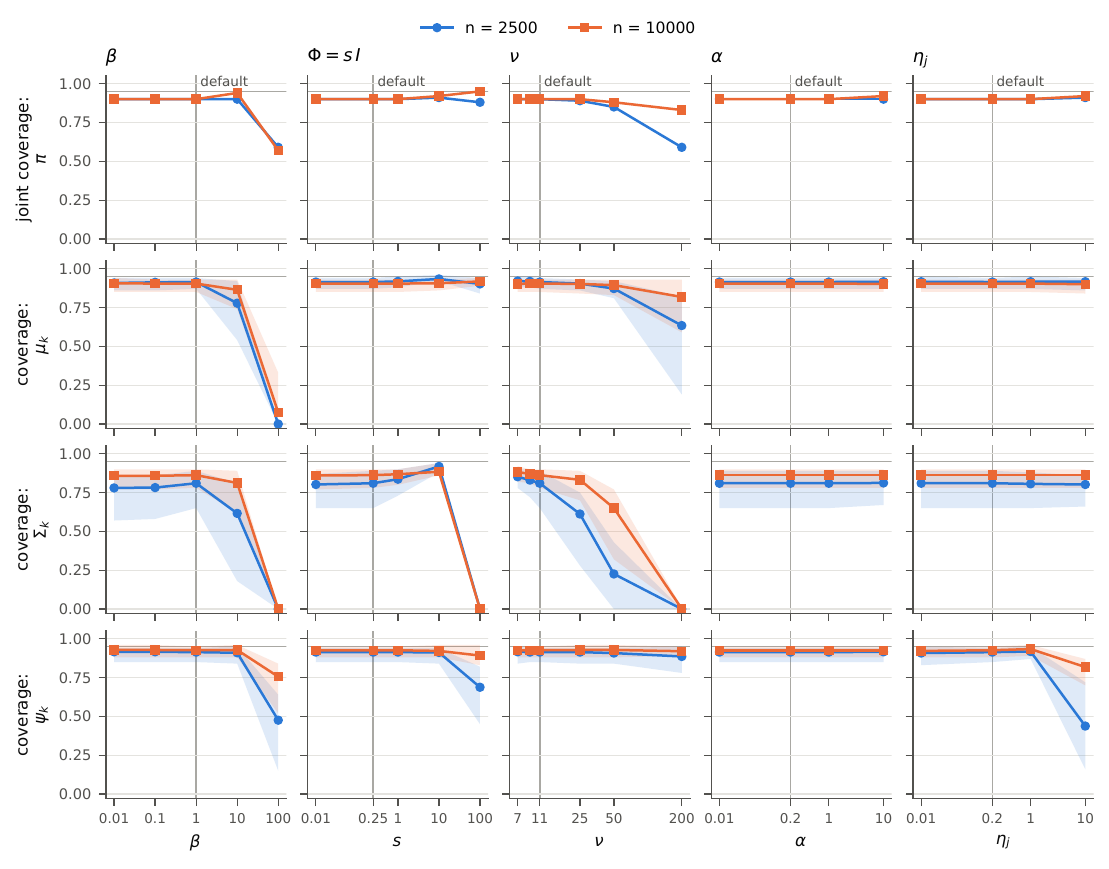}
  \caption{Sensitivity for Scenario 3: frequentist coverage of the
    95\% HDR. Top row: $\pi$, joint coverage, since $q(\pi)$ is a Dirichlet distribution.
    Other rows: marginal coverage per $k$, of $\mu_k$, $\Sigma_k$
    and $\psi_k$, uncertainty bands represent lowest and highest coverage over component $k$. Each column varies one hyperparameter with the
    others at their defaults (vertical line).}
  \label{fig:sens-cov-perk}
\end{figure}
\add{
Overall the point estimation errors are fairly stable across hyperparameter choices, as seen in \Cref{fig:sens-err}. Estimates only deteriorated when the prior 
hyperparameter exceeded the default value by about two orders of magnitude. Intuitively, the hyperparameters will start to have a noticeable effect when they become large enough to compete with the size of the smaller mixture components. For example, $\beta$ which enters the posterior hyperparameters through $\beta+n_k$, where $n_k=\sum_{i=1}^{n} r_{ik}$, the expected number of data points in component $k$, will have a noticeable effect when $\beta$ is competitive in size with $n_k$, and pull the posterior mean towards the prior mean. So with a sample size of $2500$ and $K^*=5$, the smallest mixture component will contain at most $500$ points, and it is not surprising $\beta=100$ will have a noticeable effect. Similarly, $\eta=10$ has a stronger effect on the error than $\alpha=10$ most likely because $\eta$ enters the posterior through $\eta+n_{kg}$, where $n_{kg}=\sum_{i=1}^{n} I_{\{c_{i}=g\}} r_{ik}$ whereas $\alpha$ enters through $\alpha+n_k$. The additional requirement for $c_{i}=g$ partitions the data further into $5$ parts, making $n_{kg}$ smaller compared to $n_k$ and easier for $\eta$ to `compete' with. 

Frequentist coverage follows a similar trend, but is overall more sensitive to hyperparameter changes. Notably, increasing $s$ has the effect of increasing the overall joint coverage and the joint coverage of $\Sigma$ where $s=10$, but decreases back down significantly again when $s=100$ (\Cref{fig:sens-cov}). This suggests the default choice of $s$ is arguably a little too small, but overshooting it too much makes $\Sigma$ be centred too far from the truth and therefore produce poor coverage. Similarly, increasing $\nu$ has the effect of shrinking $\Sigma$, causing the coverage to drop. The default hyperparameter choice is suitable for the coverage per $k$ as can be seen from \Cref{fig:sens-cov-perk}, which is arguably the more natural uncertainty quantification to use for mixture models. 
}
\add{
\section{Number of iterations run by each algorithm}\label{sec: iterations}
}

\begin{table*}[htbp]
\small
\add{
\centering
\begin{tabular}{ll ccc}
\toprule
$n$ & Method & Scenario 1 & Scenario 2 & Scenario 3 \\ 
\midrule
\multirow{7}{*}{2500} 
 & VI                                    & 7.33, [7.00, 8.00] & 24.4, [21.0, 28.0] & 18.4, [16.0, 23.0] \\
 & Gibbs                                 & 1560, [458, 4320] & 1560, [119, 5001] & 3112, [322, 5001] \\
 & Gibbs (iters at VI finish)$^\ddagger$ & 3.51, [3.00, 4.00] & 8.51, [7.00, 11.0] & 9.07, [7.00, 12.0] \\
 & Gibbs (\% reaching threshold)$^\dagger$ & 98\% & 86\% & 59\% \\
 & EM                                    & 11.1, [11.0, 12.0] & 24.0, [20.9, 28.0] & 19.0, [16.0, 25.0] \\
 & EM (1 step)                           & 211, [211, 212] & 224, [221, 228] & 219, [216, 225] \\
 & EM (all steps)                        & 791, [732, 832] & 2448, [2242, 2816] & 2155, [1961, 2380] \\
\midrule

\multirow{7}{*}{5000} 
 & VI                                    & 7.40, [7.00, 8.00] & 24.3, [22.0, 27.0] & 18.9, [17.0, 21.0] \\
 & Gibbs                                 & 1705, [193, 5001] & 2049, [196, 5001] & 2348, [134, 5001] \\
 & Gibbs (iters at VI finish)$^\ddagger$ & 5.00, [4.00, 6.00] & 11.6, [9.95, 14.0] & 12.9, [10.0, 16.0] \\
 & Gibbs (\% reaching threshold)$^\dagger$   & 89\% & 85\% & 76\% \\
 & EM                                    & 11.1, [11.0, 11.0] & 23.9, [21.0, 27.0] & 19.0, [17.0, 21.0] \\
 & EM (1 step)                           & 211, [211, 211] & 224, [221, 227] & 219, [217, 221] \\
 & EM (all steps)                        & 770, [750, 793] & 2235, [2139, 2344] & 2014, [1895, 2179] \\
\midrule

\multirow{7}{*}{10000} 
 & VI                                    & 7.50, [7.00, 8.00] & 24.1, [23.0, 26.0] & 19.6, [18.0, 21.0] \\
 & Gibbs                                 & 1451, [146, 4984] & 1806, [240, 5001] & 2239, [274, 5001] \\
 & Gibbs (iters at VI finish)$^\ddagger$ & 7.05, [6.00, 9.00] & 16.1, [13.0, 19.0] & 18.4, [15.0, 22.0] \\
 & Gibbs (\% reaching threshold)$^\dagger$   & 95\% & 87\% & 88\% \\
 & EM                                    & 11.2, [11.0, 12.0] & 23.6, [22.0, 26.0] & 19.4, [18.0, 21.0] \\
 & EM (1 step)                           & 211, [211, 212] & 224, [222, 226] & 219, [218, 221] \\
 & EM (all steps)                        & 754, [731, 791] & 2183, [2122, 2266] & 1964, [1878, 2060] \\
\midrule

\multirow{7}{*}{20000} 
 & VI                                    & 7.52, [7.00, 8.00] & 24.5, [23.0, 26.0] & 20.6, [19.0, 22.0] \\
 & Gibbs                                 & 1430, [162, 4035] & 1493, [235, 5001] & 1463, [237, 4928] \\
 & Gibbs (iters at VI finish)$^\ddagger$ & 9.41, [7.95, 12.0] & 22.8, [19.0, 28.0] & 25.8, [21.9, 32.0] \\
 & Gibbs (\% reaching threshold)$^\dagger$   & 100\% & 91\% & 95\% \\
 & EM                                    & 11.2, [11.0, 12.0] & 24.0, [23.0, 26.0] & 20.1, [19.0, 22.0] \\
 & EM (1 step)                           & 211, [211, 212] & 224, [223, 226] & 220, [219, 222] \\
 & EM (all steps)                        & 745, [731, 771] & 2155, [2103, 2206] & 1956, [1899, 2040] \\
\bottomrule
\multicolumn{5}{l}{\footnotesize $^\dagger$ $\%$ of runs in which Gibbs reached the error threshold} \\
\multicolumn{5}{l}{\footnotesize $^\ddagger$ Number of Gibbs iterations completed within the VI runtime.}
\end{tabular}
}
\caption{Number of iterations run by each algorithm across scenarios. For Gibbs, iterations up to the required error threshold of within $2\%$ of VI's error, up to a maximum of $5001$; for EM (1 step) and EM (all steps), the iterations of the EM fit plus those of the $B = 200$ bootstrap refits. Also shown is the number of Gibbs iterations completed within the VI runtime.}
\label{table:iterations}

\end{table*}

\add{
\Cref{table:iterations} reports the number of iterations each algorithm required. Since all three algorithms have the same per iteration complexity, $O(nK(q^2+p)+Kq^3)$, differences in runtime are driven mainly by differences in the number of iterations. VI converged in 7--25 iterations on average and the EM (without bootstrapping) in 11--24, whereas Gibbs required 1430--3112 iterations to reach the precision of VI, two orders of magnitude more. Within the time VI took to converge, Gibbs completed only 4--26 iterations. For VI and EM, the number of iterations is roughly constant in $n$. For Gibbs, the number of iterations fluctuates more significantly compared to VI with increasing $n$, though remaining fairly stable in scenarios 1 and 2 with no clear trend. In scenario 3 the number of Gibbs iterations decreases with $n$, but this appears to be driven by the proportion of runs that hit the cap of 5001 iterations, which falls from $41\%$ at $n = 2500$ to $5\%$ at $n = 20000$ (equivalently, the proportion reaching the error threshold rises from $59\%$ to $95\%$). Lastly, the cost of the bootstrap for EM is determined by the number of bootstrap samples, $B = 200$: with a single EM step per refit, EM (1 step) requires about 200 iterations beyond the fit. Running each refit to convergence (EM all steps) requires 745--2448 iterations in total, reaching similar levels to Gibbs.
}